\documentclass[sn-mathphys,Numbered]{sn-jnl}
\usepackage{graphicx}%
\usepackage{multirow}%
\usepackage{amsmath,amssymb,amsfonts}%
\usepackage{amsthm}%
\usepackage{mathrsfs}%
\usepackage[title]{appendix}%
\usepackage{xcolor}%
\usepackage{textcomp}%
\usepackage{manyfoot}%
\usepackage{booktabs}%
\usepackage{algorithm}%
\usepackage{algorithmicx}%
\usepackage{algpseudocode}%
\usepackage{comment}
\usepackage{listings}%
\usepackage{bbold} % for the identity symbol
\usepackage{tikz}

\theoremstyle{plain}%
\theoremstyle{remark}%

\theoremstyle{definition}%
\newtheorem{thm}{Theorem}
\newtheorem{lem}[thm]{Lemma}
\newtheorem{lemma}{Lemma}
\newtheorem{prop}[thm]{Proposition}
\newtheorem{cor}[thm]{Corollary}

\newcommand{\unit}{1\!\!1}
\newcommand{\lV}{\left\Vert}
\newcommand{\rV}{\right\Vert}

\newcommand{\lmk}{\left(}
\newcommand{\rmk}{\right)}
\newcommand{\Tr}{\mathop{\mathrm{Tr}}\nolimits}

\newcommand{\caH}{{\mathcal H}}
\newcommand{\caA}{{\mathcal A}}
\newcommand{\caB}{{\mathcal B}}

\newcommand{\caU}{{\mathcal U}}
\newcommand{\caS}{{\mathcal S}}
\newcommand{\caT}{{\mathcal T}}
\newcommand{\Aut}{\mathop{\mathrm{Aut}}\nolimits}

\newcommand{\Mat}{\mathop{\mathrm{M}}\nolimits}

\newcommand{\bbC}{\mathbb C}

\newcommand{\bbZ}{\mathbb Z}
\newcommand{\bbT}{\mathbb T}
\newcommand{\Ad}{\mathop{\mathrm{Ad}}\nolimits}

\newcommand{\id}{\mathop{\mathrm{id}}\nolimits}

\newcommand{\ketbra}[2]{\left\vert #1\right \rangle \! \left\langle #2\right\vert}

\newcommand{\tr}{\operatorname{Tr}}

\newcommand{\ket}[1]{\vert #1 \rangle}
\newcommand{\bra}[1]{\langle #1 \vert}
\newcommand{\scalprod}[2]{\langle #1 \vert #2 \rangle}

\newcommand{\aut}{\mathrm{Aut}}

\usetikzlibrary{decorations.pathreplacing,decorations.markings,calc,3d,positioning,math,shapes}
\tikzset{
  ->-/.style={postaction={decorate,decoration={
        markings,
        mark=at position .6 with {\arrow[#1]{stealth}}
      }}},
  -<-/.style={postaction={decorate,decoration={
        markings,
        mark=at position .6 with {\arrow[#1]{stealth reversed}}
      }}},
  virtual/.style = {gray},
  irrep/.style={
    anchor=south,
    font = \tiny,
    inner sep=2pt
  },
}

\tikzset{
  tensor/.style={
    inner sep = 0.055cm, shape = circle, draw, fill
  },
  conjugate tensor/.style={
    tensor, fill=white
  },
  t/.style={
    inner sep = 0.03cm, shape = circle, draw, fill
  },
  conjugate t/.style={
    inner sep = 0.03cm, shape = circle, draw, fill=white
  },
  mps/.style={
    tensor, blue
  },
  cmps/.style={
    conjugate tensor, draw=blue,
  },
  mpo/.style={
    tensor
  },
  cmpo/.style={
    conjugate tensor,
  },
  fusion/.style={
    line width=1mm, 
    line cap=round,
    shorten >= -0.2mm,
    shorten <= -0.2mm
  },
  cfusion/.style={
    line width=0.5pt,
    double distance = 0.8mm,
    rounded corners,  
    line cap=round,
    shorten >= -0.2mm,
    shorten <= -0.2mm
  },
  action/.style={
    blue,
    line width=1mm, 
    line cap=round,
    shorten >= -0.2mm,
    shorten <= -0.2mm
  },
  caction/.style={
    draw = blue,
    line width=0.5pt,
    double distance = 0.8mm,
    rounded corners,  
    line cap=round,
    shorten >= -0.2mm,
    shorten <= -0.2mm
  },
  halfc action/.style={
    blue!30!white,
    line width=1mm,
    line cap=round,
    shorten >= -0.2mm,
    shorten <= -0.2mm
  },
}

\tikzset{>=stealth}

\tikzset{
  every picture/.style = {
    baseline=-1mm,
    font=\scriptsize
  },
}

\begin{document}

\title[Article Title]{Classifying symmetric and symmetry-broken spin chain phases with anomalous group actions}

\author[1]{\fnm{Jos\'e} \sur{Garre Rubio}}%\email{iauthor@gmail.com}

\author[1]{\fnm{Andras} \sur{Molnar}}%\email{iiauthor@gmail.com}

\author[2]{\fnm{Yoshiko} \sur{Ogata}}%\email{iiauthor@gmail.com}

\affil[1]{\orgdiv{Faculty of Mathematics}, \orgname{University of Vienna}, \orgaddress{\street{Oskar-Morgenstern-Platz 1}, \city{Vienna}, \postcode{1090}, \country{Austria}}}

\affil[2]{\orgdiv{Research Institute for Mathematical Sciences}, \orgname{Kyoto University}, \orgaddress{\street{Kyoto}, \city{Kyoto}, \postcode{606-8502 }, \country{Japan}}}

\abstract{
We consider the classification problem of quantum spin chains invariant under local decomposable group actions, covering matrix product unitaries (MPUs), using an operator algebraic approach. We focus on finite group symmetries hosting both symmetric and symmetry broken phases. The local-decomposable group actions we consider have a $3$-cocycle class of the symmetry group associated to them. We derive invariants for our classification that naturally cover one-dimensional symmetry protected topological (SPT) phases. We prove that these invariants coincide with the ones of \cite{mposym} using matrix product states (MPSs) techniques, by explicitly working out the GNS representation of MPSs and MPUs, resulting in a useful dictionary between both approaches that could be of independent interest.
}


\maketitle

\tableofcontents

\section{Introduction}\label{sec:intro}

The phase classification problem of quantum spin systems with global symmetries has attracted humongous attention in the last two decades. The main target are geometrically local and gapped Hamiltonians placed on a lattice. The classification approach focuses on finding the invariants on the boundary of a ground state's region after applying the symmetry there \cite{EntSPT,SPTclass1,SPTclass2}. That insight was given by matrix product states (MPSs), the family of tensor network states that approximates efficiently one-dimensional ground states of local and gapped Hamiltonians\cite{Hastings}. Later, the approach was generalized to higher-dimensional tensor network states \cite{2dSPT} and also using algebraic approaches \cite{EN3co}. 

Finally, an operator algebraic approach developed by one of the authors concluded the classification problem for on-site symmetric chains \cite{O1SPT}. The completeness of this work is based on two main points. First, the framework is based on states satisfying the so called split property which is exactly satisfied by ground states of local gapped Hamiltonians \cite{Matsui}. Then, with this formalism the approximation introduced by using MPSs disappears. Second, the operator algebra approach allows to deal properly with the thermodynamic limit which is exactly where quantum phases are defined.

The result for on-site (tensor product form) global symmetries is as follows: for unique ground states, the so-called symmetry protected topological (SPT) phases, the invariants are given by the second and third cohomology group of the symmetry group in one and two spatial dimensions respectively.

However, more general symmetries than on-site actions can be considered as representations of finite groups acting on 1D lattices. These are, nonetheless, automorphisms of the quantum spin systems that can be locally decomposed using unitaries (see definition below). For example, matrix product unitary (MPU) representations of groups fall into this category\cite{MPUdef}. 

In fact, Ref.\ \cite{mposym} considers MPU representations of finite groups, and even representations of fusion categories as global symmetries of ground spaces spanned by MPSs. They derive a phase classification based on  invariants obtained at  the virtual level of the MPSs which restricts to the second cohomology class for the regular setting of SPT phases.

In this work, we generalize those findings by using the operator algebra approach which  allows us to work outside tensor networks and in the thermodynamic limit. Moreover, we prove that their invariants coincide with ours working out explicitly the connection between tensor networks and the operator algebra approach. We note that the symmetry is specified by an extra index, not just the group $G$ as in the SPT setting, associated with the representation of the finite group $G$ that takes values on the third cohomology group of $G$. In our setting, we deal with degenerate ground states where the global symmetry permutes between them. Our findings can be stated informally as follows: given a finite group $G$ and a representation of it characterized by $3$-cocycle class $[\omega]$, the different quantum phases symmetric under the pair $(G,[\omega])$ are labeled by $(H, \sigma)$: first, a subgroup $H \subset G$ that trivializes $\omega$ to a $3$-coboundary and second, an element $\sigma$ of the second cohomology group of $H$. The subgroup $H$ is the unbroken symmetry  group, the ground state degeneracy is $|G/H|$ and $\sigma$ characterizes the SPT phase of the unbroken symmetry.

We note that while finishing this article, we became aware of Ref. \cite{Kanomalous} that also shows the 3-cocycle index of global symmetries and the impossibility of having unique ground state for non-trivial ones.

\section{General theory}
\subsection{Setting and main result}\label{setting}

In this paper we consider one dimensional quantum spin systems.
{For each $z\in\bbZ$,  let $\caA_{\{z\}}$ be an isomorphic copy of $\Mat_{d}$, and for any finite subset $\Lambda\subset\bbZ$, let $\caA_{\Lambda} = \otimes_{z\in\Lambda}\caA_{\{z\}}$, which is the local algebra of observables in $\Lambda$. 
For finite $\Lambda$, the algebra $\caA_{\Lambda} $ can be regarded as the set of all bounded operators acting on
the Hilbert space $\otimes_{z\in\Lambda}{\bbC}^{d}$.
We use this identification freely.
If $\Lambda_1\subset\Lambda_2$, the algebra $\caA_{\Lambda_1}$ is naturally embedded in $\caA_{\Lambda_2}$ by tensoring its elements with the identity. 
The algebra $\caA$
 (resp. $\caA_{R}$, $\caA_L$) representing the infinite chain (resp. half-infinite chain)
is given as the inductive limit of the algebras $\caA_{\Lambda}$(
resp. $\caA_{\Lambda}$ with $\Lambda\subset[0,\infty)$
 and $\Lambda\subset(-\infty -1]$) completed by the operator norm. 
}
Let $G$ be a finite group and $\{\beta_g\}_{g\in G}$ be a set of automorphisms on
$\caA$ such that
\begin{align}
\begin{split}
\beta_g\beta_h=\beta_{gh},\quad g,h\in G.
\end{split}
\end{align}
We assume that each $\beta_g$ is local-decomposable, i.e., it satisfies:
\begin{align}\label{bdecom}
\begin{split}
\beta_g=\Ad(w_g)\lmk\beta_{gL}\otimes\beta_{gR}\rmk,
\end{split}
\end{align}
where $w_g$ is a unitary in $\caA$,
$\beta_{gL}$ an automorphism on the left infinite chain $\caA_L$
and $\beta_{gR}$ an automorphism on the right infinite chain $\caA_R$. Notice that this particular choice is arbitrary and this decomposition can be done for any chosen left and right partition. 
We may set $\beta_{eR}=\id$.
We denote by $\boldsymbol{B}_R$
the set of such $\boldsymbol{\beta_R}:=\{\beta_{gR}\}_{g\in G}$.

Let $X$ be a finite $G$-right set and  consider a set ${\boldsymbol{\omega}}:=\{\omega_x\}_{x\in X}$
of pure split states on 
$\caA$ labeled by $x\in X$.
We assume that the automorphisms permute between the set of pure split states:
\begin{align}\label{xgt}
\begin{split}
\omega_x\beta_g=\omega_{x\cdot g}.
\end{split}
\end{align}
We denote by $\caS$ the set of all such $\boldsymbol{ \omega}$.

We consider a classification of $\caS$.
We denote by $\Aut_{G,0}\caA$ the set of all
automorphisms $\alpha$ on $\caA$ allowing a decomposition
\begin{align}\label{adecom}
    \alpha=\Ad V\circ(\alpha_L\otimes \alpha_R)
\end{align}
and  commuting with $\beta_g$.
Recall that automorphisms given by $G$-invariant interactions
satisfy this condition.
We say $\boldsymbol{\omega}, \boldsymbol{\hat \omega}\in \caS$ are equivalent
and write $\boldsymbol{\omega}\sim_{G,r} \boldsymbol{\hat\omega}$
if there exists an $\alpha\in \Aut_{G,0}\caA$
such that $\hat\omega_x=\omega_x\alpha$ for all $x\in X$.
Note that $\sim_{G,r}$ is rougher than the equivalence relation considered in \cite{OCDM}.

We regard $\bigoplus_{x\in X}U(1)$ an abelian group with the pointwise multiplication, and
associate a
$G$-action
\begin{align}
   \lmk \bigoplus_{x\in X} a_x\rmk^g:=\bigoplus_{x\in X} a_{xg}.
\end{align}
For $\boldsymbol{\sigma}\in C^2\lmk G,\; \bigoplus_{x\in X}U(1)\rmk$,
$\boldsymbol{\zeta}\in C^1\lmk G, \;\bigoplus_{x\in X}U(1)\rmk$
we set
\begin{align}
\lmk d{\boldsymbol{\sigma}}(g,h,k)\rmk_x:=
\frac{\sigma_x(g,h)
\sigma_x(gh,k)}
{\sigma_{xg}(h,k)
\sigma_x(g,hk)},\quad x\in X,\quad g,h,k\in G,
\end{align}
and
\begin{align}
    \lmk d\boldsymbol{\zeta}\rmk_x(g,h):=\frac{\zeta_x(gh)}{\zeta_{xg}(h)\zeta_x(g)}
    \quad x\in X,\quad g,h,k\in G.
\end{align}
Set
\begin{align}
    \caT_0:=\left\{
(\boldsymbol{\sigma},c) \mid 
\boldsymbol{\sigma} \in C^2\lmk G,\; \bigoplus_{x\in X}U(1)\rmk,\;\;
c\in Z^3(G,\; U(1)),\;\; (d \boldsymbol{\sigma})_x=c,\;\text{for all}\; x\in X 
    \right\}.
\end{align}
We introduce an equivalence relation on $\caT_0$
: $(\hat{\boldsymbol{\sigma}},\hat c)
    \sim (\boldsymbol{\sigma}, c)$
    if and only if there exist 
    $ \boldsymbol{\zeta}\in C^1\lmk G, \bigoplus_{x\in X}U(1)\rmk$,
    $z\in C^2\lmk G,  U(1)\rmk$
such that
$\boldsymbol{\hat\sigma}=d\boldsymbol{\zeta}\cdot \boldsymbol{\sigma}\cdot {\bar z} $
and $ \hat c={dz}\cdot c$.
We denote by $\caT$ the set of all equivalence classes
of $\caT_0$ with respect to this equivalence relation.
For $(\boldsymbol{\sigma}, c)\in \caT_0$, we denote by 
$[(\boldsymbol{\sigma}, c)]_{\caT}$ the equivalence class of  
$(\boldsymbol{\sigma}, c)$.
Here is our main theorem in this general setting.
\begin{thm}\label{mainthm}
There exists a $\caT$-valued invariant $I(\boldsymbol{\omega})$ on $\boldsymbol{\omega}\in \caS$
of the classification $\sim_G$.
\end{thm}
For the rest of this section we prove this theorem.

\subsection{The $3$rd group cohomology index from local decomposable finite group actions}
\label{betas}
In this section we derive a  $3$rd group cohomology index for our $\beta$.
Let $\boldsymbol{\beta_R}:=\{\beta_{gR}\}_{g\in G}\in \boldsymbol{B}_R$.
We have
\begin{align}
\begin{split}
&\Ad(w_{gh})\lmk\beta_{ghL}\otimes\beta_{ghR}\rmk=
\beta_{gh}=\beta_g\beta_h
=\Ad(w_g)\lmk\beta_{gL}\otimes\beta_{gR}\rmk
\Ad(w_h)\lmk\beta_{hL}\otimes\beta_{hR}\rmk\\
&=\Ad\lmk w_g \lmk\beta_{gL}\otimes\beta_{gR}\rmk(w_h)\rmk
\lmk\beta_{gL}\beta_{hL}\otimes\beta_{gR}\beta_{hR}\rmk.
\end{split}
\end{align}
From this, we see
\begin{align}
\begin{split}
&\lmk\beta_{gL}\beta_{hL}\beta_{gh,L}^{-1}
\otimes\beta_{gR}\beta_{hR}\beta_{ghR}^{-1}\rmk\\
&=\Ad\lmk
\lmk w_g \lmk\beta_{gL}\otimes\beta_{gR}\rmk(w_h)\rmk^{-1}w_{gh}
\rmk.
\end{split}
\end{align}
By Lemma B.1 of \cite{2dSPT_O}, we see that there exists a unitary $v(g,h)\in \caU(\caA_R)$ such that
\begin{align}\label{vdef}
\begin{split}
\Ad(v(g,h))= \beta_{gR}\beta_{hR}\beta_{ghR}^{-1}.
\end{split}
\end{align}
We denote by $\boldsymbol{V}(\boldsymbol{\beta_R})$ 
the set of all such $\boldsymbol{v}:=\{v(g,h))\}_{g,h\in G}$.
For each $\boldsymbol{v}:=\{v(g,h))\}_{g,h\in G}\in \boldsymbol{V}(\boldsymbol{\beta_R})$ we have
\begin{align}
\begin{split}
\beta_{gR}\beta_{hR}\beta_{kR} =\Ad(v(g,h))\beta_{ghR}\beta_{kR}= \Ad(v(g,h)v(gh,k))\beta_{ghkR}
\end{split}
\end{align}
and
\begin{align}
\begin{split}
\beta_{gR}\beta_{hR}\beta_{kR}
=\beta_{gR}\Ad(v(h,k))\beta_{hkR}
=\Ad\lmk \beta_{gR}\lmk v(h,k) \rmk\rmk\beta_{gR}\beta_{hkR}\\
=\Ad\lmk \beta_{gR}\lmk v(h,k) \rmk v(g,hk)\rmk\beta_{ghkR}.
\end{split}
\end{align}
Comparing them, with $\caA_R\cap\caA_R'=\bbC\unit$,
we conclude that there exists $c(g,h,k)\in U(1)$ such that
\begin{align}\label{sc}
\begin{split}
\beta_{gR}\lmk v(h,k) \rmk v(g,hk) = c(g,h,k)v(g,h)v(gh,k)
\end{split}
\end{align}
Let us show now that $c$ satisfies the 3-cocyle condition
\begin{equation}\label{eq:3-cocycle}
    c(g,h,k) c(g,hk,l)c(h,k,l) = c(g,h,kl) c(gh,k,l).
\end{equation}
For that, let us consider the operator $v(g,h)v(gh,k)v(ghk,l)$. Using \eqref{sc} repeatedly, we can write
\begin{equation*}
    \begin{aligned}
        &v(g,h)v(gh,k)v(ghk,l) = \frac{1}{c(g,h,k)} \beta_{gR}\lmk v(h,k) \rmk v(g,hk) v(ghk,l) =\\
        &\frac{1}{c(g,h,k)c(g,hk,l)} \beta_{gR} \lmk v(h,k) v(hk,l) \rmk v(g,hkl) = \\
        & \frac{1}{c(g,h,k)c(g,hk,l)c(h,k,l)} \beta_{gR} \lmk \beta_{hR} \lmk v(k,l)\rmk \rmk \beta_{gR}\lmk v(h,kl) \rmk v(g,hkl) = \\
        & \frac{c(g,h,kl)}{c(g,h,k)c(g,hk,l)c(h,k,l)} \beta_{gR} \lmk \beta_{hR} \lmk v(k,l)\rmk \rmk v(g,h)v(gh,kl).
    \end{aligned}
\end{equation*}
Finally, considering that \eqref{vdef} implies
\begin{equation*}
    \beta_{gR} \lmk \beta_{hR} \lmk v(k,l)\rmk \rmk v(g,h) = v(g,h)\beta_{ghR} \lmk v(k,l) \rmk,
\end{equation*}
we can further write
\begin{equation*}
    \begin{aligned}
        v(g,h)v(gh,k)v(ghk,l)  &=  \frac{c(g,h,kl)}{c(g,h,k)c(g,hk,l)c(h,k,l)} v(g,h) \beta_{ghR} \lmk v(k,l)\rmk v(gh,kl) \\
        & = \frac{c(g,h,kl)c(gh,k,l)}{c(g,h,k)c(g,hk,l)c(h,k,l)} v(g,h) v(gh,k) v(ghk,l),
    \end{aligned}
\end{equation*}
and thus, as $v(g,h)v(gh,k)v(ghk,l)\neq 0$, \eqref{eq:3-cocycle} holds.
We denote this $c\in Z^3(G,U(1))$ by $c\lmk\boldsymbol{\beta_R}, \boldsymbol{v}\rmk$.

Let us show  that the equivalence class of the 3-cocycle does not depend on the choice of 
$\boldsymbol{\beta}_{R}\in \boldsymbol{B}_R$, $\boldsymbol{v}\in \boldsymbol{V}(\boldsymbol{\beta}_R)$.
Let $\boldsymbol{\hat \beta}_{R}\in \boldsymbol{B}_R$, $\boldsymbol{\hat v}:=\{\hat v(g,h)\}\in \boldsymbol{V}(\boldsymbol{\hat \beta}_R)$
be a different choice.
Comparing the two expressions ((\ref{bdecom}) and corresponding decomposition for
$\hat\beta_{gR}$) for $\beta_g$, we obtain
\begin{equation*}
    \Ad(\hat w_g^{-1} w_g) = \hat\beta_{gL}\beta_{gL}^{-1}\otimes\hat\beta_{gR}\beta_{gR}^{-1},
\end{equation*}
with some unitaries $w_g,\hat w_g\in \caA$
and thus, by  Lemma B.1 of \cite{2dSPT_O}, there exists a unitary $a_g\in\mathcal{U}(\mathcal{A}_R)$ such that
\begin{equation*}
    \Ad\lmk a_g \rmk = \hat \beta_{gR} \beta_{gR}^{-1}.
\end{equation*}
We can thus express $\Ad \lmk \hat v(g,h) \rmk $ as
\begin{equation*}
\begin{aligned}
    \Ad \lmk \hat v(g,h) \rmk &= \hat\beta_{gR} \Ad(a_h) \beta_{hR} \hat \beta_{ghR}^{-1} = \Ad \lmk \hat\beta_{gR} (a_h) \rmk \hat\beta_{gR}  \beta_{hR} \hat \beta_{ghR}^{-1} = \\
    &= \Ad \left( \hat \beta_{gR} (a_h) a_g\right) \beta_{gR} \beta_{hR} \beta_{ghR}^{-1} \Ad \lmk a_{gh}^{-1} \rmk \\
    & = \Ad \lmk  \hat\beta_{gR} (a_h) a_g v(g,h) a_{gh}^{-1}\rmk.
\end{aligned}
\end{equation*}
This means that there is $z(g,h)\in U(1)$ such that
\begin{equation*}
    \hat v(g,h)  = z(g,h)  \hat \beta_{gR} (a_h) a_g v(g,h) a_{gh}^{-1} = z(g,h) a_g \beta_{gR} (a_h)  v(g,h) a_{gh}^{-1} .
\end{equation*}
We also obtain
\begin{equation*}
\begin{aligned}
    \hat \beta_{gR}\lmk \hat v(h,k) \rmk \hat v(g,hk) &= z(h,k)z(g,hk) \hat \beta_{gR} \left( \hat \beta_{hR}  (a_{k}) a_h v(h,k) \right)  a_g v(g,hk) a_{ghk}^{-1} \\
    & = z(h,k)z(g,hk) \hat \beta_{gR}(a_h) \hat\beta_{gR}  \left( \beta_{hR}  (a_{k}) v(h,k) \right)  a_g v(g,hk) a_{ghk}^{-1} \\
    & = z(h,k)z(g,hk) \hat \beta_{gR}(a_h) a_g \beta_{gR}  \left( \beta_{hR}  (a_{k}) v(h,k) \right)  v(g,hk) a_{ghk}^{-1} \\
    & = z(h,k)z(g,hk) \hat \beta_{gR}(a_h) a_g \beta_{gR}  \left( \beta_{hR} (a_{k}) \right) \beta_{gR}\left( v(h,k) \right)  v(g,hk) a_{ghk}^{-1} \\
    & = z(h,k)z(g,hk) c(g,h,k) \hat \beta_{gR}(a_h) a_g \beta_{gR}  \left( \beta_{hR} (a_{k}) \right) v(g,h) v(gh,k) a_{ghk}^{-1}\\
    & = z(h,k)z(g,hk) c(g,h,k) \hat \beta_{gR}(a_h) a_g v(g,h)  \beta_{ghR} (a_{k})  v(gh,k) a_{ghk}^{-1}\\
    & = \frac{z(h,k)z(g,hk)}{z(gh,k)} c(g,h,k) \hat \beta_{gR}(a_h) a_g v(g,h) a_{gh}^{-1}  \hat v(gh,k) \\
    & = \frac{z(h,k)z(g,hk)}{z(g,h) z(gh,k)} c(g,h,k) \hat v(g,h)  \hat v(gh,k).
\end{aligned}
\end{equation*}
Therefore
\begin{equation*}
    \hat c(g,h,k) = \frac{z(h,k)z(g,hk)}{z(g,h) z(gh,k)} c(g,h,k).
\end{equation*}
Hence $\left[c\lmk\boldsymbol{\beta_R}, \boldsymbol{v}\rmk\right]_{H^3(G,U(1))}
=\left[c\lmk\boldsymbol{\hat \beta_R}, \boldsymbol{\hat v}\rmk\right]_{H^3(G,U(1))}$,
i.e.\ the equivalence class $[c]_{H^3(G,U(1))}$ of the 3-cocycle $c$ does not depend on the choice of $w_g$, $\beta_{gL}$, $\beta_{gR}$.

\subsection{Index of the symmetry acting on the split states}\label{sxgh}
In this section we derive an index $I(\boldsymbol{\omega})$
for each $\boldsymbol{\omega}\in \caS$.
By the split property, for each ${\boldsymbol{\omega}}:=\{\omega_x\}_{x\in X}\in \caS$,
each $\omega_x$ satisfies 
\begin{align}\label{odec}
    \omega_x\simeq_{u.e.}\omega_{xL}\otimes \omega_{xR}.
\end{align}
Here, $\omega_{xL}$ and $\omega_{xR}$ are pure states on
$\caA_L$, $\caA_R$ respectively and 
$\simeq_{u.e.}$ denotes (unitary) equivalence.
For each ${\boldsymbol{\omega}}\in \caS$, we denote by 
$\mathrm O_R(\boldsymbol{\omega})$ the set of all such ${\boldsymbol{ \omega}}_{R}:=\{\omega_{xR}\}_{x\in X}$.

Let $\boldsymbol{\omega}_R=\{\omega_{xR}\}_{x\in X}\in \mathrm O_R(\boldsymbol{\omega})$.
Consider 
$\lmk \boldsymbol{\caH}, \boldsymbol{\pi}\rmk=
(\caH_x,\pi_x)_{x\in X}$,
where $(\caH_x,\pi_x)$ is a GNS representation of 
$\omega_{xR}$ for each $x\in X$.
We denote by $\boldsymbol{GNS}(\boldsymbol{\omega_R})$
the set of such tuples of GNS representations. 
By definition, we obtain
\begin{align}
\begin{split}
\omega_{xL}\beta_{gL}\otimes\omega_{xR}\beta_{gR}
\simeq_{u.e.}
\omega_x\beta_g=\omega_{x\cdot g}\simeq_{u.e.}
\omega_{x\cdot g,L}\otimes\omega_{x\cdot g,R}
\end{split}
\end{align}
Restricting this to $\caA_R$,
we obtain
\begin{align}
\begin{split}
\omega_{x,R}\beta_{gR}
\simeq_{u.e.}
\omega_{x\cdot g,R}.
\end{split}
\end{align}
Hence for each $\lmk \boldsymbol{\caH}, \boldsymbol{\pi}\rmk=
(\caH_x,\pi_x)_{x\in X}\in \boldsymbol{GNS}(\boldsymbol{\omega_R})$, there is a unitary
\begin{align}
u_{x, g} : \caH_{xg}\to \caH_{x}
\end{align}
such that
\begin{align}\label{udef}
\Ad\lmk u_{x,g}\rmk\pi_{xg}=\pi_x\beta_{gR}.
\end{align}
We denote by ${\boldsymbol U}(\boldsymbol{\omega_R}, \boldsymbol{\beta_R},  \boldsymbol{\caH}, \boldsymbol{\pi})$
the set of such 
$\boldsymbol{u}:=\{ u_{x,g}\}_{x\in X, g\in G}$.
For each  $\boldsymbol{v}\in \boldsymbol{V}(\boldsymbol{\beta}_R)$, we have
\begin{align}
\begin{split}
\Ad\lmk u_{x,g}u_{xg, h} u_{x,gh}^{-1}\rmk\pi_x
=\pi_x\beta_{gR}\beta_{hR}\beta_{ghR}^{-1}
=\Ad\lmk \pi_x(v(g,h))\rmk\pi_x.
\end{split}
\end{align}
Because $\omega_{x, R}$ is irreducible, this means there exists $\sigma_x(g,h)\in U(1)$
such that
\begin{align}\label{eq:mps_index}
\begin{split}
u_{x,g}u_{xg, h} u_{x,gh}^{-1}
=\sigma_x(g,h)\pi_x(v(g,h)).
\end{split}
\end{align}
We denote by 
$\boldsymbol{\sigma}\lmk \boldsymbol{u}, \boldsymbol{v},
\boldsymbol{\omega_R}, \boldsymbol{\beta_R},  \boldsymbol{\caH}, \boldsymbol{\pi}\rmk$
this
$\boldsymbol{\sigma}:=\{\sigma_x(g,h)\}_{g,h\in G, x\in X}$.
We then have
\begin{align}
\begin{split}
&u_{x,g}u_{xg, h}u_{xgh, k} u_{x,ghk}^{-1}
=
u_{x,g}u_{xg, h} u_{x,gh}^{-1}
u_{x,gh}u_{xgh, k} u_{x,ghk}^{-1}\\
&=\sigma_x(g,h)\pi_x(v(g,h))
\sigma_x(gh,k)\pi_x(v(gh, k)),
\end{split}
\end{align}
and
\begin{align}
\begin{split}
&u_{x,g}u_{xg, h}u_{xgh, k} u_{x,ghk}^{-1}
=
u_{x,g}
u_{xg,h}u_{xgh, k} u_{xg,hk}^{-1}u_{xg,hk}
u_{x,ghk}^{-1}\\
&=u_{x,g}
\sigma_{xg}(h,k)\pi_{xg}(v(h,k))
u_{x,g}^{-1}
u_{x,g}
u_{xg,hk}
u_{x,ghk}^{-1}\\
&=
\sigma_{xg}(h,k)\pi_{x}\beta_{gR}(v(h,k))
\sigma_x(g,hk)\pi_x(v(g,hk))
\end{split}
\end{align}
Comparing with (\ref{sc}), we obtain
\begin{align}\label{eq:sigmaconstraint}
\begin{split}
\sigma_x(g,h)
\sigma_x(gh,k)
=c(g,h,k)\sigma_{xg}(h,k)
\sigma_x(g,hk),
\end{split}
\end{align}
for $c\lmk\boldsymbol{\beta_R}, \boldsymbol{v}\rmk=\{c(g,h,k)\}$.
We note that this equation corresponds to the pentagon equation of a module category over the fusion category ${\rm Vec}_G^c$ \cite{EO}.
This equation has several consequences. For trivial $3$-cocycle, $c=1$, it corresponds to the classification of symmetry broken phases for on-site global symmetries. In that case, if,  moreover, $X$ consists of a single point, the equation corresponds to the $2$-cocycle equation for the invariants of SPT phases \cite{OCDM}. It also shows that if $X$ consists of single point, $c$ has trivial cohomology. This means that for an automorphism with non-trivial $3$-cocycle, $c \neq 1$, $X$ cannot be a single point, i.e. there is no unique ground state phase. This result has been derived before in the MPS setting \cite{2dSPT}. Let us define $H\subseteq G$ as $H=\{ h\in H, x\cdot h= x, \forall x\in X \}$, the so-called unbroken symmetry group. Then, Eq.\eqref{eq:sigmaconstraint} restricted to the elements of $H$ reads:
\begin{align}\label{eq:sigmaconstraintH}
\begin{split}
\sigma_x(h_1,h_2)
\sigma_x(h_1h_2,h_3)
=c(h_1,h_2,h_3)\sigma_{x}(h_2,h_3)
\sigma_x(h_1,h_2h_3),
\end{split}
\end{align}
which implies that $c$ restricted to $H$ is a trivial 3-cocycle and then, as argued in \cite{mposym}, $\sigma_{x}$ restricted to $H$ is a 2-cocycle of $H$.

Finally Eq.~\eqref{eq:sigmaconstraint} also leads to a well-defined index of $\omega$.
\begin{thm}
Let $\boldsymbol{\omega}\in \caS$.
Then 
\begin{align}
    I(\omega):=\left[ 
    \boldsymbol{\sigma}\lmk \boldsymbol{u}, \boldsymbol{v},
\boldsymbol{\omega_R}, \boldsymbol{\beta_R},  \boldsymbol{\caH}, \boldsymbol{\pi}\rmk, c\lmk\boldsymbol{\beta_R}, \boldsymbol{v}\rmk\right]_{\caT}
    \in \caT
\end{align} 
is independent of the choice of $\boldsymbol{\beta}_R\in \boldsymbol{B}_R$,
$\boldsymbol{v}\in \boldsymbol{V}(\boldsymbol{\beta_R})$,
$\boldsymbol{\omega}_R\in \boldsymbol{O}_R(\boldsymbol{\omega})$,
$(\boldsymbol{\caH}, \boldsymbol{\pi})\in 
\boldsymbol{GNS}(\boldsymbol{\omega_R})$,
$\boldsymbol{u}\in{\boldsymbol U}(\boldsymbol{\omega_R}, \boldsymbol{\beta_R},  \boldsymbol{\caH}, \boldsymbol{\pi})$.
\end{thm}

\begin{proof}
    Let
    $\boldsymbol{\hat \beta}_R\in \boldsymbol{B}_R$,
$\boldsymbol{\hat v}\in \boldsymbol{V}(\boldsymbol{\hat\beta_R})$,
$\boldsymbol{\hat \omega}_R\in \boldsymbol{O}_R(\boldsymbol{\omega})$,
$(\boldsymbol{\hat \caH}, \boldsymbol{\hat \pi})\in 
\boldsymbol{GNS}(\boldsymbol{\hat \omega_R})$,
$\boldsymbol{\hat u}\in{\boldsymbol U}(\boldsymbol{\hat\omega_R}, 
\boldsymbol{\hat \beta_R},  \boldsymbol{\hat \caH}, \boldsymbol{\hat \pi})$
be another choice,
and let 
$\boldsymbol{\hat c}:=\boldsymbol{c\lmk\boldsymbol{\hat \beta_R},\boldsymbol{\hat v}\rmk}$
and $\boldsymbol{\hat \sigma}:=
\boldsymbol{\sigma}\lmk \boldsymbol{\hat u}, \boldsymbol{\hat v},
\boldsymbol{\hat \omega_R}, \boldsymbol{\hat \beta_R},  \boldsymbol{\hat \caH}, \boldsymbol{\hat \pi}\rmk$.

By the subsection \ref{betas}, there are unitaries $b_g\in \caA_R$
and $z\in C^2(G, U(1))$
such that $\hat\beta_{g,R}=\Ad b_g \beta_{g,R}$
 and 
 \begin{align}
    \hat v(g,h)=
    z(g,h)b_g \beta_{gR}(b_h) v(g,h) b_{gh}^*\in \boldsymbol{V}(\boldsymbol{\hat\beta_{R}}),
 \end{align}
 and $\hat c=c\cdot dz$.
Because $(\hat\caH_x,\hat \pi_{x}), (\caH_x, \pi_{x})$ are GNS-representations of $\omega_{x,R}$, $\hat\omega_{x,R}$ for each $x\in X$,
and $\hat\omega_{x,R}\simeq_{u.e.}\omega_{x,R}$
there exists a unitary $W_x: \caH_x\to \hat\caH_x$ 
such that $\hat\pi_x=\Ad(W_x)\circ\pi_x$.
From 
$\Ad\hat u_{x, g}\hat\pi_{xg}=\hat\pi_x\hat\beta_{g,R}$,
 we have 
\begin{align}
    \Ad\hat u_{x, g}\hat\pi_{xg}=\hat\pi_x\hat\beta_{g,R}
    =\Ad(W_x\pi_x(b_g) u_{x,g} W_{xg}^*)\hat \pi_{xg}.
\end{align}
Because $\hat \pi_{xg}$ is irreducible, there exists a
$\zeta_x(g)\in U(1)$ such that
\begin{align}
    \hat u_{x, g}
    =\zeta_x(g)W_x\pi_x(b_g) u_{x,g} W_{xg}^*.
\end{align}
Substituting this to the definition of $\boldsymbol{\hat\sigma}$, we have
\begin{align}\begin{split}
&\hat \sigma_x(g,h)\hat\pi_x(\hat v(g,h))\\
    &=\hat u_{x,g} \hat u_{xg, h}\hat u_{x,gh}^{-1}\\
&    =\zeta_x(g)\zeta_{xg}(h)\overline{\zeta_x(gh)}
    W_x\pi_x(b_g) \Ad u_{x,g} \lmk \pi_{xg}(b_h)\rmk u_{x,g}
    u_{xg,h} u_{x,gh}^*\pi_x(b_{gh}^*)W_x^*\\
    &=\zeta_x(g)\zeta_{xg}(h)\overline{\zeta_x(gh)}\sigma_x(g,h)
    W_x\pi_x\lmk b_g \beta_{gR}(b_h) v(g,h) b_{gh}^*\rmk W_x^*\\
&=\zeta_x(g)\zeta_{xg}(h)\overline{\zeta_x(gh)}\sigma_x(g,h)
\overline{z(g,h)}\hat\pi_x(\hat v(g,h)),
\end{split}
\end{align}
and $\boldsymbol{\hat \sigma}=d\boldsymbol{\zeta}\bar z \boldsymbol{\sigma}$.
Hence we obtain 
$(\sigma,c)\sim (\hat\sigma, \hat c)$.

\end{proof}

\subsection{Proof of Theorem \ref{mainthm}}

We give the proof of Theorem \ref{mainthm}.
It suffices to show that $I(\boldsymbol{\omega})$ obtained above
is an invariant.
\begin{proof}
Let $\alpha\in \Aut_{G,0}\caA$ with decomposition (\ref{adecom}).
Set $\boldsymbol{\hat\omega}:=\{ \hat\omega_x:=\omega_x\alpha\}$.
It suffices to show that $I(\boldsymbol{\hat\omega})=I(\boldsymbol{\omega})$.
Let  $\boldsymbol{\beta}_R\in \boldsymbol{B}_R$,
$\boldsymbol{v}\in \boldsymbol{V}(\boldsymbol{\beta_R})$,
$\boldsymbol{\omega}_R\in \boldsymbol{O}_R(\boldsymbol{\omega})$,
$(\boldsymbol{\caH}, \boldsymbol{\pi})\in 
\boldsymbol{GNS}(\boldsymbol{\omega_R})$,
$\boldsymbol{u}\in{\boldsymbol U}(\boldsymbol{\omega_R}, \boldsymbol{\beta_R},  \boldsymbol{\caH}, \boldsymbol{\pi})$.
Let $\boldsymbol{\sigma}:=\boldsymbol{\sigma}\lmk \boldsymbol{u}, \boldsymbol{v},
\boldsymbol{\omega_R}, \boldsymbol{\beta_R},  \boldsymbol{\caH}, \boldsymbol{\pi}\rmk$.

With a decomposition (\ref{bdecom}) for this $\boldsymbol{\beta_R}$, we have
\begin{align}
    \Ad(\beta_g(V)w_g)\circ(\beta_{gL}\alpha_L\otimes \beta_{gR}\alpha_R)
    =\beta_g\alpha=\alpha\beta_g=
    \Ad(\alpha(w_g) V) \lmk \alpha_L\beta_{gL}\otimes 
    \alpha_R\beta_{gR}\rmk.
\end{align}
From this, there exists a unitary $a_g\in \caA_R$
such that
\begin{align}
    \Ad(a_g)\beta_{gR}\alpha_R=\alpha_R\beta_{gR}.
\end{align}
Therefore, we have
\begin{align}
    \boldsymbol{\hat\beta_R}:=\{\hat \beta_{gR}:=\alpha_R^{-1}\beta_{gR}\alpha_R\}\in \boldsymbol{B}_R.
\end{align}
By 
\begin{align}
    \hat\beta_{gR}\hat\beta_{hR}\hat\beta_{ghR}^{-1}
    =\alpha_R^{-1}\beta_{gR}\beta_{hR}\beta_{ghR}^{-1}\alpha_R
    =\Ad \lmk \alpha_R^{-1}\lmk v(g,h)\rmk\rmk,
\end{align}
we have
\begin{align}
\begin{split}
    &\boldsymbol{\hat v}:=\{\hat v(g,h):=  \alpha_R^{-1}\lmk v(g,h) \rmk\}\in \boldsymbol{V}(\hat\beta_R),\\
    &c\lmk \boldsymbol{\hat \beta_R}, \boldsymbol{\hat v}\rmk=c\lmk \boldsymbol{ \beta_R}, \boldsymbol{ v}\rmk.
 \end{split}   
\end{align}
With the decomposition (\ref{odec}) for our $\boldsymbol{\omega_R}$,
we also have
\begin{align}\label{fox}
\hat\omega_x:=\omega_x\alpha\simeq\omega_{xL}\alpha_L\otimes \omega_{xR}\alpha_R.
\end{align}
Set
\begin{align}
\begin{split}
&   \boldsymbol{\hat \omega}_{R}:=\{\hat\omega_{x,R}:=\omega_{x,R}\alpha_R\},\quad
    \boldsymbol{\hat\pi}:=\{\hat\pi_x:=\pi_x\alpha_R\}.\
\end{split}
\end{align}
Clearly $\boldsymbol{\hat \omega_R}\in \boldsymbol{O}_R(\boldsymbol{\hat\omega})$,
$(\boldsymbol{\caH}, \boldsymbol{\hat\pi})\in \boldsymbol{GNS}(\boldsymbol{\hat\omega_R})$.
From 
\begin{align}
    \Ad  u_{x,g}\hat\pi_{xg}
    =\Ad u_{x,g} \pi_{xg}\alpha_R
    =\pi_{x}\beta_{gR}\alpha_R
    =\pi_x\alpha_R\hat \beta_{gR}
    =\hat\pi_{x}\hat \beta_{gR},
\end{align}
we have
$\boldsymbol{u}\in{\boldsymbol U}(\boldsymbol{\hat \omega_R}, \boldsymbol{\hat \beta_R}, 
\boldsymbol{\caH}, \boldsymbol{\hat \pi})$.
Therefore, for $\boldsymbol{\hat\sigma}:=\boldsymbol{\sigma}\lmk \boldsymbol{u}, \boldsymbol{\hat v},
\boldsymbol{\hat \omega_R}, \boldsymbol{\hat \beta_R},  \boldsymbol{\caH}, \boldsymbol{\hat \pi}\rmk$
we have
\begin{align}
\begin{split}
 \sigma_x(g,h)\pi_x( v(g,h))
=
u_{x,g} u_{xg, h}  u_{x,gh}^{-1}
=\hat \sigma_x(g,h)\hat \pi_x(\hat v(g,h))
=\hat \sigma_x(g,h)\pi_x( v(g,h)).
\end{split}
\end{align}
Therefore, we get $\boldsymbol{\hat\sigma}=\boldsymbol{\sigma}$,
and obtain
\begin{align}
\begin{split}
    &I(\boldsymbol{\hat\omega})
    =\left[ 
    \boldsymbol{\sigma}\lmk \boldsymbol{u}, \boldsymbol{\hat v},
\boldsymbol{\hat \omega_R}, \boldsymbol{\hat \beta_R},  \boldsymbol{\caH}, \boldsymbol{\hat \pi}\rmk, c\lmk\boldsymbol{\hat \beta_R}, \boldsymbol{\hat v}\rmk\right]_{\caT}\\
&=\left[ 
    \boldsymbol{\sigma}\lmk \boldsymbol{u}, \boldsymbol{v},
\boldsymbol{\omega_R}, \boldsymbol{\beta_R},  \boldsymbol{\caH}, \boldsymbol{\pi}\rmk, c\lmk\boldsymbol{\beta_R}, \boldsymbol{v}\rmk\right]_{\caT}=I(\boldsymbol{\omega}).
\end{split}
\end{align}
This completes the proof.
\end{proof}

\section{Translation invariant case}\label{tani}
In this section, we restrict our attention to the
translation invariant case.
Let $\tau$ be the space translation of $\caA$,
one site to the right.
Throughout this section, we assume that 
the action $\beta$ is translation invariant:
\begin{align}\label{ike}
    \beta_g\tau=\tau\beta_g,\quad g\in G.
\end{align}
We furthermore suppose
 that $\beta_{gR}$,  $\beta_{gL}$ in (\ref{bdecom})
 can be chosen so that they act non-trivially only on $\caA_{[r,\infty)}$
and  $\caA_{(-\infty,-r-1]}$ respectively for some $2\le r$, and we do fix 
such $\boldsymbol{\beta_R}\in \boldsymbol{B}_R$ and
$\boldsymbol{v}\in \boldsymbol{V}(\boldsymbol{\beta_R})$.

We denote by $\caS_\tau$ the set of $\boldsymbol{\omega}=\{\omega_x\}\in \caS$
such that
 \begin{align}\label{ti}
 \omega_x\tau=\omega_x,\quad \text{for all}\quad x\in X.
 \end{align}
In this setting we derive MPS-like representation of $\boldsymbol{\omega}\in\caS_{\tau}$.
Throughout this section we fix $\boldsymbol{\omega}\in \caS_\tau$
and
$\boldsymbol{\omega}_R\in \boldsymbol{O}_R(\boldsymbol{\omega})$,
$(\boldsymbol{\caH}, \boldsymbol{\pi})\in 
\boldsymbol{GNS}(\boldsymbol{\omega_R})$,
$\boldsymbol{u}\in{\boldsymbol U}(\boldsymbol{\omega_R}, \boldsymbol{\beta_R},  \boldsymbol{\caH}, \boldsymbol{\pi})$.

\subsection{Representations of \texorpdfstring{$\omega_{x}\beta_{gR}, \omega_x\beta_{gR}\beta_{hR}$}{oxg}}
We show the MPS-like representation of states, following the standard argument \cite{bjp}\cite{arv}.
\begin{prop}\label{niji}
For any $x\in X$ and $g\in G$,
there are $\{S_\mu^{(x), g}\}_{x\in X,g\in G,\mu=1,\ldots,d}\subset\caB(\caH_x)$
and a density matrix $\rho_{x}$ on $\caH_x$
such that
\begin{align}
\begin{split}
&S_\mu^{(x), g}\lmk S_\nu^{(x), g}\rmk^*
=\pi_x\beta_{gR}\lmk e_{\mu\nu}^{(0)}\rmk,\\
&
 \pi_x\beta_{gR}\tau(A)
 =
 \sum_\mu  {S_\mu^{(x), g}}  \pi_x\beta_{gR}(A)\lmk  {S_\mu^{(x), g}}\rmk^*
,\quad A\in \caA_R,\\
&( {S_\mu^{(x), g}})^* {S_\nu^{(x), g}}=\delta_{\mu,\nu}\unit,
\end{split}
\end{align}
\begin{align}
  \begin{split}
&\omega_{x}\beta_{gR}\lmk e_{\mu_0\nu_0}^{(0)}\otimes e_{\mu_1\nu_1}^{(1)}\otimes \cdots\otimes e_{\mu_N\nu_N}^{(N)}\rmk\\
   &=\Tr\rho_{x}\lmk {S_{\mu_0}^{(x), g}} {S_{\mu_1}^{(x), g}}\cdots {S_{\mu_N}^{(x), g} }
{S_{\nu_N}^{(x), g}}^*\cdots {S_{\nu_1}^{(x), g}}^*{S_{\nu_0}^{(x), g}}^*\rmk.   
  \end{split}  
\end{align}
Here $\{e_{\mu,\nu}^{(k)}\}_{\mu,\nu=1,\ldots, d}$ is the system of standard matrix units
of $\Mat_d\simeq \caA_{\{0\}}$.

\end{prop}
\begin{proof}
This follows from a standard argument \cite{bjp}, \cite{arv} but for the reader's convenience we give a proof here.    

First note that by the irreducibility of $\pi_x\beta_{gR}$, 
we have
 \begin{align}\label{pi}
 \pi_x\beta_{gR}\tau(\caA_R)'
 =\pi_x\beta_{gR}\lmk\caA_{[1,\infty)}\rmk'
 =\pi_x\beta_{gR}(\caA_{\{0\}}).
 \end{align}

 We claim that there exists an endomorphism
 $\varphi_{x,g}: \caB(\caH_x)\to \caB(\caH_x)$
 such that
 \begin{align}\label{umi}
 \varphi_{x,g}\lmk \pi_x\beta_{gR}(A)\rmk=
 \pi_x\beta_{gR}\tau(A),\quad A\in \caA_R.
 \end{align}
 To see this, note that
 \begin{align}\label{sora}
\lmk \omega_{x,L}\otimes\omega_{x,R}\rmk 
\lmk \beta_{gL}\otimes\beta_{gR}\rmk
\tau
\lmk \beta_{gL}\otimes\beta_{gR}\rmk^{-1}
\simeq_{u.e.}
\omega_x\beta_g\tau\beta_g^{-1}
=\omega_x\simeq_{u.e.}  \omega_{x,L}\otimes\omega_{x,R}.
 \end{align}
Let $(\caH_{xL},\pi_{xL},\Omega_{xL})$, $(\caH_{x},\pi_{x},\Omega_{x})$
GNS triples for $\omega_{xL}$ and $\omega_{xR}$ respectively.
Then
\begin{align}
    \lmk \caH_{xL}\otimes \caH_x, \pi_1:=\lmk \pi_{xL}\otimes\pi_x\rmk
    \lmk \beta_{gL}\otimes\beta_{gR}\rmk
\tau
\lmk \beta_{gL}\otimes\beta_{gR}\rmk^{-1},\Omega_{xL}\otimes\Omega_x
    \rmk
\end{align}
is a GNS triple of $\lmk \omega_{x,L}\otimes\omega_{x,R}\rmk 
\lmk \beta_{gL}\otimes\beta_{gR}\rmk
\tau
\lmk \beta_{gL}\otimes\beta_{gR}\rmk^{-1}$,
while
\begin{align*}
    \lmk \caH_{xL}\otimes \caH_x, \pi_2:=\pi_{xL}\otimes\pi_x,
    \Omega_{xL}\otimes\Omega_x
    \rmk
\end{align*}
is a GNS triple of $ \omega_{x,L}\otimes\omega_{x,R}$.
From (\ref{sora}),
$\pi_1$ and $\pi_2$ are unitarily equivalent.
Restricting this to $\caA_R$, 
$\pi_1\vert_{\caA_R}$ and $\pi_2\vert_{\caA_R}$ are quasi-equivalent.
Note that $\pi_1\vert_{\caA_R}$ and 
$\pi_x \beta_{gR}\tau \lmk \beta_{gR}\rmk^{-1}$
are quasi-equivalent and
$\pi_2\vert_{\caA_R}$ and $\pi_x$ are quasi-equivalent.
This and the irreducibility of $\pi_x\beta_{gR}$ implies (\ref{umi}).

This $\varphi_{x,g}$ in (\ref{umi}) is an endomorphism of $\caB(\caH_x)$ with Powers index $d$, from
(\ref{pi}).
Therefore, from \cite{bjp}, \cite{arv},
there exists $ {\widetilde{S_\mu^{(x), g}}}\in \caB(\caH_x)$, $\mu=1,\ldots,d$
such that
\begin{align}
\begin{split}
&( \widetilde{S_\mu^{(x), g}})^* \widetilde{S_\nu^{(x), g}}=\delta_{\mu,\nu}\unit,\\
&\varphi_{x,g}\lmk a\rmk=\sum_\mu  \widetilde{S_\mu^{(x), g}} a\lmk  \widetilde{S_\mu^{(x), g}}\rmk^*,\quad
a\in \caB(\caH_x).
\end{split}
\end{align}
In particular, we have
\begin{align}\label{pea}
 \pi_x\beta_{gR}\tau(A)
 =\varphi_{x,g}\lmk \pi_x\beta_{gR}(A)\rmk=
 \sum_\mu  \widetilde{S_\mu^{(x), g}}  \pi_x\beta_{gR}(A)\lmk  \widetilde{S_\mu^{(x), g}}\rmk^*
,\quad A\in \caA_R.
\end{align}
We claim that
\begin{align}\label{tanuki}
\pi_x\beta_{gR}\lmk\caA_{\{0\}}\rmk=
\bbC-\left\{ \widetilde{S_\mu^{(x), g}}\lmk \widetilde{S_\nu^{(x), g}}\rmk^*\mid \mu,\nu=1,\ldots,d
\right\}.
\end{align}
Here $\bbC-$ means $\bbC$-linear span.
Multiplying (\ref{pea}) by $\widetilde{S_\mu^{(x), g}} \lmk \widetilde{S_\nu^{(x), g}}\rmk^*$ from right and
left gives the same result,
therefore, we see that
\begin{align}
\widetilde{S_\mu^{(x), g}}\lmk \widetilde{S_\nu^{(x), g}}\rmk^*
\in \pi_x\beta_{gR}\tau(\caA_R)'
  =\pi_x\beta_{gR}(\caA_{\{0\}}),
\end{align}
from (\ref{tanuki}).
On the other hand,
for any $B\in \caA_{\{0\}}$ and $A\in \caA_R$, we have
\begin{align}
\begin{split}
&0= \left[\pi_x\beta_{gR}\tau(A),\pi_x\beta_{gR}(B)\right]\\
& =
 \sum_\mu
 \lmk
 \widetilde{S_\mu^{(x), g}}  \pi_x\beta_{gR}(A)\lmk  \widetilde{S_\mu^{(x), g}}\rmk^* \pi_x\beta_{gR}(B)
 -\pi_x\beta_{gR}(B)\widetilde{S_\mu^{(x), g}}  \pi_x\beta_{gR}(A)\lmk  \widetilde{S_\mu^{(x), g}}\rmk^*
 \rmk.
\end{split}
\end{align}
Multiplying $\lmk  \widetilde{S_\nu^{(x), g}}\rmk^* $ from left
and $ \widetilde{S_\mu^{(x), g}} $ from right,
we obtain
\begin{align}
\left[
\pi_x\beta_{gR}(A),\lmk  \widetilde{S_\nu^{(x), g}}\rmk^* \pi_x\beta_{gR}(B)\widetilde{S_\mu^{(x), g}}
\right]=0
\end{align}
for any $B\in \caA_{0}$ and $A\in \caA_R$.
Because $\pi_x\beta_{gR}$ is irreducible, we have
\begin{align}
\lmk  \widetilde{S_\nu^{(x), g}}\rmk^* \pi_x\beta_{gR}(B)\widetilde{S_\mu^{(x), g}}=c_{\nu\mu}(B)\unit,
\end{align}
with some $c_{\nu\mu}(B)\in \bbC$, for any $B\in \caA_{\{0\}}$ .
Therefore, for any $B\in \caA_{\{0\}}$, we obtain
\begin{align}
\pi_x\beta_{gR}(B)
=\sum_{\mu,\nu} \widetilde{S_\nu^{(x), g}} \lmk  \widetilde{S_\nu^{(x), g}}\rmk^*  \pi_x\beta_{gR}(B)
\widetilde{S_\mu^{(x), g}}  \lmk  \widetilde{S_\mu^{(x), g}}\rmk^*
=\sum_{\mu,\nu}c_{\nu\mu}(B)
\widetilde{S_\nu^{(x), g}}   \lmk  \widetilde{S_\mu^{(x), g}}\rmk^* ,
\end{align}
proving the claim.

Let 
$\{e_{\mu\nu}^{(0)}\}$ be the standard system of matrix units of 
$\Mat_d\simeq\caA_{\{0\}}$.
Note that $\{\widetilde{S_\mu^{(x), g}}\lmk \widetilde{S_\nu^{(x), g}}\rmk^*\}_{ \mu,\nu=1,\ldots,d}$
is a system of matrix units of $\pi_x\beta_{gR}\lmk\caA_{\{0\}}\rmk$.
Hence there is some unitary $U\in \caU_d$ such that
\begin{align}
\begin{split}
\widetilde{S_\mu^{(x), g}}\lmk \widetilde{S_\nu^{(x), g}}\rmk^*
=\pi_x\beta_{gR}\lmk U^* e_{\mu\nu}^{(0)}U\rmk.
\end{split}
\end{align}
Setting
\begin{align}
S_\mu^{(x), g}:=\sum_\lambda U_{\lambda\mu} \widetilde{S_\lambda^{(x), g}},
\end{align}
we obtain
\begin{align}
\begin{split}
&S_\mu^{(x), g}\lmk S_\nu^{(x), g}\rmk^*
=\sum_{\lambda,\zeta} U_{\lambda\mu} \overline{U_{\zeta\nu} }
\widetilde{S_\lambda^{(x), g}}\lmk \widetilde{S_\zeta^{(x), g}}\rmk^*\\
&=\sum_{\lambda,\zeta} U_{\lambda\mu} \overline{U_{\zeta\nu} }
\pi_x\beta_{gR}\lmk U^* e_{\lambda\zeta}^{(0)}U\rmk
=\pi_x\beta_{gR}\lmk U^* Ue_{\mu\nu}^{(0)}U^*U\rmk
=\pi_x\beta_{gR}\lmk e_{\mu\nu}^{(0)}\rmk,\\
&
 \pi_x\beta_{gR}\tau(A)
 =
 \sum_\mu  {S_\mu^{(x), g}}  \pi_x\beta_{gR}(A)\lmk  {S_\mu^{(x), g}}\rmk^*
,\quad A\in \caA_R,\\
&( {S_\mu^{(x), g}})^* {S_\nu^{(x), g}}=\delta_{\mu,\nu}\unit.
\end{split}
\end{align}
From this, we obtain
\begin{align}\begin{split}
 &   \pi_x\beta_{gR} \lmk e_{\mu_0\nu_0}^{(0)}\otimes e_{\mu_1\nu_1}^{(1)}\otimes \cdots\otimes e_{\mu_N\nu_N}^{(N)}\rmk
   =\pi_x\beta_{gR} \lmk e_{\mu_0\nu_0}^{(0)}\tau \lmk e_{\mu_1\nu_1}^{(0)}\rmk \cdots \tau^N \lmk e_{\mu_N\nu_N}^{(0)}\rmk\rmk,\\
&   = {S_{\mu_0}^{(x), g}} {S_{\mu_1}^{(x), g}}\cdots {S_{\mu_N}^{(x), g} }
{S_{\nu_N}^{(x), g}}^*\cdots {S_{\nu_1}^{(x), g}}^*{S_{\nu_0}^{(x), g}}^*.
\end{split}
\end{align}

Note that $\omega_x\vert_{\caA_R}$ is quasi-equivalent to $\omega_{xR}$ because of
the definition of $\omega_{xR}$ and the purity of $\omega_x$.
Therefore, there exists a density matrix $\rho_{x}$ on $\caH_x$
such that 
\begin{align*}
    \omega_x\vert_{\caA_R}=\Tr \rho_{x}\pi_{x}.
\end{align*}
Then we have
\begin{align*}
    \begin{split}
&        \omega_x\beta_{gR}\lmk e_{\mu_0\nu_0}^{(0)}\otimes e_{\mu_1\nu_1}^{(1)}\otimes \cdots\otimes e_{\mu_N\nu_N}^{(N)}\rmk
        =\Tr \rho_{x}\pi_{x}\beta_{gR}\lmk e_{\mu_0\nu_0}^{(0)}\otimes e_{\mu_1\nu_1}^{(1)}\otimes \cdots\otimes e_{\mu_N\nu_N}^{(N)}\rmk\\
        &= \Tr \rho_{x}
        \lmk {S_{\mu_0}^{(x), g}} {S_{\mu_1}^{(x), g}}\cdots {S_{\mu_N}^{(x), g} }
{S_{\nu_N}^{(x), g}}^*\cdots {S_{\nu_1}^{(x), g}}^*{S_{\nu_0}^{(x), g}}^*.
        \rmk.
    \end{split}
\end{align*}
This completes the proof of Proposition (\ref{niji}).
\end{proof}

Now we consider the action of $u_{x,g}$ on
$S_\mu^{(x) g}$s.

\begin{lem}\label{yama}

For any $x\in X$, 
$g\in G$
there exits $c_2((x), g)\in U(1)$ such that
\begin{align}
\label{xsc}
\Ad\lmk u_{xg^{-1},g}^*\rmk \lmk S_\mu^{(xg^{-1}),g}\rmk
=c_2((x),g)S_\mu^{(x), e}.
\end{align}
\end{lem}
\begin{proof}
Consider 
\begin{align}
\begin{split}
X_\mu^{(x),g}:=\Ad\lmk u_{x,g}^*\rmk \lmk S_\mu^{(x),g}\rmk\in \caB(\caH_{xg}),\quad
x\in X.
\end{split}
\end{align}
It is straightforward to check that $X_\mu^{(xg^{-1}), g}$ satisfies
the same relation as $S_\mu^{(x), e}$s
in Proposition \ref{niji}.
Therefore, by the standard argument (see \cite{ot} for example), there exists $c_2((x),g)\in U(1)$ such that
\begin{align}\label{xsc}
\Ad\lmk u_{xg^{-1},g}^*\rmk \lmk S_\mu^{(xg^{-1}),g}\rmk
=
X_\mu^{(xg^{-1}), g}=c_2((x),g)S_\mu^{(x), e}.
\end{align}
\end{proof}

We also introduce $S^{(x),g,h}$ corresponding to the state $\omega_x \beta_g \beta_h$:

\begin{prop}\label{niji2}
For any $x\in X$ and $g,h\in G$,
there are $\{S_\mu^{(x), g,h}\}_{x\in X,g,h\in G,\mu=1,\ldots,d}\subset\caB(\caH_x)$
and a density matrix $\rho_{x,g,h}$ on $\caH_x$
such that
\begin{align}
\begin{split}
&S_\mu^{(x), g,h}\lmk S_\nu^{(x), g,h}\rmk^*
=\pi_x\beta_{gR}\beta_{hR}\lmk e_{\mu\nu}^{(0)}\rmk,\\
&
 \pi_x\beta_{gR}\beta_{hR}\tau(A)
 =
 \sum_\mu  {S_\mu^{(x), g,h}}  \pi_x\beta_{gR}\beta_{hR}(A)\lmk  {S_\mu^{(x), g,h}}\rmk^*
,\quad A\in \caA_R,\\
&( {S_\mu^{(x), g,h}})^* {S_\nu^{(x), g,h}}=\delta_{\mu,\nu}\unit,
\end{split}
\end{align}
\begin{align}
  \begin{split}
&\omega_{x}\beta_{gR}\beta_{hR}\lmk e_{\mu_0\nu_0}^{(0)}\otimes e_{\mu_1\nu_1}^{(1)}\otimes \cdots\otimes e_{\mu_N\nu_N}^{(N)}\rmk\\
   &=\Tr\rho_{x,g,h}\lmk {S_{\mu_0}^{(x), g,h}} {S_{\mu_1}^{(x), g,h}}\cdots {S_{\mu_N}^{(x), g,h} }
{S_{\nu_N}^{(x), g,h}}^*\cdots {S_{\nu_1}^{(x), g,h}}^*{S_{\nu_0}^{(x), g,h}}^*\rmk.   
  \end{split}  
\end{align}
Here $\{e_{\mu,\nu}^{(k)}\}_{\mu,\nu=1,\ldots, d}$ is the system of standard matrix units
of $\Mat_d$.
\end{prop}

Next, similar to Lemma~\ref{yama}, we have

\begin{lemma}\label{yama2}
    There exists $c_3(xg,g,h)\in U(1)$ such that 
\begin{equation*}
    S^{(x),g,h} = c_3(xg,g,h) \cdot \Ad \left(\pi_x (v(g,h))\right) \left( S^{(x),gh} \right).
\end{equation*}
\end{lemma}
We also have
\begin{lemma}\label{yama3}
    There exists $c_2'(x,g,h)\in U(1)$ such that 
    \begin{equation*}
        S^{(x),g,h} = c_2'(x,g,h)\cdot  \Ad (u_{x,g}) \left( S^{(xg),h} \right).
    \end{equation*}
\end{lemma}
Let us express now $S_\mu^{(x),g,h}$ in two different ways using Lemma~\ref{yama2} and \ref{yama3}. First, using  Lemma~\ref{yama3} we obtain
\begin{align*}
S_\mu^{(x),g,h} = c_2'(x,g,h)\cdot \Ad( u_{x,g} ) \left(S_\mu^{(xg), h} \right)
\end{align*}
Using now Eq.~\eqref{xsc}, we obtain
\begin{equation*}
   S_\mu^{(x),g,h} = c_2'(x,g,h) c_2(xgh,h)\cdot \Ad( u_{x,g} u_{xg,h}) \left(S_\mu^{(xgh), e} \right) .
\end{equation*}
Let us express now $S_\mu^{(x),g,h}$ with the help of Lemma~\ref{yama2}:
\begin{align*}
S_\mu^{(x),g,h} =  c_3(xg,g,h) \cdot \Ad \left(\pi_x (v(g,h) )\right) \left(S_\mu^{(x), gh} \right).
\end{align*}
Using again Eq.~\eqref{xsc}, we obtain
\begin{align*}
S_\mu^{(x),g,h} = c_3(xg,g,h) c_2(xgh,gh) \cdot \Ad \left(\pi_x (v(g,h) u_{x,gh})\right) \left(S_\mu^{(xgh), e} \right).
\end{align*}
Finally comparing the two different expressions for $S^{(x),g,h}$, we obtain
\begin{align*}
&c_2'(x,g,h) c_2(xgh,h)\cdot \Ad( u_{x,g} u_{xg,h}) \left(S_\mu^{(xgh), e} \right) \\
= &c_3(xg,g,h) c_2(xgh,gh) \cdot \Ad \left(\pi_x (v(g,h) u_{x,gh})\right) \left(S_\mu^{(xgh), e} \right).
\end{align*}

\subsection{Reduction to MPS}

Now we apply the general theory in the previous subsection to MPS.
Assume now that $\omega_x$ are injective MPS.

For each $x\in X$, let  $\Phi_x$ be a parent interaction of $\omega_x$.
Let $P^{(x), g}$ be the orthogonal projection onto
\[
\cap_{X\subset [0,\infty)} \ker\pi_x\beta_{gR}\lmk\Phi_{xg}(X)\rmk\subset \caH_x,
\]
for $x\in X$ and $g\in G$.
Set 
\begin{align*}
    B_\mu^{(x),g}:=P^{(x), g} S_\mu^{(x), g}P^{(x), g}.
\end{align*}
for $x\in X$ and $g\in G$.
By \cite{O1}, 
we know that $P^{(x),e }$ is equal to the support projection of the density matrix
$\rho_x$, and that
\begin{align}\label{kuma}
    B_{\mu}^{(x), e}=P^{(x), e}S_\mu^{(x), e} 
\end{align}
and that  the tuple of matrices on $P^{(x),e}\caH_x$
$\{B_{\mu}^{(x), e}\}_\mu$ gives an injective MPS.
Because 
\[
\Ad(u_{x,g})\lmk \pi_{xg}\lmk\Phi_{xg}(X)\rmk\rmk=\pi_x\beta_{gR}\lmk\Phi_{xg}(X)\rmk,
\]
we see that
\[
\Ad\lmk u_{x,g}^*\rmk (P^{(x), g})=P^{(xg), e}.
\]
By this, (\ref{xsc}), and (\ref{kuma}), we have
\begin{align}\label{inu}
B^{(x), g}_\mu
=
c_2((xg),g) \Ad u_{x,g}\lmk B_\mu^{(xg), e}\rmk.
\end{align}
From this, we have
\begin{align}\label{kuma}
    B_{\mu}^{(x), g}=P^{(x), g}S_\mu^{(x), g},
\end{align}
and the tuple of matrices $\{B_\mu^{(x), g}\}_\mu$ on $P^{(x), g}\caH_x$
gives an injective tensor.

Similarly, let $P^{(x), g, h}$ be an orthogonal projection onto
\begin{align*}
    \cap_{X\subset [0,\infty)} \ker\pi_x\beta_{gR}\beta_{hR}\lmk \Phi_{xgh}(X)\rmk.
\end{align*}
Then because $\pi_x\beta_{gR}\beta_{hR}\lmk \Phi_{xgh}(X)\rmk
=\Ad \pi_x\lmk  v(g,h)\rmk\lmk \pi_x\beta_{gh, R}(\Phi_{xgh}(X)) \rmk$,
we have
\begin{align*}
    \Ad\lmk  \pi_x\lmk v(g,h)\rmk\rmk\lmk P^{(x),gh}\rmk=P^{(x), g,h}.
\end{align*}
Setting
\begin{align*}
B_\mu^{(x), g, h}:=P^{(x), g,h} S_\mu^{(x), g,h} P^{(x), g,h},
\end{align*}
we get
\begin{align}\label{neko}
    B_\mu^{(x), g, h}=P^{(x), g,h} S_\mu^{(x), g,h}
    =c_3(xg, g,h) \Ad\lmk  \pi_x\lmk v(g,h)\rmk\rmk\lmk B_\mu^{(x), gh}\rmk.
\end{align}
By this formula, the tuple of matrices $\{B_\mu^{(x), g, h}\}$ 
on $P^{(x), g, h}\caH_x$
gives an injective tensor.

Finally, by $\pi_x\beta_{gR}\beta_{hR}\lmk\Phi_{xgh}(X)\rmk=\Ad\lmk u_{x,g}\rmk\pi_{xg}\beta_{hR}\lmk\Phi_{xgh}(X)\rmk$,
we get
\begin{align*}
    \Ad\lmk u_{x, g}\rmk \lmk P^{(xg), h}\rmk=P^{(x), g,h}.
\end{align*}
Therefore, we obtain
\begin{align}
    B_\mu^{(x), g,h}
    =c_2'(x,g,h) \Ad u_{x,g} \lmk B_\mu^{(xg), h}\rmk.
\end{align}

Let us express now $B_\mu^{(x),g,h}$ in two different ways.
First we obtain
\begin{align*}
B_\mu^{(x),g,h} = c_2'(x,g,h)\cdot \Ad( u_{x,g} ) \left(B_\mu^{(xg), h} \right)
\end{align*}
Using now Eq.~\eqref{inu}, we obtain
\begin{equation*}
   B_\mu^{(x),g,h} = c_2'(x,g,h) c_2(xgh,h)\cdot \Ad( u_{x,g} u_{xg,h}) \left(B_\mu^{(xgh), e} \right) .
\end{equation*}
Let us express now $B_\mu^{(x),g,h}$ with the help of (\ref{neko})
\begin{align*}
B_\mu^{(x),g,h} =  c_3(xg,g,h) \cdot \Ad \left(\pi_x (v(g,h) )\right) \left(B_\mu^{(x), gh} \right).
\end{align*}
Using again Eq.~\eqref{inu}, we obtain
\begin{align*}
B_\mu^{(x),g,h} = c_3(xg,g,h) c_2(xgh,gh) \cdot \Ad \left(\pi_x (v(g,h) u_{x,gh})\right) \left(B_\mu^{(xgh), e} \right).
\end{align*}
Finally comparing the two different expressions for $S^{(x),g,h}$, we obtain
\begin{align*}
&c_2'(x,g,h) c_2(xgh,h)\cdot \Ad( u_{x,g} u_{xg,h}) \left(B_\mu^{(xgh), e} \right) \\
= &c_3(xg,g,h) c_2(xgh,gh) \cdot \Ad \left(\pi_x (v(g,h) u_{x,gh})\right) \left(B_\mu^{(xgh), e} \right).
\end{align*}
From injectivity, we derive \eqref{eq:mps_index} from  this.

\section{Matrix product state approach and its GNS representation}

In this section we develop a dictionary  between the tensor network approach and the operator algebra formalism used in the  previous sections. We show how the GNS representation of MPS works writing explicitly all the ingredients (see also \cite{UncleHamiltonian}). We hope this section helps to connect both approaches.

\subsection{The finite size MPS and MPU setup}

We assume that for every $g\in G$ we are given an integer $D_g$ and a tensor $u(g)\in \Mat_{d}\otimes \Mat_{D_g}$,

\begin{equation*}
    u(g) = \sum_{ij} \ket{i}\bra{j} \otimes u_{ij}(g) = \sum_{\alpha\beta} u_{\alpha\beta}(g) \otimes \ket{\alpha}\bra{\beta} ,
\end{equation*}
such that the finite size matrix product operator (MPO) $U_n(g) \in\Mat_d^{\otimes n}$ generated by this tensor,
\begin{align*}
    U_n(g) = &\sum_{ij} \tr \{u_{i_1j_1}(g) \cdot u_{i_2j_2}(g) \cdots  u_{i_nj_n}(g) \} \cdot \ket{i_1i_2\dots i_n}\bra{j_1j_2\dots j_n} \\ 
    =&\sum_{\alpha_1 \dots \alpha_n} u_{\alpha_1\alpha_2}(g)\otimes u_{\alpha_2\alpha_3}(g) \otimes \dots \otimes u_{\alpha_n\alpha_1}(g) \ ,
\end{align*}
form a unitary representation of $G$: $U_n(g) U_n(h) = U_n(gh)$, $U_n(1)=\mathbb 1^{\otimes n}$ (note that this is also an MPO with $D_1=1$), and  $U_n(g)^* U_n(g) = \mathbb 1^{\otimes n}$ for all $g,h\in G$ and $n\in\mathbb N$ so that it is a matrix product unitary (MPU).

We also assume that for every $x\in X$ we are given an integer $D_x$ and an MPS tensor $a(x)\in \mathbb C^{d} \otimes \Mat_{D_x}$,
\begin{equation*}
    a(x) = \sum_{\alpha\beta} a_{\alpha\beta}(x) \otimes \ket{\alpha}\bra{\beta} = \sum_i \ket{i} \otimes a_i(x),
\end{equation*}
such that the finite size MPS, $\psi_n(x)\in(\mathbb C^d)^{\otimes n}$, defined by 
\begin{align*}
    \psi_n(x) & = \sum_{i} \tr \{a_{i_1}(x)\cdot a_{i_2}(x)\cdots  a_{i_n}(x) \} \cdot \ket{i_1i_2\dots i_n}\\
    &= \sum_{\alpha_1\dots \alpha_n} a_{\alpha_1 \alpha_2}(x) \otimes a_{\alpha_2\alpha_3}(x)\otimes \dots \otimes a_{\alpha_n\alpha_1}(x),
\end{align*}
satisfies the equations $U_n(g)^* \psi_n(x) = \psi_n(xg)$ for all $x\in X$, $g\in G$, and $n\in\mathbb N$.

We further assume that the MPS tensors $a(x)$ and $u(g)$ are injective after blocking, that is, there is $K_0\in\mathbb N$ such that for all $k\geq K_0$ the maps $\Mat_{D_x}\to \mathbb C$ and $\Mat_{D_g}\to\Mat_d$ defined by
\begin{align*}
    X &\mapsto \sum_{\alpha_1 \dots \alpha_k} \bra{\alpha_{k+1}} X \ket{\alpha_1} \cdot a_{\alpha_1\alpha_2}(x) \otimes \dots \otimes a_{\alpha_k\alpha_{k+1}}(x),\\
    Y &\mapsto \sum_{\alpha_1 \dots \alpha_k} \bra{\alpha_{k+1}} Y \ket{\alpha_1} \cdot u_{\alpha_1\alpha_2}(g) \otimes \dots \otimes u_{\alpha_k\alpha_{k+1}}(g),\\
\end{align*}
are injective\footnote{Note that if these maps are injective for some $k\in\mathbb N$, then they are also injective for any $n>k$, and thus this assumption is redundant.}. The transfer matrices of these MPS and MPO tensors are defined as the completely positive maps 
\begin{align*}
    T_x(\rho) = \sum_i a_i(x) \rho a_i(x)^*  
    \quad \text{and} \quad
    T_{g}(\rho) =  \sum_{ij} u_{ij}(g) \rho u_{ij}(g)^* . 
\end{align*}
Let us assume that the MPS tensors are normalized in such a way that the spectral radius of $T_x$ is $1$; as $U_n$ is unitary, $\tr(T_g^n) = \tr(U_n^*U_n)=d^n$, and thus the spectral radius of $T_g$ is $d$. Through the Perron-Frobenius theorem, the spectral radius of each of these maps is an eigenvalue of the map. Due to the injectivity condition, this eigenvalue is non-degenerate and the corresponding eigenvector is positive and full rank \cite{PerronFrobenius}. Using the fact that given an invertible matrix $X$ the MPS tensors defined by matrices $a_i$ and $Xa_iX^{-1}$ generate the same MPS for all system sizes, we can assume  w.l.o.g.\ that these eigenvectors are the identity (that is, the tensor is in the right canonical form \cite{FCS,MPSReps}), 
\begin{align}\label{eq:rfp}
    \sum_{i} a_i(x)a_i(x)^* = \mathbb 1 , 
    \quad \text{and} \quad 
    \sum_{ij} u_{ij}(g)u_{ij}(g)^* = d\cdot\mathbb 1 . 
\end{align}
Similarly, let $\rho(x)$ and $\rho(g)$ be the positive full rank matrices uniquely defined by the equations
\begin{equation}\label{eqlfp}
     \sum_i a_i(x)^* \rho(x) a_i(x) = \rho(x) \quad \text{and} \quad  \sum_{ij} u_{ij}(g)^* \rho(g) u_{ij}(g) = d \cdot \rho(g).
\end{equation}
Furthermore, we will write $a^{(n)}(x)$ and $u^{(n)}(g)$ for the tensors obtained by blocking $n$ sites:
\begin{align*}
    a_{\alpha\beta}^{(n)}(x)  &= \sum_{\gamma_1\dots \gamma_n} a_{\alpha\gamma_1}(x) \otimes a_{\gamma_1\gamma_2}(x) \otimes \dots \otimes a_{\gamma_{n-1}\beta}(x), \\
    u_{\alpha\beta}^{(n)}(g)  &= \sum_{\gamma_1\dots \gamma_n} u_{\alpha\gamma_1}(g) \otimes u_{\gamma_1\gamma_2}(g) \otimes \dots \otimes u_{\gamma_{n-1}\beta}(g),
\end{align*}

\subsection{Graphical notation}

As some of the formulas that we will write are quite cumbersome, we will use a graphical language customary in the study of tensor networks to express those equations. In this notation tensors are denoted with various shapes, and their indices are denoted as lines attached to the shape. For example, the tensor $u(g)$ is denoted by
  \begin{equation}
    u(g) = \sum_{\alpha,\beta=1}^{D_g} u_{\alpha\beta}(g) \otimes \ket{\alpha}\bra{\beta}  =
    \begin{tikzpicture}
      \node[irrep] at (0.2,0) {$g$};
      \node[mpo] (t) at (0,0) {};
      \foreach \x/\d in {0/-<-, 90/->-, 180/->-, 270/-<-}{
        \draw[\d] (t) --++ (\x:0.4);
      }
    \end{tikzpicture} \ ,
  \end{equation}
  where the vertical indices correspond to the indices of the matrix $u_{\alpha\beta}\in\Mat_d$, while the horizontal indices correspond to the indices of $\ketbra{\alpha}{\beta}\in\Mat_D$. We have attached a label $g$ on the horizontal line to denote the $g$-dependence of the tensor. The arrows differentiate between the input and output indices of the matrices: the vertical matrix acts from bottom to top, while the horizontal one from right to left. Similarly, the MPS tensor $a(x)$ is denoted as 
  \begin{equation}
    a(x) = \sum_{\alpha,\beta=1}^{D_x} a_{\alpha\beta}(g) \otimes \ket{\alpha}\bra{\beta}  =
    \begin{tikzpicture}
      \node[irrep] at (0.2,0) {$x$};
      \node[mps] (t) at (0,0) {};
      \foreach \x/\d/\c in {0/-<-/blue, 90/->-/black, 180/->-/blue}{
        \draw[\d,\c] (t) --++ (\x:0.4);
      }
    \end{tikzpicture} \ .
  \end{equation}
  Here, the index $x$ indicates the $x$-dependence of the tensor $a$ and it is written on the horizontal line. 
  
  As usual, joining lines corresponds to index contraction. Let us remark here that when working with matrices we think about column vectors, so the product is opposite to composition: $AB$ means $B$ acts first, $A$ second, or graphically,
  \begin{equation*}
    \begin{tikzpicture}[scale=0.6]
        \node[tensor,label=below:$A$] (t) at (0,0) {};
        \node[tensor,label=below:$B$] (r) at (1,0) {};
        \draw[-<-] (t)--(r);
        \draw[-<-] (r)--++(1,0);
        \draw[->-] (t)--++(-1,0);
    \end{tikzpicture} =
    \begin{tikzpicture}[scale=0.6]
        \node[tensor,label=below:$AB$] (t) at (0,0) {};
        \draw[-<-] (t)--++(1,0);
        \draw[->-] (t)--++(-1,0);
    \end{tikzpicture} \ .
  \end{equation*}
  From now on, we will only indicate explicitly the reading direction of the matrices when they are not from left to right, or bottom to top, respectively. In this notation, the MPS $\psi_n(x)$ and the MPU $U_n(g)$ reads as
  \begin{align*}
    \psi_n(x) &= \sum_{\alpha_1 \dots \alpha_n} a_{\alpha_1\alpha_2}(x)\otimes a_{\alpha_2\alpha_3}(x) \otimes \dots \otimes a_{\alpha_n\alpha_1}(x) =
    \begin{tikzpicture}
      \def\d{0.6};
      \draw[blue] (-0.5,0) rectangle (3*\d+0.5,-0.5);
      \node[irrep] at (\d/2,0) {$x$};
      \foreach \x in {0,1,3}{
        \draw (\x*\d,0) -- (\x*\d,0.4);
        \node[mps] (t) at (\x*\d,0) {};
      }
      \node[fill=white] at (\d*2,0) {$\dots$};
    \end{tikzpicture} \ , \\
    U_n(g) &= \sum_{\alpha_1 \dots \alpha_n} u_{\alpha_1\alpha_2}(g)\otimes u_{\alpha_2\alpha_3}(g) \otimes \dots \otimes u_{\alpha_n\alpha_1}(g) =
    \begin{tikzpicture}
      \def\d{0.6};
      \draw (-0.5,0) rectangle (3*\d+0.5,-0.5);
      \node[irrep] at (\d/2,0) {$g$};
      \foreach \x in {0,1,3}{
        \draw (\x*\d,-0.4) -- (\x*\d,0.4);
        \node[mpo] (t) at (\x*\d,0) {};
      }
      \node[fill=white] at (\d*2,0) {$\dots$};
    \end{tikzpicture} \ .
  \end{align*}
  In this notation, the equation $U_n(g) U_n(h) = U_n(gh)$ reads as
  \begin{equation}\label{eq:MPU_mult}
     U_n(g) U_n(h) =
    \begin{tikzpicture}[baseline = 1mm,xscale=0.6]
      \draw (-0.8,0) rectangle (3.8,-0.5);
      \draw (-1,0.4) rectangle (4,-0.6);
      \node[irrep] at (0.5,0.4) {$g$};
      \node[irrep] at (0.5,0) {$h$};
      \foreach \x in {0,1,3}{
        \draw (\x,-0.4) -- (\x,0.8);
        \node[mpo] (t) at (\x,0) {};
        \node[mpo] (t) at (\x,0.4) {};
      }
      \node[fill=white] at (2,0) {$\dots$};
      \node[fill=white] at (2,0.4) {$\dots$};
    \end{tikzpicture} \  = 
    \begin{tikzpicture}[xscale=0.6]
      \draw (-0.5,0) rectangle (4,-0.5);
      \node[irrep] at (0.5,0) {$gh$};
      \foreach \x in {0,1,3}{
        \draw (\x,-0.4) -- (\x,0.4);
        \node[mpo] (t) at (\x,0) {};
      }
      \node[fill=white] at (2,0) {$\dots$};
    \end{tikzpicture} \ = U_n(gh).
  \end{equation}
  The adjoint of the MPS $\psi_n(x)$ (that is, the linear functional $\phi\mapsto\langle\psi_n(x)|\phi\rangle$) and of the MPU $U_n(g)$ are described by the MPS and MPO tensors
  \begin{align*}
    u(g)^* = \sum_{\alpha,\beta=1}^{D_g} u_{\alpha\beta}(g)^* \otimes \ket{\beta}\bra{\alpha}  =
    \begin{tikzpicture}
      \node[irrep] at (0.2,0) {$g$};
      \node[cmpo] (t) at (0,0) {};
      \foreach \x/\d in {0/->-, 90/->-, 180/-<-, 270/-<-}{
        \draw[\d] (t) --++ (\x:0.4);
      }
    \end{tikzpicture} \ , \\
    a(x)^* = \sum_{\alpha,\beta=1}^{D_x} a_{\alpha\beta}(g)^* \otimes \ket{\beta}\bra{\alpha}  =
    \begin{tikzpicture}
      \node[irrep] at (0.2,0) {$x$};
      \node[cmps] (t) at (0,0) {};
      \foreach \x/\d/\c in {0/->-/blue, 270/-<-/black, 180/-<-/blue}{
        \draw[\d,\c] (t) --++ (\x:0.4);
      }
    \end{tikzpicture} \ ,
  \end{align*}
  where again the arrows indicate which index is input and which is output. Specifically, on both of the horizontal lines, the index $\ket{\beta}$ of the matrix unit is on the right side of the tensor, while the index $\bra{\alpha}$ of the matrix unit is on the left side of the tensor. From now on, we do not display the orientation of these tensors; they will always be oriented from bottom to top and left to right. Using these tensors the equation $U_n(g)^* \psi_n(x) = \psi_n(xg)$ reads as
\begin{equation}\label{eq:MPS_sym}
    \begin{tikzpicture}[baseline = 3mm]
      \def\d{0.6};
      \draw[blue] (-0.5,0) rectangle (3*\d+0.5,-0.5);
      \draw (-0.7,0.4) rectangle (3*\d+0.7,-0.7);
      \node[irrep] at (\d/2,0.4) {$g$};
      \node[irrep] at (\d/2,0) {$x$};
      \foreach \x in {0,1,3}{
        \draw (\x*\d,0) -- (\x*\d,0.8);
        \node[mps] (t) at (\x*\d,0) {};
        \node[cmpo] (t) at (\x*\d,0.4) {};
      }
      \node[fill=white] at (\d*2,0) {$\dots$};
      \node[fill=white] at (\d*2,0.4) {$\dots$};
    \end{tikzpicture} \ = \ 
    \begin{tikzpicture}
      \def\d{0.6};
      \draw[blue] (-0.5,0) rectangle (3*\d+0.5,-0.5);
      \node[irrep] at (\d/2,0) {$xg$};
      \foreach \x in {0,1,3}{
        \draw (\x*\d,0) -- (\x*\d,0.4);
        \node[mps] (t) at (\x*\d,0) {};
      }
      \node[fill=white] at (\d*2,0) {$\dots$};
    \end{tikzpicture} \ .
\end{equation}
  As $U_n(g)^\dagger = U_n(g^{-1})$, and both $u(g^{-1})$ and $u(g)^*$ are injective MPO tensors, using the fundamental theorem of MPS we obtain that there is $X_g \in \mathbb{C}^D \otimes \mathbb{C}^D$ and $X_g^{-1} \in \left(\mathbb{C}^D\right)^*\otimes\left(\mathbb{C}^D\right)^*$ such that
  \begin{equation}\label{eq:conjugate_gauge}
    \begin{tikzpicture}
      \node[irrep] at (0.3,0) {$g^{-1}$};
      \node[irrep] at (-0.3,0) {$g^{-1}$};
      \node[mpo] (t) at (0,0) {};
      \foreach \x/\d in {0/-<-, 90/->-, 180/->-, 270/-<-}{
        \draw[\d] (t) --++ (\x:0.6);
      }
    \end{tikzpicture} \  = \
    \begin{tikzpicture}
      \node[irrep] at (0.3,0) {$g$};
      \node[irrep] at (-0.3,0) {$g$};
      \node[irrep] at (0.9,0) {$g^{-1}$};
      \node[irrep] at (-0.9,0) {$g^{-1}$};
      \node[cmpo] (t) at (0,0) {};
      \node[tensor,fill=gray,label=below:$X_g\vphantom{X_g^{-1}}$] (l) at (-0.6,0) {};
      \node[tensor,fill=gray,label=below:$X_g^{-1}$] (r) at (0.6,0) {};
      \draw[-<-] (t)--(l);
      \draw[->-] (t)--(r);
      \draw[->-] (l)--++(-0.6,0);
      \draw[-<-] (r)--++(0.6,0);
      \draw[->-] (t)--++(0,0.6);
      \draw[-<-] (t)--++(0,-0.6);
    \end{tikzpicture}, \quad
    \begin{tikzpicture}
      \node[irrep] at (-0.3,0) {$g$};
      \node[irrep] at (0.3,0) {$g^{-1}$};
      \node[irrep] at (-0.9,0) {$g^{-1}$};
      \node[tensor,fill=gray,label=below:$X_g\vphantom{X_g^{-1}}$] (l) at (-0.6,0) {};
      \node[tensor,fill=gray,label=below:$\ X_g^{-1}$] (r) at (0,0) {};
      \draw[->-] (l)--(r);
      \draw[->-] (l)--++(-0.6,0);
      \draw[-<-] (r)--++(0.6,0);
    \end{tikzpicture}  = \mathbb{1}
    \quad \text{and} \quad
    \begin{tikzpicture}
      \node[irrep] at (-0.3,0) {$g^{-1}$};
      \node[irrep] at (0.3,0) {$g$};
      \node[irrep] at (-0.9,0) {$g$};
      \node[tensor,fill=gray,label=below:$X_g^{-1}$] (l) at (-0.6,0) {};
      \node[tensor,fill=gray,label=below:$\ X_g\vphantom{X_g^{-1}}$] (r) at (0,0) {};
      \draw[-<-] (l)--(r);
      \draw[-<-] (l)--++(-0.6,0);
      \draw[->-] (r)--++(0.6,0);
    \end{tikzpicture}  = \mathbb{1}.
  \end{equation}

The eigenvalue equations \eqref{eq:rfp} and \eqref{eqlfp} read as 
\begin{equation} \label{eq:graphical lfp}
\begin{aligned}
    \begin{tikzpicture}[baseline=1.5mm]
        \def\d{0.6}
        \draw (\d/2,0.5) -- (-\d/2,0.5) -- (-\d/2,0) -- (\d/2,0);
        \draw (0,0)--(0,0.5);
        \node[mps] at (0,0) {};
        \node[mps,label=left:$\rho{(x)}$] at (-\d/2,0.25) {};
        \node[cmps] at (0,0.5) {};
    \end{tikzpicture} = 
    \begin{tikzpicture}[baseline=1.5mm]
        \def\d{0.6}
        \draw (0,0.5) -- (-\d/2,0.5) -- (-\d/2,0) -- (0,0);
        \node[cmps,label=left:$\rho{(x)}$] at (-\d/2,0.25) {};
    \end{tikzpicture} \ &,  \quad
    \begin{tikzpicture}[baseline=1.5mm,xscale=-1]
        \def\d{0.6}
        \draw (\d/2,0.5) -- (-\d/2,0.5) -- (-\d/2,0) -- (\d/2,0);
        \draw (0,0)--(0,0.5);
        \node[mps] at (0,0) {};
        \node[cmps] at (0,0.5) {};
    \end{tikzpicture} = 
    \begin{tikzpicture}[baseline=1.5mm,xscale=-1]
        \def\d{0.6}
        \draw (0,0.5) -- (-\d/2,0.5) -- (-\d/2,0) -- (0,0);
    \end{tikzpicture}, \\ 
  \begin{tikzpicture}
    \draw (0,-0.6) rectangle (0.3,0.6);
    \draw (0.6,0.2) -- (-0.6,0.2) -- (-0.6,-0.2) -- (0.6,-0.2);
    \node[mpo] at (0,0.2) {};
    \node[cmpo] at (0,-0.2) {};
    \node[tensor,label=left:$\rho(g)$] at (-0.6,0) {};
  \end{tikzpicture} = 
  \begin{tikzpicture}[baseline=1.5mm]
    \def\d{0.6}
    \draw (0,0.5) -- (-\d/2,0.5) -- (-\d/2,0) -- (0,0);
    \node[mpo,label=left:$\rho{(g)}$] at (-\d/2,0.25) {};
  \end{tikzpicture} \ &,  \quad  
  \begin{tikzpicture}[xscale=-1]
    \draw (0,-0.6) rectangle (0.3,0.6);
    \draw (0.6,0.2) -- (-0.6,0.2) -- (-0.6,-0.2) -- (0.6,-0.2);
    \node[mpo] at (0,0.2) {};
    \node[cmpo] at (0,-0.2) {};
  \end{tikzpicture} =
    \begin{tikzpicture}[baseline=1.5mm]
        \def\d{0.6}
        \draw (0,0.5) -- (\d/2,0.5) -- (\d/2,0) -- (0,0);
    \end{tikzpicture} \ .  
\end{aligned}
\end{equation}

  Graphically, injectivity reads as the map
  \begin{equation*}
    X \mapsto
    \begin{tikzpicture}
      \def\d{0.6};
      \draw[decorate, decoration=brace] (-0.1,0.5) -- (3*\d+0.1,0.5) node[midway,above] {$K$};
      \draw (-\d-0.5,0) rectangle (3*\d+0.5,-0.5);
      \node[tensor,label=below:$X$] at (-\d,0) {};
      \foreach \x in {0,1,3}{
        \draw (\x*\d,-0.4) -- (\x*\d,0.4);
        \node[mpo] (t) at (\x*\d,0) {};
      }
      \node[fill=white] at (\d*2,0) {$\dots$};
    \end{tikzpicture} \ .
  \end{equation*}
  being injective.
  
\subsection{The main tools for the finite size MPS setup}

In the following we repeatedly use the following lemma. For a proof, see \cite{MPSdec}.

\begin{lemma}\label{lem:main_mps_tool}
    Let $a\in \mathbb C^d\otimes \Mat_D$ be an MPS tensor that becomes injective after blocking and $b\in \mathbb C^d \otimes \Mat_{D'}$ be an MPS tensor such that for all $n\in\mathbb{N}$,
    \begin{equation}\label{eq:MPS_equal}
    \begin{tikzpicture}
      \def\d{0.6};
      \draw (-0.5,0) rectangle (3*\d+0.5,-0.5);
      \foreach \x in {0,1,3}{
        \draw (\x*\d,0) -- (\x*\d,0.4);
        \node[tensor,label={[label distance=-2pt]below:$a$}] (t) at (\x*\d,0) {};
      }
      \node[fill=white] at (\d*2,0) {$\dots$};
      \draw[decorate, decoration=brace] (-0.1,0.5) -- (3*\d+0.1,0.5) node[midway,above] {$n$};
    \end{tikzpicture}   =
    \begin{tikzpicture}
      \def\d{0.6};
      \draw (-0.5,0) rectangle (3*\d+0.5,-0.5);
      \foreach \x in {0,1,3}{
        \draw (\x*\d,0) -- (\x*\d,0.4);
        \node[tensor,label={[label distance=-2pt]below:$b$}] (t) at (\x*\d,0) {};
      }
      \node[fill=white] at (\d*2,0) {$\dots$};
      \draw[decorate, decoration=brace] (-0.1,0.5) -- (3*\d+0.1,0.5) node[midway,above] {$n$};
    \end{tikzpicture} \ .
    \end{equation}
    Then there is $V \in \Mat_{D,D'}$ and $W\in\Mat_{D',D}$ such that $VW=\mathbb{1}_D$ and for all $n\in \mathbb{N}$,
    \begin{equation}\label{eq:reduction}
    \begin{tikzpicture}
      \def\d{0.6};
      \draw (-\d-0.5,0) -- (4*\d+0.5,0);
      \foreach \x in {0,1,3}{
        \draw (\x*\d,0) -- (\x*\d,0.4);
        \node[tensor,label={[label distance=-2pt]below:$b$}] (t) at (\x*\d,0) {};
      }
      \node[tensor,label={[label distance=-1pt]below:$V$}] (t) at (-\d,0) {};
      \node[tensor,label={[label distance=-1pt]below:$W$}] (t) at (4*\d,0) {};
      \node[fill=white] at (\d*2,0) {$\dots$};
      \draw[decorate, decoration=brace] (-0.1,0.5) -- (3*\d+0.1,0.5) node[midway,above] {$n$};
    \end{tikzpicture}   =
    \begin{tikzpicture}
      \def\d{0.6};
      \draw (-0.5,0) -- (3*\d+0.5,0);
      \foreach \x in {0,1,3}{
        \draw (\x*\d,0) -- (\x*\d,0.4);
        \node[tensor,label={[label distance=-2pt]below:$a$}] (t) at (\x*\d,0) {};
      }
      \node[fill=white] at (\d*2,0) {$\dots$};
      \draw[decorate, decoration=brace] (-0.1,0.5) -- (3*\d+0.1,0.5) node[midway,above] {$n$};
    \end{tikzpicture}  \ .
    \end{equation}
    Moreover, if $a,b$ and $V,W$ are such that both Eq.~\eqref{eq:MPS_equal}  and Eq.~\eqref{eq:reduction} holds, then there is $M\in\mathbb{N}$, referred to as the \emph{nilpotency length} corresponding to $V$ and $W$,  such that for all  $m,k,n\in\mathbb{N}$, $m,k\geq M$,
    \begin{equation}\label{eq:VW_factorize}
    \begin{tikzpicture}
      \def\d{0.6};
      \draw (-0.4,0) -- (3*\d+0.4,0);
      \foreach \x in {0,1,3}{
        \draw (\x*\d,0) -- (\x*\d,0.4);
        \node[tensor,label={[label distance=-1pt]below:$b\vphantom{V}$}] (t) at (\x*\d,0) {};
      }
      \node[fill=white] at (\d*2,0) {$\dots$};
      \draw[decorate, decoration=brace] (-0.1,0.5) -- (3*\d+0.1,0.5) node[midway,above] {$n+m+k$};
    \end{tikzpicture}   =
    \begin{tikzpicture}
      \def\d{0.6};
      \draw (-0.4,0) -- (10*\d+0.4,0);
      \foreach \x/\l in {0/b,2/b,4/a,6/a,8/b,10/b}{
        \draw (\x*\d,0) -- (\x*\d,0.4);
        \node[tensor,label={[label distance=-1pt]below:$\l\vphantom{V}$}] (t) at (\x*\d,0) {};
      }
      \node[tensor,label={[label distance=-1pt]below:$W$}] (t) at (3*\d,0) {};
      \node[tensor,label={[label distance=-1pt]below:$V$}] (t) at (7*\d,0) {};
      \node[fill=white] at (\d,0) {$\dots$};
      \node[fill=white] at (\d*5,0) {$\dots$};
      \node[fill=white] at (\d*9,0) {$\dots$};
      \draw[decorate, decoration=brace] (-0.1,0.5) -- (2*\d+0.1,0.5) node[midway,above] {$m$};
      \draw[decorate, decoration=brace] (4*\d-0.1,0.5) -- (6*\d+0.1,0.5) node[midway,above] {$n$};
      \draw[decorate, decoration=brace] (8*\d-0.1,0.5) -- (10*\d+0.1,0.5) node[midway,above] {$k$};
    \end{tikzpicture} \ .
    \end{equation}
\end{lemma}
  
Note that as every MPU is an injective MPS, this lemma also applies to MPUs. We will also use the following simple consequence of injectivity:

  \begin{lemma}\label{lem:injectivity_factorize}
      Let $a\in\mathbb C^d \otimes \Mat_D$ be an MPS tensor that is injective after blocking $K$ tensors. Let $\mathcal{H}_L$ and $\mathcal{H}_R$ be arbitrary finite dimensional vector spaces and $v,w\in \mathbb \mathcal{H}_L \otimes \mathbb C^D$ and $u,z\in \mathcal{H}_R \otimes \mathbb C^D$ be vectors such that
      \begin{equation*}
          \sum_{\alpha\beta} v_{\alpha} \otimes a^{(n)}_{\alpha\beta}  \otimes u_{\beta} = \sum_{\alpha\beta} w_\alpha \otimes a^{(n)}_{\alpha\beta} \otimes z_{\beta}
      \end{equation*}
      holds for some $n\geq K$. Then $\exists \mu \in \mathbb C$, $\mu\neq 0$ such that
      \begin{equation*}
          v = \mu w \quad\text{and} \quad u = \mu^{-1} z.
      \end{equation*}
  \end{lemma}
  
  Graphically, the statement of the lemma reads as
  \begin{equation*}
      \begin{tikzpicture}[xscale=0.6,yscale=0.4]
          \draw (0,0) -- (4,0);
          \foreach \x/\s/\l in {0/{red}/v,1//a,3//a,4/{blue}/u}{
            \draw[\s] (\x,0) --++(0,1);
            \node[mpo] at (\x,0) {};
            \node[anchor=north,inner sep=5pt] at (\x,0) {$\l$};
          }
          \node[fill=white] at (2,0) {$\dots$};
      \end{tikzpicture} = 
      \begin{tikzpicture}[xscale=0.6,yscale=0.4]
          \draw (0,0) -- (4,0);
          \foreach \x/\s/\l in {0/{red}/w,1//a,3//a,4/{blue}/z}{
            \draw[\s] (\x,0) --++(0,1);
            \node[mpo] at (\x,0) {};
            \node[anchor=north,inner sep=5pt] at (\x,0) {$\l$};
          }
          \node[fill=white] at (2,0) {$\dots$};
      \end{tikzpicture} \quad \Rightarrow \quad v=\mu w \quad \text{and} \quad u = \mu^{-1} z. 
  \end{equation*}

  \begin{proof}
      Fix a basis on $\mathcal{H}_L$ and $\mathcal{H}_R$, and apply an element of the dual basis on each space to obtain
      \begin{equation*}
          \sum_{\alpha\beta} v_{\alpha}^iu_{\beta}^j \cdot a^{(n)}_{\alpha\beta}    = \sum_{\alpha\beta} w_\alpha^iz_{\beta}^j \cdot a^{(n)}_{\alpha\beta} . 
      \end{equation*}
      Due to injectivity of $a^{(n)}$, this implies
      \begin{equation*}
          v_{\alpha}^iu_{\beta}^j = w_\alpha^iz_{\beta}^j,
      \end{equation*}
      for all $i,j$ and $\alpha,\beta$. This is equivalent to the statement of the lemma. 
  \end{proof}

  This lemma lets us now prove a simple corollary of Lemma \ref{lem:main_mps_tool}:
  \begin{cor}\label{cor:zipper}
      Let $a\in \mathbb{C}^d \otimes \Mat_D$, $b\in \mathbb{C}^d \otimes \Mat_{D'}$, $V\in\Mat_{D,D'}$, $W\in\Mat_{D',D}$ and $M\in\mathbb N$ such that both Eq.~\eqref{eq:MPS_equal} and Eq.~\eqref{eq:reduction} holds for all $n\geq 0$. Then for all $k>M$,
      \begin{align*}
        \sum_{\alpha\beta} b^{(k)}_{\alpha\beta} \otimes Ve_{\alpha\beta} &= \sum_{\alpha \beta \gamma\delta} a_{\alpha\beta} \otimes b^{(k-1)}_{\gamma\delta} \otimes e_{\alpha\beta}Ve_{\gamma \delta}, \\  
        \sum_{\alpha\beta} b^{(k)}_{\alpha\beta} \otimes e_{\alpha\beta}W &= \sum_{\alpha \beta \gamma\delta}  b^{(k-1)}_{\alpha\beta} \otimes a_{\gamma\delta} \otimes e_{\alpha\beta}We_{\gamma \delta},
      \end{align*}
      or graphically,
    \begin{equation}\label{eq:reduction_size_independence}
    \begin{aligned}
    \begin{tikzpicture}
      \def\d{0.5};
      \draw (-\d-0.4,0) -- (3*\d+0.4,0);
      \foreach \x in {0,1,3}{
        \draw (\x*\d,0) -- (\x*\d,0.4);
        \node[tensor,label={[label distance=-1pt]below:$b\vphantom{V}$}] (t) at (\x*\d,0) {};
      }
      \node[tensor,label={[label distance=-1pt]below:$V$}] (t) at (-\d,0) {};
      \node[fill=white] at (\d*2,0) {$\dots$};
      \draw[decorate, decoration=brace] (-0.1,0.5) -- (3*\d+0.1,0.5) node[midway,above] {$k$};
    \end{tikzpicture}   =
    \begin{tikzpicture}
      \def\d{0.5};
      \draw (-\d-0.4,0) -- (3*\d+0.4,0);
      \foreach \x/\l in {-1/a,1/b,3/b}{
        \draw (\x*\d,0) -- (\x*\d,0.4);
        \node[tensor,label={[label distance=-1pt]below:$\l\vphantom{V}$}] (t) at (\x*\d,0) {};
      }
      \node[tensor,label={[label distance=-1pt]below:$V$}] (t) at (0,0) {};
      \node[fill=white] at (\d*2,0) {$\dots$};
      \draw[decorate, decoration=brace] (\d-0.1,0.5) -- (3*\d+0.1,0.5) node[midway,above] {$k-1$};
    \end{tikzpicture}, \\
    \begin{tikzpicture}
      \def\d{0.5};
      \draw (-0.4,0) -- (4*\d+0.4,0);
      \foreach \x in {0,2,3}{
        \draw (\x*\d,0) -- (\x*\d,0.4);
        \node[tensor,label={[label distance=-1pt]below:$b\vphantom{W}$}] (t) at (\x*\d,0) {};
      }
      \node[tensor,label={[label distance=-1pt]below:$W$}] (t) at (4*\d,0) {};
      \node[fill=white] at (\d,0) {$\dots$};
      \draw[decorate, decoration=brace] (-0.1,0.5) -- (3*\d+0.1,0.5) node[midway,above] {$k$};
    \end{tikzpicture}   = 
    \begin{tikzpicture}
      \def\d{0.5};
      \draw (-0.4,0) -- (4*\d+0.4,0);
      \foreach \x/\l in {0/b,2/b,4/a}{
        \draw (\x*\d,0) -- (\x*\d,0.4);
        \node[tensor,label={[label distance=-1pt]below:$\l\vphantom{W}$}] (t) at (\x*\d,0) {};
      }
      \node[tensor,label={[label distance=-1pt]below:$W$}] (t) at (3*\d,0) {};
      \node[fill=white] at (\d,0) {$\dots$};
      \draw[decorate, decoration=brace] (-0.1,0.5) -- (2*\d+0.1,0.5) node[midway,above] {$k-1$};
    \end{tikzpicture} \ .
    \end{aligned}
    \end{equation}
  \end{cor}

  \begin{proof}
    Let us apply Lemma~\ref{lem:main_mps_tool} to this situation. As $V$ and $W$ satisfy Eq.~\eqref{eq:reduction}, there is an $M\in\mathbb N$ such that Eq.~\eqref{eq:VW_factorize} holds for any $n$ and $m,k\geq M$. Let us consider this equation twice, once  with $m,n,k$ and once  with $m,n+1,k-1$ such that $m,k>M$ and $n\geq L$, where $L$ is the injectivity length of $a$. We obtain
    \begin{align*}
    &
    \begin{tikzpicture}
      \def\d{0.6};
      \draw (-0.4,0) -- (10*\d+0.4,0);
      \foreach \x/\l in {0/b,2/b,4/a,6/a,8/b,10/b}{
        \draw (\x*\d,0) -- (\x*\d,0.4);
        \node[tensor,label={[label distance=-1pt]below:$\l\vphantom{W}$}] (t) at (\x*\d,0) {};
      }
      \node[tensor,label={[label distance=-1pt]below:$W$}] (t) at (3*\d,0) {};
      \node[tensor,label={[label distance=-1pt]below:$V$}] (t) at (7*\d,0) {};
      \node[fill=white] at (\d,0) {$\dots$};
      \node[fill=white] at (\d*5,0) {$\dots$};
      \node[fill=white] at (\d*9,0) {$\dots$};
      \draw[decorate, decoration=brace] (-0.1,0.5) -- (2*\d+0.1,0.5) node[midway,above] {$m$};
      \draw[decorate, decoration=brace] (4*\d-0.1,0.5) -- (6*\d+0.1,0.5) node[midway,above] {$n$};
      \draw[decorate, decoration=brace] (8*\d-0.1,0.5) -- (10*\d+0.1,0.5) node[midway,above] {$k$};
    \end{tikzpicture} \ =\\&
    \begin{tikzpicture}
      \def\d{0.6};
      \draw (-0.4,0) -- (11*\d+0.4,0);
      \foreach \x/\l in {0/b,2/b,4/a,6/a,7/a,9/b,11/b}{
        \draw (\x*\d,0) -- (\x*\d,0.4);
        \node[tensor,label={[label distance=-1pt]below:$\l\vphantom{V}$}] (t) at (\x*\d,0) {};
      }
      \node[tensor,label={[label distance=-1pt]below:$W$}] (t) at (3*\d,0) {};
      \node[tensor,label={[label distance=-1pt]below:$V$}] (t) at (8*\d,0) {};
      \node[fill=white] at (\d,0) {$\dots$};
      \node[fill=white] at (\d*5,0) {$\dots$};
      \node[fill=white] at (\d*10,0) {$\dots$};
      \draw[decorate, decoration=brace] (-0.1,0.5) -- (2*\d+0.1,0.5) node[midway,above] {$m$};
      \draw[decorate, decoration=brace] (4*\d-0.1,0.5) -- (6*\d+0.1,0.5) node[midway,above] {$n$};
      \draw[decorate, decoration=brace] (9*\d-0.1,0.5) -- (11*\d+0.1,0.5) node[midway,above] {$k-1$};
    \end{tikzpicture} \ .
    \end{align*}
  Given that $n$ that is larger than the injectivity length of the tensor $a$, using Lemma~\ref{lem:injectivity_factorize} and recognizing that the two tensors on the left side of $a^{(n)}$ are the same in both sides of the equation, we conclude that for all $k>M$ 
    \begin{equation*}
    \begin{tikzpicture}
      \def\d{0.5};
      \draw (-\d-0.4,0) -- (3*\d+0.4,0);
      \foreach \x in {0,1,3}{
        \draw (\x*\d,0) -- (\x*\d,0.4);
        \node[tensor,label={[label distance=-1pt]below:$b\vphantom{V}$}] (t) at (\x*\d,0) {};
      }
      \node[tensor,label={[label distance=-1pt]below:$V$}] (t) at (-\d,0) {};
      \node[fill=white] at (\d*2,0) {$\dots$};
      \draw[decorate, decoration=brace] (-0.1,0.5) -- (3*\d+0.1,0.5) node[midway,above] {$k$};
    \end{tikzpicture}   = 
    \begin{tikzpicture}
      \def\d{0.5};
      \draw (-\d-0.4,0) -- (3*\d+0.4,0);
      \foreach \x/\l in {-1/a,1/b,3/b}{
        \draw (\x*\d,0) -- (\x*\d,0.4);
        \node[tensor,label={[label distance=-1pt]below:$\l\vphantom{V}$}] (t) at (\x*\d,0) {};
      }
      \node[tensor,label={[label distance=-1pt]below:$V$}] (t) at (0,0) {};
      \node[fill=white] at (\d*2,0) {$\dots$};
      \draw[decorate, decoration=brace] (\d-0.1,0.5) -- (3*\d+0.1,0.5) node[midway,above] {$k-1$};
    \end{tikzpicture}. 
    \end{equation*}
    The other equation is obtained similarly, by changing $m$ instead of $k$.
  \end{proof}

  Finally we also obtain that, while  $V$ and $W$ are not unique, different choices can be compared to each other:
  \begin{lemma}\label{lem:fusion_unique}
      Let $a,b$ be two MPS tensors and $V,W$ and $\hat V, \hat W$ be two pairs of operators such that $VW=\hat V \hat W=\mathbb 1$ and such that both Eq.~\eqref{eq:MPS_equal} and Eq.~\eqref{eq:reduction} holds for all $m,k>M$ and $n$. Then $\exists \lambda\in\mathbb C$, $\lambda\neq 0$ such that $\forall m\geq M$
  \begin{equation}\label{eq:reduction_uniqueness}
  \begin{aligned}
  \begin{tikzpicture}
    \def\d{0.5};
    \draw (-\d-0.4,0) -- (3*\d+0.4,0);
    \foreach \x in {0,1,3}{
      \draw (\x*\d,0) -- (\x*\d,0.4);
      \node[tensor,label={[label distance=-1pt]below:$b\vphantom{V}$}] (t) at (\x*\d,0) {};
    }
    \node[tensor,label={[label distance=-1pt]below:$V$}] (t) at (-\d,0) {};
    \node[fill=white] at (\d*2,0) {$\dots$};
    \draw[decorate, decoration=brace] (-0.1,0.5) -- (3*\d+0.1,0.5) node[midway,above] {$m$};
  \end{tikzpicture}   &= \lambda\cdot
  \begin{tikzpicture}
    \def\d{0.5};
    \draw (-\d-0.4,0) -- (3*\d+0.4,0);
    \foreach \x in {0,1,3}{
      \draw (\x*\d,0) -- (\x*\d,0.4);
      \node[tensor,label={[label distance=-1pt]below:$b\vphantom{\hat V}$}] (t) at (\x*\d,0) {};
    }
    \node[tensor,label={[label distance=-1pt]below:$\hat{V}$}] (t) at (-\d,0) {};
    \node[fill=white] at (\d*2,0) {$\dots$};
    \draw[decorate, decoration=brace] (-0.1,0.5) -- (3*\d+0.1,0.5) node[midway,above] {$m$};
  \end{tikzpicture}, \\
  \begin{tikzpicture}
    \def\d{0.5};
    \draw (-0.4,0) -- (4*\d+0.4,0);
    \foreach \x in {0,2,3}{
      \draw (\x*\d,0) -- (\x*\d,0.4);
      \node[tensor,label={[label distance=-1pt]below:$b\vphantom{W}$}] (t) at (\x*\d,0) {};
    }
    \node[tensor,label={[label distance=-1pt]below:$W$}] (t) at (4*\d,0) {};
    \node[fill=white] at (\d,0) {$\dots$};
    \draw[decorate, decoration=brace] (-0.1,0.5) -- (3*\d+0.1,0.5) node[midway,above] {$m$};
  \end{tikzpicture}   &= \frac{1}{\lambda} \cdot
  \begin{tikzpicture}
    \def\d{0.5};
    \draw (-0.4,0) -- (4*\d+0.4,0);
    \foreach \x in {0,2,3}{
      \draw (\x*\d,0) -- (\x*\d,0.4);
      \node[tensor,label={[label distance=-1pt]below:$b\vphantom{\hat W}$}] (t) at (\x*\d,0) {};
    }
    \node[tensor,label={[label distance=-1pt]below:$\hat{W}$}] (t) at (4*\d,0) {};
    \node[fill=white] at (\d,0) {$\dots$};
    \draw[decorate, decoration=brace] (-0.1,0.5) -- (3*\d+0.1,0.5) node[midway,above] {$m$};
  \end{tikzpicture} \ .
  \end{aligned}
  \end{equation}
  \end{lemma}

\begin{proof}
    As both of the pairs $(V,W)$ and $(\hat V, \hat W)$ satisfy Eq.~\eqref{eq:reduction}, there is an $M\in\mathbb N$ such that Eq.~\eqref{eq:VW_factorize} holds for any $n$ and $m,k\geq M$, for both $(V,W)$ and $(\hat V, \hat W)$. Therefore
    \begin{align*}
    &
    \begin{tikzpicture}
      \def\d{0.6};
      \draw (-0.4,0) -- (10*\d+0.4,0);
      \foreach \x/\l in {0/b,2/b,4/a,6/a,8/b,10/b}{
        \draw (\x*\d,0) -- (\x*\d,0.4);
        \node[tensor,label={[label distance=-1pt]below:$\l\vphantom{V}$}] (t) at (\x*\d,0) {};
      }
      \node[tensor,label={[label distance=-1pt]below:$W$}] (t) at (3*\d,0) {};
      \node[tensor,label={[label distance=-1pt]below:$V$}] (t) at (7*\d,0) {};
      \node[fill=white] at (\d,0) {$\dots$};
      \node[fill=white] at (\d*5,0) {$\dots$};
      \node[fill=white] at (\d*9,0) {$\dots$};
      \draw[decorate, decoration=brace] (-0.1,0.5) -- (2*\d+0.1,0.5) node[midway,above] {$m$};
      \draw[decorate, decoration=brace] (4*\d-0.1,0.5) -- (6*\d+0.1,0.5) node[midway,above] {$n$};
      \draw[decorate, decoration=brace] (8*\d-0.1,0.5) -- (10*\d+0.1,0.5) node[midway,above] {$k$};
    \end{tikzpicture} \  = \\ &
    \begin{tikzpicture}
      \def\d{0.6};
      \draw (-0.4,0) -- (10*\d+0.4,0);
      \foreach \x/\l in {0/b,2/b,4/a,6/a,8/b,10/b}{
        \draw (\x*\d,0) -- (\x*\d,0.4);
        \node[tensor,label={[label distance=-1pt]below:$\l\vphantom{\hat V}$}] (t) at (\x*\d,0) {};
      }
      \node[tensor,label={[label distance=-1pt]below:$\hat W$}] (t) at (3*\d,0) {};
      \node[tensor,label={[label distance=-1pt]below:$\hat V$}] (t) at (7*\d,0) {};
      \node[fill=white] at (\d,0) {$\dots$};
      \node[fill=white] at (\d*5,0) {$\dots$};
      \node[fill=white] at (\d*9,0) {$\dots$};
      \draw[decorate, decoration=brace] (-0.1,0.5) -- (2*\d+0.1,0.5) node[midway,above] {$m$};
      \draw[decorate, decoration=brace] (4*\d-0.1,0.5) -- (6*\d+0.1,0.5) node[midway,above] {$n$};
      \draw[decorate, decoration=brace] (8*\d-0.1,0.5) -- (10*\d+0.1,0.5) node[midway,above] {$k$};
    \end{tikzpicture} \ .
    \end{align*}
    This equation is true with $n$ that is larger than the injectivity length of $a$, and thus using Lemma~\ref{lem:injectivity_factorize} leads to
  \begin{align*}
  \begin{tikzpicture}
    \def\d{0.5};
    \draw (-\d-0.4,0) -- (3*\d+0.4,0);
    \foreach \x in {0,1,3}{
      \draw (\x*\d,0) -- (\x*\d,0.4);
      \node[tensor,label={[label distance=-1pt]below:$b\vphantom{V}$}] (t) at (\x*\d,0) {};
    }
    \node[tensor,label={[label distance=-1pt]below:$V$}] (t) at (-\d,0) {};
    \node[fill=white] at (\d*2,0) {$\dots$};
    \draw[decorate, decoration=brace] (-0.1,0.5) -- (3*\d+0.1,0.5) node[midway,above] {$k$};
  \end{tikzpicture}   &= \lambda_{m,k}\cdot
  \begin{tikzpicture}
    \def\d{0.5};
    \draw (-\d-0.4,0) -- (3*\d+0.4,0);
    \foreach \x in {0,1,3}{
      \draw (\x*\d,0) -- (\x*\d,0.4);
      \node[tensor,label={[label distance=-1pt]below:$b\vphantom{\hat V}$}] (t) at (\x*\d,0) {};
    }
    \node[tensor,label={[label distance=-1pt]below:$\hat{V}$}] (t) at (-\d,0) {};
    \node[fill=white] at (\d*2,0) {$\dots$};
    \draw[decorate, decoration=brace] (-0.1,0.5) -- (3*\d+0.1,0.5) node[midway,above] {$k$};
  \end{tikzpicture}, \\
  \begin{tikzpicture}
    \def\d{0.5};
    \draw (-0.4,0) -- (4*\d+0.4,0);
    \foreach \x in {0,2,3}{
      \draw (\x*\d,0) -- (\x*\d,0.4);
      \node[tensor,label={[label distance=-1pt]below:$b\vphantom{W}$}] (t) at (\x*\d,0) {};
    }
    \node[tensor,label={[label distance=-1pt]below:$W$}] (t) at (4*\d,0) {};
    \node[fill=white] at (\d,0) {$\dots$};
    \draw[decorate, decoration=brace] (-0.1,0.5) -- (3*\d+0.1,0.5) node[midway,above] {$m$};
  \end{tikzpicture}   &= \frac{1}{\lambda_{m,k}} \cdot
  \begin{tikzpicture}
    \def\d{0.5};
    \draw (-0.4,0) -- (4*\d+0.4,0);
    \foreach \x in {0,2,3}{
      \draw (\x*\d,0) -- (\x*\d,0.4);
      \node[tensor,label={[label distance=-1pt]below:$b\vphantom{\hat W}$}] (t) at (\x*\d,0) {};
    }
    \node[tensor,label={[label distance=-1pt]below:$\hat{W}$}] (t) at (4*\d,0) {};
    \node[fill=white] at (\d,0) {$\dots$};
    \draw[decorate, decoration=brace] (-0.1,0.5) -- (3*\d+0.1,0.5) node[midway,above] {$m$};
  \end{tikzpicture} \ .
  \end{align*}
  The first equation implies that $\lambda_{m,k}$ in independent of $m$, while the second implies that $\lambda_{m,k}$ is independent of $k$, and thus $\lambda_{m,k}$ is independent from both $m$ and $k$.
\end{proof}

\subsection{Connection between the $3$-cocycle index of finite MPU and the one in Section \ref{betas}}

In this subsection we show that the $3$-cocycle index derived for $\beta_g$ in Section\ref{betas} coincides with the one of the finite size MPU when the $\beta_g$ is defined as a ``conjugation with an MPU representation of $g$'' (finite state automaton). We first defined the index corresponding to the finite size MPU representation of a finite group $G$. Then we define $\beta_g$ rigorously using MPUs and we show that $\beta_g$ defined in this way satisfies the local decomposable assumption. We finally prove that the $3$-cocycle obtained from this calculation coincides with the cocycle obtained from the finite size MPU representation of the group $G$.

\subsubsection{Third cohomology index of the finite size MPU representations}

  We note that the index measuring translations defined in Ref.\ \cite{MPUdef} for MPU representations of a finite group is trivial since the index is additive under MPU multiplication so that $0={\rm ind}(U_e)={\rm ind}(U^{|G|}_g) = |G|\times {\rm ind}(U_g)$.

Let us use Lemma \ref{lem:main_mps_tool} for the MPU representation of the finite group $G$. The equation $U_gU_h = U_{gh}$ guarantees that there exists a pair of rank-three tensors $V(g,h)$ and $W(g,h)$ such that $V(g,h) W(g,h) = \mathbb{1}_{gh}$ and such that for all $n\in\mathbb N$,
  \begin{equation}\label{eq:fusion_tensors}
    \begin{tikzpicture}[baseline=1mm]
      \def\d{0.6};
      \draw (-0.5,0) -- (2*\d+0.5,0);
      \draw (-0.5,0.4) -- (2*\d+0.5,0.4);
      \foreach \x in {0,2}{
        \draw (\x*\d,-0.4) -- (\x*\d,0.8);
        \node[mpo] (t) at (\x*\d,0) {};
        \node[mpo] (t) at (\x*\d,0.4) {};
      }
      \node[fill=white] at (\d*1,0) {$\dots$};
      \node[fill=white] at (\d*1,0.4) {$\dots$};
      \draw[fusion] (-0.5,0)--++(0,0.4);
      \draw[fusion] (2*\d+0.5,0)--++(0,0.4);
      \draw (-0.5,0.2)--++(-0.5,0);
      \draw (2*\d+0.5,0.2)--++(0.5,0);
      \node[irrep] at (-0.25,0) {$h$};
      \node[irrep] at (-0.25,0.4) {$g$};
      \node[irrep] at (-0.75,0.2) {$gh$};
      \node[irrep] at (2*\d+0.25,0) {$h$};
      \node[irrep] at (2*\d+0.25,0.4) {$g$};
      \node[irrep] at (2*\d+0.75,0.2) {$gh$};
    \draw[decorate, decoration=brace] (-0.1,0.9) -- (2*\d+0.1,0.9) node[midway,above] {$n$};
    \end{tikzpicture} \  = \
    \begin{tikzpicture}
      \def\d{0.6};
      \draw (-0.5,0) -- (2*\d+0.5,0);
      \foreach \x in {0,2}{
        \draw (\x*\d,-0.4) -- (\x*\d,0.4);
        \node[mpo] (t) at (\x*\d,0) {};
      }
      \node[fill=white] at (\d*1,0) {$\dots$};
      \node[irrep] at (-0.25,0) {$gh$};representations
      \node[irrep] at (2*\d+0.25,0) {$gh$};
    \draw[decorate, decoration=brace] (-0.1,0.5) -- (2*\d+0.1,0.5) node[midway,above] {$n$};
    \end{tikzpicture} \ .
  \end{equation}
From now on, let us fix such a pair of tensor for each $g,h\in G$. We call these tensors \emph{fusion tensors}. The nilpotency length corresponding to these operators is $M(g,h)$ and we will write $M=\max_{g,h} M(g,h)$. Considering the product of three group elements, we observe that
  \begin{equation*}
    \begin{tikzpicture}[baseline=4mm]
      \def\d{0.6};
      \draw (-0.5,0) -- (2*\d+0.5,0);
      \draw (-0.5,0.4) -- (2*\d+0.5,0.4);
      \draw (-1,0.8) -- (2*\d+1,0.8);
      \foreach \x in {0,2}{
        \draw (\x*\d,-0.4) -- (\x*\d,1.2);
        \node[mpo] (t) at (\x*\d,0) {};
        \node[mpo] (t) at (\x*\d,0.4) {};
        \node[mpo] (t) at (\x*\d,0.8) {};
      }
      \node[fill=white] at (\d*1,0) {$\dots$};
      \node[fill=white] at (\d*1,0.4) {$\dots$};
      \node[fill=white] at (\d*1,0.8) {$\dots$};
      \draw[fusion] (-0.5,0)--++(0,0.4);
      \draw (-0.5,0.2)--++(-0.5,0);
      \draw[fusion] (2*\d+0.5,0)--++(0,0.4);
      \draw (2*\d+0.5,0.2)--++(0.5,0);
      \draw[fusion] (-1,0.2)--++(0,0.6);
      \draw (-1,0.5)--++(-0.5,0);
      \draw[fusion] (2*\d+1,0.2)--++(0,0.6);
      \draw (2*\d+1,0.5)--++(0.5,0);
      \node[irrep] at (-0.25,0.8) {$g$};
      \node[irrep] at (-0.25,0.4) {$h$};
      \node[irrep] at (-0.25,0) {$k$};
      \node[irrep] at (-0.75,0.2) {$hk$};
      \node[irrep] at (-1.3,0.5) {$ghk$};
      \node[irrep] at (2*\d+0.25,0.8) {$g$};
      \node[irrep] at (2*\d+0.25,0.4) {$h$};
      \node[irrep] at (2*\d+0.25,0) {$k$};
      \node[irrep] at (2*\d+0.75,0.2) {$hk$};
      \node[irrep] at (2*\d+1.3,0.5) {$ghk$};
      \draw[decorate, decoration=brace] (-0.1,1.3) -- (2*\d+0.1,1.3) node[midway,above] {$n$};
    \end{tikzpicture} \  = \
    \begin{tikzpicture}
      \def\d{0.6};
      \draw (-0.5,0) -- (2*\d+0.5,0);
      \foreach \x in {0,2}{
        \draw (\x*\d,-0.4) -- (\x*\d,0.4);
        \node[mpo] (t) at (\x*\d,0) {};
      }
      \node[fill=white] at (\d*1,0) {$\dots$};
      \node[irrep] at (-0.25,0) {$ghk$};
      \node[irrep] at (2*\d+0.25,0) {$ghk$};
      \draw[decorate, decoration=brace] (-0.1,0.5) -- (2*\d+0.1,0.5) node[midway,above] {$n$};
    \end{tikzpicture} \ = \
    \begin{tikzpicture}[baseline=2mm]
      \def\d{0.6};
      \draw (-1,0) -- (2*\d+1,0);
      \draw (-0.5,0.4) -- (2*\d+0.5,0.4);
      \draw (-0.5,0.8) -- (2*\d+0.5,0.8);
      \foreach \x in {0,2}{
        \draw (\x*\d,-0.4) -- (\x*\d,1.2);
        \node[mpo] (t) at (\x*\d,0) {};
        \node[mpo] (t) at (\x*\d,0.4) {};
        \node[mpo] (t) at (\x*\d,0.8) {};
      }
      \node[fill=white] at (\d*1,0) {$\dots$};
      \node[fill=white] at (\d*1,0.4) {$\dots$};
      \node[fill=white] at (\d*1,0.8) {$\dots$};
      \draw[fusion] (-0.5,0.4)--++(0,0.4);
      \draw (-0.5,0.6)--++(-0.5,0);
      \draw[fusion] (2*\d+0.5,0.4)--++(0,0.4);
      \draw (2*\d+0.5,0.6)--++(0.5,0);
      \draw[fusion] (-1,0)--++(0,0.6);
      \draw (-1,0.3)--++(-0.5,0);
      \draw[fusion] (2*\d+1,0)--++(0,0.6);
      \draw (2*\d+1,0.3)--++(0.5,0);
      \foreach \x in {-0.25,2*\d+0.35}{
        \node[irrep] at (\x,0.8) {$g$};
        \node[irrep] at (\x,0.4) {$h$};
        \node[irrep] at (\x,0) {$k$};
      }
      \node[irrep] at (-0.75,0.6) {$gh$};
      \node[irrep] at (-1.3,0.3) {$ghk$};
      \node[irrep] at (2*\d+0.75,0.6) {$gh$};
      \node[irrep] at (2*\d+1.3,0.3) {$ghk$};
      \draw[decorate, decoration=brace] (-0.1,1.3) -- (2*\d+0.1,1.3) node[midway,above] {$n$};
    \end{tikzpicture} \  .
  \end{equation*}
These equations hold for $n=0$ as well in the sense that
  \begin{equation*}
    \begin{tikzpicture}[baseline=4mm]
      \draw (-0.5,0) -- (0.5,0);
      \draw (-0.5,0.4) -- (0.5,0.4);
      \draw (-1,0.8) -- (1,0.8);
      \draw[fusion] (-0.5,0)--++(0,0.4);
      \draw (-0.5,0.2)--++(-0.5,0);
      \draw[fusion] (0.5,0)--++(0,0.4);
      \draw (0.5,0.2)--++(0.5,0);
      \draw[fusion] (-1,0.2)--++(0,0.6);
      \draw (-1,0.5)--++(-0.5,0);
      \draw[fusion] (1,0.2)--++(0,0.6);
      \draw (1,0.5)--++(0.5,0);
      \node[irrep] at (-0.25,0.8) {$g$};
      \node[irrep] at (-0.25,0.4) {$h$};
      \node[irrep] at (-0.25,0) {$k$};
      \node[irrep] at (-0.75,0.2) {$hk$};
      \node[irrep] at (-1.3,0.5) {$ghk$};
      \node[irrep] at (0.25,0.8) {$g$};
      \node[irrep] at (0.25,0.4) {$h$};
      \node[irrep] at (0.25,0) {$k$};
      \node[irrep] at (0.75,0.2) {$hk$};
      \node[irrep] at (1.3,0.5) {$ghk$};
    \end{tikzpicture} \  = \mathbb{1}_{ghk} = \
    \begin{tikzpicture}[baseline=2mm]
      \draw (-1,0) -- (1,0);
      \draw (-0.5,0.4) -- (0.5,0.4);
      \draw (-0.5,0.8) -- (0.5,0.8);
      \draw[fusion] (-0.5,0.4)--++(0,0.4);
      \draw (-0.5,0.6)--++(-0.5,0);
      \draw[fusion] (0.5,0.4)--++(0,0.4);
      \draw (0.5,0.6)--++(0.5,0);
      \draw[fusion] (-1,0)--++(0,0.6);
      \draw (-1,0.3)--++(-0.5,0);
      \draw[fusion] (1,0)--++(0,0.6);
      \draw (1,0.3)--++(0.5,0);
      \foreach \x in {-0.25,0.35}{
        \node[irrep] at (\x,0.8) {$g$};
        \node[irrep] at (\x,0.4) {$h$};
        \node[irrep] at (\x,0) {$k$};
      }
      \node[irrep] at (-0.75,0.6) {$gh$};
      \node[irrep] at (-1.3,0.3) {$ghk$};
      \node[irrep] at (0.75,0.6) {$gh$};
      \node[irrep] at (1.3,0.3) {$ghk$};
    \end{tikzpicture} \  .
  \end{equation*}
Using Lemma \ref{lem:fusion_unique}, there is $\omega(g,h,k)\in\mathbb{C}$ and $M\in \mathbb{N}$ such that for all $m>M$
  \begin{equation}\label{eq:mps_cocycle_def}
  \begin{aligned}
    \begin{tikzpicture}[baseline=4mm]
      \def\d{0.6};
      \draw (-0.5,0) -- (2*\d+0.5,0);
      \draw (-0.5,0.4) -- (2*\d+0.5,0.4);
      \draw (-0.5,0.8) -- (2*\d+1,0.8);
      \foreach \x in {0,2}{
        \draw (\x*\d,-0.4) -- (\x*\d,1.2);
        \node[mpo] (t) at (\x*\d,0) {};
        \node[mpo] (t) at (\x*\d,0.4) {};
        \node[mpo] (t) at (\x*\d,0.8) {};
      }
      \node[fill=white] at (\d*1,0) {$\dots$};
      \node[fill=white] at (\d*1,0.4) {$\dots$};
      \node[fill=white] at (\d*1,0.8) {$\dots$};
      \draw[fusion] (2*\d+0.5,0)--++(0,0.4);
      \draw (2*\d+0.5,0.2)--++(0.5,0);
      \draw[fusion] (2*\d+1,0.2)--++(0,0.6);
      \draw (2*\d+1,0.5)--++(0.5,0);
      \node[irrep] at (-0.25,0.8) {$g$};
      \node[irrep] at (-0.25,0.4) {$h$};
      \node[irrep] at (-0.25,0) {$k$};
      \node[irrep] at (2*\d+0.25,0.8) {$g$};
      \node[irrep] at (2*\d+0.25,0.4) {$h$};
      \node[irrep] at (2*\d+0.25,0) {$k$};
      \node[irrep] at (2*\d+0.75,0.2) {$hk$};
      \node[irrep] at (2*\d+1.3,0.5) {$ghk$};
      \draw[decorate, decoration=brace] (-0.1,1.3) -- (2*\d+0.1,1.3) node[midway,above] {$m$};
    \end{tikzpicture} \  = \ \omega(g,h,k) \cdot
    \begin{tikzpicture}[baseline=2mm]
      \def\d{0.6};
      \draw (-0.5,0) -- (2*\d+1,0);
      \draw (-0.5,0.4) -- (2*\d+0.5,0.4);
      \draw (-0.5,0.8) -- (2*\d+0.5,0.8);
      \foreach \x in {0,2}{
        \draw (\x*\d,-0.4) -- (\x*\d,1.2);
        \node[mpo] (t) at (\x*\d,0) {};
        \node[mpo] (t) at (\x*\d,0.4) {};
        \node[mpo] (t) at (\x*\d,0.8) {};
      }
      \node[fill=white] at (\d*1,0) {$\dots$};
      \node[fill=white] at (\d*1,0.4) {$\dots$};
      \node[fill=white] at (\d*1,0.8) {$\dots$};
      \draw[fusion] (2*\d+0.5,0.4)--++(0,0.4);
      \draw (2*\d+0.5,0.6)--++(0.5,0);
      \draw[fusion] (2*\d+1,0)--++(0,0.6);
      \draw (2*\d+1,0.3)--++(0.5,0);
      \node[irrep] at (-0.25,0.8) {$g$};
      \node[irrep] at (-0.25,0.4) {$h$};
      \node[irrep] at (-0.25,0) {$k$};
      \node[irrep] at (2*\d+0.25,0.8) {$g$};
      \node[irrep] at (2*\d+0.25,0.4) {$h$};
      \node[irrep] at (2*\d+0.25,0) {$k$};
      \node[irrep] at (2*\d+0.75,0.6) {$gh$};
      \node[irrep] at (2*\d+1.3,0.3) {$ghk$};
      \draw[decorate, decoration=brace] (-0.1,1.3) -- (2*\d+0.1,1.3) node[midway,above] {$m$};
    \end{tikzpicture}, \\
    \begin{tikzpicture}[baseline=4mm,xscale=-1]
      \def\d{0.6};
      \draw (-0.5,0) -- (2*\d+0.5,0);
      \draw (-0.5,0.4) -- (2*\d+0.5,0.4);
      \draw (-0.5,0.8) -- (2*\d+1,0.8);
      \foreach \x in {0,2}{
        \draw (\x*\d,-0.4) -- (\x*\d,1.2);
        \node[mpo] (t) at (\x*\d,0) {};
        \node[mpo] (t) at (\x*\d,0.4) {};
        \node[mpo] (t) at (\x*\d,0.8) {};
      }
      \node[fill=white] at (\d*1,0) {$\dots$};
      \node[fill=white] at (\d*1,0.4) {$\dots$};
      \node[fill=white] at (\d*1,0.8) {$\dots$};
      \draw[fusion] (2*\d+0.5,0)--++(0,0.4);
      \draw (2*\d+0.5,0.2)--++(0.5,0);
      \draw[fusion] (2*\d+1,0.2)--++(0,0.6);
      \draw (2*\d+1,0.5)--++(0.5,0);
      \node[irrep] at (-0.25,0.8) {$g$};
      \node[irrep] at (-0.25,0.4) {$h$};
      \node[irrep] at (-0.25,0) {$k$};
      \node[irrep] at (2*\d+0.25,0.8) {$g$};
      \node[irrep] at (2*\d+0.25,0.4) {$h$};
      \node[irrep] at (2*\d+0.25,0) {$k$};
      \node[irrep] at (2*\d+0.75,0.2) {$hk$};
      \node[irrep] at (2*\d+1.3,0.5) {$ghk$};
      \draw[decorate, decoration={brace,mirror}] (-0.1,1.3) -- (2*\d+0.1,1.3) node[midway,above] {$m$};
    \end{tikzpicture} \  = \ \frac{1}{\omega(g,h,k)} \cdot
    \begin{tikzpicture}[baseline=2mm,xscale=-1]
      \def\d{0.6};
      \draw (-0.5,0) -- (2*\d+1,0);
      \draw (-0.5,0.4) -- (2*\d+0.5,0.4);
      \draw (-0.5,0.8) -- (2*\d+0.5,0.8);
      \foreach \x in {0,2}{
        \draw (\x*\d,-0.4) -- (\x*\d,1.2);
        \node[mpo] (t) at (\x*\d,0) {};
        \node[mpo] (t) at (\x*\d,0.4) {};
        \node[mpo] (t) at (\x*\d,0.8) {};
      }
      \node[fill=white] at (\d*1,0) {$\dots$};
      \node[fill=white] at (\d*1,0.4) {$\dots$};
      \node[fill=white] at (\d*1,0.8) {$\dots$};
      \draw[fusion] (2*\d+0.5,0.4)--++(0,0.4);
      \draw (2*\d+0.5,0.6)--++(0.5,0);
      \draw[fusion] (2*\d+1,0)--++(0,0.6);
      \draw (2*\d+1,0.3)--++(0.5,0);
      \node[irrep] at (-0.25,0.8) {$g$};
      \node[irrep] at (-0.25,0.4) {$h$};
      \node[irrep] at (-0.25,0) {$k$};
      \node[irrep] at (2*\d+0.25,0.8) {$g$};
      \node[irrep] at (2*\d+0.25,0.4) {$h$};
      \node[irrep] at (2*\d+0.25,0) {$k$};
      \node[irrep] at (2*\d+0.75,0.6) {$gh$};
      \node[irrep] at (2*\d+1.3,0.3) {$ghk$};
      \draw[decorate, decoration={brace,mirror}] (-0.1,1.3) -- (2*\d+0.1,1.3) node[midway,above] {$m$};
    \end{tikzpicture} \  .
  \end{aligned}
  \end{equation}
  One can show (see for example \cite{2dSPT}) that $\omega$ satisfies the 3-cocycle condition Eq.~\eqref{eq:3-cocycle}. Different choices of $V(g,h)$ and $W(g,h)$ lead to different $\omega$, but all of these $\omega$ are related to each other through a 2-coboundary. For instance, as $U_n(1)=\id^{\otimes n}$ is a rank-one MPO, the tensors $V(g,1),V(1,g),W(g,1),W(1,g)$ are all rank-two instead of rank three. In fact, one can choose them such that
  \begin{equation*}
    V(g,1)=W(g,1) = V(1,g) = W(1,g) = \id_{D_g}.
  \end{equation*}
  With this choice,
  \begin{equation*}
    \omega(g,h,1) = \omega(g,1,h) = \omega(1,g,h) = 1.
  \end{equation*}

\subsubsection{ Defining \texorpdfstring{$\beta_g$}{beta\_g} from MPUs}

In this section we define an automorphism group, its elements denoted by $\beta_g$ ($g\in G$),  of the quasi-local observables using the an MPU representation of a finite group $G$. We define $\beta_g$ first on local observables then, we show that it is an automorphism, that it is norm-contractive and also that $\beta_g\beta_h = \beta_{gh}$. Finally we extend $\beta_g$ to the whole set of quasi-local observables.

Let us recall that given an MPU representation of a finite group $G$ we have obtained certain rank-three tensors $V(g,h)$ and $W(g,h)$, called fusion tensors, such that \eqref{eq:fusion_tensors} holds. If $h=g^{-1}$, then $U_n(gh)=\id^{\otimes n}$ is an MPO with bond dimension one, and thus $V(g,g^{-1})$ and $W(g,g^{-1})$ are rank-two tensors instead of rank-three. These tensors satisfy the equation (see \eqref{eq:reduction})
\begin{equation}
  \begin{tikzpicture}[baseline=1mm]
    \def\d{0.6};
    \draw (-0.6,0) -- (2*\d+0.6,0);
    \draw (-0.6,0.4) -- (2*\d+0.6,0.4);
    \foreach \x in {0,2}{
      \draw (\x*\d,-0.4) -- (\x*\d,0.8);
      \node[mpo] (t) at (\x*\d,0) {};
      \node[mpo] (t) at (\x*\d,0.4) {};
    }
    \node[fill=white] at (\d*1,0) {$\dots$};
    \node[fill=white] at (\d*1,0.4) {$\dots$};
    \draw[fusion] (-0.6,0)--++(0,0.4);
    \draw[fusion] (2*\d+0.6,0)--++(0,0.4);
    \node[irrep] at (-0.25,0) {$g^{-1}$};
    \node[irrep] at (-0.25,0.4) {$g$};
    \node[irrep] at (2*\d+0.3,0) {$g^{-1}$};
    \node[irrep] at (2*\d+0.3,0.4) {$g$};
    \draw[decorate, decoration=brace] (-0.1,0.9) -- (2*\d+0.1,0.9) node[midway,above] {$n$};
  \end{tikzpicture} \  = \id^{\otimes n}
\end{equation}
for all $n\in \mathbb{N}$, as well as the equation
\begin{equation}\label{eq:simple_VW_factorize}
  \begin{tikzpicture}[baseline=1mm]
    \def\d{0.6};
    \draw (-0.6,0) -- (2*\d+0.6,0);
    \draw (-0.6,0.4) -- (2*\d+0.6,0.4);
    \foreach \x in {0,2}{
      \draw (\x*\d,-0.4) -- (\x*\d,0.8);
      \node[mpo] (t) at (\x*\d,0) {};
      \node[mpo] (t) at (\x*\d,0.4) {};
    }
    \node[fill=white] at (\d*1,0) {$\dots$};
    \node[fill=white] at (\d*1,0.4) {$\dots$};
    \node[irrep] at (-0.3,0) {$g^{-1}$};
    \node[irrep] at (-0.3,0.4) {$g$};
    \node[irrep] at (2*\d+0.3,0) {$g^{-1}$};
    \node[irrep] at (2*\d+0.3,0.4) {$g$};
    \draw[decorate, decoration=brace] (-0.1,0.9) -- (2*\d+0.1,0.9) node[midway,above] {$n+k+m$};
  \end{tikzpicture} \  =
  \begin{tikzpicture}[baseline=1mm]
    \def\d{0.6};
    \draw (-0.6,0) -- (2*\d+0.6,0);
    \draw (-0.6,0.4) -- (2*\d+0.6,0.4);
    \draw (8*\d-0.6,0) -- (10*\d+0.6,0);
    \draw (8*\d-0.6,0.4) -- (10*\d+0.6,0.4);
    \foreach \x in {0,2,8,10}{
      \draw (\x*\d,-0.4) -- (\x*\d,0.8);
      \node[mpo] (t) at (\x*\d,0) {};
      \node[mpo] (t) at (\x*\d,0.4) {};
    }
    \foreach \x in {4,6}{
      \draw (\x*\d,-0.4) -- (\x*\d,0.8);
    }
    \draw[fusion] (8*\d-0.6,0)--++(0,0.4);
    \draw[fusion] (2*\d+0.6,0)--++(0,0.4);
    \node[fill=white] at (\d*1,0) {$\dots$};
    \node[fill=white] at (\d*1,0.4) {$\dots$};
    \node[fill=white] at (\d*9,0) {$\dots$};
    \node[fill=white] at (\d*9,0.4) {$\dots$};
    \node[fill=white] at (\d*5,0.2) {$\dots$};
    \node[irrep] at (-0.3,0) {$g^{-1}$};
    \node[irrep] at (-0.3,0.4) {$g$};
    \node[irrep] at (2*\d+0.3,0) {$g^{-1}$};
    \node[irrep] at (2*\d+0.3,0.4) {$g$};
    \node[irrep] at (8*\d-0.3,0) {$g^{-1}$};
    \node[irrep] at (8*\d-0.3,0.4) {$g$};
    \node[irrep] at (10*\d+0.3,0) {$g^{-1}$};
    \node[irrep] at (10*\d+0.3,0.4) {$g$};
    \draw[decorate, decoration=brace] (-0.1,0.9) -- (2*\d+0.1,0.9) node[midway,above] {$k$};
    \draw[decorate, decoration=brace] (4*\d-0.1,0.9) -- (6*\d+0.1,0.9) node[midway,above] {$n$};
    \draw[decorate, decoration=brace] (8*\d-0.1,0.9) -- (10*\d+0.1,0.9) node[midway,above] {$m$};
  \end{tikzpicture} \
\end{equation}
for all $n\geq 0$, $k,m\geq M$. In particular, \eqref{eq:reduction_size_independence} reads as (remember that $\mu_{m,k}=1$)
\begin{equation}\label{eq:trivial_zipper}
  \begin{aligned}
  \begin{tikzpicture}[baseline=1mm]
    \def\d{0.6};
    \draw (-0.6,0) -- (2*\d+0.6,0);
    \draw (-0.6,0.4) -- (2*\d+0.6,0.4);
    \foreach \x in {0,2}{
      \draw (\x*\d,-0.4) -- (\x*\d,0.8);
      \node[mpo] (t) at (\x*\d,0) {};
      \node[mpo] (t) at (\x*\d,0.4) {};
    }
    \draw[fusion] (2*\d+0.6,0)--++(0,0.4);
    \node[fill=white] at (\d*1,0) {$\dots$};
    \node[fill=white] at (\d*1,0.4) {$\dots$};
    \node[irrep] at (-0.3,0) {$g^{-1}$};
    \node[irrep] at (-0.3,0.4) {$g$};
    \node[irrep] at (2*\d+0.3,0) {$g^{-1}$};
    \node[irrep] at (2*\d+0.3,0.4) {$g$};
    \draw[decorate, decoration=brace] (-0.1,0.9) -- (2*\d+0.1,0.9) node[midway,above] {$k$};
  \end{tikzpicture} \  =  \
  \begin{tikzpicture}[baseline=1mm]
    \def\d{0.6};
    \draw (-0.6,0) -- (2*\d+0.6,0);
    \draw (-0.6,0.4) -- (2*\d+0.6,0.4);
    \foreach \x in {0,2}{
      \draw (\x*\d,-0.4) -- (\x*\d,0.8);
      \node[mpo] (t) at (\x*\d,0) {};
      \node[mpo] (t) at (\x*\d,0.4) {};
    }
    \draw[fusion] (2*\d+0.6,0)--++(0,0.4);
    \node[fill=white] at (\d*1,0) {$\dots$};
    \node[fill=white] at (\d*1,0.4) {$\dots$};
    \node[irrep] at (-0.3,0) {$g^{-1}$};
    \node[irrep] at (-0.3,0.4) {$g$};
    \node[irrep] at (2*\d+0.3,0) {$g^{-1}$};
    \node[irrep] at (2*\d+0.3,0.4) {$g$};
    \draw[decorate, decoration=brace] (-0.1,0.9) -- (2*\d+0.1,0.9) node[midway,above] {$k-1$};
  \end{tikzpicture}  \otimes \id, \\
  \begin{tikzpicture}[baseline=1mm,xscale=-1]
    \def\d{0.6};
    \draw (-0.6,0) -- (2*\d+0.6,0);
    \draw (-0.6,0.4) -- (2*\d+0.6,0.4);
    \foreach \x in {0,2}{
      \draw (\x*\d,-0.4) -- (\x*\d,0.8);
      \node[mpo] (t) at (\x*\d,0) {};
      \node[mpo] (t) at (\x*\d,0.4) {};
    }
    \draw[fusion] (2*\d+0.6,0)--++(0,0.4);
    \node[fill=white] at (\d*1,0) {$\dots$};
    \node[fill=white] at (\d*1,0.4) {$\dots$};
    \node[irrep] at (-0.3,0) {$g^{-1}$};
    \node[irrep] at (-0.3,0.4) {$g$};
    \node[irrep] at (2*\d+0.3,0) {$g^{-1}$};
    \node[irrep] at (2*\d+0.3,0.4) {$g$};
    \draw[decorate, decoration={brace,mirror}] (-0.1,0.9) -- (2*\d+0.1,0.9) node[midway,above] {$k$};
  \end{tikzpicture} \  =  \
  \begin{tikzpicture}[baseline=1mm,xscale=-1]
    \def\d{0.6};
    \draw (-0.6,0) -- (2*\d+0.6,0);
    \draw (-0.6,0.4) -- (2*\d+0.6,0.4);
    \foreach \x in {0,2}{
      \draw (\x*\d,-0.4) -- (\x*\d,0.8);
      \node[mpo] (t) at (\x*\d,0) {};
      \node[mpo] (t) at (\x*\d,0.4) {};
    }
    \draw[fusion] (2*\d+0.6,0)--++(0,0.4);
    \node[fill=white] at (\d*1,0) {$\dots$};
    \node[fill=white] at (\d*1,0.4) {$\dots$};
    \node[irrep] at (-0.3,0) {$g^{-1}$};
    \node[irrep] at (-0.3,0.4) {$g$};
    \node[irrep] at (2*\d+0.3,0) {$g^{-1}$};
    \node[irrep] at (2*\d+0.3,0.4) {$g$};
    \draw[decorate, decoration={brace,mirror}] (-0.1,0.9) -- (2*\d+0.1,0.9) node[midway,above] {$k-1$};
  \end{tikzpicture}  \otimes \id,
  \end{aligned}
\end{equation}
for all $k>M$.  

Let us now define the  maps $\beta_g^{n,m,k}:\Mat_d^{\otimes n}\to \Mat_d^{\otimes (n+m+k)}$ for every $n$ and $m,k\geq M$ as
\begin{equation}\label{eq:beta_def_with_mps}
   \beta_g^{n,m,k} : \
  \begin{tikzpicture}
    \def\d{0.6};
    \draw (-\d,-0.5)--(-\d,0.5);
    \draw (0,-0.5)--(0,0.5);
    \draw (\d,-0.5)--(\d,0.5);
    \node[draw, fill=white,rounded corners, minimum width = \d*3cm] at (0,0) {$X$};
    \draw[decorate, decoration={brace}] (-\d-0.1,0.6) -- (\d+0.1,0.6) node[midway,above] {$n$};
  \end{tikzpicture} \ \mapsto \
  \begin{tikzpicture}
    \def\d{0.6};
    \draw (-\d*3-0.5,0.5) -- (\d*3+0.5,0.5);
    \draw (-\d*3-0.5,-0.5) -- (\d*3+0.5,-0.5);
    \node[irrep] at (-2.5*\d,0.5) {$g$};
    \node[irrep] at (-2.5*\d,-0.5) {$g^{-1}$};
    \foreach \x in {-3,-2,-1,0,1,2,3}{
      \draw (\d*\x,-0.9)--(\d*\x,0.9);
      \node[mpo] at (\d*\x,0.5) {};
      \node[mpo] at (\d*\x,-0.5) {};
    }
    \node[draw, fill=white,rounded corners, minimum width = \d*3cm] at (0,0) {$X$};
    \draw[fusion] (-\d*3-0.5,-0.5) -- (-\d*3-0.5,0.5);
    \draw[fusion] (\d*3+0.5,-0.5) -- (\d*3+0.5,0.5);
    \draw[decorate, decoration={brace}] (-3*\d-0.1,1) -- (-2*\d+0.1,1) node[midway,above] {$m$};
    \draw[decorate, decoration={brace}] (2*\d-0.1,1) -- (3*\d+0.1,1) node[midway,above] {$k$};
    \draw[decorate, decoration={brace}] (-1*\d-0.1,1) -- (1*\d+0.1,1) node[midway,above] {$n$};
  \end{tikzpicture} \ ,
\end{equation}
for any $X\in \Mat_d^{\otimes n}$, for any $n$. Through the application of \eqref{eq:trivial_zipper} (remember that we have shown that $\mu=1$), we see that $\beta^{n,m,k}$ behaves nicely when tensoring with the identity:
\begin{align*}
    \beta_g^{n+1,m,k}(\mathbb{1} \otimes X ) = \beta_g^{n,m+1,k} (X) = \mathbb{1} \otimes \beta_g^{n,m,k}(X), \\
    \beta_g^{n+1,m,k}(X\otimes\mathbb{1}) = \beta_g^{n,m,k+1}(X) = \beta_g^{n,m,k}(X)\otimes \mathbb{1}.
\end{align*}
Because of these equations we can define a single map $\beta_g$ acting on the algebra of local observables consistently. To denote $\beta_g:\mathcal{A}_{loc}\to\mathcal{A}_{loc}$ we will use the same graphical notation as for $\beta_g^{n,m,k}$, but without specifying $n,m$ and $k$. In the following we show that $\beta_g$ is a $*$-automorphism of $\mathcal{A}_{loc}$ that is norm-contractive and that $g\mapsto\beta_g$ is a representation of $G$. Strictly speaking, we prove all these properties for $\beta_g^{n,m,k}$, but the proof trivially lifts to $\beta_g$.

Notice now that $\beta_g(\mathbb{1}) = \mathbb{1}$. One can then easily see that this operation is an algebra isomorphism, as
\begin{align*}
  \beta_g(X)\beta_g(Y) =
  \begin{tikzpicture}[baseline=-8.5mm]
    \def\d{0.6};
    \draw (-\d*3-0.5,0.5) -- (\d*3+0.5,0.5);
    \draw (-\d*3-0.5,-0.5) -- (\d*3+0.5,-0.5);
    \draw (-\d*3-0.5,-1) -- (\d*3+0.5,-1);
    \draw (-\d*3-0.5,-2) -- (\d*3+0.5,-2);
    \node[irrep] at (-2.5*\d,0.5) {$g$};
    \node[irrep] at (-2.5*\d,-0.5) {$g^{-1}$};
    \node[irrep] at (-2.5*\d,-1) {$g$};
    \node[irrep] at (-2.5*\d,-2) {$g^{-1}$};
    \foreach \x in {-3,-2,-1,0,1,2,3}{
      \draw (\d*\x,-2.4)--(\d*\x,0.9);
      \node[mpo] at (\d*\x,0.5) {};
      \node[mpo] at (\d*\x,-0.5) {};
      \node[mpo] at (\d*\x,-1) {};
      \node[mpo] at (\d*\x,-2) {};
    }
    \node[draw, fill=white,rounded corners, minimum width = \d*3cm] at (0,0) {$X$};
    \node[draw, fill=white,rounded corners, minimum width = \d*3cm] at (0,-1.5) {$Y$};
    \draw[fusion] (-\d*3-0.5,-0.5) -- (-\d*3-0.5,0.5);
    \draw[fusion] (\d*3+0.5,-0.5) -- (\d*3+0.5,0.5);
    \draw[fusion] (-\d*3-0.5,-2) -- (-\d*3-0.5,-1);
    \draw[fusion] (\d*3+0.5,-2) -- (\d*3+0.5,-1);
  \end{tikzpicture} \ =
  \begin{tikzpicture}[baseline=-8.5mm]
    \def\d{0.6};
    \draw (-\d*3-0.5,0.5) -- (\d*3+0.5,0.5);
    \draw (-\d*3-0.5,-0.5) -- (-\d*1.5,-0.5);
    \draw (\d*1.5,-0.5) -- (\d*3+0.5,-0.5);
    \draw (-\d*3-0.5,-1) -- (-\d*1.5,-1);
    \draw (\d*1.5,-1) -- (\d*3+0.5,-1);
    \draw (-\d*3-0.5,-2) -- (\d*3+0.5,-2);
    \node[irrep] at (-2.5*\d,0.5) {$g$};
    \node[irrep] at (-2.5*\d,-0.5) {$g^{-1}$};
    \node[irrep] at (-2.5*\d,-1) {$g$};
    \node[irrep] at (-2.5*\d,-2) {$g^{-1}$};
    \foreach \x in {-3,-2,-1,0,1,2,3}{
      \draw (\d*\x,-2.4)--(\d*\x,0.9);
      \node[mpo] at (\d*\x,0.5) {};
      \node[mpo] at (\d*\x,-2) {};
    }
    \foreach \x in {-3,-2,2,3}{
      \node[mpo] at (\d*\x,-0.5) {};
      \node[mpo] at (\d*\x,-1) {};
    }
    \node[draw, fill=white,rounded corners, minimum width = \d*3cm] at (0,0) {$X$};
    \node[draw, fill=white,rounded corners, minimum width = \d*3cm] at (0,-1.5) {$Y$};
    \draw[fusion] (-\d*3-0.5,-0.5) -- (-\d*3-0.5,0.5);
    \draw[fusion] (\d*3+0.5,-0.5) -- (\d*3+0.5,0.5);
    \draw[fusion] (-\d*3-0.5,-2) -- (-\d*3-0.5,-1);
    \draw[fusion] (\d*3+0.5,-2) -- (\d*3+0.5,-1);
    \draw[fusion] (-\d*1.5,-1) -- (-\d*1.5,-0.5);
    \draw[fusion] (\d*1.5,-1) -- (\d*1.5,-0.5);
  \end{tikzpicture}  = \\
  \begin{tikzpicture}[baseline=-8.5mm]
    \def\d{0.6};
    \draw (-\d*3-0.5,0.5) -- (\d*3+0.5,0.5);
    \draw (-\d*3-0.5,-0.5) -- (-\d*1.5,-0.5);
    \draw (\d*1.5,-0.5) -- (\d*3+0.5,-0.5);
    \draw (-\d*3-0.5,-1) -- (-\d*1.5,-1);
    \draw (\d*1.5,-1) -- (\d*3+0.5,-1);
    \draw (-\d*3-0.5,-1.5) -- (\d*3+0.5,-1.5);
    \node[irrep] at (-2.5*\d,0.5) {$g$};
    \node[irrep] at (-2.5*\d,-0.5) {$g^{-1}$};
    \node[irrep] at (-2.5*\d,-1) {$g$};
    \node[irrep] at (-2.5*\d,-1.5) {$g^{-1}$};
    \foreach \x in {-3,-2,-1,0,1,2,3}{
      \draw (\d*\x,-1.9)--(\d*\x,0.9);
      \node[mpo] at (\d*\x,0.5) {};
      \node[mpo] at (\d*\x,-1.5) {};
    }
    \foreach \x in {-3,-2,2,3}{
      \node[mpo] at (\d*\x,-0.5) {};
      \node[mpo] at (\d*\x,-1) {};
    }
    \node[draw, fill=white,rounded corners, minimum width = \d*3cm] at (0,0) {$X\cdot Y$};
    \draw[fusion] (-\d*3-0.5,-0.5) -- (-\d*3-0.5,0.5);
    \draw[fusion] (\d*3+0.5,-0.5) -- (\d*3+0.5,0.5);
    \draw[fusion] (-\d*3-0.5,-1.5) -- (-\d*3-0.5,-1);
    \draw[fusion] (\d*3+0.5,-1.5) -- (\d*3+0.5,-1);
    \draw[fusion] (-\d*1.5,-1) -- (-\d*1.5,-0.5);
    \draw[fusion] (\d*1.5,-1) -- (\d*1.5,-0.5);
  \end{tikzpicture}  =
  \begin{tikzpicture}[baseline=-8.5mm]
    \def\d{0.6};
    \draw (-\d*3-0.5,0.5) -- (\d*3+0.5,0.5);
    \draw (-\d*3-0.5,-0.5) -- (\d*3+0.5,-0.5);
    \draw (-\d*3-0.5,-1) -- (\d*3+0.5,-1);
    \draw (-\d*3-0.5,-1.5) -- (\d*3+0.5,-1.5);
    \node[irrep] at (-2.5*\d,0.5) {$g$};
    \node[irrep] at (-2.5*\d,-0.5) {$g^{-1}$};
    \node[irrep] at (-2.5*\d,-1) {$g$};
    \node[irrep] at (-2.5*\d,-1.5) {$g^{-1}$};
    \foreach \x in {-3,-2,-1,0,1,2,3}{
      \draw (\d*\x,-1.9)--(\d*\x,0.9);
      \node[mpo] at (\d*\x,0.5) {};
      \node[mpo] at (\d*\x,-0.5) {};
      \node[mpo] at (\d*\x,-1) {};
      \node[mpo] at (\d*\x,-1.5) {};
    }
    \node[draw, fill=white,rounded corners, minimum width = \d*3cm] at (0,0) {$X\cdot Y$};
    \draw[fusion] (-\d*3-0.5,-0.5) -- (-\d*3-0.5,0.5);
    \draw[fusion] (\d*3+0.5,-0.5) -- (\d*3+0.5,0.5);
    \draw[fusion] (-\d*3-0.5,-1.5) -- (-\d*3-0.5,-1);
    \draw[fusion] (\d*3+0.5,-1.5) -- (\d*3+0.5,-1);
  \end{tikzpicture} \ = \beta_g(XY)
\end{align*}
where in the second equality we have used \eqref{eq:simple_VW_factorize}, and in the fourth one we have used it again in the other direction. In the last equality we have used $\beta_g(\mathbb{1})=\mathbb{1}$.

Let us check that $g\mapsto \beta_g$ is a group representation. First, $\beta_1=\id$ trivially. Second, we need to check that $\beta_g(\beta_h(X)) = \beta_{gh}(X)$. We can write $\beta_g(\beta_h(X))$ as
\begin{equation}
   \beta_g(\beta_h(X)) =
  \begin{tikzpicture}
    \def\d{0.6};
    \foreach \y/\x/\l in {-0.9/1/{g^{-1}},-0.5/0.5/{h^{-1}},0.5/0.5/{h},0.9/1/{g}}{
      \draw (-\d*3-\x,\y) -- (\d*3+\x,\y);
    \node[irrep] at (-2.5*\d,\y) {$\l$};
    }
    \foreach \x in {-3,-2,-1,0,1,2,3}{
      \draw (\d*\x,-1.3)--(\d*\x,1.3);
      \node[mpo] at (\d*\x,0.5) {};
      \node[mpo] at (\d*\x,-0.5) {};
      \node[mpo] at (\d*\x,0.9) {};
      \node[mpo] at (\d*\x,-0.9) {};
    }
    \node[draw, fill=white,rounded corners, minimum width = \d*3cm] at (0,0) {$X$};
    \draw[fusion] (-\d*3-0.5,-0.5) -- (-\d*3-0.5,0.5);
    \draw[fusion] (\d*3+0.5,-0.5) -- (\d*3+0.5,0.5);
    \draw[fusion] (-\d*3-1,-0.9) -- (-\d*3-1,0.9);
    \draw[fusion] (\d*3+1,-0.9) -- (\d*3+1,0.9);
  \end{tikzpicture} \ .
\end{equation}
Notice that, using Eq.~\eqref{eq:mps_cocycle_def}, we obtain
\begin{align*}
  \begin{tikzpicture}[xscale=-1]
    \def\d{0.6};
    \foreach \y/\x/\l in {-0.6/1/{g^{-1}},-0.2/0.5/{h^{-1}},0.2/0.5/{h},0.6/1/{g}}{
      \draw (-\x,\y) -- (\d+0.4,\y);
      \node[irrep] at (0.5*\d,\y) {$\l$};
    }
    \foreach \x in {0,1}{
      \draw (\d*\x,-0.9)--(\d*\x,0.9);
      \node[mpo] at (\d*\x,0.2) {};
      \node[mpo] at (\d*\x,-0.2) {};
      \node[mpo] at (\d*\x,0.6) {};
      \node[mpo] at (\d*\x,-0.6) {};
    }
    \draw[fusion] (-0.5,-0.2) -- (-0.5,0.2);
    \draw[fusion] (-1,-0.6) -- (-1,0.6);
  \end{tikzpicture} \ =
  \begin{tikzpicture}[xscale=-1]
    \def\d{0.6};
    \foreach \y/\x/\l in {-0.6/1/{g^{-1}},-0.2/0.5/{h^{-1}},0.2/0.5/{h},0.6/1.5/{g}}{
      \draw (-\x,\y) -- (\d+0.4,\y);
      \node[irrep] at (0.5*\d,\y) {$\l$};
    }
    \foreach \x in {0,1}{
      \draw (\d*\x,-0.9)--(\d*\x,0.9);
      \node[mpo] at (\d*\x,0.2) {};
      \node[mpo] at (\d*\x,-0.2) {};
      \node[mpo] at (\d*\x,0.6) {};
      \node[mpo] at (\d*\x,-0.6) {};
    }
    \draw[fusion] (-0.5,-0.2) -- (-0.5,0.2);
    \draw[densely dotted] (-0.5,0) --++(-0.5,0);
    \draw[fusion] (-1,-0.6) -- (-1,0);
    \draw (-1,-0.3) --++(-0.5,0);
    \draw[fusion] (-1.5,-0.3) -- (-1.5,0.6);
  \end{tikzpicture} \ = \frac{1}{\omega(h,h^{-1},g^{-1})}
  \begin{tikzpicture}[xscale=-1]
    \def\d{0.6};
    \foreach \y/\x/\l in {-0.6/0.5/{g^{-1}},-0.2/0.5/{h^{-1}},0.2/1/{h},0.6/1.5/{g}}{
      \draw (-\x,\y) -- (\d+0.4,\y);
      \node[irrep] at (0.5*\d,\y) {$\l$};
    }
    \foreach \x in {0,1}{
      \draw (\d*\x,-0.9)--(\d*\x,0.9);
      \node[mpo] at (\d*\x,0.2) {};
      \node[mpo] at (\d*\x,-0.2) {};
      \node[mpo] at (\d*\x,0.6) {};
      \node[mpo] at (\d*\x,-0.6) {};
    }
    \draw[fusion] (-0.5,-0.6) -- (-0.5,-0.2);
    \draw (-0.5,-0.4) --++(-0.5,0);
    \draw[fusion] (-1,-0.4) -- (-1,0.2);
    \draw (-1,-0.1) --++(-0.5,0);
    \draw[fusion] (-1.5,-0.1) -- (-1.5,0.6);
  \end{tikzpicture} \ ,
\end{align*}
and thus, using Eq.~\eqref{eq:mps_cocycle_def} again, we obtain that
\begin{align*}
  \begin{tikzpicture}[xscale=-1]
    \def\d{0.6};
    \foreach \y/\x/\l in {-0.6/1/{g^{-1}},-0.2/0.5/{h^{-1}},0.2/0.5/{h},0.6/1/{g}}{
      \draw (-\x,\y) -- (\d+0.4,\y);
      \node[irrep] at (0.5*\d,\y) {$\l$};
    }
    \foreach \x in {0,1}{
      \draw (\d*\x,-0.9)--(\d*\x,0.9);
      \node[mpo] at (\d*\x,0.2) {};
      \node[mpo] at (\d*\x,-0.2) {};
      \node[mpo] at (\d*\x,0.6) {};
      \node[mpo] at (\d*\x,-0.6) {};
    }
    \draw[fusion] (-0.5,-0.2) -- (-0.5,0.2);
    \draw[fusion] (-1,-0.6) -- (-1,0.6);
  \end{tikzpicture} \ =
  \frac{\omega(g,h,h^{-1})}{\omega(h,h^{-1},g^{-1})}
  \begin{tikzpicture}[xscale=-1]
    \def\d{0.6};
    \foreach \y/\x/\l in {-0.6/0.5/{g^{-1}},-0.2/0.5/{h^{-1}},0.2/0.5/{h},0.6/0.5/{g}}{
      \draw (-\x,\y) -- (\d+0.4,\y);
      \node[irrep] at (0.5*\d,\y) {$\l$};
    }
    \foreach \x in {0,1}{
      \draw (\d*\x,-0.9)--(\d*\x,0.9);
      \node[mpo] at (\d*\x,0.2) {};
      \node[mpo] at (\d*\x,-0.2) {};
      \node[mpo] at (\d*\x,0.6) {};
      \node[mpo] at (\d*\x,-0.6) {};
    }
    \draw[fusion] (-0.5,-0.6) -- (-0.5,-0.2);
    \draw (-0.5,-0.4) --++(-0.5,0);
    \draw[fusion] (-0.5,0.6) -- (-0.5,0.2);
    \draw (-0.5,0.4) --++(-0.5,0);
    \draw[fusion] (-1,-0.4) -- (-1,0.4);
  \end{tikzpicture} \ .
\end{align*}
Similarly,
\begin{align*}
  \begin{tikzpicture}[xscale=1]
    \def\d{0.6};
    \foreach \y/\x/\l in {-0.6/1/{g^{-1}},-0.2/0.5/{h^{-1}},0.2/0.5/{h},0.6/1/{g}}{
      \draw (-\x,\y) -- (\d+0.4,\y);
      \node[irrep] at (0.5*\d,\y) {$\l$};
    }
    \foreach \x in {0,1}{
      \draw (\d*\x,-0.9)--(\d*\x,0.9);
      \node[mpo] at (\d*\x,0.2) {};
      \node[mpo] at (\d*\x,-0.2) {};
      \node[mpo] at (\d*\x,0.6) {};
      \node[mpo] at (\d*\x,-0.6) {};
    }
    \draw[fusion] (-0.5,-0.2) -- (-0.5,0.2);
    \draw[fusion] (-1,-0.6) -- (-1,0.6);
  \end{tikzpicture} \ =
  \frac{\omega(h,h^{-1},g^{-1})}{\omega(g,h,h^{-1})}
  \begin{tikzpicture}[xscale=1]
    \def\d{0.6};
    \foreach \y/\x/\l in {-0.6/0.5/{g^{-1}},-0.2/0.5/{h^{-1}},0.2/0.5/{h},0.6/0.5/{g}}{
      \draw (-\x,\y) -- (\d+0.4,\y);
      \node[irrep] at (0.5*\d,\y) {$\l$};
    }
    \foreach \x in {0,1}{
      \draw (\d*\x,-0.9)--(\d*\x,0.9);
      \node[mpo] at (\d*\x,0.2) {};
      \node[mpo] at (\d*\x,-0.2) {};
      \node[mpo] at (\d*\x,0.6) {};
      \node[mpo] at (\d*\x,-0.6) {};
    }
    \draw[fusion] (-0.5,-0.6) -- (-0.5,-0.2);
    \draw (-0.5,-0.4) --++(-0.5,0);
    \draw[fusion] (-0.5,0.6) -- (-0.5,0.2);
    \draw (-0.5,0.4) --++(-0.5,0);
    \draw[fusion] (-1,-0.4) -- (-1,0.4);
  \end{tikzpicture} \ .
\end{align*}
Therefore, 
\begin{equation*}
   \beta_g(\beta_h(X)) =
  \begin{tikzpicture}
    \def\d{0.6};
    \foreach \y/\x/\l in {-0.9/0.5/{g^{-1}},-0.5/0.5/{h^{-1}},0.5/0.5/{h},0.9/0.5/{g}}{
      \draw (-\d*3-\x,\y) -- (\d*3+\x,\y);
    \node[irrep] at (-2.5*\d,\y) {$\l$};
    }
    \foreach \x in {-3,-2,-1,0,1,2,3}{
      \draw (\d*\x,-1.3)--(\d*\x,1.3);
      \node[mpo] at (\d*\x,0.5) {};
      \node[mpo] at (\d*\x,-0.5) {};
      \node[mpo] at (\d*\x,0.9) {};
      \node[mpo] at (\d*\x,-0.9) {};
    }
    \node[draw, fill=white,rounded corners, minimum width = \d*3cm] at (0,0) {$X$};
    \draw[fusion] (-\d*3-0.5,-0.9) -- (-\d*3-0.5,-0.5);
    \draw (-\d*3-0.5,-0.7) --++ (-0.5,0);
    \draw[fusion] (\d*3+0.5,-0.9) -- (\d*3+0.5,-0.5);
    \draw (\d*3+0.5,-0.7) --++ (0.5,0);
    \draw[fusion] (-\d*3-0.5,0.5) -- (-\d*3-0.5,0.9);
    \draw (-\d*3-0.5,0.7) --++ (-0.5,0);
    \draw[fusion] (\d*3+0.5,0.5) -- (\d*3+0.5,0.9);
    \draw (\d*3+0.5,0.7) --++ (0.5,0);
    \draw[fusion] (-\d*3-1,-0.7) -- (-\d*3-1,0.7);
    \draw[fusion] (\d*3+1,-0.7) -- (\d*3+1,0.7);
  \end{tikzpicture} \ = \beta_{gh}(X),
\end{equation*}
where in the second equality we have used \eqref{eq:fusion_tensors}.

Let us show now that the operation $\beta_g$ is a $*$-isomorphism and that it is norm-contractive. 
Combining \eqref{eq:trivial_zipper} with \eqref{eq:conjugate_gauge}, we obtain
\begin{equation}\label{eq:gray_g_g_zipper}
  \begin{aligned}
  \begin{tikzpicture}[baseline=1mm]
    \def\d{0.6};
    \draw (-0.6,0) -- (3*\d,0);
    \draw (-0.6,0.4) -- (3*\d,0.4);
    \foreach \x in {0,2}{
      \draw (\x*\d,-0.4) -- (\x*\d,0.8);
      \node[cmpo] (t) at (\x*\d,0) {};
      \node[mpo] (t) at (\x*\d,0.4) {};
    }
    \draw[fusion,gray] (3*\d,0)--++(0,0.4);
    \node[fill=white] at (\d*1,0) {$\dots$};
    \node[fill=white] at (\d*1,0.4) {$\dots$};
    \node[irrep] at (-0.3,0) {$g$};
    \node[irrep] at (-0.3,0.4) {$g$};
    \node[irrep] at (2.5*\d,0) {$g$};
    \node[irrep] at (2.5*\d,0.4) {$g$};
    \draw[decorate, decoration={brace}] (-0.1,0.9) -- (2*\d+0.1,0.9) node[midway,above] {$k$};
  \end{tikzpicture} \  &= \
  \begin{tikzpicture}[baseline=1mm]
    \def\d{0.6};
    \draw (-0.6,0) -- (3*\d,0);
    \draw (-0.6,0.4) -- (3*\d,0.4);
    \foreach \x in {0,2}{
      \draw (\x*\d,-0.4) -- (\x*\d,0.8);
      \node[cmpo] (t) at (\x*\d,0) {};
      \node[mpo] (t) at (\x*\d,0.4) {};
    }
    \draw[fusion,gray] (3*\d,0)--++(0,0.4);
    \node[fill=white] at (\d*1,0) {$\dots$};
    \node[fill=white] at (\d*1,0.4) {$\dots$};
    \node[irrep] at (-0.3,0) {$g$};
    \node[irrep] at (-0.3,0.4) {$g$};
    \node[irrep] at (2.5*\d,0) {$g$};
    \node[irrep] at (2.5*\d,0.4) {$g$};
    \draw[decorate, decoration={brace}] (-0.1,0.9) -- (2*\d+0.1,0.9) node[midway,above] {$k-1$};
  \end{tikzpicture}  \otimes \mathbb 1, \\
  \begin{tikzpicture}[baseline=1mm,xscale=-1]
    \def\d{0.6};
    \draw (-0.6,0) -- (3*\d,0);
    \draw (-0.6,0.4) -- (3*\d,0.4);
    \foreach \x in {0,2}{
      \draw (\x*\d,-0.4) -- (\x*\d,0.8);
      \node[cmpo] (t) at (\x*\d,0) {};
      \node[mpo] (t) at (\x*\d,0.4) {};
    }
    \draw[fusion,gray] (3*\d,0)--++(0,0.4);
    \node[fill=white] at (\d*1,0) {$\dots$};
    \node[fill=white] at (\d*1,0.4) {$\dots$};
    \node[irrep] at (-0.3,0) {$g$};
    \node[irrep] at (-0.3,0.4) {$g$};
    \node[irrep] at (2.5*\d,0) {$g$};
    \node[irrep] at (2.5*\d,0.4) {$g$};
    \draw[decorate, decoration={brace,mirror}] (-0.1,0.9) -- (2*\d+0.1,0.9) node[midway,above] {$k$};
  \end{tikzpicture} \  &=  \mathbb 1 \otimes \
  \begin{tikzpicture}[baseline=1mm,xscale=-1]
    \def\d{0.6};
    \draw (-0.6,0) -- (3*\d,0);
    \draw (-0.6,0.4) -- (3*\d,0.4);
    \foreach \x in {0,2}{
      \draw (\x*\d,-0.4) -- (\x*\d,0.8);
      \node[cmpo] (t) at (\x*\d,0) {};
      \node[mpo] (t) at (\x*\d,0.4) {};
    }
    \draw[fusion,gray] (3*\d,0)--++(0,0.4);
    \node[fill=white] at (\d*1,0) {$\dots$};
    \node[fill=white] at (\d*1,0.4) {$\dots$};
    \node[irrep] at (-0.3,0) {$g$};
    \node[irrep] at (-0.3,0.4) {$g$};
    \node[irrep] at (2.5*\d,0) {$g$};
    \node[irrep] at (2.5*\d,0.4) {$g$};
    \draw[decorate, decoration={brace,mirror}] (-0.1,0.9) -- (2*\d+0.1,0.9) node[midway,above] {$k-1$};
  \end{tikzpicture} ,
  \end{aligned}
\end{equation}
for all $k$ large enough where we have introduced
\begin{equation}\label{eq:gray_fusion}
    \begin{tikzpicture}[xscale=0.5,yscale=0.25,baseline=-1mm]
        \draw (0,-1) --++ (-1,0);
        \draw (0,1) --++ (-1,0);
        \draw[fusion,gray] (0,-1) -- (0,1);
        \node[irrep] at (-0.5,-1) {$g$};
        \node[irrep] at (-0.5,1) {$g$};
    \end{tikzpicture} = 
    \begin{tikzpicture}[xscale=0.5,yscale=0.25,baseline=-1mm]
        \draw (0,1) --++ (-2.2,0);
        \draw (0,-1) --++ (-2.2,0);
         \node[mpo,fill=gray,label=below:$X_g$] at (-1.2,-1) {};
        \draw[fusion] (0,-1) -- (0,1);
        \node[irrep] at (-0.5,1) {$g$};
        \node[irrep] at (-0.6,-1) {$g^{-1}$};
        \node[irrep] at (-1.7,-1) {$g$};
    \end{tikzpicture}
    \quad \text{and} \quad  
    \begin{tikzpicture}[xscale=-0.5,yscale=0.25,baseline=-1mm]
        \draw (0,-1) --++ (-1,0);
        \draw (0,1) --++ (-1,0);
        \draw[fusion,gray] (0,-1) -- (0,1);
        \node[irrep] at (-0.5,-1) {$g$};
        \node[irrep] at (-0.5,1) {$g$};
    \end{tikzpicture} = 
    \begin{tikzpicture}[xscale=-0.5,yscale=0.25,baseline=-1mm]
        \draw (0,1) --++ (-2.2,0);
        \draw (0,-1) --++ (-2.2,0);
         \node[mpo,fill=gray,label=below:$X_g$] at (-1.2,-1) {};
        \draw[fusion] (0,-1) -- (0,1);
        \node[irrep] at (-0.5,1) {$g$};
        \node[irrep] at (-0.6,-1) {$g^{-1}$};
        \node[irrep] at (-1.7,-1) {$g$};
    \end{tikzpicture} \ .
\end{equation}
Tracing the physical indices in this equation, we obtain that the matrices 
\begin{equation}\label{eq:rho_l_rho_r}
  \begin{tikzpicture}[xscale=-1]
    \draw (0.2,0.2) -- (0,0.2)--(0,-0.2)--(0.2,-0.2);
    \node[tensor, label=right:$\rho_g^r$] at (0,0) {};
  \end{tikzpicture} =\frac{1}{d^M}\
  \begin{tikzpicture}[baseline=1mm]
    \def\d{0.6};
    \draw (-0.6,0) -- (3*\d,0);
    \draw (-0.6,0.4) -- (3*\d,0.4);
    \foreach \x in {0,2}{
      \draw (\x*\d,-0.4) rectangle (\x*\d-0.2,0.8);
      \node[cmpo] (t) at (\x*\d,0) {};
      \node[mpo] (t) at (\x*\d,0.4) {};
    }
    \draw[fusion,gray] (3*\d,0)--++(0,0.4);
    \node[fill=white] at (\d*1,0) {$\dots$};
    \node[fill=white] at (\d*1,0.4) {$\dots$};
    \node[irrep] at (-0.3,0) {$g$};
    \node[irrep] at (-0.3,0.4) {$g$};
    \node[irrep] at (2*\d+0.3,0) {$g$};
    \node[irrep] at (2*\d+0.3,0.4) {$g$};
    \draw[decorate, decoration={brace}] (-0.3,0.9) -- (2*\d+0.1,0.9) node[midway,above] {$M$};
  \end{tikzpicture} \quad \text{and} \quad
  \begin{tikzpicture}
    \draw (0.2,0.2) -- (0,0.2)--(0,-0.2)--(0.2,-0.2);
    \node[tensor, label=left:$\rho_g^l$] at (0,0) {};
  \end{tikzpicture} = \frac{1}{d^{M}}\
  \begin{tikzpicture}[baseline=1mm,xscale=-1]
    \def\d{0.6};
    \draw (-0.6,0) -- (3*\d,0);
    \draw (-0.6,0.4) -- (3*\d,0.4);
    \foreach \x in {0,2}{
      \draw (\x*\d,-0.4) rectangle (\x*\d-0.2,0.8);
      \node[cmpo] (t) at (\x*\d,0) {};
      \node[mpo] (t) at (\x*\d,0.4) {};
    }
    \draw[fusion,gray] (3*\d,0)--++(0,0.4);
    \node[fill=white] at (\d*1,0) {$\dots$};
    \node[fill=white] at (\d*1,0.4) {$\dots$};
    \node[irrep] at (-0.3,0) {$g$};
    \node[irrep] at (-0.3,0.4) {$g$};
    \node[irrep] at (2*\d+0.3,0) {$g$};
    \node[irrep] at (2*\d+0.3,0.4) {$g$};
    \draw[decorate, decoration={brace,mirror}] (-0.3,0.9) -- (2*\d+0.1,0.9) node[midway,above] {$M$};
  \end{tikzpicture}
\end{equation}
are right and left eigenvectors of the transfer matrix $T$ of $u(g)$, respectively, with eigenvalue $d$. Therefore they are proportional to the right and left Frobenius-Perron eigenvectors of $T$. Notice that $\tr \rho_g^l\rho_g^r = 1$, and thus, in fact,
\begin{equation}\label{eq:boundary_positive}
    \rho^l_g \otimes \rho^r_g = \rho_g \otimes \mathbb{1}.
\end{equation}
Using the phase degree of freedom of $X_g$, we can also assume that $\rho^l_g = \rho_g$ and $\rho^r_g =  \mathbb 1$. Notice now that using \eqref{eq:conjugate_gauge} in \eqref{eq:beta_def_with_mps} and tracing out a sufficient amount of physical indices at the two ends, we can equally write
\begin{equation*}
  \beta_g(X) =
  \begin{tikzpicture}
    \def\d{0.6};
    \draw (-\d*3-0.5,0.5) rectangle (\d*3+0.5,-0.5);
    \node[irrep] at (-2.5*\d,0.5) {$g$};
    \node[irrep] at (-2.5*\d,-0.5) {$g$};
    \foreach \x in {-3,-2,-1,0,1,2,3}{
      \draw (\d*\x,-0.9)--(\d*\x,0.9);
      \node[mpo] at (\d*\x,0.5) {};
      \node[cmpo] at (\d*\x,-0.5) {};
    }
    \node[draw, fill=white,rounded corners, minimum width = \d*3cm] at (0,0) {$X$};
    \node[tensor,label=left:$\rho_g$] at (-\d*3-0.5,0) {};
  \end{tikzpicture} \ .
\end{equation*}
As $\rho_g $ is positive,  $\beta_g(X)$ can be written in the form $\beta_g(X) = \sum_i U_{g,i} X U_{g,i}^*$ with $\sum_i U_{g,i} U_{g,i}^\dagger = \id$. This means that
\begin{equation*}
    \beta_g(X^*) = \left(\sum_i U_{g,i} X U_{g,i}^*\right)^* = \sum_i U_{g,i} X^* U_{g,i}^* = \beta_g(X)^*,
\end{equation*}
i.e.\ that it is a $*$-isomorphism, and that
\begin{equation*}
    \|\beta_g(X)\| \leq \|\beta_g(\mathbb{1})\| \cdot \|X\|  = \|X\|,
\end{equation*}
i.e.\ that it is norm-contractive.

As $\beta_g$ is norm-contractive, it can be extended to the whole set of quasi-local operators such that it remains an automorphism and a group representation.  We call this map $\beta_g$ the conjugation by $U_g$.

\subsubsection{The \texorpdfstring{$\beta_g$}{beta\_g} defined from MPUs decomposes locally}

Let us show now that $\beta_g$ satisfies
\begin{equation}
  \beta_g  = \aut(w_g) \cdot (\beta_{g,L} \otimes \beta_{g,R}),
\end{equation}
for some unitary $w_g$. For that, let $L$ be larger than $M$ and let us define a new MPU tensor $u^{(L)}(g) \in \Mat_d^{\otimes L} \otimes \Mat_D$ as
\begin{equation*}
  u^{(L)}(g) = \sum_{\alpha} u_{\alpha_1\alpha_2}(g) \otimes \dots \otimes u_{\alpha_L\alpha_{L+1}}(g) \otimes \ket{\alpha_1}\bra{\alpha_{L+1}},
\end{equation*}
or graphically,
\begin{equation*}
  \begin{tikzpicture}
    \draw (-0.4,0) -- (0.4,0);
    \draw[very thick] (0,-0.4) -- (0,0.4);
    \node[mpo] at (0,0) {};
  \end{tikzpicture} =
  \begin{tikzpicture}
    \def\d{0.6};
    \draw (-0.4,0) -- (3*\d+0.4,0);
    \foreach \x in {0,\d,3*\d}{
      \draw (\x,-0.4) -- (\x,0.4);
      \node[mpo] at (\x,0) {};
    }
    \node[fill=white] at (2*\d,0) {$\dots$};
  \end{tikzpicture} \ .
\end{equation*}
The conjugate tensor is defined similarly:
\begin{equation*}
  \begin{tikzpicture}
    \draw (-0.4,0) -- (0.4,0);
    \draw[very thick] (0,-0.4) -- (0,0.4);
    \node[cmpo] at (0,0) {};
  \end{tikzpicture} =
  \begin{tikzpicture}
    \def\d{0.6};
    \draw (-0.4,0) -- (3*\d+0.4,0);
    \foreach \x in {0,\d,3*\d}{
      \draw (\x,-0.4) -- (\x,0.4);
      \node[cmpo] at (\x,0) {};
    }
    \node[fill=white] at (2*\d,0) {$\dots$};
  \end{tikzpicture} \ .
\end{equation*}
Considering \eqref{eq:simple_VW_factorize} once with $n=0$ and once with $n=2M$, we obtain that these tensors satisfy
\begin{equation*}
  \begin{aligned}
  \begin{tikzpicture}[baseline=1.5mm,xscale=0.6]
    \draw (-0.5,0) -- (4.5,0);
    \draw (-0.5,0.5) -- (4.5,0.5);
    \foreach \x/\s in {0/very thick,1/,3/,4/very thick}{
        \draw[\s] (\x,-0.4)-- (\x,0.9);
        \node[cmpo] at (\x,0) {};
        \node[mpo] at (\x,0.5) {};
    }
    \node[fill=white] at (2,0) {$\dots$};
    \node[fill=white] at (2,0.5) {$\dots$};
    \draw[decorate, decoration={brace}] (1-0.1,1) -- (3+0.1,1) node[midway,above] {$2M$};
  \end{tikzpicture} =
  \begin{tikzpicture}[xscale=0.6,baseline=1.5mm]
    \draw (-0.5,0) -- (3.7,0);
    \draw (-0.5,0.5) -- (3.7,0.5);
    \foreach \x/\s in {0/very thick,1/,3/}{
        \draw[\s] (\x,-0.4)-- (\x,0.9);
        \node[cmpo] at (\x,0) {};
        \node[mpo] at (\x,0.5) {};
    }
    \node[fill=white] at (2,0) {$\dots$};
    \node[fill=white] at (2,0.5) {$\dots$};
    \draw[decorate, decoration={brace}] (1-0.1,1) -- (3+0.1,1) node[midway,above] {$M$};
    \draw[fusion,gray] (3.7,0) --(3.7,0.5);
  \end{tikzpicture} \otimes
  \begin{tikzpicture}[baseline=1.5mm,xscale=-0.6]
    \draw (-0.5,0) -- (3.7,0);
    \draw (-0.5,0.5) -- (3.7,0.5);
    \foreach \x/\s in {0/very thick,1/,3/}{
        \draw[\s] (\x,-0.4)-- (\x,0.9);
        \node[cmpo] at (\x,0) {};
        \node[mpo] at (\x,0.5) {};
    }
    \node[fill=white] at (2,0) {$\dots$};
    \node[fill=white] at (2,0.5) {$\dots$};
    \draw[decorate, decoration={brace,mirror}] (1-0.1,1) -- (3+0.1,1) node[midway,above] {$M$};
    \draw[fusion,gray] (3.7,0) --(3.7,0.5);
  \end{tikzpicture} =
  \begin{tikzpicture}[baseline=1.5mm,xscale=0.6]
    \draw (-0.5,0) -- (0.7,0);
    \draw (-0.5,0.5) -- (0.7,0.5);
    \foreach \x in {0}{
        \draw[very thick] (\x,-0.4)-- (\x,0.9);
        \node[cmpo] at (\x,0) {};
        \node[mpo] at (\x,0.5) {};
    }
    \draw[fusion,gray] (0.7,0) --(0.7,0.5);
  \end{tikzpicture} \otimes
  \mathbb{1}^{2M} \otimes
  \begin{tikzpicture}[xscale=-0.6,baseline=1.5mm]
    \draw (-0.5,0) -- (0.7,0);
    \draw (-0.5,0.5) -- (0.7,0.5);
    \foreach \x in {0}{
        \draw[very thick] (\x,-0.4)-- (\x,0.9);
        \node[cmpo] at (\x,0) {};
        \node[mpo] at (\x,0.5) {};
    }
    \draw[fusion,gray] (0.7,0) --(0.7,0.5);
  \end{tikzpicture} \ .
  \end{aligned}
\end{equation*}
Tracing out the middle $2M$ physical index in the last equality and dividing by $d^{2M}$, using Eq.~\eqref{eq:rho_l_rho_r} together with the fact that $\rho_g^l=\rho_g$ and $\rho_g^r=\mathbb{1}$ we obtain that
\begin{equation}\label{eq:factorize_after_blocking}
  \begin{tikzpicture}[baseline=1.5mm]
    \def\d{1.8};
    \draw[very thick] (0,-0.4)-- (0,0.9);
    \draw[very thick] (\d,-0.4)-- (\d,0.9);
    \draw (-0.5,0) -- (0.5,0) -- (0.5,0.5)-- (-0.5,0.5);
    \draw (\d+0.5,0) -- (\d-0.5,0) -- (\d-0.5,0.5) node[tensor,midway,label=left:$\rho_g$] {} -- (\d+0.5,0.5);
    \node[cmpo] at (0,0) {};
    \node[cmpo] at (\d,0) {};
    \node[mpo] at (0,0.5) {};
    \node[mpo] at (\d,0.5) {};
  \end{tikzpicture} \  =
  \begin{tikzpicture}[baseline=1.5mm,xscale=0.6]
    \draw (-0.5,0) -- (0.7,0);
    \draw (-0.5,0.5) -- (0.7,0.5);
    \foreach \x in {0}{
        \draw[very thick] (\x,-0.4)-- (\x,0.9);
        \node[cmpo] at (\x,0) {};
        \node[mpo] at (\x,0.5) {};
    }
    \draw[fusion,gray] (0.7,0) --(0.7,0.5);
  \end{tikzpicture} \otimes
  \begin{tikzpicture}[xscale=-0.6,baseline=1.5mm]
    \draw (-0.5,0) -- (0.7,0);
    \draw (-0.5,0.5) -- (0.7,0.5);
    \foreach \x in {0}{
        \draw[very thick] (\x,-0.4)-- (\x,0.9);
        \node[cmpo] at (\x,0) {};
        \node[mpo] at (\x,0.5) {};
    }
    \draw[fusion,gray] (0.7,0) --(0.7,0.5);
  \end{tikzpicture} \   
  =
  \begin{tikzpicture}[baseline=1.5mm]
    \def\d{0.6};
    \draw (-0.5,0) -- (\d+0.5,0);
    \draw (-0.5,0.5) -- (\d+0.5,0.5);
    \draw[very thick] (0,-0.4)-- (0,0.9);
    \draw[very thick] (\d,-0.4)-- (\d,0.9);
    \node[cmpo] at (0,0) {};
    \node[cmpo] at (\d,0) {};
    \node[mpo] at (0,0.5) {};
    \node[mpo] at (\d,0.5) {};
  \end{tikzpicture} ,
\end{equation}
where in the last equality we have used again \eqref{eq:simple_VW_factorize} with $n=0$ and $m,k=L$. Similarly,
\begin{equation}\label{eq:mpo_product_local_identity_after_blocking}
  \begin{tikzpicture}[baseline=1.5mm]
    \def\d{0.6};
    \draw (-0.5,0) -- (\d+0.5,0) --++(0,0.5);
    \draw (\d+0.5,0.5) -- (-0.5,0.5) --++(0,-0.5)  node[tensor,midway,label=left:$\rho_g$] {};
    \draw[very thick] (0,-0.4)-- (0,0.9);
    \draw[very thick] (\d,-0.4)-- (\d,0.9);
    \node[cmpo] at (0,0) {};
    \node[cmpo] at (\d,0) {};
    \node[mpo] at (0,0.5) {};
    \node[mpo] at (\d,0.5) {};
  \end{tikzpicture} =
  \begin{tikzpicture}[baseline=1.5mm,xscale=0.6]
    \draw (-1,0) -- (2,0);
    \draw (-1,0.5) -- (2,0.5);
    \draw[very thick] (0,-0.4)-- (0,0.9);
    \draw[very thick] (1,-0.4)-- (1,0.9);
    \node[cmpo] at (0,0) {};
    \node[cmpo] at (1,0) {};
    \node[mpo] at (0,0.5) {};
    \node[mpo] at (1,0.5) {};
    \draw[fusion,gray] (2,0) --(2,0.5);
    \draw[fusion,gray] (-1,0) --(-1,0.5);
  \end{tikzpicture} = \mathbb{1}^{2L}.
\end{equation}
Similar arguments show that there are $\hat{\rho}^L_g$, $\hat{\rho}_g^R$ positive definite matrices such that
\begin{equation}\label{eq:rho_hat_definition}
  \begin{tikzpicture}[baseline=1.5mm]
    \def\d{0.6};
    \draw (-0.5,0) -- (\d+0.5,0);
    \draw (-0.5,0.5) -- (\d+0.5,0.5);
    \draw[very thick] (0,-0.4)-- (0,0.9);
    \draw[very thick] (\d,-0.4)-- (\d,0.9);
    \node[mpo] at (0,0) {};
    \node[mpo] at (\d,0) {};
    \node[cmpo] at (0,0.5) {};
    \node[cmpo] at (\d,0.5) {};
  \end{tikzpicture} =
  \begin{tikzpicture}[baseline=1.5mm]
    \def\d{2.5};
    \draw[very thick] (0,-0.4)-- (0,0.9);
    \draw[very thick] (\d,-0.4)-- (\d,0.9);
    \draw (-0.5,0) -- (0.5,0) -- (0.5,0.5) node[tensor,midway,label=right:$\hat{\rho}_g^r$] {} -- (-0.5,0.5);
    \draw (\d+0.5,0) -- (\d-0.5,0) -- (\d-0.5,0.5) node[tensor,midway,label=left:$\hat{\rho}_g^l$] {} -- (\d+0.5,0.5);
    \node[mpo] at (0,0) {};
    \node[mpo] at (\d,0) {};
    \node[cmpo] at (0,0.5) {};
    \node[cmpo] at (\d,0.5) {};
  \end{tikzpicture} \quad \text{and} \quad
  \begin{tikzpicture}[baseline=1.5mm]
    \def\d{0.6};
    \draw (-0.5,0) rectangle (\d+0.5,0.5);
    \node[tensor,label=left:$\hat{\rho}_g^l$] (t1) at (-0.5,0.25) {};
    \node[tensor,label=right:$\hat{\rho}_g^r$] (t2) at (\d+0.5,0.25) {};
    \draw[very thick] (0,-0.4)-- (0,0.9);
    \draw[very thick] (\d,-0.4)-- (\d,0.9);
    \node[mpo] at (0,0) {};
    \node[mpo] at (\d,0) {};
    \node[cmpo] at (0,0.5) {};
    \node[cmpo] at (\d,0.5) {};
  \end{tikzpicture} = \mathbb{1} \otimes \mathbb{1},
\end{equation}

Let us reproduce the results in \cite{MPUdef} to be certain that we cite them right. Following \cite{MPUdef}, let us consider a minimal rank decomposition of the MPO tensors:
\begin{equation}\label{eq:mpo_tensors_scmidt_decomposition}
  \begin{tikzpicture}
    \draw (-0.4,0)--(0.4,0);
    \draw[very thick] (0,-0.4)--(0,0.4);
    \node[mpo] at (0,0) {};
  \end{tikzpicture} =
  \begin{tikzpicture}
    \node[mpo] (t) at (0.2,0.2) {};
    \node[mpo] (r) at (-0.2,-0.2) {};
    \draw[red] (t)--(r);
    \draw (t)--++(0.4,0);
    \draw[very thick] (t)--++(0,0.4);
    \draw (r)--++(-0.4,0);
    \draw[very thick] (r)--++(0,-0.4);
  \end{tikzpicture} =
  \begin{tikzpicture}
    \node[mpo] (t) at (0.2,-0.2) {};
    \node[mpo] (r) at (-0.2,0.2) {};
    \draw[blue] (t)--(r);
    \draw (t)--++(0.4,0);
    \draw[very thick] (t)--++(0,-0.4);
    \draw (r)--++(-0.4,0);
    \draw[very thick] (r)--++(0,0.4);
  \end{tikzpicture} \quad \text{and} \quad
  \begin{tikzpicture}
    \draw (-0.4,0)--(0.4,0);
    \draw[very thick] (0,-0.4)--(0,0.4);
    \node[cmpo] at (0,0) {};
  \end{tikzpicture} =
  \begin{tikzpicture}
    \node[cmpo] (t) at (0.2,0.2) {};
    \node[cmpo] (r) at (-0.2,-0.2) {};
    \draw[blue] (t)--(r);
    \draw (t)--++(0.4,0);
    \draw[very thick] (t)--++(0,0.4);
    \draw (r)--++(-0.4,0);
    \draw[very thick] (r)--++(0,-0.4);
  \end{tikzpicture} =
  \begin{tikzpicture}
    \node[cmpo] (t) at (0.2,-0.2) {};
    \node[cmpo] (r) at (-0.2,0.2) {};
    \draw[red] (t)--(r);
    \draw (t)--++(0.4,0);
    \draw[very thick] (t)--++(0,-0.4);
    \draw (r)--++(-0.4,0);
    \draw[very thick] (r)--++(0,0.4);
  \end{tikzpicture} \ .
\end{equation}
As the two tensors are the adjoint of each other, we can choose the minimal rank decompositions such that the two decompositions denoted by red (resp.\ blue) lines are adjoint to each other. By formulas,
\begin{equation*}
\begin{aligned}
    &\sum_{\alpha \beta} u_{\alpha \beta}^{(L)} (g) \otimes \ket{\alpha}\otimes \bra{\beta} = S_{2314}\left(\sum_{i=1}^{l_g(L)} \ket{x_i} \, {\color{red} \otimes}\, \ket{y_i} \right) = S_{1324}\left(\sum_{i=1}^{r_g(L)} \ket{z_i} \,{\color{blue} \otimes}\, \ket{v_i} \right),\\
    &\sum_{\alpha \beta} \left(u_{\alpha \beta}^{(L)} (g)\right)^\dagger  \otimes \bra{\alpha}\otimes \ket{\beta} = S_{2314}\left(\sum_{i=1}^{r_g(L)} \bra{z_i} \,{\color{blue}\otimes}\, \bra{v_i} \right)  = S_{1324}\left(\sum_{i=1}^{r_g(L)} \bra{x_i} \,{\color{red}\otimes}\, \bra{y_i} \right),
\end{aligned}
\end{equation*}
where $S_{ijkl}$ denotes the rearrangement of the tensor components according to the permutation $1\mapsto i$, $2\mapsto j$, $3\mapsto k$ and $4\mapsto l$. The vectors $\ket{x_i}$, $\ket{y_i}$, $\ket{z_i}$ and $\ket{w_i}$ are elements of the vector spaces $\ket{x_i} \in V^* \otimes W$, $\ket{y_i}\in V\otimes W^*$, $\ket{z_i} \in V^*\otimes W$ and $\ket{v_i}\in V^*\otimes W^*$, where $V = \left(\mathbb{C}^d\right)^{\otimes L}$, and $W=\mathbb{C}^D$. Here $\bra{x}$ denotes the linear functional $\ket{y}\mapsto \scalprod{x}{y}$, and thus $\ket{x}\mapsto \bra{x}$ is an anti-linear map, and we have noticed that for a matrix $u\in V\otimes V^*$, $u=\sum_{ij} u_{ij} \ket{i}\bra{j}$, we can write $u^\dagger = S_{12}\left(\sum_{ij} \bar{u}_{ij} \bra{i}\otimes \ket{j}\right)$. We have colored the tensor product sign to make it easier to identify which formula belongs to which picture. The rank-three tensors in \eqref{eq:mpo_tensors_scmidt_decomposition} are then defined as
\begin{equation*}
\begin{aligned}
    \begin{tikzpicture}
        \node[mpo] (t) at (0,0) {};
        \draw[very thick] (t)--++(0,-0.4);
        \draw (t)--++(-0.4,0);
        \draw[red] (t)--++(45:0.4);
    \end{tikzpicture} = \sum_{i=1}^{l_g(L)} \ket{x_i} \otimes \ket{i} \in V^* \otimes W \otimes U_{l,g,L}, \quad
    \begin{tikzpicture}
        \node[cmpo] (t) at (0,0) {};
        \draw[very thick] (t)--++(0,0.4);
        \draw (t)--++(-0.4,0);
        \draw[red] (t)--++(-45:0.4);
    \end{tikzpicture} = \sum_{i=1}^{l_g(L)} \bra{x_i} \otimes \bra{i} \in V \otimes W^* \otimes U_{l,g,L}^*, \\
    \begin{tikzpicture}
        \node[mpo] (t) at (0,0) {};
        \draw[very thick] (t)--++(0,0.4);
        \draw (t)--++(0.4,0);
        \draw[red] (t)--++(-135:0.4);
    \end{tikzpicture} = \sum_{i=1}^{l_g(L)} \ket{y_i} \otimes \bra{i} \in V \otimes W^* \otimes U_{l,g,L}^*, \quad
    \begin{tikzpicture}
        \node[cmpo] (t) at (0,0) {};
        \draw[very thick] (t)--++(0,-0.4);
        \draw (t)--++(0.4,0);
        \draw[red] (t)--++(135:0.4);
    \end{tikzpicture} = \sum_{i=1}^{l_g(L)} \bra{y_i} \otimes \ket{i} \in V^* \otimes W \otimes U_{l,g,L}, \\
    \begin{tikzpicture}
        \node[mpo] (t) at (0,0) {};
        \draw[very thick] (t)--++(0,0.4);
        \draw (t)--++(-0.4,0);
        \draw[blue] (t)--++(-45:0.4);
    \end{tikzpicture} = \sum_{i=1}^{r_g(L)} \ket{z_i} \otimes \ket{i} \in V \otimes W \otimes U_{r,g,L}, \quad
    \begin{tikzpicture}
        \node[cmpo] (t) at (0,0) {};
        \draw[very thick] (t)--++(0,-0.4);
        \draw (t)--++(-0.4,0);
        \draw[blue] (t)--++(45:0.4);
    \end{tikzpicture} = \sum_{i=1}^{r_g(L)} \bra{x_i} \otimes \bra{i} \in V^* \otimes W^* \otimes U_{r,g,L}^*, \\
    \begin{tikzpicture}
        \node[mpo] (t) at (0,0) {};
        \draw[very thick] (t)--++(0,-0.4);
        \draw (t)--++(0.4,0);
        \draw[blue] (t)--++(135:0.4);
    \end{tikzpicture} = \sum_{i=1}^{r_g(L)} \ket{v_i} \otimes \bra{i} \in V^* \otimes W^* \otimes U_{r,g,L}^*, \quad
    \begin{tikzpicture}
        \node[cmpo] (t) at (0,0) {};
        \draw[very thick] (t)--++(0,0.4);
        \draw (t)--++(0.4,0);
        \draw[blue] (t)--++(-135:0.4);
    \end{tikzpicture} = \sum_{i=1}^{r_g(L)} \bra{v_i} \otimes \ket{i} \in V \otimes W \otimes U_{r,g,L},
\end{aligned}
\end{equation*}
where we have introduced the complex inner product vector spaces $U_{l,g,L}$ and $U_{r,g,L}$ with dimensions $l_g(L)$ and $r_g(L)$, respectively, and orthonormal basis $\ket{i}$ ($i=1,\dots, l_g(L)$, and $i=1,\dots, r_g(L)$, respectively).

 As $\rho_g, \hat\rho_g^l,\mathbb{1}$ and $\hat\rho_g^r$ are full rank, and positive, the overlap matrices $\Lambda_1,\Lambda_2,\Lambda_3$ and $\Lambda_4$ defined as
\begin{equation*}
  \begin{aligned}
    &\bra{x_i} \mathbb{1}\otimes \rho_g \ket{x_j} = \left(\Lambda_1\right)_{ij},\\
    & \bra{y_i} \mathbb{1}\otimes \left(\hat\rho_g^r\right)^T \ket{y_j}  = \left(\Lambda_2\right)_{ij},\\
    &\bra{z_i} \mathbb{1}\otimes \hat\rho_g^l \ket{z_j} =\left(\Lambda_3\right)_{ij},\\
    &\bra{v_i} \mathbb{1}\otimes \mathbb{1} \ket{v_j} = \left(\Lambda_4\right)_{ij},
  \end{aligned}
\end{equation*}
are positive definite matrices.
As $\Lambda_1$ and $\Lambda_4$ are positive definite, we can write
\begin{equation*}
  \Lambda_1 = M_1^\dagger M_1 \quad \text{and} \quad \Lambda_4 = M_4^\dagger M_4,
\end{equation*}
with invertible matrices $M_1$ and $M_4$. We can thus modify the minimal rank decomposition to
\begin{equation}
  \begin{aligned}
    &\ket{\hat{x}_i} = \sum_j \left(M_1^{-1}\right)_{ji} \ket{x_j} \quad \text{and} \quad \ket{\hat y_i} = \sum_j \left(M_1\right)_{ij} \ket{y_j},\\
    &\ket{\hat{v}_i} = \sum_{j} \left(M_4^{-1}\right)_{ji} \ket{\hat{v}_i} \quad \text{and} \quad   \ket{\hat{z}_i} = \sum_{j} \left(M_4\right)_{ij} \ket{\hat{z}_i} .
  \end{aligned}
\end{equation}
The vectors $\ket{\hat x_i}$  and $\ket{\hat z_i}$ are now orthonormal:
\begin{equation*}
  \scalprod{\hat{x}_i}{\hat{x}_j} = \delta_{ij} \quad \text{and} \quad \scalprod{\hat{v}_i}{\hat{v}_j} = \delta_{ij}.
\end{equation*}
The vectors $\ket{y_i}$ and $\ket{z_i}$ are not orthonormal, but the matrices
\begin{equation*}
  \left(\Lambda_r\right)_{ij} = \scalprod{\hat{y}_i}{\hat y_j} = \left(\bar M_{1} \Lambda_2 M_1^T\right)_{ij} \quad \text{and} \quad
  \left(\Lambda_l\right)_{ij} = \scalprod{\hat{z}_i}{\hat z_j} = \left(\bar M_{4} \Lambda_3 M_4^T\right)_{ij}
\end{equation*}
are positive definite. Choosing this minimal rank decomposition in \eqref{eq:mpo_tensors_scmidt_decomposition}, the above equations can be expressed using the graphical language as
\begin{equation}\label{eq:schmidt_orthogonality}
    \begin{tikzpicture}[scale=0.5,baseline=-3.5mm,yscale=-1]
      \draw[very thick] (0,0)--(0,1);
      \draw[rounded corners=2mm] (0,0)--(-1,0)--(-1,1) node[tensor,midway,label=left:$\rho_g$] {} --(0,1);
      \draw[red] (0,0) --++ (-45:0.8);
      \draw[red] (0,1) --++ (45:0.8);
      \node[mpo] at (0,0) {};
      \node[cmpo] at (0,1) {};
    \end{tikzpicture} = \mathbb 1, \quad
    \begin{tikzpicture}[scale=0.5,baseline=1.5mm,xscale=-1]
      \draw[very thick] (0,0)--(0,1);
      \draw[rounded corners=2mm] (0,0)--(-1,0)--(-1,1) node[tensor,midway,label=right:$\hat\rho_g^r$] {} --(0,1);
      \draw[red] (0,0) --++ (-45:0.8);
      \draw[red] (0,1) --++ (45:0.8);
      \node[mpo] at (0,0) {};
      \node[cmpo] at (0,1) {};
    \end{tikzpicture} = \Lambda_r, \quad
    \begin{tikzpicture}[scale=0.5,baseline=1.5mm]
      \draw[very thick] (0,0)--(0,1);
      \draw[rounded corners=2mm] (0,0)--(-1,0)--(-1,1) node[tensor,midway,label=left:$\hat\rho_g^l$] {} --(0,1);
      \draw[blue] (0,0) --++ (-45:0.8);
      \draw[blue] (0,1) --++ (45:0.8);
      \node[mpo] at (0,0) {};
      \node[cmpo] at (0,1) {};
    \end{tikzpicture} = \Lambda_l \quad \text{and} \quad
    \begin{tikzpicture}[scale=0.5,baseline=-3.5mm,xscale=-1,yscale=-1]
      \draw[very thick] (0,0)--(0,1);
      \draw[rounded corners=2mm] (0,0)--(-1,0)--(-1,1) --(0,1);
      \draw[blue] (0,0) --++ (-45:0.8);
      \draw[blue] (0,1) --++ (45:0.8);
      \node[mpo] at (0,0) {};
      \node[cmpo] at (0,1) {};
    \end{tikzpicture} = \mathbb 1.
\end{equation}

Using the decomposition \eqref{eq:mpo_tensors_scmidt_decomposition} together with these equations (notice that by construction $\Lambda^r$ and $\Lambda^l$ are invertible), \eqref{eq:factorize_after_blocking} can be simplified to
\begin{equation*}
  \begin{tikzpicture}[baseline=1.5mm]
    \node[cmpo] (t11) at (0,0) {};
    \node[cmpo] (t21) at (0.6,0) {};
    \node[mpo] (t12) at (0,0.5) {};
    \node[mpo] (t22) at (0.6,0.5) {};
    \draw[very thick] (t11)--(t12);
    \draw[very thick] (t21)--(t22);
    \draw (t11)--(t21);
    \draw (t12)--(t22);
    \foreach \t/\d/\c in {t11/225/blue,t21/-45/red,t22/45/red,t12/135/blue}{
      \draw[\c] (\t)--++(\d:0.4);
    }
  \end{tikzpicture} =
  \begin{tikzpicture}[baseline=1.5mm]
    \node[cmpo] (t11) at (0,0) {};
    \node[cmpo] (t21) at (2.0,0) {};
    \node[mpo] (t12) at (0,0.5) {};
    \node[mpo] (t22) at (2.0,0.5) {};
    \draw[very thick] (t11)--(t12);
    \draw[very thick] (t21)--(t22);
    \draw[rounded corners=2mm] (t12)--++(0.5,0)--++(0,-0.5) -- (t11);
    \draw[rounded corners=2mm] (t22)--++(-0.5,0)--++(0,-0.5) node[midway,tensor,label=left:$\rho_g$] {} -- (t21);
    \foreach \t/\d/\c in {t11/225/blue,t21/-45/red,t22/45/red,t12/135/blue}{
      \draw[\c] (\t)--++(\d:0.4);
    }
  \end{tikzpicture}   = \mathbb 1 \otimes \mathbb 1.
\end{equation*}
Similarly, \eqref{eq:mpo_product_local_identity_after_blocking} simplifies to
\begin{equation}\label{eq:v_dagger_v_is_identity}
  \begin{tikzpicture}[baseline=1.5mm]
    \def\d{0.6};
    \node[cmpo] (t1) at (0,0) {};
    \node[cmpo] (t2) at (\d,0) {};
    \node[mpo] (r1) at (0,0.6) {};
    \node[mpo] (r2) at (\d,0.6) {};
    \draw (r1)--(r2);
    \draw (t1)--(t2);
    \draw[red] (r1) to[out=-135,in=135,looseness=2] (t1);
    \draw[blue] (r2) to[out=-45,in=45,looseness=2] (t2);
    \foreach \t/\y in {r1/1,r2/1,t1/-1,t2/-1}{
        \draw[very thick] (\t)--++(0,\y*0.4);
    }
  \end{tikzpicture} = \mathbb{1} \otimes \mathbb{1},
\end{equation}
and \eqref{eq:rho_hat_definition} simplifies to
\begin{equation*}
  \begin{tikzpicture}[baseline=1.5mm]
    \node[mpo] (t11) at (0,0) {};
    \node[mpo] (t21) at (0.6,0) {};
    \node[cmpo] (t12) at (0,0.5) {};
    \node[cmpo] (t22) at (0.6,0.5) {};
    \draw[very thick] (t11)--(t12);
    \draw[very thick] (t21)--(t22);
    \draw (t11)--(t21);
    \draw (t12)--(t22);
    \foreach \t/\d/\c in {t11/225/red,t21/-45/blue,t22/45/blue,t12/135/red}{
      \draw[\c] (\t)--++(\d:0.4);
    }
  \end{tikzpicture} =
  \begin{tikzpicture}[baseline=1.5mm]
    \node[mpo] (t11) at (0,0) {};
    \node[mpo] (t21) at (2.6,0) {};
    \node[cmpo] (t12) at (0,0.5) {};
    \node[cmpo] (t22) at (2.6,0.5) {};
    \draw[very thick] (t11)--(t12);
    \draw[very thick] (t21)--(t22);
    \draw[rounded corners=2mm] (t12)--++(0.5,0)--++(0,-0.5) node[midway,tensor,label=right:$\hat{\rho}_g^r$] {} -- (t11);
    \draw[rounded corners=2mm] (t22)--++(-0.5,0)--++(0,-0.5) node[midway,tensor,label=left:$\hat\rho_g^l$] {} -- (t21);
    \foreach \t/\d/\c in {t11/225/red,t21/-45/blue,t22/45/blue,t12/135/red}{
      \draw[\c] (\t)--++(\d:0.4);
    }
  \end{tikzpicture}   = \Lambda_r \otimes \Lambda_l \quad \text{and} \quad
  \begin{tikzpicture}[baseline=1.5mm]
    \def\d{0.6};
    \node[mpo] (t1) at (0,0) {};
    \node[mpo] (t2) at (\d,0) {};
    \node[cmpo] (r1) at (0,0.6) {};
    \node[cmpo] (r2) at (\d,0.6) {};
    \draw (r1)--(r2);
    \draw (t1)--(t2);
    \draw[draw=blue] (r1) to[out=-135,in=135,looseness=2] node[draw=black,tensor,midway,label=left:$\Lambda_l$] {} (t1);
    \draw[draw=red] (r2) to[out=-45,in=45,looseness=2] node[draw=black,tensor,midway,label=right:$\Lambda_r$] {} (t2);
    \foreach \t/\y in {r1/1,r2/1,t1/-1,t2/-1}{
        \draw[very thick] (\t)--++(0,\y*0.4);
    }
  \end{tikzpicture} = \mathbb{1} \otimes \mathbb{1}.
\end{equation*}
Let us finally notice that
\begin{equation*}
  \Lambda_r^2 \otimes \Lambda_l^2 =
  \begin{tikzpicture}[baseline=6.5mm]
    \foreach \y/\f in {0/black,1/white,2/black,3/white}{
      \node[tensor,fill=\f] (t1\y) at (0,0.5*\y) {};
      \node[tensor,fill=\f] (t2\y) at (0.6,0.5*\y) {};
      \draw (t1\y) -- (t2\y);
    }
    \draw[very thick] (t10)--(t11);
    \draw[very thick] (t20)--(t21);
    \draw[very thick] (t12)--(t13);
    \draw[very thick] (t22)--(t23);
    \foreach \t/\d/\c in {t10/225/red,t20/-45/blue,t23/45/blue,t13/135/red}{
      \draw[\c] (\t)--++(\d:0.4);
    }
    \draw[draw=red] (t11) to[out=135,in=-135,looseness=2] (t12);
    \draw[draw=blue] (t21) to[out=45,in=-45,looseness=2] (t22);
  \end{tikzpicture} =
  \begin{tikzpicture}[baseline=1.5mm]
    \node[mpo] (t11) at (0,0) {};
    \node[mpo] (t21) at (0.6,0) {};
    \node[cmpo] (t12) at (0,0.5) {};
    \node[cmpo] (t22) at (0.6,0.5) {};
    \draw[very thick] (t11)--(t12);
    \draw[very thick] (t21)--(t22);
    \draw (t11)--(t21);
    \draw (t12)--(t22);
    \foreach \t/\d/\c in {t11/225/red,t21/-45/blue,t22/45/blue,t12/135/red}{
      \draw[\c] (\t)--++(\d:0.4);
    }
  \end{tikzpicture} = \Lambda_r \otimes \Lambda_l,
\end{equation*}
where in the second equality we have used \eqref{eq:v_dagger_v_is_identity}. As $\Lambda_l$ and $\Lambda_r$ are positive definite, this implies that $\Lambda_l = \mathbb{1}$ and $\Lambda_r=\mathbb{1}$.
We thus obtain that the matrices
\begin{equation}
  w_g =
  \begin{tikzpicture}
    \node[mpo] (t) at (0,0) {};
    \node[mpo] (r) at (0.5,0) {};
    \draw (t)--(r);
    \draw[very thick] (t)--++(0,0.4);
    \draw[red] (t)--++(-135:0.4);
    \draw[very thick] (r)--++(0,0.4);
    \draw[blue] (r)--++(-45:0.4);
  \end{tikzpicture}
  \quad \text{and} \quad
  v_g =
  \begin{tikzpicture}
    \node[mpo] (t) at (0,0) {};
    \node[mpo] (r) at (0.5,0) {};
    \draw (t)--(r);
    \draw[very thick] (t)--++(0,-0.4);
    \draw[blue] (t)--++(135:0.4);
    \draw[very thick] (r)--++(0,-0.4);
    \draw[red] (r)--++(45:0.4);
  \end{tikzpicture},
\end{equation}
are unitary matrices, and in particular, the product of the two Schmidt ranks $l_g(L)$ and $r_g(L)$ (corresponding to the blue and the red cut, respectively) is $l_g(L) r_g(L) = d^{2L}$.

In \cite{MPUdef} it is proven that $i_g = l_g(L)/r_g(L)$ is independent of $L$ and that $i_{g^n} = (i_g)^n$, for any $n\in\mathbb N$. If $G$ is a finite group, then, as $i_1 = 1$, this implies $i_g = 1$ as well, and thus $l_g(L) =  r_g(L) = d^L$. Therefore in the above decomposition of $w_g$ and $v_g$  all free indices have dimension $d^L$. As the rank of the minimal rank decomposition is $d^L$, we can write
\begin{equation}
  \begin{tikzpicture}
    \node[mpo] (t) at (0,0) {};
    \draw[very thick] (t)--++(0,-0.4);
    \draw[blue] (t)--++(135:0.4);
    \draw (t)--++(0.4,0);
  \end{tikzpicture} \equiv
  \begin{tikzpicture}
    \node[mpo] (t) at (0,0) {};
    \draw[very thick] (t)--++(0,0.4);
    \draw[very thick] (t)--++(0,-0.4);
    \draw (t)--++(0.4,0);
  \end{tikzpicture} \ \quad \text{and} \quad
  \begin{tikzpicture}
    \node[cmpo] (t) at (0,0) {};
    \draw[very thick] (t)--++(0,0.4);
    \draw[blue] (t)--++(-135:0.4);
    \draw (t)--++(0.4,0);
  \end{tikzpicture} \equiv
  \begin{tikzpicture}
    \node[cmpo] (t) at (0,0) {};
    \draw[very thick] (t)--++(0,0.4);
    \draw[very thick] (t)--++(0,-0.4);
    \draw (t)--++(0.4,0);
  \end{tikzpicture} \ .
\end{equation}
Let us now define the maps  $\beta_{g,R}^{n,m}$ and $\beta_{g,L}^{n,m}$ for $n>M$ and $m\geq M$ through
\begin{align}
   &\beta_{g,R}^{n,m} : \
  \begin{tikzpicture}
    \def\d{0.6};
    \draw (-\d,-0.5)--(-\d,0.5);
    \draw (0,-0.5)--(0,0.5);
    \draw (\d,-0.5)--(\d,0.5);
    \node[draw, fill=white,rounded corners, minimum width = \d*3cm] at (0,0) {$X$};
    \draw[decorate, decoration={brace}] (-\d-0.1,0.6) -- (\d+0.1,0.6) node[midway,above] {$n$};
  \end{tikzpicture} \ \mapsto \
  \begin{tikzpicture}
    \def\d{0.6};
    \draw (-\d,0.5) -- (\d*3+0.5,0.5) -- (\d*3+0.5,-0.5) -- (-\d,-0.5);
    \node[irrep] at (2.5*\d,0.5) {$g$};
    \node[irrep] at (2.5*\d,-0.5) {$g$};
    \draw (-\d,-0.9)--(-\d,0.9);
    \node[mpo] at (-\d,0.5) {};
    \node[cmpo] at (-\d,-0.5) {};
    \foreach \x in {0,1,2,3}{
      \draw (\d*\x,-0.9)--(\d*\x,0.9);
      \node[mpo] at (\d*\x,0.5) {};
      \node[cmpo] at (\d*\x,-0.5) {};
    }
    \node[draw, fill=white,rounded corners, minimum width = \d*3cm] at (0,0) {$X$};
    \draw[decorate, decoration={brace}] (-\d-0.1,1) -- (\d+0.1,1) node[midway,above] {$n$};
    \draw[decorate, decoration={brace}] (2*\d-0.1,1) -- (3*\d+0.1,1) node[midway,above] {$m$};
  \end{tikzpicture} \  =
  \begin{tikzpicture}
    \def\d{0.6};
    \draw (-\d,0.5) -- (\d*3+0.5,0.5);
    \draw (-\d,-0.5) -- (\d*3+0.5,-0.5);
    \node[irrep] at (2.5*\d,0.5) {$g$};
    \node[irrep] at (2.5*\d,-0.5) {$g^{-1}$};
    \draw (-\d,-0.9)--(-\d,0.9);
    \node[mpo] at (-\d,0.5) {};
    \node[tensor,fill=gray] at (-\d,-0.5) {};
    \foreach \x in {0,1,2,3}{
      \draw (\d*\x,-0.9)--(\d*\x,0.9);
      \node[mpo] at (\d*\x,0.5) {};
      \node[mpo] at (\d*\x,-0.5) {};
    }
    \node[draw, fill=white,rounded corners, minimum width = \d*3cm] at (0,0) {$X$};
    \draw[fusion] (\d*3+0.5,-0.5) -- (\d*3+0.5,0.5);
    \draw[decorate, decoration={brace}] (-\d-0.1,1) -- (\d+0.1,1) node[midway,above] {$n$};
    \draw[decorate, decoration={brace}] (2*\d-0.1,1) -- (3*\d+0.1,1) node[midway,above] {$m$};
  \end{tikzpicture} \ , \\
   &\beta_{g,L}^{n,m} : \
  \begin{tikzpicture}
    \def\d{0.6};
    \draw (-\d,-0.5)--(-\d,0.5);
    \draw (0,-0.5)--(0,0.5);
    \draw (\d,-0.5)--(\d,0.5);
    \node[draw, fill=white,rounded corners, minimum width = \d*3cm] at (0,0) {$X$};
    \draw[decorate, decoration={brace}] (-\d-0.1,0.6) -- (\d+0.1,0.6) node[midway,above] {$n$};
  \end{tikzpicture} \ \mapsto \
  \begin{tikzpicture}[xscale=-1]
    \def\d{0.6};
    \draw (-\d,0.5) -- (\d*3+0.5,0.5) -- (\d*3+0.5,-0.5) -- (-\d,-0.5);
    \node[irrep] at (2.5*\d,0.5) {$g$};
    \node[irrep] at (2.5*\d,-0.5) {$g$};
    \draw (-\d,-0.9)--(-\d,0.9);
    \node[mpo] at (-\d,0.5) {};
    \node[cmpo] at (-\d,-0.5) {};
    \foreach \x in {0,1,2,3}{
      \draw (\d*\x,-0.9)--(\d*\x,0.9);
      \node[mpo] at (\d*\x,0.5) {};
      \node[cmpo] at (\d*\x,-0.5) {};
    }
    \node[draw, fill=white,rounded corners, minimum width = \d*3cm] at (0,0) {$X$};
    \node[tensor,label=left:$\rho_g$] at (\d*3+0.5,0) {};
    \draw[decorate, decoration={brace,mirror}] (-\d-0.1,1) -- (\d+0.1,1) node[midway,above] {$n$};
    \draw[decorate, decoration={brace,mirror}] (2*\d-0.1,1) -- (3*\d+0.1,1) node[midway,above] {$m$};
  \end{tikzpicture} \  =
  \begin{tikzpicture}[xscale=-1]
    \def\d{0.6};
    \draw (-\d,0.5) -- (\d*3+0.5,0.5);
    \draw (-\d,-0.5) -- (\d*3+0.5,-0.5);
    \node[irrep] at (2.5*\d,0.5) {$g$};
    \node[irrep] at (2.5*\d,-0.5) {$g^{-1}$};
    \draw (-\d,-0.9)--(-\d,0.9);
    \node[mpo] at (-\d,0.5) {};
    \node[tensor,fill=gray] at (-\d,-0.5) {};
    \foreach \x in {0,1,2,3}{
      \draw (\d*\x,-0.9)--(\d*\x,0.9);
      \node[mpo] at (\d*\x,0.5) {};
      \node[mpo] at (\d*\x,-0.5) {};
    }
    \node[draw, fill=white,rounded corners, minimum width = \d*3cm] at (0,0) {$X$};
    \draw[fusion] (\d*3+0.5,-0.5) -- (\d*3+0.5,0.5);
    \draw[decorate, decoration={brace,mirror}] (-\d-0.1,1) -- (\d+0.1,1) node[midway,above] {$n$};
    \draw[decorate, decoration={brace,mirror}] (2*\d-0.1,1) -- (3*\d+0.1,1) node[midway,above] {$m$};
  \end{tikzpicture} \ ,
\end{align}
where the gray tensors are defined as the gauge $X_g$ ($X_g^{-1}$) applied on the white tensors.
Using \eqref{eq:trivial_zipper} (remember that $\mu=1$) we obtain that these maps behave well under tensoring with the identity:
\begin{align*}
    \beta_{g,R}^{n,m}(X \otimes \mathbb 1) = \beta_{g,R}^{n,m+1}(X) = \beta_{g,R}^{n,m}(X) \otimes\mathbb{1} \\
    \beta_{g,L}^{m,n}(\mathbb{1} \otimes X ) = \beta_{g,R}^{m+1,n}(X) = \mathbb{1} \otimes \beta_{g,R}^{n,m}(X) ,
\end{align*}
and thus we can define maps $\beta_{g,L}$ and $\beta_{g,R}$ acting on $\mathcal{A}_{loc}$ that correspond to these two families of maps. We denote these maps with the same graphical notation as the maps $\beta_{g,L/R}^{n,m}$, but without specifying $n$ and $m$.
Using these maps, we arrive at the following decomposition of $\beta_g$:
\begin{equation*}
  \beta_g(X) =
  \begin{tikzpicture}
    \def\d{0.6};
    \draw (-\d*2.5-0.5,0.5) rectangle (\d*2.5+0.5,-0.5);
    \node[irrep] at (-2*\d,0.5) {$g$};
    \node[irrep] at (-2*\d,-0.5) {$g$};
    \foreach \x/\s in {-2.5/,-1.5/,-0.5/{very thick},0.5/{very thick},1.5/,2.5/}{
      \draw[\s] (\d*\x,-0.9)--(\d*\x,0.9);
      \node[mpo] at (\d*\x,0.5) {};
      \node[cmpo] at (\d*\x,-0.5) {};
    }
    \node[draw, fill=white,rounded corners, minimum width = \d*2cm] at (0,0) {$X$};
    \node[tensor,label=left:$\rho_g$] at (-\d*2.5-0.5,0) {};
  \end{tikzpicture} \ =
  \begin{tikzpicture}
    \def\d{0.6};
    \draw (\d*0.5,0.5) --++ (2*\d+0.5,0) --++ (0,-1) -- (0.5*\d,-0.5);
    \draw (-\d*0.5,0.5) --++ (-2*\d-0.5,0) --++ (0,-1) -- (-0.5*\d,-0.5);
    \node[irrep] at (-2*\d,0.5) {$g$};
    \node[irrep] at (-2*\d,-0.5) {$g$};
    \foreach \x/\s in {-2.5/,-1.5/,-0.5/{very thick},0.5/{very thick},1.5/,2.5/}{
      \draw[\s] (\d*\x,-0.9)--(\d*\x,0.9);
      \node[mpo] at (\d*\x,0.5) {};
      \node[cmpo] at (\d*\x,-0.5) {};
    }
    \node[draw, fill=white,rounded corners, minimum width = \d*2cm] at (0,0) {$X$};
    \draw[very thick] (-0.5*\d,1)--++(0,0.5);
    \draw[very thick] (0.5*\d,1)--++(0,0.5);
    \draw[very thick] (-0.5*\d,-1)--++(0,-0.5);
    \draw[very thick] (0.5*\d,-1)--++(0,-0.5);
    \node[draw, fill=white,rounded corners, minimum width = \d*2cm,inner sep=1pt] at (0,1) {$w(g)$};
    \node[draw, fill=white,rounded corners, minimum width = \d*2cm,inner sep=1pt] at (0,-1) {$w(g)^\dagger$};
    \node[tensor,label=left:$\rho_g$] at (-\d*2.5-0.5,0) {};
  \end{tikzpicture} \ ,
\end{equation*}
or
\begin{equation*}
\beta_g(X) = \Ad(w_g) \circ (\beta_{g,L} \otimes \beta_{g,R})(X).
\end{equation*}
Let us prove now that $\beta_{g,L}$ and $\beta_{g,R}$ are $*$-automorphisms that are norm contractive, and thus they can be extended to the quasi-local observables such that this equation still holds. First note that using \eqref{eq:schmidt_orthogonality}, we obtain that for any $n\geq 0$ and $m\geq L$
\begin{equation*}
\begin{aligned}
  \begin{tikzpicture}[baseline=1mm]
    \def\dx{0.6}
    \def\dy{0.4}
    \draw[very thick] (0,-0.4) -- (0,\dy+0.4);
    \draw (0,\dy) --(3*\dx+0.5,\dy);
    \draw (0,0) --(3*\dx+0.5,0);
    \node[cmpo] at (0,0) {};
    \node[mpo] at (0,\dy) {};
    \foreach \x in {1,2,3}{
      \draw (\x*\dx,-0.4) -- (\x*\dx,\dy+0.4);
      \node[cmpo] at (\x*\dx,0) {};
      \node[mpo] at (\x*\dx,\dy) {};
    }
    \node[irrep] at (\dx*0.5,\dy) {$g$};
    \node[irrep] at (\dx*0.5,0) {$g$};
    \draw[decorate, decoration={brace}] (2*\dx-0.1,\dy+0.5) -- (3*\dx+0.1,\dy+0.5) node[midway,above] {$m$};
    \draw[decorate, decoration={brace}] (0*\dx-0.1,\dy+0.5) -- (1*\dx+0.1,\dy+0.5) node[midway,above] {$L+n$};
  \end{tikzpicture}  =
  \begin{tikzpicture}[baseline=1mm]
    \def\dx{0.6}
    \def\dy{0.4}
    \draw[very thick] (0,-0.4) -- (0,\dy+0.4);
    \draw (3*\dx+0.5,0) --(-0.5,0)--++(0,\dy) node[tensor,midway,label=left:$\rho_g$] {} --(3*\dx+0.5,\dy);
    \node[cmpo] at (0,0) {};
    \node[mpo] at (0,\dy) {};
    \foreach \x in {1,2,3}{
      \draw (\x*\dx,-0.4) -- (\x*\dx,\dy+0.4);
      \node[cmpo] at (\x*\dx,0) {};
      \node[mpo] at (\x*\dx,\dy) {};
    }
    \node[irrep] at (\dx*0.5,\dy) {$g$};
    \node[irrep] at (\dx*0.5,0) {$g$};
    \draw[decorate, decoration={brace}] (2*\dx-0.1,\dy+0.5) -- (3*\dx+0.1,\dy+0.5) node[midway,above] {$m$};
    \draw[decorate, decoration={brace}] (0*\dx-0.1,\dy+0.5) -- (1*\dx+0.1,\dy+0.5) node[midway,above] {$L+n$};
  \end{tikzpicture}  =\id^{\otimes (L+n)} \otimes
  \begin{tikzpicture}[baseline=1mm]
    \def\dx{0.6}
    \def\dy{0.4}
    \draw (3*\dx+0.5,0) --(2*\dx-0.5,0)--++(0,\dy) node[tensor,midway,label=left:$\rho_g$] {} --(3*\dx+0.5,\dy);
    \foreach \x in {2,3}{
      \draw (\x*\dx,-0.4) -- (\x*\dx,\dy+0.4);
      \node[cmpo] at (\x*\dx,0) {};
      \node[mpo] at (\x*\dx,\dy) {};
    }
    \node[irrep] at (\dx*2.5,\dy) {$g$};
    \node[irrep] at (\dx*2.5,0) {$g$};
    \draw[decorate, decoration={brace}] (2*\dx-0.1,\dy+0.5) -- (3*\dx+0.1,\dy+0.5) node[midway,above] {$m$};
  \end{tikzpicture} , \\
  \begin{tikzpicture}[baseline=1mm,xscale=-1]
    \def\dx{0.6}
    \def\dy{0.4}
    \draw[very thick] (0,-0.4) -- (0,\dy+0.4);
    \draw (0,\dy) --(3*\dx+0.5,\dy);
    \draw (0,0) --(3*\dx+0.5,0);
    \node[cmpo] at (0,0) {};
    \node[mpo] at (0,\dy) {};
    \foreach \x in {1,2,3}{
      \draw (\x*\dx,-0.4) -- (\x*\dx,\dy+0.4);
      \node[cmpo] at (\x*\dx,0) {};
      \node[mpo] at (\x*\dx,\dy) {};
    }
    \node[irrep] at (\dx*0.5,\dy) {$g$};
    \node[irrep] at (\dx*0.5,0) {$g$};
    \draw[decorate, decoration={brace,mirror}] (2*\dx-0.1,\dy+0.5) -- (3*\dx+0.1,\dy+0.5) node[midway,above] {$m$};
    \draw[decorate, decoration={brace,mirror}] (0*\dx-0.1,\dy+0.5) -- (1*\dx+0.1,\dy+0.5) node[midway,above] {$L+n$};
  \end{tikzpicture}  =
  \begin{tikzpicture}[baseline=1mm,xscale=-1]
    \def\dx{0.6}
    \def\dy{0.4}
    \draw[very thick] (0,-0.4) -- (0,\dy+0.4);
    \draw (3*\dx+0.5,0) --(-0.5,0)--++(0,\dy) --(3*\dx+0.5,\dy);
    \node[cmpo] at (0,0) {};
    \node[mpo] at (0,\dy) {};
    \foreach \x in {1,2,3}{
      \draw (\x*\dx,-0.4) -- (\x*\dx,\dy+0.4);
      \node[cmpo] at (\x*\dx,0) {};
      \node[mpo] at (\x*\dx,\dy) {};
    }
    \node[irrep] at (\dx*0.5,\dy) {$g$};
    \node[irrep] at (\dx*0.5,0) {$g$};
    \draw[decorate, decoration={brace,mirror}] (2*\dx-0.1,\dy+0.5) -- (3*\dx+0.1,\dy+0.5) node[midway,above] {$m$};
    \draw[decorate, decoration={brace,mirror}] (0*\dx-0.1,\dy+0.5) -- (1*\dx+0.1,\dy+0.5) node[midway,above] {$L+n$};
  \end{tikzpicture}  =
  \begin{tikzpicture}[baseline=1mm,xscale=-1]
    \def\dx{0.6}
    \def\dy{0.4}
    \draw (3*\dx+0.5,0) --(2*\dx-0.5,0)--++(0,\dy) --(3*\dx+0.5,\dy);
    \foreach \x in {2,3}{
      \draw (\x*\dx,-0.4) -- (\x*\dx,\dy+0.4);
      \node[cmpo] at (\x*\dx,0) {};
      \node[mpo] at (\x*\dx,\dy) {};
    }
    \node[irrep] at (\dx*2.5,\dy) {$g$};
    \node[irrep] at (\dx*2.5,0) {$g$};
    \draw[decorate, decoration={brace,mirror}] (2*\dx-0.1,\dy+0.5) -- (3*\dx+0.1,\dy+0.5) node[midway,above] {$m$};
  \end{tikzpicture} \otimes \id^{\otimes (L+n)}.
\end{aligned}
\end{equation*}
Similarly, for any $n\geq 0$,
\begin{equation*}
\begin{aligned}
  \begin{tikzpicture}[baseline=1mm]
    \def\dx{0.6}
    \def\dy{0.4}
    \draw[very thick] (0,-0.4) -- (0,\dy+0.4);
    \draw (0,0) --(3*\dx+0.5,0)--++(0,\dy) --(0,\dy);
    \node[cmpo] at (0,0) {};
    \node[mpo] at (0,\dy) {};
    \foreach \x in {1,2,3}{
      \draw (\x*\dx,-0.4) -- (\x*\dx,\dy+0.4);
      \node[cmpo] at (\x*\dx,0) {};
      \node[mpo] at (\x*\dx,\dy) {};
    }
    \node[irrep] at (\dx*0.5,\dy) {$g$};
    \node[irrep] at (\dx*0.5,0) {$g$};
    \draw[decorate, decoration={brace}] (1*\dx-0.1,\dy+0.5) -- (3*\dx+0.1,\dy+0.5) node[midway,above] {$n$};
  \end{tikzpicture}  =
  \begin{tikzpicture}[baseline=1mm]
    \def\dx{0.6}
    \def\dy{0.4}
    \draw[very thick] (0,-0.4) -- (0,\dy+0.4);
    \draw (0,0) --(3*\dx+0.5,0)--++(0,\dy)--(0,\dy);
    \draw (0,0) --(-0.5,0)--++(0,\dy) node[tensor,midway,label=left:$\rho_g$] {} --(0,\dy);
    \node[cmpo] at (0,0) {};
    \node[mpo] at (0,\dy) {};
    \foreach \x in {1,2,3}{
      \draw (\x*\dx,-0.4) -- (\x*\dx,\dy+0.4);
      \node[cmpo] at (\x*\dx,0) {};
      \node[mpo] at (\x*\dx,\dy) {};
    }
    \node[irrep] at (\dx*0.5,\dy) {$g$};
    \node[irrep] at (\dx*0.5,0) {$g$};
    \draw[decorate, decoration={brace}] (1*\dx-0.1,\dy+0.5) -- (3*\dx+0.1,\dy+0.5) node[midway,above] {$n$};
  \end{tikzpicture}  =\mathbb 1,\\
  \begin{tikzpicture}[baseline=1mm,xscale=-1]
    \def\dx{0.6}
    \def\dy{0.4}
    \draw[very thick] (0,-0.4) -- (0,\dy+0.4);
    \draw (0,0) --(3*\dx+0.5,0)--++(0,\dy) node[tensor,midway,label=left:$\rho_g$] {} --(0,\dy);
    \node[cmpo] at (0,0) {};
    \node[mpo] at (0,\dy) {};
    \foreach \x in {1,2,3}{
      \draw (\x*\dx,-0.4) -- (\x*\dx,\dy+0.4);
      \node[cmpo] at (\x*\dx,0) {};
      \node[mpo] at (\x*\dx,\dy) {};
    }
    \node[irrep] at (\dx*0.5,\dy) {$g$};
    \node[irrep] at (\dx*0.5,0) {$g$};
    \draw[decorate, decoration={brace,mirror}] (1*\dx-0.1,\dy+0.5) -- (3*\dx+0.1,\dy+0.5) node[midway,above] {$n$};
  \end{tikzpicture}  =
  \begin{tikzpicture}[baseline=1mm,xscale=-1]
    \def\dx{0.6}
    \def\dy{0.4}
    \draw[very thick] (0,-0.4) -- (0,\dy+0.4);
    \draw (0,0) --(3*\dx+0.5,0)--++(0,\dy) node[tensor,midway,label=left:$\rho_g$] {} --(0,\dy);
    \draw (0,0) --(-0.5,0)--++(0,\dy)--(0,\dy);
    \node[cmpo] at (0,0) {};
    \node[mpo] at (0,\dy) {};
    \foreach \x in {1,2,3}{
      \draw (\x*\dx,-0.4) -- (\x*\dx,\dy+0.4);
      \node[cmpo] at (\x*\dx,0) {};
      \node[mpo] at (\x*\dx,\dy) {};
    }
    \node[irrep] at (\dx*0.5,\dy) {$g$};
    \node[irrep] at (\dx*0.5,0) {$g$};
    \draw[decorate, decoration={brace,mirror}] (1*\dx-0.1,\dy+0.5) -- (3*\dx+0.1,\dy+0.5) node[midway,above] {$n$};
  \end{tikzpicture}  =\mathbb 1.\\
\end{aligned}
\end{equation*}
These equations directly imply that $\beta_{g,L}(\mathbb 1) = \mathbb 1$ and $\beta_{g,R}(\mathbb 1) = \mathbb 1$. To see that $\beta_{g,L}$ and $\beta_{g,R}$ are norm contractive $*$-automorphisms, note that
\begin{equation*}
    \beta_{g,L}(X) = \Ad(w_g)^{-1} \circ \beta_g(\mathbb 1 \otimes X),
\end{equation*}
and thus $\beta_{g,L}$ (and $\beta_{g,R}$ as well) arises as a composition of norm-contractive $*$-automorphisms.

\subsubsection{Equivalence between  both 3-cocycle indices}\label{3eq}

Notice that
\begin{equation}\label{eq:zipper_g_g_inv}
  \begin{tikzpicture}[baseline=1mm]
    \def\dx{0.6}
    \def\dy{0.4}
    \draw[very thick] (0,-0.4) -- (0,\dy+0.4);
    \draw (0,0) --(3*\dx+0.5,0)--++(0,\dy)--(0,\dy);
    \node[cmpo] at (0,0) {};
    \node[mpo] at (0,\dy) {};
    \foreach \x in {1,2,3}{
      \draw (\x*\dx,-0.4) -- (\x*\dx,\dy+0.4);
      \node[cmpo] at (\x*\dx,0) {};
      \node[mpo] at (\x*\dx,\dy) {};
    }
    \node[irrep] at (\dx*0.5,\dy) {$g$};
    \node[irrep] at (\dx*0.5,0) {$g$};
  \end{tikzpicture}  =
  \begin{tikzpicture}[baseline=1mm]
    \def\dx{0.6}
    \def\dy{0.4}
    \draw[very thick] (0,-0.4) -- (0,\dy+0.4);
    \foreach \y/\f in {0/gray,1/black}{
      \draw (0,\y*\dy) --++ (3*\dx+0.4,0);
      \node[tensor,fill=\f] at (0,\y*\dy) {};
    }
    \draw[fusion] (3*\dx+0.4,0) --++(0,\dy);
    \foreach \x in {1,2,3}{
      \draw (\x*\dx,-0.4) -- (\x*\dx,\dy+0.4);
      \foreach \y in {0,1}{
        \node[mpo] at (\x*\dx,\y*\dy) {};
      }
    }
    \node[irrep] at (\dx*0.5,\dy) {$g$};
    \node[irrep] at (\dx*0.5,0) {$g^{-1}$};
  \end{tikzpicture}  =\mathbb 1.
\end{equation}
Let us multiply this equation with an operator of MPO form $X_{g}\in \Mat_d^{\otimes (L+n)} \otimes \mathbb{C}^D$ as follows:
\begin{equation}
  X_{g} =
  \begin{tikzpicture}
    \def\dx{0.6}
    \def\dy{0.4}
    \draw[very thick] (0,-0.4) -- (0,0.4);
    \draw (0,0) --++ (3*\dx+0.4,0);
    \node[mpo] at (0,0) {};
    \foreach \x in {1,2,3}{
      \draw (\x*\dx,-0.4) -- (\x*\dx,0.4);
      \node[mpo] at (\x*\dx,0) {};
    }
    \node[irrep] at (\dx*0.5,0) {$g$};
  \end{tikzpicture} \ .
\end{equation}
We obtain that
\begin{equation}
  X_g =
  \begin{tikzpicture}[baseline=3mm]
    \def\dx{0.6}
    \def\dy{0.4}
    \draw[very thick] (0,-0.4) -- (0,2*\dy+0.4);
    \foreach \y/\f/\l in {0/black/0.8,1/gray/0.4,2/black/0.4}{
      \draw (0,\y*\dy) --++ (3*\dx+\l,0);
      \node[tensor,fill=\f] at (0,\y*\dy) {};
    }
    \draw[densely dotted] (3*\dx+0.4,1.5*\dy) --++(0.4,0);
    \draw (3*\dx+0.8,0.75*\dy) --++(0.4,0);
    \draw[fusion] (3*\dx+0.4,\dy) --++(0,\dy);
    \draw[fusion] (3*\dx+0.8,0) --++(0,1.5*\dy);
    \foreach \x in {1,2,3}{
      \draw (\x*\dx,-0.4) -- (\x*\dx,2*\dy+0.4);
      \foreach \y in {0,1,2}{
        \node[mpo] at (\x*\dx,\y*\dy) {};
      }
    }
    \node[irrep] at (\dx*0.5,2*\dy) {$g$};
    \node[irrep] at (\dx*0.5,\dy) {$g^{-1}$};
    \node[irrep] at (\dx*0.5,0) {$g$};
  \end{tikzpicture} \ = \frac{1}{\omega(g,g^{-1},g)} \
  \begin{tikzpicture}[baseline=3mm]
    \def\dx{0.6}
    \def\dy{0.4}
    \draw[very thick] (0,-0.4) -- (0,2*\dy+0.4);
    \foreach \y/\f/\l in {0/black/0.4,1/gray/0.4,2/black/0.8}{
      \draw (0,\y*\dy) --++ (3*\dx+\l,0);
      \node[tensor,fill=\f] at (0,\y*\dy) {};
    }
    \draw[densely dotted] (3*\dx+0.4,0.5*\dy) --++(0.4,0);
    \draw (3*\dx+0.8,1.25*\dy) --++(0.4,0);
    \draw[fusion] (3*\dx+0.4,0) --++(0,\dy);
    \draw[fusion] (3*\dx+0.8,0.5*\dy) --++(0,1.5*\dy);
    \foreach \x in {1,2,3}{
      \draw (\x*\dx,-0.4) -- (\x*\dx,2*\dy+0.4);
      \foreach \y in {0,1,2}{
        \node[mpo] at (\x*\dx,\y*\dy) {};
      }
    }
    \node[irrep] at (\dx*0.5,2*\dy) {$g$};
    \node[irrep] at (\dx*0.5,\dy) {$g^{-1}$};
    \node[irrep] at (\dx*0.5,0) {$g$};
  \end{tikzpicture} \ = \frac{1}{\omega(g,g^{-1},g)}\
  \begin{tikzpicture}[baseline=3mm]
    \def\dx{0.6}
    \def\dy{0.4}
    \draw[very thick] (0,-0.4) -- (0,2*\dy+0.4);
    \foreach \y/\f/\l in {0/black/0.4,1/gray/0.4,2/black/0.4}{
      \draw (0,\y*\dy) --++ (3*\dx+\l,0);
      \node[tensor,fill=\f] at (0,\y*\dy) {};
    }
    \draw[fusion] (3*\dx+0.4,0) --++(0,\dy);
    \foreach \x in {1,2,3}{
      \draw (\x*\dx,-0.4) -- (\x*\dx,2*\dy+0.4);
      \foreach \y in {0,1,2}{
        \node[mpo] at (\x*\dx,\y*\dy) {};
      }
    }
    \node[irrep] at (\dx*0.5,2*\dy) {$g$};
    \node[irrep] at (\dx*0.5,\dy) {$g^{-1}$};
    \node[irrep] at (\dx*0.5,0) {$g$};
  \end{tikzpicture} \ ,
\end{equation}
and thus, as $X_g$ is left invertible,
\begin{equation*}
  \begin{tikzpicture}[baseline=1mm]
    \def\dx{0.6}
    \def\dy{0.4}
    \draw[very thick] (0,-0.4) -- (0,\dy+0.4);
    \draw (0,0) --(3*\dx+0.5,0)--++(0,\dy) node[tensor,midway,label=right:$\hat{\rho}_g^r$] {} --(0,\dy);
    \node[mpo] at (0,0) {};
    \node[cmpo] at (0,\dy) {};
    \foreach \x in {1,2,3}{
      \draw (\x*\dx,-0.4) -- (\x*\dx,\dy+0.4);
      \node[mpo] at (\x*\dx,0) {};
      \node[cmpo] at (\x*\dx,\dy) {};
    }
    \node[irrep] at (\dx*0.5,\dy) {$g$};
    \node[irrep] at (\dx*0.5,0) {$g$};
  \end{tikzpicture}  =\mathbb 1,
\end{equation*}
we obtain that
\begin{equation}
  \begin{tikzpicture}[baseline=1mm]
    \def\dx{0.6}
    \def\dy{0.4}
    \draw[very thick] (0,-0.4) -- (0,\dy+0.4);
    \foreach \y/\f in {0/black,1/gray}{
      \draw (0,\y*\dy) --++ (3*\dx+0.4,0);
      \node[tensor,fill=\f] at (0,\y*\dy) {};
    }
    \draw[fusion] (3*\dx+0.4,0) --++(0,\dy);
    \foreach \x in {1,2,3}{
      \draw (\x*\dx,-0.4) -- (\x*\dx,\dy+0.4);
      \foreach \y in {0,1}{
        \node[mpo] at (\x*\dx,\y*\dy) {};
      }
    }
    \node[irrep] at (\dx*0.5,\dy) {$g^{-1}$};
    \node[irrep] at (\dx*0.5,0) {$g$};
  \end{tikzpicture}  = \omega(g,g^{-1},g) \cdot \mathbb 1 .
\end{equation}

 We can now express $v(g,h)$ explicitly. It can be chosen (notice the phase degree of freedom) as:
\begin{equation}\label{eq:v graphical}
  v(g,h) =
  \begin{tikzpicture}[baseline=3mm]
    \def\dx{0.6}
    \def\dy{0.4}
    \draw[very thick] (0,-0.4) -- (0,2*\dy+0.4);
    \foreach \y/\f/\l in {0/gray/0.8,1/black/0.4,2/black/0.4}{
      \draw (0,\y*\dy) --++ (3*\dx+\l,0);
      \node[tensor,fill=\f] at (0,\y*\dy) {};
    }
    \draw (3*\dx+0.4,1.5*\dy) --++(0.4,0);
    \draw[fusion] (3*\dx+0.4,\dy) --++(0,\dy);
    \draw[fusion] (3*\dx+0.8,0) --++(0,1.5*\dy);
    \foreach \x in {1,2,3}{
      \draw (\x*\dx,-0.4) -- (\x*\dx,2*\dy+0.4);
      \foreach \y in {0,1,2}{
        \node[mpo] at (\x*\dx,\y*\dy) {};
      }
    }
    \node[irrep] at (\dx*0.5,2*\dy) {$g$};
    \node[irrep] at (\dx*0.5,\dy) {$h$};
    \node[irrep] at (\dx*0.5,0) {$(gh)^{-1}$};
  \end{tikzpicture} \ .
\end{equation}
Let us multiply this expression with $X_{gh}$. We obtain
\begin{align}
  (v(g,h) \otimes 1) \cdot X_{gh} &=
  \begin{tikzpicture}[baseline=3mm]
    \def\dx{0.6}
    \def\dy{0.4}
    \draw[very thick] (0,-\dy-0.4) -- (0,2*\dy+0.4);
    \foreach \y/\f/\l in {-1/black/1.2,0/gray/0.8,1/black/0.4,2/black/0.4}{
      \draw (0,\y*\dy) --++ (3*\dx+\l,0);
      \node[tensor,fill=\f] at (0,\y*\dy) {};
    }
    \draw (3*\dx+0.4,1.5*\dy) --++(0.4,0);
    \draw[densely dotted] (3*\dx+0.8,0.75*\dy) --++(0.4,0);
    \draw (3*\dx+1.2,-0.175*\dy) --++(0.4,0);
    \draw[fusion] (3*\dx+0.4,\dy) --++(0,\dy);
    \draw[fusion] (3*\dx+0.8,0) --++(0,1.5*\dy);
    \draw[fusion] (3*\dx+1.2,-\dy) --++(0,1.75*\dy);
    \foreach \x in {1,2,3}{
      \draw (\x*\dx,-\dy-0.4) -- (\x*\dx,2*\dy+0.4);
      \foreach \y in {-1,0,1,2}{
        \node[mpo] at (\x*\dx,\y*\dy) {};
      }
    }
    \node[irrep] at (\dx*0.5,2*\dy) {$g$};
    \node[irrep] at (\dx*0.5,\dy) {$h$};
    \node[irrep] at (\dx*0.5,0) {$(gh)^{-1}$};
    \node[irrep] at (\dx*0.5,-\dy) {$gh$};
  \end{tikzpicture} \ = \frac{1}{\omega(gh,(gh)^{-1},gh)}
  \begin{tikzpicture}[baseline=1mm]
    \def\dx{0.6}
    \def\dy{0.4}
    \draw[very thick] (0,-\dy-0.4) -- (0,2*\dy+0.4);
    \foreach \y/\f/\l in {-1/black/0.4,0/gray/0.4,1/black/0.4,2/black/0.4}{
      \draw (0,\y*\dy) --++ (3*\dx+\l,0);
      \node[tensor,fill=\f] at (0,\y*\dy) {};
    }
    \draw (3*\dx+0.4,1.5*\dy) --++(0.4,0);
    \draw (3*\dx+0.8,0.5*\dy) --++(0.4,0);
    \draw[densely dotted] (3*\dx+0.4,-0.5*\dy) --++(0.4,0);
    \draw[fusion] (3*\dx+0.4,\dy) --++(0,\dy);
    \draw[fusion] (3*\dx+0.8,-0.5*\dy) --++(0,2*\dy);
    \draw[fusion] (3*\dx+0.4,-\dy) --++(0,\dy);
    \foreach \x in {1,2,3}{
      \draw (\x*\dx,-\dy-0.4) -- (\x*\dx,2*\dy+0.4);
      \foreach \y in {-1,0,1,2}{
        \node[mpo] at (\x*\dx,\y*\dy) {};
      }
    }
    \node[irrep] at (\dx*0.5,2*\dy) {$g$};
    \node[irrep] at (\dx*0.5,\dy) {$h$};
    \node[irrep] at (\dx*0.5,0) {$(gh)^{-1}$};
    \node[irrep] at (\dx*0.5,-\dy) {$gh$};
  \end{tikzpicture} \ = \\
  & =
  \begin{tikzpicture}[baseline=5mm]
    \def\dx{0.6}
    \def\dy{0.4}
    \draw[very thick] (0,\dy-0.4) -- (0,2*\dy+0.4);
    \foreach \y/\f/\l in {1/black/0.4,2/black/0.4}{
      \draw (0,\y*\dy) --++ (3*\dx+\l,0);
      \node[tensor,fill=\f] at (0,\y*\dy) {};
    }
    \draw (3*\dx+0.4,1.5*\dy) --++(0.4,0);
    \draw[fusion] (3*\dx+0.4,\dy) --++(0,\dy);
    \foreach \x in {1,2,3}{
      \draw (\x*\dx,\dy-0.4) -- (\x*\dx,2*\dy+0.4);
      \foreach \y in {1,2}{
        \node[mpo] at (\x*\dx,\y*\dy) {};
      }
    }
    \node[irrep] at (\dx*0.5,2*\dy) {$g$};
    \node[irrep] at (\dx*0.5,\dy) {$h$};
  \end{tikzpicture} \ .
\end{align}
Let us consider now the MPO
\begin{equation}\label{eq:cocycle_rhs}
  \begin{tikzpicture}[baseline=3mm]
    \def\dx{0.6}
    \def\dy{0.4}
    \draw[very thick] (0,-0.4) -- (0,2*\dy+0.4);
    \foreach \y/\l in {0/0.8,1/0.4,2/0.4}{
      \draw (0,\y*\dy) --++ (3*\dx+\l,0);
      \node[mpo] at (0,\y*\dy) {};
    }
    \draw (3*\dx+0.4,1.5*\dy) --++(0.4,0);
    \draw (3*\dx+0.8,0.75*\dy) --++(0.4,0);
    \draw[fusion] (3*\dx+0.4,\dy) --++(0,\dy);
    \draw[fusion] (3*\dx+0.8,0) --++(0,1.5*\dy);
    \foreach \x in {1,2,3}{
      \draw (\x*\dx,-0.4) -- (\x*\dx,2*\dy+0.4);
      \foreach \y in {0,1,2}{
        \node[mpo] at (\x*\dx,\y*\dy) {};
      }
    }
    \node[irrep] at (\dx*0.5,2*\dy) {$g$};
    \node[irrep] at (\dx*0.5,\dy) {$h$};
    \node[irrep] at (\dx*0.5,0) {$k$};
  \end{tikzpicture} \ =
  \begin{tikzpicture}[baseline=3mm]
    \def\dx{0.6}
    \def\dy{0.4}
    \draw[very thick] (0,-0.4) -- (0,2*\dy+0.4);
    \foreach \y/\l in {0/0.4,1/0.4}{
      \draw (0,\y*\dy) --++ (3*\dx+\l,0);
      \node[mpo] at (0,\y*\dy) {};
    }
    \draw (3*\dx+0.4,0.5*\dy) --++(0.4,0);
    \draw[fusion] (3*\dx+0.4,0) --++(0,\dy);
    \foreach \x in {1,2,3}{
      \draw (\x*\dx,-0.4) -- (\x*\dx,2*\dy+0.4);
      \foreach \y in {0,1}{
        \node[mpo] at (\x*\dx,\y*\dy) {};
      }
    }
    \node[irrep] at (\dx*0.5,\dy) {$gh$};
    \node[irrep] at (\dx*0.5,0) {$k$};
    \node[draw, fill=white,rounded corners, minimum width = \dx*3.5cm,inner sep=1pt] at (1.5*\dx,2*\dy) {$v(g,h)$};
  \end{tikzpicture} \ =
  \begin{tikzpicture}[baseline=3mm]
    \def\dx{0.6}
    \def\dy{0.4}
    \draw[very thick] (0,-0.4) -- (0,2*\dy+0.4+0.3);
    \draw (0,0) --++ (3*\dx+0.4,0);
    \node[mpo] at (0,0) {};
    \foreach \x in {1,2,3}{
      \draw (\x*\dx,-0.4) -- (\x*\dx,2*\dy+0.4+0.3);
      \node[mpo] at (\x*\dx,0) {};
    }
    \node[irrep] at (\dx*0.5,0) {$ghk$};
    \node[draw, fill=white,rounded corners, minimum width = \dx*3.5cm,inner sep=1pt] at (1.5*\dx,2*\dy+0.2) {$v(g,h)$};
    \node[draw, fill=white,rounded corners, minimum width = \dx*3.5cm,inner sep=1pt] at (1.5*\dx,1*\dy+0.1) {$v(gh,k)$};
  \end{tikzpicture} \ .
\end{equation}
It can also be written as
\begin{equation*}
  \begin{tikzpicture}[baseline=3mm]
    \def\dx{0.6}
    \def\dy{0.4}
    \draw[very thick] (0,-0.4) -- (0,2*\dy+0.4);
    \foreach \y/\l in {0/0.8,1/0.4,2/0.4}{
      \draw (0,\y*\dy) --++ (3*\dx+\l,0);
      \node[mpo] at (0,\y*\dy) {};
    }
    \draw (3*\dx+0.4,1.5*\dy) --++(0.4,0);
    \draw (3*\dx+0.8,0.75*\dy) --++(0.4,0);
    \draw[fusion] (3*\dx+0.4,\dy) --++(0,\dy);
    \draw[fusion] (3*\dx+0.8,0) --++(0,1.5*\dy);
    \foreach \x in {1,2,3}{
      \draw (\x*\dx,-0.4) -- (\x*\dx,2*\dy+0.4);
      \foreach \y in {0,1,2}{
        \node[mpo] at (\x*\dx,\y*\dy) {};
      }
    }
    \node[irrep] at (\dx*0.5,2*\dy) {$g$};
    \node[irrep] at (\dx*0.5,\dy) {$h$};
    \node[irrep] at (\dx*0.5,0) {$k$};
  \end{tikzpicture} \ = \omega(g,h,k)\cdot
  \begin{tikzpicture}[baseline=3mm]
    \def\dx{0.6}
    \def\dy{0.4}
    \draw[very thick] (0,-0.4) -- (0,2*\dy+0.4);
    \foreach \y/\l in {0/0.4,1/0.4,2/0.8}{
      \draw (0,\y*\dy) --++ (3*\dx+\l,0);
      \node[mpo] at (0,\y*\dy) {};
    }
    \draw (3*\dx+0.8,1.25*\dy) --++(0.4,0);
    \draw (3*\dx+0.4,0.5*\dy) --++(0.4,0);
    \draw[fusion] (3*\dx+0.4,0) --++(0,\dy);
    \draw[fusion] (3*\dx+0.8,0.5*\dy) --++(0,1.5*\dy);
    \foreach \x in {1,2,3}{
      \draw (\x*\dx,-0.4) -- (\x*\dx,2*\dy+0.4);
      \foreach \y in {0,1,2}{
        \node[mpo] at (\x*\dx,\y*\dy) {};
      }
    }
    \node[irrep] at (\dx*0.5,2*\dy) {$g$};
    \node[irrep] at (\dx*0.5,\dy) {$h$};
    \node[irrep] at (\dx*0.5,0) {$k$};
  \end{tikzpicture} \ = \frac{1}{\omega(g,h,k)} \cdot
  \begin{tikzpicture}[baseline=3mm]
    \def\dx{0.6}
    \def\dy{0.4}
    \draw[very thick] (0,-0.4) -- (0,2*\dy+0.4);
    \foreach \y/\l in {0/0.4,2/0.4}{
      \draw (0,\y*\dy) --++ (3*\dx+\l,0);
      \node[mpo] at (0,\y*\dy) {};
    }
    \draw (3*\dx+0.4,\dy) --++(0.4,0);
    \draw[fusion] (3*\dx+0.4,0) --++(0,2*\dy);
    \foreach \x in {1,2,3}{
      \draw (\x*\dx,-0.4) -- (\x*\dx,2*\dy+0.4);
      \foreach \y in {0,2}{
        \node[mpo] at (\x*\dx,\y*\dy) {};
      }
    }
    \node[irrep] at (\dx*0.5,2*\dy) {$g$};
    \node[irrep] at (\dx*0.5,0) {$hk$};
    \node[draw, fill=white,rounded corners, minimum width = \dx*3.5cm,inner sep=1pt] at (1.5*\dx,\dy) {$v(h,k)$};
  \end{tikzpicture} \ .
\end{equation*}
Inserting the identity below $v(h,k)$, we can also write
\begin{equation*}
  \begin{tikzpicture}[baseline=3mm]
    \def\dx{0.6}
    \def\dy{0.4}
    \draw[very thick] (0,-0.4) -- (0,2*\dy+0.4);
    \foreach \y/\l in {0/0.8,1/0.4,2/0.4}{
      \draw (0,\y*\dy) --++ (3*\dx+\l,0);
      \node[mpo] at (0,\y*\dy) {};
    }
    \draw (3*\dx+0.4,1.5*\dy) --++(0.4,0);
    \draw (3*\dx+0.8,0.75*\dy) --++(0.4,0);
    \draw[fusion] (3*\dx+0.4,\dy) --++(0,\dy);
    \draw[fusion] (3*\dx+0.8,0) --++(0,1.5*\dy);
    \foreach \x in {1,2,3}{
      \draw (\x*\dx,-0.4) -- (\x*\dx,2*\dy+0.4);
      \foreach \y in {0,1,2}{
        \node[mpo] at (\x*\dx,\y*\dy) {};
      }
    }
    \node[irrep] at (\dx*0.5,2*\dy) {$g$};
    \node[irrep] at (\dx*0.5,\dy) {$h$};
    \node[irrep] at (\dx*0.5,0) {$k$};
  \end{tikzpicture} \ = \frac{1}{\omega(g,h,k)\cdot\omega(g,g^{-1},g)}\cdot
  \begin{tikzpicture}[baseline=3mm]
    \def\dx{0.6}
    \def\dy{0.4}
    \draw[very thick] (0,-2*\dy-0.4) -- (0,2*\dy+0.4);
    \foreach \y/\f/\l in {-2/black/1.2,-1/black/0.4,0/gray/0.4,2/black/0.8}{
      \draw (0,\y*\dy) --++ (3*\dx+\l,0);
      \node[tensor,fill=\f] at (0,\y*\dy) {};
    }
    \draw[fusion] (3*\dx+0.4,-\dy) --++(0,\dy);
    \draw[densely dotted] (3*\dx+0.4,-0.5*\dy) --++(0.4,0);
    \draw[fusion] (3*\dx+0.8,-0.5*\dy) --++(0,2.5*\dy);
    \draw (3*\dx+0.8,0.75*\dy) --++(0.4,0);
    \draw[fusion] (3*\dx+1.2,-2*\dy) --++(0,2.75*\dy);
    \draw (3*\dx+1.2,-0.625*\dy) --++(0.4,0);
    \foreach \x in {1,2,3}{
      \draw (\x*\dx,-2*\dy-0.4) -- (\x*\dx,2*\dy+0.4);
      \foreach \y in {-2,-1,0,2}{
        \node[mpo] at (\x*\dx,\y*\dy) {};
      }
    }
    \node[irrep] at (\dx*0.5,2*\dy) {$g$};
    \node[irrep] at (\dx*0.5,0) {$g^{-1}$};
    \node[irrep] at (\dx*0.5,-\dy) {$g$};
    \node[irrep] at (\dx*0.5,-2*\dy) {$hk$};
    \node[draw, fill=white,rounded corners, minimum width = \dx*3.5cm,inner sep=1pt] at (1.5*\dx,\dy) {$v(h,k)$};
  \end{tikzpicture} \ .
\end{equation*}
Reordering the fusion tensors, we obtain
\begin{equation*}
   \omega(g,h,k) \cdot
  \begin{tikzpicture}[baseline=3mm]
    \def\dx{0.6}
    \def\dy{0.4}
    \draw[very thick] (0,-0.4) -- (0,2*\dy+0.4);
    \foreach \y/\l in {0/0.8,1/0.4,2/0.4}{
      \draw (0,\y*\dy) --++ (3*\dx+\l,0);
      \node[mpo] at (0,\y*\dy) {};
    }
    \draw (3*\dx+0.4,1.5*\dy) --++(0.4,0);
    \draw (3*\dx+0.8,0.75*\dy) --++(0.4,0);
    \draw[fusion] (3*\dx+0.4,\dy) --++(0,\dy);
    \draw[fusion] (3*\dx+0.8,0) --++(0,1.5*\dy);
    \foreach \x in {1,2,3}{
      \draw (\x*\dx,-0.4) -- (\x*\dx,2*\dy+0.4);
      \foreach \y in {0,1,2}{
        \node[mpo] at (\x*\dx,\y*\dy) {};
      }
    }
    \node[irrep] at (\dx*0.5,2*\dy) {$g$};
    \node[irrep] at (\dx*0.5,\dy) {$h$};
    \node[irrep] at (\dx*0.5,0) {$k$};
  \end{tikzpicture} \ =
  \begin{tikzpicture}[baseline=-1mm]
    \def\dx{0.6}
    \def\dy{0.4}
    \draw[very thick] (0,-2*\dy-0.4) -- (0,2*\dy+0.4);
    \foreach \y/\f/\l in {-2/black/0.4,-1/black/0.4,0/gray/0.4,2/black/0.4}{
      \draw (0,\y*\dy) --++ (3*\dx+\l,0);
      \node[tensor,fill=\f] at (0,\y*\dy) {};
    }
    \draw[fusion] (3*\dx+0.4,0) --++(0,2*\dy);
    \draw[fusion] (3*\dx+0.4,-2*\dy) --++(0,\dy);
    \draw (3*\dx+0.4,-1.5*\dy) --++(0.4,0);
    \foreach \x in {1,2,3}{
      \draw (\x*\dx,-2*\dy-0.4) -- (\x*\dx,2*\dy+0.4);
      \foreach \y in {-2,-1,0,2}{
        \node[mpo] at (\x*\dx,\y*\dy) {};
      }
    }
    \node[irrep] at (\dx*0.5,2*\dy) {$g$};
    \node[irrep] at (\dx*0.5,0) {$g^{-1}$};
    \node[irrep] at (\dx*0.5,-\dy) {$g$};
    \node[irrep] at (\dx*0.5,-2*\dy) {$hk$};
    \node[draw, fill=white,rounded corners, minimum width = \dx*3.5cm,inner sep=1pt] at (1.5*\dx,\dy) {$v(h,k)$};
  \end{tikzpicture} \ =
  \begin{tikzpicture}[baseline=3mm]
    \def\dx{0.6}
    \def\dy{0.4}
    \draw[very thick] (0,-0.4) -- (0,2*\dy+0.4+0.3);
    \draw (0,0) --++ (3*\dx+0.4,0);
    \node[mpo] at (0,0) {};
    \foreach \x in {1,2,3}{
      \draw (\x*\dx,-0.4) -- (\x*\dx,2*\dy+0.4+0.3);
      \node[mpo] at (\x*\dx,0) {};
    }
    \node[irrep] at (\dx*0.5,0) {$ghk$};
    \node[draw, fill=white,rounded corners, minimum width = \dx*3.5cm,inner sep=1pt] at (1.5*\dx,2*\dy+0.2) {$\beta_{gR}(v(g,h))$};
    \node[draw, fill=white,rounded corners, minimum width = \dx*3.5cm,inner sep=1pt] at (1.5*\dx,1*\dy+0.1) {$v(g,hk)$};
  \end{tikzpicture} \ .
\end{equation*}
Comparing the this equation with \eqref{eq:cocycle_rhs}, using invertibility of $X_{ghk}$, we obtain the cocycle equation
\begin{equation}
  \beta_{gR}(v(h,k)) \cdot v(g,hk) = \omega(g,h,k) \cdot v(g,h)\cdot v(gh,k),
\end{equation}
and thus the two cocycles coincide.

\subsection{Injective MPSs satisfy the assumptions}

Let $A^{(x)}$ be injective MPS tensors for each $x\in X$ and $\psi_n(x)$ be the corresponding MPS on $n$ sites such that for each $n\in\mathbb N$, $x\in X$ and $g\in G$,
\begin{equation*}
    U_n(g)^* \psi_n(x) = \psi_n(xg),
\end{equation*}
or graphically,
\begin{equation}\label{eq:MPS_sym}
    \begin{tikzpicture}[baseline = 3mm]
      \def\d{0.6};
      \draw[blue] (-0.5,0) rectangle (3*\d+0.5,-0.5);
      \draw (-0.7,0.4) rectangle (3*\d+0.7,-0.7);
      \node[irrep] at (\d/2,0.4) {$g$};
      \node[irrep] at (\d/2,0) {$x$};
      \foreach \x in {0,1,3}{
        \draw (\x*\d,0) -- (\x*\d,0.8);
        \node[cmps] (t) at (\x*\d,0) {};
        \node[cmpo] (t) at (\x*\d,0.4) {};
      }
      \node[fill=white] at (\d*2,0) {$\dots$};
      \node[fill=white] at (\d*2,0.4) {$\dots$};
    \end{tikzpicture} \ = \ 
    \begin{tikzpicture}
      \def\d{0.6};
      \draw[blue] (-0.5,0) rectangle (3*\d+0.5,-0.5);
      \node[irrep] at (\d/2,0) {$xg$};
      \foreach \x in {0,1,3}{
        \draw (\x*\d,0) -- (\x*\d,0.4);
        \node[cmps] (t) at (\x*\d,0) {};
      }
      \node[fill=white] at (\d*2,0) {$\dots$};
    \end{tikzpicture} \ .
\end{equation}

Let us assume that $A^{(x)}$ is in the right canonical form, $\sum_{\mu} A_\mu^{(x)} (A_{\mu}^{(x)})^* = \mathbb 1$, or graphically,
\begin{equation*}
    \begin{tikzpicture}[baseline=1.5mm]
        \def\d{0.6}
        \draw (-\d/2,0.5) -- (\d/2,0.5) -- (\d/2,0) -- (-\d/2,0);
        \draw (0,0)--(0,0.5);
        \node[cmps] at (0,0) {};
        \node[cmps] at (0,0.5) {};
    \end{tikzpicture} = \mathbb 1.
\end{equation*}
Let $\omega_{x}$ be the translation invariant pure state defined by the matrices $A^{(x)}_{\mu}$, that is,
\begin{equation*}
    \omega_{x}(e^{(0)}_{\nu_0\mu_0} e^{(1)}_{\nu_1\mu_1} \dots e^{(k)}_{\nu_k\mu_k}) = \mathrm{tr} \{ \rho^{(x)} A^{(x)}_{\mu_0} A^{(x)}_{\mu_1} \dots A^{(x)}_{\mu_k} (A^{(x)}_{\nu_k})^* \dots (A^{(x)}_{\nu_1})^*(A^{(x)}_{\nu_0})^* \},
\end{equation*}
or graphically, 
\begin{equation}\label{eq:mps_to_state}
  \omega_{x}(X) = 
  \begin{tikzpicture}
    \def\d{0.6};
    \draw[blue] (-\d,0.5) rectangle (\d*2+0.5,-0.5);
      \node[mps,label=left:$\rho^{(x)}$] at (-\d,0) {};
    \foreach \x in {0,1,2}{
      \draw (\d*\x,-0.5)--(\d*\x,0.5);
      \node[cmps] at (\d*\x,0.5) {};
      \node[mps] at (\d*\x,-0.5) {};
    }
    \node[draw, fill=white,rounded corners, minimum width = \d*3cm] at (\d,0) {$X$};
  \end{tikzpicture} \ ,
\end{equation}
where $\rho^{(x)}$ is the left fixed point of the transfer matrix,
\begin{equation*}
    \begin{tikzpicture}[baseline=1.5mm]
        \def\d{0.6}
        \draw (\d/2,0.5) -- (-\d/2,0.5) -- (-\d/2,0) -- (\d/2,0);
        \draw (0,0)--(0,0.5);
        \node[mps] at (0,0) {};
        \node[mps,label=left:$\rho^{(x)}$] at (-\d/2,0.25) {};
        \node[cmps] at (0,0.5) {};
    \end{tikzpicture} = \sum_{\mu} (A_\mu^{(x)})^* \rho^{(x)} A_\mu^{(x)} = \rho^{(x)}.
\end{equation*}
Let us show that these states satisfy $\omega_x \beta_g = \omega_{xg}$. For that, observe first that Lemma~\ref{lem:main_mps_tool} can be applied on the situation described by \eqref{eq:MPS_sym}, and thus exist operators $V(x,g)$ and $W(x,g)$, denoted graphically as
\begin{equation}\label{fusiontensors}
    V(x,g) = 
    \begin{tikzpicture}[xscale=0.5,yscale=0.25,baseline=1.5mm]
        \draw[blue] (0,0) --++ (-1,0);
        \draw[blue] (0,1) --++ (1,0);
        \draw (0,2) --++ (-1,0);
        \draw[action] (0,0) -- (0,2);
        \node[irrep] at (-0.5,0) {$x$};
        \node[irrep] at (-0.5,2) {$g$};
        \node[irrep] at (0.5,1) {$xg$};
    \end{tikzpicture}
    \quad \text{and} \quad 
    W(x,g) = 
    \begin{tikzpicture}[xscale=-0.5,yscale=0.25,baseline=1.5mm]
        \draw[blue] (0,0) --++ (-1,0);
        \draw[blue] (0,1) --++ (1,0);
        \draw (0,2) --++ (-1,0);
        \draw[action] (0,0) -- (0,2);
        \node[irrep] at (-0.5,0) {$x$};
        \node[irrep] at (-0.5,2) {$g$};
        \node[irrep] at (0.5,1) {$xg$};
    \end{tikzpicture} \ ,
\end{equation}
such that that $V(x,g)W(x,g)=\mathbb 1$ and such that for all $n> 0$,
  \begin{equation}\label{eq:mps_fusion_tensors}
    \begin{tikzpicture}[baseline=1mm]
      \def\d{0.6};
      \draw[blue] (-0.5,0) -- (2*\d+0.5,0);
      \draw (-0.5,0.4) -- (2*\d+0.5,0.4);
      \foreach \x in {0,2}{
        \draw (\x*\d,0) -- (\x*\d,0.8);
        \node[mps] (t) at (\x*\d,0) {};
        \node[cmpo] (t) at (\x*\d,0.4) {};
      }
      \node[fill=white] at (\d*1,0) {$\dots$};
      \node[fill=white] at (\d*1,0.4) {$\dots$};
      \draw[action] (-0.5,0)--++(0,0.4);
      \draw[action] (2*\d+0.5,0)--++(0,0.4);
      \draw (-0.5,0.2)--++(-0.5,0);
      \draw (2*\d+0.5,0.2)--++(0.5,0);
      \node[irrep] at (-0.25,0) {$x$};
      \node[irrep] at (-0.25,0.4) {$g$};
      \node[irrep] at (-0.75,0.2) {$xg$};
      \node[irrep] at (2*\d+0.25,0) {$x$};
      \node[irrep] at (2*\d+0.25,0.4) {$g$};
      \node[irrep] at (2*\d+0.75,0.2) {$xg$};
    \draw[decorate, decoration=brace] (-0.1,0.9) -- (2*\d+0.1,0.9) node[midway,above] {$n$};
    \end{tikzpicture} \  = \
    \begin{tikzpicture}
      \def\d{0.6};
      \draw[blue] (-0.5,0) -- (2*\d+0.5,0);
      \foreach \x in {0,2}{
        \draw (\x*\d,0) -- (\x*\d,0.4);
        \node[mps] (t) at (\x*\d,0) {};
      }
      \node[fill=white] at (\d*1,0) {$\dots$};
      \node[irrep] at (-0.25,0) {$xg$};
      \node[irrep] at (2*\d+0.25,0) {$xg$};
    \draw[decorate, decoration=brace] (-0.1,0.5) -- (2*\d+0.1,0.5) node[midway,above] {$n$};
    \end{tikzpicture} \ .
  \end{equation}
Additionally, there is $M'>0$ such any $n\geq 0$, $k>M'$ and $m>M'$,
\begin{equation}\label{eq:MPO_MPS_compress}
  \begin{tikzpicture}[baseline=1mm]
    \def\d{0.6};
    \draw[blue] (-0.6,0) -- (2*\d+0.6,0);
    \draw (-0.6,0.4) -- (2*\d+0.6,0.4);
    \foreach \x in {0,2}{
      \draw (\x*\d,0) -- (\x*\d,0.8);
      \node[mps] (t) at (\x*\d,0) {};
      \node[cmpo] (t) at (\x*\d,0.4) {};
    }
    \node[fill=white] at (\d*1,0) {$\dots$};
    \node[fill=white] at (\d*1,0.4) {$\dots$};
    \node[irrep] at (-0.3,0) {$x$};
    \node[irrep] at (-0.3,0.4) {$g$};
    \node[irrep] at (2*\d+0.3,0) {$x$};
    \node[irrep] at (2*\d+0.3,0.4) {$g$};
    \draw[decorate, decoration=brace] (-0.1,0.9) -- (2*\d+0.1,0.9) node[midway,above] {$n+k+m$};
  \end{tikzpicture} \  =
  \begin{tikzpicture}[baseline=1mm]
    \def\d{0.6};
    \draw[blue] (-0.6,0) -- (2*\d+0.6,0);
    \draw (-0.6,0.4) -- (2*\d+0.6,0.4);
    \draw[blue] (8*\d-0.6,0) -- (10*\d+0.6,0);
    \draw (8*\d-0.6,0.4) -- (10*\d+0.6,0.4);
    \draw[blue] (2*\d+0.6,0.2) -- (8*\d-0.6,0.2);
    \foreach \x in {0,2,8,10}{
      \draw (\x*\d,0) -- (\x*\d,0.8);
      \node[mps] (t) at (\x*\d,0) {};
      \node[cmpo] (t) at (\x*\d,0.4) {};
    }
    \foreach \x in {4,6}{
      \draw (\x*\d,0.2) -- (\x*\d,0.8);
      \node[mps] (t) at (\x*\d,0.2) {};
    }
    \draw[action] (8*\d-0.6,0)--++(0,0.4);
    \draw[action] (2*\d+0.6,0)--++(0,0.4);
    \node[fill=white] at (\d*1,0) {$\dots$};
    \node[fill=white] at (\d*1,0.4) {$\dots$};
    \node[fill=white] at (\d*9,0) {$\dots$};
    \node[fill=white] at (\d*9,0.4) {$\dots$};
    \node[fill=white] at (\d*5,0.2) {$\dots$};
    \node[irrep] at (-0.3,0) {$x$};
    \node[irrep] at (-0.3,0.4) {$g$};
    \node[irrep] at (2*\d+0.3,0) {$x$};
    \node[irrep] at (2*\d+0.3,0.4) {$g$};
    \node[irrep] at (8*\d-0.3,0) {$x$};
    \node[irrep] at (8*\d-0.3,0.4) {$g$};
    \node[irrep] at (10*\d+0.3,0) {$x$};
    \node[irrep] at (10*\d+0.3,0.4) {$g$};
    \draw[decorate, decoration=brace] (-0.1,0.9) -- (2*\d+0.1,0.9) node[midway,above] {$k$};
    \draw[decorate, decoration=brace] (4*\d-0.1,0.9) -- (6*\d+0.1,0.9) node[midway,above] {$n$};
    \draw[decorate, decoration=brace] (8*\d-0.1,0.9) -- (10*\d+0.1,0.9) node[midway,above] {$m$};
  \end{tikzpicture} \ .
\end{equation}
Using Corollary~\ref{cor:zipper}, the equations 
\begin{equation}\label{eq:action_zipper}
  \begin{aligned}
  \begin{tikzpicture}[baseline=1mm]
    \def\d{0.6};
    \draw[blue] (-0.6,0) -- (2*\d+0.6,0);
    \draw (-0.6,0.4) -- (2*\d+0.6,0.4);
    \draw[blue] (2*\d+0.6,0.2) --++ (0.5,0);
    \foreach \x in {0,2}{
      \draw (\x*\d,0) -- (\x*\d,0.8);
      \node[mps] (t) at (\x*\d,0) {};
      \node[cmpo] (t) at (\x*\d,0.4) {};
    }
    \draw[action] (2*\d+0.6,0)--++(0,0.4);
    \node[fill=white] at (\d*1,0) {$\dots$};
    \node[fill=white] at (\d*1,0.4) {$\dots$};
    \node[irrep] at (-0.3,0) {$x$};
    \node[irrep] at (-0.3,0.4) {$g$};
    \node[irrep] at (2*\d+0.3,0) {$x$};
    \node[irrep] at (2*\d+0.3,0.4) {$g$};
    \draw[decorate, decoration=brace] (-0.1,0.9) -- (2*\d+0.1,0.9) node[midway,above] {$k$};
  \end{tikzpicture} \  =  \
  \begin{tikzpicture}[baseline=1mm]
    \def\d{0.6};
    \draw[blue] (-0.6,0) -- (2*\d+0.6,0);
    \draw (-0.6,0.4) -- (2*\d+0.6,0.4);
    \draw[blue] (2*\d+0.6,0.2) --++ (\d+0.5,0);
    \draw (4*\d,0.2)--++(0,0.6);
    \node[mps] at (4*\d,0.2) {};
    \foreach \x in {0,2}{
      \draw (\x*\d,0) -- (\x*\d,0.8);
      \node[mps] (t) at (\x*\d,0) {};
      \node[cmpo] (t) at (\x*\d,0.4) {};
    }
    \draw[action] (2*\d+0.6,0)--++(0,0.4);
    \node[fill=white] at (\d*1,0) {$\dots$};
    \node[fill=white] at (\d*1,0.4) {$\dots$};
    \node[irrep] at (-0.3,0) {$x$};
    \node[irrep] at (-0.3,0.4) {$g$};
    \node[irrep] at (2*\d+0.3,0) {$x$};
    \node[irrep] at (2*\d+0.3,0.4) {$g$};
    \draw[decorate, decoration=brace] (-0.1,0.9) -- (2*\d+0.1,0.9) node[midway,above] {$k-1$};
  \end{tikzpicture} , \\
  \begin{tikzpicture}[baseline=1mm,xscale=-1]
    \def\d{0.6};
    \draw[blue] (-0.6,0) -- (2*\d+0.6,0);
    \draw (-0.6,0.4) -- (2*\d+0.6,0.4);
    \draw[blue] (2*\d+0.6,0.2) --++ (0.5,0);
    \foreach \x in {0,2}{
      \draw (\x*\d,0) -- (\x*\d,0.8);
      \node[mps] (t) at (\x*\d,0) {};
      \node[cmpo] (t) at (\x*\d,0.4) {};
    }
    \draw[action] (2*\d+0.6,0)--++(0,0.4);
    \node[fill=white] at (\d*1,0) {$\dots$};
    \node[fill=white] at (\d*1,0.4) {$\dots$};
    \node[irrep] at (-0.3,0) {$x$};
    \node[irrep] at (-0.3,0.4) {$g$};
    \node[irrep] at (2*\d+0.3,0) {$x$};
    \node[irrep] at (2*\d+0.3,0.4) {$g$};
    \draw[decorate, decoration={brace,mirror}] (-0.1,0.9) -- (2*\d+0.1,0.9) node[midway,above] {$k$};
  \end{tikzpicture} \  =  \
  \begin{tikzpicture}[baseline=1mm,xscale=-1]
    \def\d{0.6};
    \draw[blue] (-0.6,0) -- (2*\d+0.6,0);
    \draw (-0.6,0.4) -- (2*\d+0.6,0.4);
    \draw[blue] (2*\d+0.6,0.2) --++ (\d+0.5,0);
    \draw (4*\d,0.2)--++(0,0.6);
    \node[mps] at (4*\d,0.2) {};
    \foreach \x in {0,2}{
      \draw (\x*\d,0) -- (\x*\d,0.8);
      \node[mps] (t) at (\x*\d,0) {};
      \node[cmpo] (t) at (\x*\d,0.4) {};
    }
    \draw[action] (2*\d+0.6,0)--++(0,0.4);
    \node[fill=white] at (\d*1,0) {$\dots$};
    \node[fill=white] at (\d*1,0.4) {$\dots$};
    \node[irrep] at (-0.3,0) {$x$};
    \node[irrep] at (-0.3,0.4) {$g$};
    \node[irrep] at (2*\d+0.3,0) {$x$};
    \node[irrep] at (2*\d+0.3,0.4) {$g$};
    \draw[decorate, decoration={brace,mirror}] (-0.1,0.9) -- (2*\d+0.1,0.9) node[midway,above] {$k-1$};
  \end{tikzpicture}  \ ,
  \end{aligned}
\end{equation}
are also satisfied for any $k>M'$. For $g=e$ the MPO is one-dimensional, and the tensors $V(x,e)$ and $W(x,e)$ are chosen to be the identity:
\begin{equation}\label{eq:trivial_action_tensors}
    V(x,e) = 
    \begin{tikzpicture}[xscale=0.5,yscale=0.25,baseline=1.5mm]
        \draw[blue] (0,0) --++ (-1,0);
        \draw[blue] (0,1) --++ (1,0);
        \draw[densely dotted] (0,2) --++ (-1,0);
        \draw[action] (0,0) -- (0,2);
        \node[irrep] at (-0.5,0) {$x$};
        \node[irrep] at (-0.5,2) {$e$};
        \node[irrep] at (0.5,1) {$x$};
    \end{tikzpicture} = \mathbb 1 
    \quad \text{and} \quad 
    W(x,e) = 
    \begin{tikzpicture}[xscale=-0.5,yscale=0.25,baseline=1.5mm]
        \draw[blue] (0,0) --++ (-1,0);
        \draw[blue] (0,1) --++ (1,0);
        \draw[densely dotted] (0,2) --++ (-1,0);
        \draw[action] (0,0) -- (0,2);
        \node[irrep] at (-0.5,0) {$x$};
        \node[irrep] at (-0.5,2) {$e$};
        \node[irrep] at (0.5,1) {$x$};
    \end{tikzpicture} = \mathbb 1 \ .
\end{equation}
Let us denote $V(x,g)^*$ and $W(x,g)^*$ as (notice that the tensors are rotated $180$ degrees w.r.t.\ V and W)
\begin{equation*}
    V(x,g)^* = 
    \begin{tikzpicture}[xscale=0.5,yscale=0.25,baseline=-3.5mm,rotate=180]
        \draw[blue] (0,0) --++ (-1,0);
        \draw[blue] (0,1) --++ (1,0);
        \draw (0,2) --++ (-1,0);
        \draw[caction] (0,0) -- (0,2);
        \node[irrep] at (-0.5,0) {$x$};
        \node[irrep] at (-0.5,2) {$g$};
        \node[irrep] at (0.5,1) {$xg$};
    \end{tikzpicture}
    \quad \text{and} \quad 
    W(x,g)^* = 
    \begin{tikzpicture}[xscale=-0.5,yscale=0.25,baseline=-3.5mm,rotate=180]
        \draw[blue] (0,0) --++ (-1,0);
        \draw[blue] (0,1) --++ (1,0);
        \draw (0,2) --++ (-1,0);
        \draw[caction] (0,0) -- (0,2);
        \node[irrep] at (-0.5,0) {$x$};
        \node[irrep] at (-0.5,2) {$g$};
        \node[irrep] at (0.5,1) {$xg$};
    \end{tikzpicture} \ .
\end{equation*}
These tensors satisfy the adjoint of the Eqns.~\eqref{eq:mps_fusion_tensors}, \eqref{eq:MPO_MPS_compress} and \eqref{eq:action_zipper}. We will need the following lemma:
\begin{lemma}
For large enough $k$, 
\begin{equation}\label{eq:MPS_MPO_ends}
  \begin{tikzpicture}[baseline=5mm]
    \def\d{0.6};
    \draw[blue]  (2*\d+0.6,0) -- (-\d-0.6,0) --++ (0,1.2) node[midway,tensor,blue,label=left:$\color{black}\rho^{(x)}$] {}-- (2*\d+0.6,1.2);
    \draw (-0.6,0.4) -- (2*\d+0.6,0.4);
    \draw[blue] (2*\d+0.6,0.2) --++ (0.5,0);
    \draw (-0.6,0.8) -- (2*\d+0.6,0.8);
    \draw[blue] (2*\d+0.6,1) --++ (0.5,0);
    \foreach \x in {0,2}{
      \draw (\x*\d,0) -- (\x*\d,1.2);
      \node[mps] (t) at (\x*\d,0) {};
      \node[cmpo] (t) at (\x*\d,0.4) {};
      \node[cmps] (t) at (\x*\d,1.2) {};
      \node[mpo] (t) at (\x*\d,0.8) {};
    }
    \draw[action] (2*\d+0.6,0)--++(0,0.4);
    \draw[caction] (2*\d+0.6,0.8)--++(0,0.4);
    \draw[fusion] (-0.6,0.4)--++(0,0.4);
    \foreach \y in {0,0.4,0.8,1.2}{
        \node[fill=white] at (\d*1,\y) {$\dots$};
    }
    \node[irrep] at (-0.3,0) {$x$};
    \node[irrep] at (-0.3,0.4) {$g$};
    \node[irrep] at (2*\d+0.3,0) {$x$};
    \node[irrep] at (2*\d+0.3,0.4) {$g$};
    \draw[decorate, decoration={brace,mirror}] (-0.1,-0.15) -- (2*\d+0.1,-0.15) node[midway,below] {$k$};
  \end{tikzpicture} \ \otimes \ 
  \begin{tikzpicture}[baseline=5mm,xscale=-1]
    \def\d{0.6};
    \draw[blue]  (2*\d+0.6,0) -- (-\d-0.6,0) --++ (0,1.2) -- (2*\d+0.6,1.2);
    \draw (-0.6,0.4) -- (2*\d+0.6,0.4);
    \draw[blue] (2*\d+0.6,0.2) --++ (0.5,0);
    \draw (-0.6,0.8) -- (2*\d+0.6,0.8);
    \draw[blue] (2*\d+0.6,1) --++ (0.5,0);
    \foreach \x in {0,2}{
      \draw (\x*\d,0) -- (\x*\d,1.2);
      \node[mps] (t) at (\x*\d,0) {};
      \node[cmpo] (t) at (\x*\d,0.4) {};
      \node[cmps] (t) at (\x*\d,1.2) {};
      \node[mpo] (t) at (\x*\d,0.8) {};
    }
    \draw[action] (2*\d+0.6,0)--++(0,0.4);
    \draw[caction] (2*\d+0.6,0.8)--++(0,0.4);
    \draw[fusion] (-0.6,0.4)--++(0,0.4);
    \foreach \y in {0,0.4,0.8,1.2}{
        \node[fill=white] at (\d*1,\y) {$\dots$};
    }
    \node[irrep] at (-0.3,0) {$x$};
    \node[irrep] at (-0.3,0.4) {$g$};
    \node[irrep] at (2*\d+0.3,0) {$x$};
    \node[irrep] at (2*\d+0.3,0.4) {$g$};
    \draw[decorate, decoration={brace}] (-0.1,-0.15) -- (2*\d+0.1,-0.15) node[midway,below] {$k$};
  \end{tikzpicture}\ = \rho^{(xg)} \otimes \mathbb 1 \ .  
\end{equation}    
\end{lemma}
\begin{proof}
Eqs.~\eqref{eq:action_zipper} and their conjugate imply for large enough $k$
\begin{align*}
  \begin{tikzpicture}[baseline=5mm,xscale=-1]
    \def\d{0.6};
    \draw[blue]  (2*\d+0.6,0) -- (-\d-0.6,0) --++ (0,1.2) -- (2*\d+0.6,1.2);
    \draw (-0.6,0.4) -- (2*\d+0.6,0.4);
    \draw[blue] (2*\d+0.6,0.2) --++ (0.5,0);
    \draw (-0.6,0.8) -- (2*\d+0.6,0.8);
    \draw[blue] (2*\d+0.6,1) --++ (0.5,0);
    \foreach \x in {0,2}{
      \draw (\x*\d,0) -- (\x*\d,1.2);
      \node[mps] (t) at (\x*\d,0) {};
      \node[cmpo] (t) at (\x*\d,0.4) {};
      \node[cmps] (t) at (\x*\d,1.2) {};
      \node[mpo] (t) at (\x*\d,0.8) {};
    }
    \draw[action] (2*\d+0.6,0)--++(0,0.4);
    \draw[caction] (2*\d+0.6,0.8)--++(0,0.4);
    \draw[fusion] (-0.6,0.4)--++(0,0.4);
    \foreach \y in {0,0.4,0.8,1.2}{
        \node[fill=white] at (\d*1,\y) {$\dots$};
    }
    \foreach \y/\l in {0/x,0.4/g,0.8/g,1.2/x}{
        \node[irrep] at (-0.3,\y) {$\l$};
        \node[irrep] at (2*\d+0.3,\y) {$\l$};
    }
    \draw[decorate, decoration={brace}] (-0.1,-0.15) -- (2*\d+0.1,-0.15) node[midway,below] {$k$};
  \end{tikzpicture} & = 
  \begin{tikzpicture}[baseline=5mm,xscale=-0.6]
    \draw[blue]  (3,0) -- (-2,0) --++ (0,1.2) -- (3,1.2);
    \draw (-1,0.4) -- (3,0.4);
    \draw[blue] (3,0.2) --++ (2,0);
    \draw (-1,0.8) -- (3,0.8);
    \draw[blue] (3,1) --++ (2,0);
    \draw[blue] (4,0.2)--(4,1);
    \node[mps] at (4,0.2) {};
    \node[cmps] at (4,1) {};
    \foreach \x in {0,2}{
      \draw (\x,0) -- (\x,1.2);
      \node[mps] (t) at (\x,0) {};
      \node[cmpo] (t) at (\x,0.4) {};
      \node[cmps] (t) at (\x,1.2) {};
      \node[mpo] (t) at (\x,0.8) {};
    }
    \draw[action] (3,0)--++(0,0.4);
    \draw[caction] (3,0.8)--++(0,0.4);
    \draw[fusion] (-1,0.4)--++(0,0.4);
    \foreach \y in {0,0.4,0.8,1.2}{
        \node[fill=white] at (1,\y) {$\dots$};
    }
    \foreach \y/\l in {0/x,0.4/g,0.8/g,1.2/x}{
        \node[irrep] at (-0.3,\y) {$\l$};
        \node[irrep] at (2.3,\y) {$\l$};
    }
    \draw[decorate, decoration={brace}] (-0.1,-0.15) -- (2.1,-0.15) node[midway,below] {$k-1$};
  \end{tikzpicture}\ , \\
  \begin{tikzpicture}[baseline=5mm]
    \def\d{0.6};
    \draw[blue]  (2*\d+0.6,0) -- (-\d-0.6,0) --++ (0,1.2) node[midway,tensor,blue,label=left:$\color{black}\rho^{(x)}$] {}-- (2*\d+0.6,1.2);
    \draw (-0.6,0.4) -- (2*\d+0.6,0.4);
    \draw[blue] (2*\d+0.6,0.2) --++ (0.5,0);
    \draw (-0.6,0.8) -- (2*\d+0.6,0.8);
    \draw[blue] (2*\d+0.6,1) --++ (0.5,0);
    \foreach \x in {0,2}{
      \draw (\x*\d,0) -- (\x*\d,1.2);
      \node[mps] (t) at (\x*\d,0) {};
      \node[cmpo] (t) at (\x*\d,0.4) {};
      \node[cmps] (t) at (\x*\d,1.2) {};
      \node[mpo] (t) at (\x*\d,0.8) {};
    }
    \draw[action] (2*\d+0.6,0)--++(0,0.4);
    \draw[caction] (2*\d+0.6,0.8)--++(0,0.4);
    \draw[fusion] (-0.6,0.4)--++(0,0.4);
    \foreach \y in {0,0.4,0.8,1.2}{
        \node[fill=white] at (\d*1,\y) {$\dots$};
    }
    \foreach \y/\l in {0/x,0.4/g,0.8/g,1.2/x}{
        \node[irrep] at (-0.3,\y) {$\l$};
        \node[irrep] at (2*\d+0.3,\y) {$\l$};
    }
    \draw[decorate, decoration={brace,mirror}] (-0.1,-0.15) -- (2*\d+0.1,-0.15) node[midway,below] {$k$};
  \end{tikzpicture}  & =    
  \begin{tikzpicture}[baseline=5mm,xscale=0.6]
    \draw[blue]  (3,0) -- (-2,0) --++ (0,1.2) node[midway,tensor,blue,label=left:$\color{black}\rho^{(x)}$] {}-- (3,1.2);
    \draw (-1,0.4) -- (3,0.4);
    \draw[blue] (3,0.2) --++ (2,0);
    \draw (-1,0.8) -- (3,0.8);
    \draw[blue] (3,1) --++ (2,0);
    \node[mps] (a) at (4,0.2) {};
    \node[cmps] (b) at (4,1) {};
    \draw (a)--(b);
    \foreach \x in {0,2}{
      \draw (\x,0) -- (\x,1.2);
      \node[mps] (t) at (\x,0) {};
      \node[cmpo] (t) at (\x,0.4) {};
      \node[cmps] (t) at (\x,1.2) {};
      \node[mpo] (t) at (\x,0.8) {};
    }
    \draw[action] (3,0)--++(0,0.4);
    \draw[caction] (3,0.8)--++(0,0.4);
    \draw[fusion] (-1,0.4)--++(0,0.4);
    \foreach \y in {0,0.4,0.8,1.2}{
        \node[fill=white] at (1,\y) {$\dots$};
    }
    \foreach \y/\l in {0/x,0.4/g,0.8/g,1.2/x}{
        \node[irrep] at (-0.5,\y) {$\l$};
        \node[irrep] at (2.5,\y) {$\l$};
    }
    \draw[decorate, decoration={brace,mirror}] (-0.1,-0.15) -- (2.1,-0.15) node[midway,below] {$k-1$};
  \end{tikzpicture} \  ,
\end{align*}
and thus these matrices are proportional to the right and left fixed point of the transfer matrix belonging to the state $\omega_{xg}$; their tensor product is thus $\kappa\cdot \rho^{(xg)}\otimes \id$ for some $\kappa\in\mathbb C$. But
\begin{equation*}
    1 = \omega_x\circ\beta_g(\mathbb 1) = 
  \begin{tikzpicture}[baseline=5mm,xscale=0.6]
    \def\d{0.6};
    \draw[blue]  (4,0) rectangle (-2,1.2);
    \draw (-1,0.4) rectangle (3,0.8);
    \draw[fusion] (-1,0.4)--++(0,0.4);
    \draw[fusion] (3,0.4)--++(0,0.4);
    \foreach \x in {0,2}{
      \draw (\x,0) -- (\x,1.2);
      \node[mps] (t) at (\x,0) {};
      \node[cmpo] (t) at (\x,0.4) {};
      \node[cmps] (t) at (\x,1.2) {};
      \node[mpo] (t) at (\x,0.8) {};
    }
    \foreach \y in {0,0.4,0.8,1.2}{
        \node[fill=white] at (1,\y) {$\dots$};
    }
    \foreach \y/\l in {0/x,0.4/g,0.8/g,1.2/x}{
        \node[irrep] at (-0.5,\y) {$\l$};
        \node[irrep] at (2.5,\y) {$\l$};
    }
    \draw[decorate, decoration={brace,mirror}] (-0.1,-0.15) -- (2.1,-0.15) node[midway,below] {$2k$};
  \end{tikzpicture}  \ ,  
\end{equation*}
and thus, using Eq.~\eqref{eq:MPO_MPS_compress} with $m=k$ and $n=0$, we obtain 
\begin{equation*}
  1 = 
  \begin{tikzpicture}[baseline=5mm,xscale=0.6]
    \draw[blue]  (3,0) -- (-2,0) --++ (0,1.2) node[midway,tensor,blue,label=left:$\color{black}\rho^{(x)}$] {}-- (3,1.2);
    \draw[blue]  (4,0) -- (9,0) --++ (0,1.2) -- (4,1.2);
    \draw (3,0.4) --++ (-4,0)--++(0,0.4)--++(4,0);
    \draw (4,0.4) --++ (4,0)--++(0,0.4)--++(-4,0);
    \draw[blue] (3,0.2)--++(1,0);
    \draw[blue] (3,1)--++(1,0);
    \foreach \x in {0,2,5,7}{
      \draw (\x,0) -- (\x,1.2);
      \node[mps] (t) at (\x,0) {};
      \node[cmpo] (t) at (\x,0.4) {};
      \node[cmps] (t) at (\x,1.2) {};
      \node[mpo] (t) at (\x,0.8) {};
    }
    \draw[action] (3,0)--++(0,0.4);
    \draw[caction] (3,0.8)--++(0,0.4);
    \draw[action] (4,0)--++(0,0.4);
    \draw[caction] (4,0.8)--++(0,0.4);
    \draw[fusion] (-1,0.4)--++(0,0.4);
    \draw[fusion] (8,0.4)--++(0,0.4);
    \foreach \y in {0,0.4,0.8,1.2}{
        \node[fill=white] at (1,\y) {$\dots$};
        \node[fill=white] at (6,\y) {$\dots$};
    }
    \foreach \y/\l in {0/x,0.4/g,0.8/g,1.2/x}{
        \node[irrep] at (-0.5,\y) {$\l$};
        \node[irrep] at (2.5,\y) {$\l$};
    }
    \draw[decorate, decoration={brace,mirror}] (-0.1,-0.15) -- (2.1,-0.15) node[midway,below] {$k$};
    \draw[decorate, decoration={brace,mirror}] (-0.1,-0.15) -- (2.1,-0.15) node[midway,below] {$k$};
  \end{tikzpicture} \  = 
    \kappa \cdot \tr \rho^{(xg)},
\end{equation*}
so $\kappa = 1$.
\end{proof}

We are now in the position to show that  $\omega_x \beta_g = \omega_{gx}$. For that, consider a local operator $X$ and express $\omega_x \beta_g (X)$ as: 
\begin{equation*}
  \omega_x\beta_g(X) = 
  \begin{tikzpicture}
    \def\d{0.6};
    \draw (-3*\d-0.6,0.5) rectangle (\d*5+0.5,-0.5);
    \draw[blue] (-3*\d-0.9,0.9) rectangle (\d*5+0.7,-0.9);
    \node[tensor,blue,label=left:$\color{black}\rho^{(x)}$] at (-3*\d-0.9,0) {};
    \node[tensor,label=right:$\color{black}\rho_g$] at (-3*\d-0.6,0) {};
    \foreach \x in {-3,-1,0,2,3,5}{
      \draw (\d*\x,-0.9)--(\d*\x,0.9);
      \node[mpo] at (\d*\x,0.5) {};
      \node[cmpo] at (\d*\x,-0.5) {};
      \node[cmps] at (\d*\x,0.9) {};
      \node[mps] at (\d*\x,-0.9) {};
    }
    \foreach \x in {-2,1,4}{
        \foreach \y in {-0.9,-0.5,0.5,0.9}{
            \node[fill=white] at (\x*\d,\y) {$\dots$};
        }
    }
    \node[draw, fill=white,rounded corners, minimum width = \d*3cm] at (\d,0) {$X$};
    \draw[decorate, decoration=brace] (-3*\d-0.1,1.1) -- (-\d+0.1,1.1) node[midway,above] {$m$};
    \draw[decorate, decoration=brace] (3*\d-0.1,1.1) -- (5*\d+0.1,1.1) node[midway,above] {$m$};
    \draw[decorate, decoration=brace] (-0.1,1.1) -- (2*\d+0.1,1.1) node[midway,above] {$k$};
  \end{tikzpicture} \ .
\end{equation*}
The value of this expression is independent of $m$ given that it is sufficiently large. Choosing now a large enough $m$ and applying  \eqref{eq:MPO_MPS_compress}, we obtain that 
\begin{equation*}
  \omega_x\beta_g(X) = 
  \begin{tikzpicture}
    \def\d{0.6};
    \draw (-\d,0.5) -- (-4*\d-0.6,0.5) --++(0,-1) --(-\d,-0.5); 
    \draw (3*\d,0.5) -- (6*\d+0.5,0.5) --++(0,-1) --(3*\d,-0.5); 
    \draw[blue] (-\d,0.9) -- (-4*\d-0.9,0.9) --++(0,-1.8) --(-\d,-0.9); 
    \draw[blue] (3*\d,0.9) -- (6*\d+0.7,0.9) --++(0,-1.8) --(3*\d,-0.9); 
    \node[tensor,blue,label=left:$\color{black}\rho^{(x)}$] at (-4*\d-0.9,0) {};
    \node[tensor,label=right:$\color{black}\rho_g$] at (-4*\d-0.6,0) {};
    \foreach \y in {0.7,-0.7}{
        \draw (-\d,\y) -- (3*\d,\y); 
        \node[fill=white] at (\d,\y) {$\dots$};
    }
    \foreach \x in {-1,3}{
        \draw[caction] (\x*\d,0.5)--++(0,0.4);
        \draw[action] (\x*\d,-0.5)--++(0,-0.4);
    }
    \foreach \x in {-4,-2,4,6}{
      \draw (\d*\x,-0.9)--(\d*\x,0.9);
      \node[mpo] at (\d*\x,0.5) {};
      \node[cmpo] at (\d*\x,-0.5) {};
      \node[cmps] at (\d*\x,0.9) {};
      \node[mps] at (\d*\x,-0.9) {};
    }
    \foreach \x in {0,2}{
      \draw (\d*\x,-0.7)--(\d*\x,0.7);
      \node[cmps] at (\d*\x,0.7) {};
      \node[mps] at (\d*\x,-0.7) {};        
    }
    \foreach \x in {-3,5}{
        \foreach \y in {-0.9,-0.5,0.5,0.9}{
            \node[fill=white] at (\x*\d,\y) {$\dots$};
        }
    }
    \node[draw, fill=white,rounded corners, minimum width = \d*3cm] at (\d,0) {$X$};
    \draw[decorate, decoration=brace] (-4*\d-0.1,1.1) -- (-2*\d+0.1,1.1) node[midway,above] {$m$};
    \draw[decorate, decoration=brace] (4*\d-0.1,1.1) -- (6*\d+0.1,1.1) node[midway,above] {$m$};
    \draw[decorate, decoration=brace] (-0.1,1.1) -- (2*\d+0.1,1.1) node[midway,above] {$k$};
  \end{tikzpicture} \ .
\end{equation*}

Finally using \eqref{eq:MPS_MPO_ends} we obtain that $\omega_x \beta_g(X) = \omega_{xg}(X)$ for any local operator $X$, and thus $\omega_x \beta_g = \omega_{xg}$.

\subsection{Injective MPS are split}

For the MPS $\omega_x$ generated by injective tensor of normal form $a(x)$, let 
$\Phi_{x}$ be the parent interaction associated to $a(x)$
given by Definition 1.4 of
\cite{O1}. By \cite{FCS}, $\omega_x$ is a frustration free gapped ground state for $\Phi_{x}$. Therefore, by \cite{Matsui2}, $\omega_x$ satisfies the split property.
As a result, we may take a GNS representation of $\omega_x$
of the form
$(\caH=\caH_L\otimes\caH_R,
\pi=\pi_L\otimes\pi_R, \Omega)$.
Here $\pi_L$ (resp. $\pi_R$)
is an irreducible representation of $\caA_L$ (resp. $\caA_R$)
on $\caH_L$ (resp. $\caH_R$).
Let
\begin{align}
    \Omega=\sum_{j}\sqrt{\lambda_j}\xi_j^{(L)}\otimes\xi_j^{(R)}
\end{align}
be a Schmidt decomposition, with
$\lambda_j\neq 0$.
Choose and fix unit vectors
$\xi_{j}^L$, $\xi_j^{R}$
in the decomposition and
set
\begin{align}
    \omega_L(A_L):
=\lmk {\xi_j^{L}},{A_L\xi_j^L}\rmk,\quad
A_L\in\caA_L,\quad
\omega_R(A_R):
=\lmk {\xi_j^{R}},{A_R\xi_j^R}\rmk,\quad
A_R\in\caA_R.
\end{align}
They give pure states
$\omega_L$, $\omega_R$ on $\caA_L$, $\caA_R$.
Because $\omega$ is frustration free, we have 
\begin{align}
\omega_L(\Phi_x(X_L))=0,\quad X_L\subset (-\infty,-1],\quad \omega_R(\Phi_x(X_R))=0,\quad X_R\subset [0,\infty).
\end{align}
Therefore, by Lemma 3.16 of \cite{O1},
there exist one rank projections
$p_L,p_R$ in $\Mat_{D_x}$
such that
\begin{align}
    \begin{split}
       & \omega_R\lmk
    e_{\mu_0\nu_0}^{(0)}\otimes e_{\mu_1\nu_1}^{(1)}\otimes
    \cdots e_{\mu_{l}\nu_l}^{(l)}\rmk
        =\Tr p_R a_{\mu_0}(x)\cdots a_{\mu_l}(x) a_{\nu_l}^*(x)\cdots a_{\nu_0}(x)^*,\\
        &\omega_L\lmk
        e_{\mu_{-l}\nu_{-l}}^{(-l)}\otimes \cdots e_{\mu_{-1}\nu_{-1}}^{(-1)}\rmk
        =\Tr p_L 
        \rho_x^{-\frac 12}
        a_{\nu_{-1}}(x)^*\cdots
        a_{\nu_{-l}}(x)^*
        \rho_x a_{\mu_{-l}}(x)\cdots
        a_{\mu_{-1}}(x)(x)\rho_x^{-\frac 12}.
    \end{split}
\end{align}
By Kadison transitivity, there exists a unitary $U\in \caA$
such that
\begin{align}
    \omega\Ad U=\omega_L\otimes\omega_R.
\end{align}
As $p_L$ and $p_R$ are rank one, $\omega_L(X)$ and $\omega_R(X)$ are described as
    \begin{equation}\label{eq:mps_to_right_state}
      \omega_{L}(X) = 
      \begin{tikzpicture}[xscale=-1]
        \def\d{0.6};
        \draw[blue] (-\d-0.2,0.5) -- (\d*2+0.5,0.5) -- (\d*2+0.5,-0.5) -- (-\d-0.2,-0.5);
          \node[t,blue,fill=white,label=above:$\bra{\phi_L}$] at (-\d-0.2,0.5) {};
          \node[t,blue,label=above:$\ket{\phi_L}$] at (-\d-0.2,-0.5) {};
        \node[tensor,blue,label=left:$\color{black}\rho^{(x)}$] at (2*\d+0.5,0) {};
        \foreach \x in {0,1,2}{
          \draw (\d*\x,-0.5)--(\d*\x,0.5);
          \node[cmps] at (\d*\x,0.5) {};
          \node[mps] at (\d*\x,-0.5) {};
        }
        \node[draw, fill=white,rounded corners, minimum width = \d*3cm] at (\d,0) {$X$};
    		\draw[decorate, decoration={brace,mirror}] (-0.1,0.7) -- (2*\d+0.1,0.7) node[midway,above] {$k$};
      \end{tikzpicture} \ , \quad
      \omega_{R}(X) = 
      \begin{tikzpicture}
        \def\d{0.6};
        \draw[blue] (-\d-0.2,0.5) -- (\d*2+0.5,0.5) -- (\d*2+0.5,-0.5) -- (-\d-0.2,-0.5);
          \node[t,blue,fill=white,label=above:$\ket{\phi_R}$] at (-\d-0.2,0.5) {};
          \node[t,blue,label=above:$\bra{\phi_R}$] at (-\d-0.2,-0.5) {};
        \foreach \x in {0,1,2}{
          \draw (\d*\x,-0.5)--(\d*\x,0.5);
          \node[cmps] at (\d*\x,0.5) {};
          \node[mps] at (\d*\x,-0.5) {};
        }
        \node[draw, fill=white,rounded corners, minimum width = \d*3cm] at (\d,0) {$X$};
    		\draw[decorate, decoration=brace] (-0.1,0.7) -- (2*\d+0.1,0.7) node[midway,above] {$k$};
      \end{tikzpicture} \ ,
    \end{equation}
for any $X\in\mathcal{A}$ that is localized on $k$ particles. Here we have written $p_L=\ket{\Phi_L}\bra{\Phi_L}$ and $p_R=\ket{\Phi_R}\bra{\Phi_R}$.

\subsection{Index from the finite size MPU action on the MPS}

The equation 
\begin{equation*}
    \begin{tikzpicture}[baseline = 3mm]
      \def\d{0.6};
      \draw[blue] (-0.5,0) rectangle (3*\d+0.5,-0.5);
      \draw (-0.7,0.4) rectangle (3*\d+0.7,-0.7);
      \draw (-0.9,0.8) rectangle (3*\d+0.9,-0.9);
      \node[irrep] at (\d/2,0.8) {$g$};
      \node[irrep] at (\d/2,0.4) {$h$};
      \node[irrep] at (\d/2,0) {$xgh$};
      \foreach \x in {0,1,3}{
        \draw (\x*\d,0) -- (\x*\d,1.2);
        \node[mps] (t) at (\x*\d,0) {};
        \node[mpo] (t) at (\x*\d,0.4) {};
        \node[mpo] (t) at (\x*\d,0.8) {};
      }
      \node[fill=white] at (\d*2,0) {$\dots$};
      \node[fill=white] at (\d*2,0.4) {$\dots$};
    \end{tikzpicture} \ = \ 
    \begin{tikzpicture}
      \def\d{0.6};
      \draw[blue] (-0.5,0) rectangle (3*\d+0.5,-0.5);
      \node[irrep] at (\d/2,0) {$x$};
      \foreach \x in {0,1,3}{
        \draw (\x*\d,0) -- (\x*\d,0.4);
        \node[mps] (t) at (\x*\d,0) {};
      }
      \node[fill=white] at (\d*2,0) {$\dots$};
    \end{tikzpicture} \ ,
\end{equation*}
hold for all system size $n$, and thus there exist $V(g,h,x)$ and $W(g,h,x)$ such that Eq.~\eqref{eq:reduction} holds. In fact, we can construct these matrices directly, in two different ways. First, let us introduce
    \begin{equation}\label{eq:lightblue_action}
        \begin{tikzpicture}[xscale=0.5,yscale=-0.25,baseline=1.5mm,rotate=180]
            \draw[blue] (0,0) --++ (-1,0);
            \draw[blue] (0,1) --++ (1,0);
            \draw (0,2) --++ (-1,0);
            \draw[halfc action] (0,0) -- (0,2);
            \node[irrep] at (-0.5,0) {$xg$};
            \node[irrep] at (-0.5,2) {$g$};
            \node[irrep] at (0.5,1) {$x$};
        \end{tikzpicture} = 
        \begin{tikzpicture}[xscale=0.5,yscale=-0.25,baseline=1.5mm,rotate=180]
            \draw[blue] (0,0) --++ (-2,0);
            \draw[blue] (0,1) --++ (1,0);
            \draw (0,2) --++ (-2,0) node[midway,mpo,fill=gray,label=above:$X_g$] {};
            \draw[action] (0,0) -- (0,2);
            \node[irrep] at (-0.5,0) {$xg$};
            \node[irrep] at (-0.5,2) {$\ g^{-1}$};
            \node[irrep] at (-1.5,2) {$g$};
            \node[irrep] at (0.5,1) {$x$};
        \end{tikzpicture}
        \quad \text{and} \quad  
        \begin{tikzpicture}[xscale=-0.5,yscale=-0.25,baseline=1.5mm,rotate=180]
            \draw[blue] (0,0) --++ (-1,0);
            \draw[blue] (0,1) --++ (1,0);
            \draw (0,2) --++ (-1,0);
            \draw[halfc action] (0,0) -- (0,2);
            \node[irrep] at (-0.5,0) {$xg$};
            \node[irrep] at (-0.5,2) {$g$};
            \node[irrep] at (0.5,1) {$x$};
        \end{tikzpicture} \ = 
        \begin{tikzpicture}[xscale=-0.5,yscale=-0.25,baseline=1.5mm,rotate=180]
            \draw[blue] (0,0) --++ (-2.2,0);
            \draw[blue] (0,1) --++ (1,0);
            \draw (0,2) --++ (-2.2,0);
            \node[mpo,fill=gray,label=above:$X_g^{-1}$] at (-1.2,2) {};
            \draw[action] (0,0) -- (0,2);
            \node[irrep] at (-0.5,0) {$xg$};
            \node[irrep] at (-0.6,2) {$g^{-1}$};
            \node[irrep] at (-1.7,2) {$g$};
            \node[irrep] at (0.5,1) {$x$};
        \end{tikzpicture} \ .
    \end{equation}
  these tensors satisfy 
  \begin{equation}\label{eq:lightblue_action_reduction}
    \begin{tikzpicture}[baseline=1mm,yscale=1]
      \def\d{0.6};
      \draw[blue] (-0.5,0) -- (2*\d+0.5,0);
      \draw (-0.5,0.4) -- (2*\d+0.5,0.4);
      \foreach \x in {0,2}{
        \draw (\x*\d,0) -- (\x*\d,0.8);
        \node[mps] (t) at (\x*\d,0) {};
        \node[mpo] (t) at (\x*\d,0.4) {};
      }
      \node[fill=white] at (\d*1,0) {$\dots$};
      \node[fill=white] at (\d*1,0.4) {$\dots$};
      \draw (-0.5,0.2)--++(-0.5,0);
      \draw (2*\d+0.5,0.2)--++(0.5,0);
      \draw[halfc action] (-0.5,0)--++(0,0.4);
      \draw[halfc action] (2*\d+0.5,0)--++(0,0.4);
      \node[irrep] at (-0.25,0) {$xg$};
      \node[irrep] at (-0.25,0.4) {$g$};
      \node[irrep] at (-0.75,0.2) {$x$};
      \node[irrep] at (2*\d+0.25,0) {$xg$};
      \node[irrep] at (2*\d+0.25,0.4) {$g$};
      \node[irrep] at (2*\d+0.75,0.2) {$x$};
    \draw[decorate, decoration={brace}] (-0.1,0.9) -- (2*\d+0.1,0.9) node[midway,above] {$n$};
    \end{tikzpicture} \  = \
    \begin{tikzpicture}[yscale=1]
      \def\d{0.6};
      \draw[blue] (-0.5,0) -- (2*\d+0.5,0);
      \foreach \x in {0,2}{
        \draw (\x*\d,0) -- (\x*\d,0.4);
        \node[mps] (t) at (\x*\d,0) {};
      }
      \node[fill=white] at (\d*1,0) {$\dots$};
      \node[irrep] at (-0.25,0) {$x$};
      \node[irrep] at (2*\d+0.25,0) {$x$};
    \draw[decorate, decoration={brace}] (-0.1,0.5) -- (2*\d+0.1,0.5) node[midway,above] {$n$};
    \end{tikzpicture} \ ,
  \end{equation}
and thus we can consider
\begin{equation*}
    V(g,h,x) = 
    \begin{tikzpicture}[xscale=-0.5,yscale=0.4]
      \draw[blue] (0,-0.5)--(1,-0.5) node[midway,irrep] {$\color{black}xg$};
      \draw[blue] (1,0.25)--(2,0.25) node[midway,irrep] {$\color{black}x\ $};
      \draw[blue] (0,-1) -- (-1,-1) node[midway,irrep] {$\color{black}xgh$};
      \draw (0,0) -- (-1,0) node[midway,irrep] {$h$};
      \draw (1,1) -- (-1,1);
      \node[irrep] at (-0.5,1) {$g$};
      \draw[halfc action] (0,-1)--(0,0);
      \draw[halfc action] (1,-0.5)--(1,1);
    \end{tikzpicture} \quad \text{and} \quad 
    W(g,h,x) = 
    \begin{tikzpicture}[xscale=0.5,yscale=0.4]
      \draw[blue] (0,-0.5)--(1,-0.5) node[midway,irrep] {$\color{black}xg$};
      \draw[blue] (1,0.25)--(2,0.25) node[midway,irrep] {$\ \color{black}x$};
      \draw[blue] (0,-1) -- (-1,-1) node[midway,irrep] {$\color{black}xgh$};
      \draw (0,0) -- (-1,0) node[midway,irrep] {$h$};
      \draw (1,1) -- (-1,1);
      \node[irrep] at (-0.5,1) {$g$};
      \draw[halfc action] (0,-1)--(0,0);
      \draw[halfc action] (1,-0.5)--(1,1);
    \end{tikzpicture} \ ,
\end{equation*}
 and second, 
\begin{equation*}
    \hat V(g,h,x) = 
    \begin{tikzpicture}[xscale=-0.5,yscale=-0.4]
      \draw (0,-0.5)--(1,-0.5) node[midway,irrep] {$\color{black}gh$};
      \draw (1,0.25)--(2,0.25) node[midway,irrep] {$\color{black}x$};
      \draw (0,-1) -- (-1,-1) node[midway,irrep] {$\color{black}g$};
      \draw (0,0) -- (-1,0) node[midway,irrep] {$h$};
      \draw[blue] (1,1) -- (-1,1);
      \node[irrep] at (-0.5,1) {$xgh$};
      \draw[fusion] (0,-1)--(0,0);
      \draw[halfc action] (1,-0.5)--(1,1);
    \end{tikzpicture} \quad \text{and} \quad 
    \hat W(g,h,x) = 
    \begin{tikzpicture}[xscale=0.5,yscale=-0.4]
      \draw (0,-0.5)--(1,-0.5) node[midway,irrep] {$\color{black}gh$};
      \draw[blue] (1,0.25)--(2,0.25) node[midway,irrep] {$\color{black}x$};
      \draw (0,-1) -- (-1,-1) node[midway,irrep] {$\color{black}g$};
      \draw (0,0) -- (-1,0) node[midway,irrep] {$h$};
      \draw[blue] (1,1) -- (-1,1);
      \node[irrep] at (-0.5,1) {$xgh$};
      \draw[fusion] (0,-1)--(0,0);
      \draw[halfc action] (1,-0.5)--(1,1);
    \end{tikzpicture} \ .
\end{equation*}
Using then Lemma~\ref{lem:fusion_unique}, we obtain that there is $\hat \sigma_x(g,h)$ such that
  \begin{equation}\label{eq:mps_index_2}
  \begin{aligned}
    \begin{tikzpicture}[baseline=4mm]
      \def\d{0.6};
      \draw[blue] (-0.5,0) -- (2*\d+0.5,0);
      \draw (-0.5,0.4) -- (2*\d+0.5,0.4);
      \draw (-0.5,0.8) -- (2*\d+1,0.8);
      \foreach \x in {0,2}{
        \draw (\x*\d,0) -- (\x*\d,1.2);
        \node[mps] (t) at (\x*\d,0) {};
        \node[mpo] (t) at (\x*\d,0.4) {};
        \node[mpo] (t) at (\x*\d,0.8) {};
      }
      \node[fill=white] at (\d*1,0) {$\dots$};
      \node[fill=white] at (\d*1,0.4) {$\dots$};
      \node[fill=white] at (\d*1,0.8) {$\dots$};
      \draw[blue] (2*\d+0.5,0.2)--++(0.5,0);
      \draw[halfc action] (2*\d+0.5,0)--++(0,0.4);
      \draw[blue] (2*\d+1,0.5)--++(0.5,0);
      \draw[halfc action] (2*\d+1,0.2)--++(0,0.6);
      \node[irrep] at (-0.25,0.8) {$g$};
      \node[irrep] at (-0.25,0.4) {$h$};
      \node[irrep] at (-0.25,0) {$xgh$};
      \node[irrep] at (2*\d+0.25,0.8) {$g$};
      \node[irrep] at (2*\d+0.25,0.4) {$h$};
      \node[irrep] at (2*\d+0.25,0) {$xgh$};
      \node[irrep] at (2*\d+0.75,0.2) {$xg$};
      \node[irrep] at (2*\d+1.3,0.5) {$x$};
      \draw[decorate, decoration=brace] (-0.1,1.3) -- (2*\d+0.1,1.3) node[midway,above] {$m$};
    \end{tikzpicture} \  = \ \hat\sigma_x(g,h) \cdot
    \begin{tikzpicture}[baseline=2mm]
      \def\d{0.6};
      \draw[blue] (-0.5,0) -- (2*\d+1,0);
      \draw (-0.5,0.4) -- (2*\d+0.5,0.4);
      \draw (-0.5,0.8) -- (2*\d+0.5,0.8);
      \foreach \x in {0,2}{
        \draw (\x*\d,0) -- (\x*\d,1.2);
        \node[mps] (t) at (\x*\d,0) {};
        \node[mpo] (t) at (\x*\d,0.4) {};
        \node[mpo] (t) at (\x*\d,0.8) {};
      }
      \node[fill=white] at (\d*1,0) {$\dots$};
      \node[fill=white] at (\d*1,0.4) {$\dots$};
      \node[fill=white] at (\d*1,0.8) {$\dots$};
      \draw (2*\d+0.5,0.6)--++(0.5,0);
      \draw[fusion] (2*\d+0.5,0.4)--++(0,0.4);
      \draw[blue] (2*\d+1,0.3)--++(0.5,0);
      \draw[halfc action] (2*\d+1,0)--++(0,0.6);
      \node[irrep] at (-0.25,0.8) {$g$};
      \node[irrep] at (-0.25,0.4) {$h$};
      \node[irrep] at (-0.25,0) {$xgh$};
      \node[irrep] at (2*\d+0.25,0.8) {$g$};
      \node[irrep] at (2*\d+0.25,0.4) {$h$};
      \node[irrep] at (2*\d+0.25,0) {$xgh$};
      \node[irrep] at (2*\d+0.75,0.6) {$gh$};
      \node[irrep] at (2*\d+1.3,0.3) {$x$};
      \draw[decorate, decoration=brace] (-0.1,1.3) -- (2*\d+0.1,1.3) node[midway,above] {$m$};
    \end{tikzpicture}, \\
    \begin{tikzpicture}[baseline=4mm,xscale=-1]
      \def\d{0.6};
      \draw[blue] (-0.5,0) -- (2*\d+0.5,0);
      \draw (-0.5,0.4) -- (2*\d+0.5,0.4);
      \draw (-0.5,0.8) -- (2*\d+1,0.8);
      \foreach \x in {0,2}{
        \draw (\x*\d,0) -- (\x*\d,1.2);
        \node[mps] (t) at (\x*\d,0) {};
        \node[mpo] (t) at (\x*\d,0.4) {};
        \node[mpo] (t) at (\x*\d,0.8) {};
      }
      \node[fill=white] at (\d*1,0) {$\dots$};
      \node[fill=white] at (\d*1,0.4) {$\dots$};
      \node[fill=white] at (\d*1,0.8) {$\dots$};
      \draw[blue] (2*\d+0.5,0.2)--++(0.5,0);
      \draw[halfc action] (2*\d+0.5,0)--++(0,0.4);
      \draw[blue] (2*\d+1,0.5)--++(0.5,0);
      \draw[halfc action] (2*\d+1,0.2)--++(0,0.6);
      \node[irrep] at (-0.25,0.8) {$g$};
      \node[irrep] at (-0.25,0.4) {$h$};
      \node[irrep] at (-0.25,0) {$xgh$};
      \node[irrep] at (2*\d+0.25,0.8) {$g$};
      \node[irrep] at (2*\d+0.25,0.4) {$h$};
      \node[irrep] at (2*\d+0.25,0) {$xgh$};
      \node[irrep] at (2*\d+0.75,0.2) {$xg$};
      \node[irrep] at (2*\d+1.3,0.5) {$x$};
      \draw[decorate, decoration={brace,mirror}] (-0.1,1.3) -- (2*\d+0.1,1.3) node[midway,above] {$m$};
    \end{tikzpicture} \  = \ \frac{1}{\hat\sigma_x(g,h)} \cdot
    \begin{tikzpicture}[baseline=2mm,xscale=-1]
      \def\d{0.6};
      \draw[blue] (-0.5,0) -- (2*\d+1,0);
      \draw (-0.5,0.4) -- (2*\d+0.5,0.4);
      \draw (-0.5,0.8) -- (2*\d+0.5,0.8);
      \foreach \x in {0,2}{
        \draw (\x*\d,0) -- (\x*\d,1.2);
        \node[mps] (t) at (\x*\d,0) {};
        \node[mpo] (t) at (\x*\d,0.4) {};
        \node[mpo] (t) at (\x*\d,0.8) {};
      }
      \node[fill=white] at (\d*1,0) {$\dots$};
      \node[fill=white] at (\d*1,0.4) {$\dots$};
      \node[fill=white] at (\d*1,0.8) {$\dots$};
      \draw (2*\d+0.5,0.6)--++(0.5,0);
      \draw[fusion] (2*\d+0.5,0.4)--++(0,0.4);
      \draw[blue] (2*\d+1,0.3)--++(0.5,0);
      \draw[halfc action] (2*\d+1,0)--++(0,0.6);
      \node[irrep] at (-0.25,0.8) {$g$};
      \node[irrep] at (-0.25,0.4) {$h$};
      \node[irrep] at (-0.25,0) {$xgh$};
      \node[irrep] at (2*\d+0.25,0.8) {$g$};
      \node[irrep] at (2*\d+0.25,0.4) {$h$};
      \node[irrep] at (2*\d+0.25,0) {$xgh$};
      \node[irrep] at (2*\d+0.75,0.6) {$gh$};
      \node[irrep] at (2*\d+1.3,0.3) {$x$};
      \draw[decorate, decoration={brace,mirror}] (-0.1,1.3) -- (2*\d+0.1,1.3) node[midway,above] {$m$};
    \end{tikzpicture} \  .
  \end{aligned}
  \end{equation}
It is straightforward to check that $\hat{\sigma}_x(g,h)$ satisfies Eq.~\eqref{eq:sigmaconstraint} with the three-cocycle $\omega$, that is,
\begin{align}
\hat\sigma_x(g,h) \hat\sigma_x(gh,k) =\omega(g,h,k)\hat\sigma_{xg}(h,k) \hat\sigma_x(g,hk).
\end{align}
The value of $\hat\sigma_x(g,h)$ depends on the concrete choice of the action and fusion tensors; given fixed fusion tensors, different choices of the action tensors, using Lemma~\ref{lem:fusion_unique}, lead to other values, denoted by $\hat\sigma_x'(g,h)$, that are related to $\hat\sigma_x(g,h)$ by
\begin{equation}\label{equivclassigmaMPS}  
\hat{\sigma}_x'(g,h)  = \hat{\sigma}_x(g,h) \cdot \frac{\alpha(x,gh)}{\alpha(xg,h)\alpha(x,g)}\ \forall g,h\in  G \ , 
\end{equation}
for some values $\alpha(x,g)$. We will call $\hat{\sigma}$ and $\hat{\sigma}'$ equivalent,
\begin{equation}\label{equivclassigmaMPS} \hat{\sigma}_x(g,h)  \sim \hat{\sigma}_x(g,h) \cdot \frac{\alpha(x,gh)}{\alpha(xg,h)\alpha(x,g)}\ \forall g,h\in  G \ .
\end{equation}
This defines the equivalent classes of $\hat{\sigma}_x(g,h)$ denoted by $[\hat{\sigma}_x(g,h)]$ satisfying Eq. \eqref{eq:sigmaconstraint}. In the following we show that the index defined this way coincides with the index defined through the general formalism.

\subsection{The GNS representation of a MPS}

Here we give an explicit construction for the GNS representation of a MPS. Intuitively, the Hilbert space is defined as finite range deformations of the infinite MPS and the action of the local operators is the obvious one; the cyclic vector is the MPS itself. 

Let $\omega$ be the state defined by the injective/normal MPS $A$ given by matrices $A_{i}\in\mathcal M_D$, $i=1\dots d$. Another way to think of the MPS tensor is as a collection of vectors $a_{\alpha\beta}\in\mathbb C^d$, with $\alpha,\beta=1\dots D$ defined implicitly through the equation $A = \sum_i \ket{i} \otimes A_i = \sum_{\alpha\beta} a_{\alpha\beta} \otimes e_{\alpha\beta}$. Due to the injectivity condition, the transfer matrix has a unique left and right fixed point, both of which are positive and full rank. W.l.o.g.\ we assume that the MPS in the right canonical form, that is, the right fixed point is the identity, while the left fixed point is an invertible positive operator $\rho\in\Mat_D$: 
\begin{equation}\label{eq:mps_fixpoint}
    \sum_i A_i  A_i^\dagger = \mathbb 1_{\Mat_D} \quad \text{and} \quad \sum_i A_i^\dagger \rho A_i = \rho,
\end{equation}
or graphically,
\begin{equation*}
    \begin{tikzpicture}[baseline=1.5mm]
        \def\d{0.6}
        \draw (-\d/2,0.5) -- (\d/2,0.5) -- (\d/2,0) -- (-\d/2,0);
        \draw (0,0)--(0,0.5);
        \node[mps] at (0,0) {};
        \node[cmps] at (0,0.5) {};
    \end{tikzpicture} = \mathbb 1 
    \quad \text{and} \quad
    \begin{tikzpicture}[baseline=1.5mm]
        \def\d{0.6}
        \draw (\d/2,0.5) -- (-\d/2,0.5) -- (-\d/2,0) node[tensor,blue,midway,label=left:$\rho$] {}-- (\d/2,0);
        \draw (0,0)--(0,0.5);
        \node[mps] at (0,0) {};
        \node[cmps] at (0,0.5) {};
    \end{tikzpicture} = \rho.
\end{equation*}
Let us express these equations with the vectors $a_{\alpha\beta}$:
\begin{equation*}
    \sum_{\beta} \langle a_{\gamma \beta}| a_{\alpha \beta} \rangle = \delta_{\alpha \gamma} \quad \text{and} \quad \sum_{\beta\gamma} \rho_{\gamma \beta} \cdot  \langle a_{\gamma \delta} | a_{\beta \alpha}\rangle = \rho_{\delta \alpha}.   
\end{equation*}

For every interval $I\subset \mathbb Z$,  let $\mathcal{K}_I:=\otimes_{i\in I} \mathbb C^d$ and let us define finite dimensional Hilbert spaces $\mathcal{H}_I$ as
\begin{equation*}
    \mathcal{H}_I = \mathcal{K}_I \otimes \Mat_D,
\end{equation*}
with scalar product 
\begin{equation}\label{eq:interval_scalar_product}
    \langle v \otimes m | w \otimes n \rangle_I := \langle v|w \rangle \cdot \mathrm{tr} (m^\dagger \rho  n).   
\end{equation}
Elements of this space are depicted as 
\begin{equation*}
    \begin{tikzpicture}[xscale=0.4,yscale=0.3]
        \draw (-1,0) --(4,0);
        \foreach \x in {0,1,3}{
          \draw (\x,0)--++(0,1);
        }
        \node at (2,0.6) {$\dots$};
        \draw[fusion,blue] (0,0)--(3,0);
    \end{tikzpicture} \in \mathcal{H}_I.
\end{equation*}
Let $J$ be the interval obtained by extending $I$ by one point to the right. Let us define the map $    \phi_{J\leftarrow I} : \mathcal{H}_I \to \mathcal{H}_J $ as the linear extension of 
\begin{equation*}
\phi_{J\leftarrow I}\left(v \otimes m\right) = \sum_{\alpha \beta } v \otimes a_{\alpha\beta} \otimes m e_{\alpha \beta}.
\end{equation*}
Graphically,
\begin{equation*}
    \phi_{J\leftarrow I} : 
    \begin{tikzpicture}[xscale=0.4,yscale=0.3]
        \draw (-1,0) --(4,0);
        \foreach \x in {0,1,3}{
          \draw (\x,0)--++(0,1);
        }
        \node at (2,0.6) {$\dots$};
        \draw[fusion,blue] (0,0)--(3,0);
    \end{tikzpicture} \mapsto
    \begin{tikzpicture}[xscale=0.4,yscale=0.3]
        \draw (-1,0) --(5,0);
        \foreach \x in {0,1,3,4}{
          \draw (\x,0)--++(0,1);
        }
        \node at (2,0.6) {$\dots$};
        \node[mps] at (4,0) {};
        \draw[fusion,blue] (0,0)--(3,0);
    \end{tikzpicture} \ .
\end{equation*}
This map respects the scalar product of any two vectors as:
\begin{align*}
   &\left\langle \phi_{J\leftarrow I}\left(v\otimes m\right) \middle | \phi_{J\leftarrow I}\left(w\otimes n \right) \right\rangle_J = \sum_{\alpha\beta \gamma \delta} \langle v| w\rangle \cdot \langle a_{\alpha \beta}| a_{\gamma\delta}\rangle \cdot \mathrm{tr}(e_{\beta \alpha} m^\dagger \rho ne_{\gamma\delta}) = \\
   &\sum_{\alpha\beta \gamma } \langle v| w\rangle \cdot \langle a_{\alpha \beta}| a_{\gamma\beta}\rangle \cdot \mathrm{tr}(  m^\dagger \rho ne_{\gamma\alpha}) = 
   \sum_{\alpha\gamma } \langle v| w\rangle \cdot \delta_{\alpha\gamma} \cdot \mathrm{tr}(  m^\dagger \rho ne_{\gamma\alpha}) =  \\
   & \langle v| w\rangle \cdot \mathrm{tr}(  m^\dagger \rho n) = \langle v\otimes m| w\otimes n\rangle_I.
\end{align*}
The same calculation in graphical language reads as
\begin{equation*}
    \begin{tikzpicture}[xscale=0.4,yscale=0.4,baseline=1mm]
        \draw (-1,0)  rectangle (5,1);
        \foreach \x in {0,1,3,4}{
          \draw (\x,0)--++(0,1);
        }
        \node at (2,0.5) {$\dots$};
        \node[mps] at (4,0) {};
        \node[cmps] at (4,1) {};
        \node[mps,label=left:$\rho$] at (-1,0.5) {};
        \draw[fusion,blue] (0,0)--(3,0);
        \draw[cfusion,blue] (0,1)--(3,1);
    \end{tikzpicture} \  = 
    \begin{tikzpicture}[xscale=0.4,yscale=0.4,baseline=1mm]
        \draw (-1,0)  rectangle (4,1);
        \foreach \x in {0,1,3}{
          \draw (\x,0)--++(0,1);
        }
        \node at (2,0.5) {$\dots$};
        \node[mps,label=left:$\rho$] at (-1,0.5) {};
        \draw[fusion,blue] (0,0)--(3,0);
        \draw[cfusion,blue] (0,1)--(3,1);
    \end{tikzpicture} \  ,
\end{equation*}
where in the equality we have used Eq.~\eqref{eq:mps_fixpoint}.

Similarly, if  $J$ is the interval obtained by extending $I$ by one point to the left, then let us define the map $\phi_{J\leftarrow I} : \mathcal{H}_I \to \mathcal{H}_J $ as
\begin{equation*}
\phi_{J\leftarrow I}\left(v\otimes m \right) = \sum_{\alpha \beta } a_{\alpha \beta} \otimes v   \otimes e_{\alpha \beta} m.
\end{equation*}
Graphically,
\begin{equation*}
    \phi_{J\leftarrow I} : 
    \begin{tikzpicture}[xscale=0.4,yscale=0.3]
        \draw (-1,0) --(4,0);
        \foreach \x in {0,1,3}{
          \draw (\x,0)--++(0,1);
        }
        \node at (2,0.6) {$\dots$};
        \draw[fusion,blue] (0,0)--(3,0);
    \end{tikzpicture} \mapsto
    \begin{tikzpicture}[xscale=-0.4,yscale=0.3]
        \draw (-1,0) --(5,0);
        \foreach \x in {0,1,3,4}{
          \draw (\x,0)--++(0,1);
        }
        \node at (2,0.6) {$\dots$};
        \node[mps] at (4,0) {};
        \draw[fusion,blue] (0,0)--(3,0);
    \end{tikzpicture} \ .
\end{equation*}
This map also preserves the scalar product:
\begin{align*}
   &\left\langle \phi_{J\leftarrow I}\left(v\otimes m\right) \middle | \phi_{J\leftarrow I}\left(w\otimes n \right) \right\rangle_J = \sum_{\alpha\beta \gamma \delta} \langle a_{\alpha \beta}| a_{\gamma\delta}\rangle \cdot \langle v| w\rangle \cdot  \mathrm{tr}(m^\dagger e_{\beta \alpha}  \rho e_{\gamma\delta} n) = \\
   &\langle v| w\rangle  \sum_{\alpha\beta \gamma\delta }  \rho_{\alpha\gamma}\langle a_{\alpha \beta}| a_{\gamma\delta}\rangle \cdot \mathrm{tr}(  m^\dagger e_{\beta \delta} n) = 
   \langle v| w\rangle \sum_{\beta \delta }  \rho_{\beta\delta} \cdot \mathrm{tr}(  m^\dagger e_{\beta\delta} n) =  \\
   & \langle v| w\rangle \cdot \mathrm{tr}(  m^\dagger \rho n) = \langle v\otimes m| w\otimes n\rangle_I.
\end{align*}
The same calculation in graphical language reads as
\begin{equation*}
    \begin{tikzpicture}[xscale=-0.4,yscale=0.4,baseline=1mm]
        \draw (-1,0)  rectangle (5,1);
        \foreach \x in {0,1,3,4}{
          \draw (\x,0)--++(0,1);
        }
        \node at (2,0.5) {$\dots$};
        \node[mps] at (4,0) {};
        \node[cmps] at (4,1) {};
        \node[mps,label=left:$\rho$] at (5,0.5) {};
        \draw[fusion,blue] (0,0)--(3,0);
        \draw[cfusion,blue] (0,1)--(3,1);
    \end{tikzpicture} \  = 
    \begin{tikzpicture}[xscale=-0.4,yscale=0.4,baseline=1mm]
        \draw (-1,0)  rectangle (4,1);
        \foreach \x in {0,1,3}{
          \draw (\x,0)--++(0,1);
        }
        \node at (2,0.5) {$\dots$};
        \node[mps,label=left:$\rho$] at (4,0.5) {};
        \draw[fusion,blue] (0,0)--(3,0);
        \draw[cfusion,blue] (0,1)--(3,1);
    \end{tikzpicture} \  ,
\end{equation*}
where in the equality we have used Eq.~\eqref{eq:mps_fixpoint}.

Concatenating these maps, we obtain a unique map $\phi_{J\leftarrow I}$ for any two intervals $I$ and $J$ such that $I\subset J$. Graphically,
\begin{equation*}
    \phi_{J\leftarrow I} : 
    \begin{tikzpicture}[xscale=0.5,yscale=0.3]
        \draw (-1,0) --(4,0);
        \foreach \x in {0,1,3}{
          \draw (\x,0)--++(0,1);
        }
        \node at (2,0.6) {$\dots$};
        \draw[fusion,blue] (0,0)--(3,0);
        \draw[decorate, decoration={brace}] (-0.1,1.1) -- (3.1,1.1) node[midway,above] {$I$};
    \end{tikzpicture} \mapsto
    \begin{tikzpicture}[xscale=0.4,yscale=0.3]
        \draw (-4,0) --(7,0);
        \foreach \x in {-3,-1,0,1,3,4,6}{
          \draw (\x,0)--++(0,1);
        }
        \node at (2,0.6) {$\dots$};
        \node[fill=white,inner sep=0.3pt] at (-2,0) {$\dots$};
        \node[fill=white,inner sep=0.3pt] at (5,0) {$\dots$};
        \node[mps] at (-1,0) {};
        \node[mps] at (-3,0) {};
        \node[mps] at (4,0) {};
        \node[mps] at (6,0) {};
        \draw[fusion,blue] (0,0)--(3,0);
        \draw[decorate, decoration={brace}] (-3.1,1.1) -- (6.1,1.1) node[midway,above] {$J$};
    \end{tikzpicture} \ .
\end{equation*}
The maps $\phi_{J\leftarrow I}$ are isometries and they satisfy the relation 
\begin{equation*}
    \phi_{K \leftarrow I} = \phi_{K\leftarrow J} \circ \phi_{ J \leftarrow I} \quad \forall I\subset J\subset K.
\end{equation*}
Using these maps, we can construct the direct limit 
\begin{equation*}
    \mathcal{H}^{loc} = \lim_{I\to \mathbb Z} \mathcal{H}_I.
\end{equation*}
As $\phi_{J\leftarrow I}$ are an isometries, we can extend the scalar product $\langle . | . \rangle_I$ to this vector space. We can then complete $\mathcal{H}^{loc}$ w.r.t.\ this scalar product obtaining a Hilbert space $\mathcal{H}$.
Let $I\subset \bbZ$ be a finite interval, $\phi_{\infty\leftarrow I} : \caH_I\to\caH$ be the embedding of $\mathcal{H}_I$ into $\mathcal{H}$  and $O\in \caA_I$.
For any $J\supset I$ finite interval we set
\begin{align}\label{pio}
    \pi^{J}\lmk O\rmk:=
    O\otimes\unit_{\lmk \bbC^{d}\rmk^{J\setminus I}}\otimes \unit_{\Mat_D},
\end{align}
and we define $\pi(O)$ as
\begin{align}
\pi(O)\phi_{\infty\leftarrow J}\lmk \xi\rmk:=\phi_{\infty\leftarrow J} \lmk\pi^{J}(O) \xi\rmk
\end{align} 
for any $\xi\in \caH_{\rm J}$.
This defines a well-defined bounded operator on $\caH_{\rm loc}$
whose norm satisfies $\lV \pi(O)\rV\le \lV O\rV$.
Because $\caH_{\rm loc}$ is dense in $\caH$, it extends uniquely to a bounded operator
on $\caH$, which we denote by the same symbol.

Let $\Omega_I$ denote the MPS on the interval $I = [n,n+1,\dots , m]$, i.e.\ 
\begin{equation*}
    \Omega_I = \sum_{\alpha \dots \omega} a_{\alpha\beta}^{(n)} a_{\beta\gamma}^{(n+1)}\dots a_{\zeta \omega}^{(m)} \otimes e_{\alpha\omega}. 
\end{equation*}
Graphically,
\begin{equation*}
    \Omega_I = 
    \begin{tikzpicture}
      \def\d{0.6};
      \draw[blue] (-0.5,0) -- (3*\d+0.5,0);
      \foreach \x in {0,1,3}{
        \draw (\x*\d,0) -- (\x*\d,0.4);
        \node[mps] (t) at (\x*\d,0) {};
      }
      \node[fill=white] at (\d*2,0) {$\dots$};
      \draw[decorate, decoration={brace}] (-0.1,0.5) -- (3*\d+0.1,0.5) node[midway,above] {$I$};
    \end{tikzpicture} \ .
\end{equation*}
These vectors obviously satisfy $\phi_{J\leftarrow I}(\Omega_I) = \Omega_J$. Their limit $\Omega=\lim_{I\to \mathbb Z} \Omega_I$ thus exist in $\mathcal{H}$. It is the infinite MPS. It is cyclic for the representation $\pi$: on any (large enough) finite interval $I$, through injectivity of $A$, $\mathrm{Span} \{(O\otimes \mathbb 1)\Omega_I | O\in \mathcal{A}_I\} = \mathcal{H}_I$, and thus $\pi\lmk \mathcal{A}_{loc}\rmk \Omega \subset \pi\lmk \mathcal{A}\rmk\Omega$ is dense in $\mathcal{H}$. Finally notice that, by construction, for any finite interval $I$ and $O\in \mathcal{A}_I$, $\langle \Omega_I | O \Omega_I\rangle = \omega(O)$, and thus $\langle \Omega | O \Omega\rangle = \omega(O)$ holds as well. 

Therefore the triple $(\mathcal{H},\pi,\Omega)$ is a GNS triple for the state $\omega$ defined by the MPS.

Let us note that the construction also works for half-infinite MPS. In this case we can consider only intervals such that one of their endpoint is $\{0\}$.

\subsection{The invariants of the phases coincide}

In this section we give an explicit representation of the unitaries $u_{x,g}: \mathcal{H}_{xg} \to \mathcal{H}_{x}$ defined in Eq.~\eqref{udef} for states and automorphisms that are given by MPS and MPO. We will show then that the index derived from the MPS theory \cite{mposym} coincides with the index derived in this paper. 

Let $G$ be a finite group and $X$ be a finite $G$-set. 
Let $A_x$ be normal MPS tensors for all $x\in X$ and $B_g$ be injective MPO tensors for all $g\in G$ such that Eq.~\eqref{eq:MPU_mult} and Eq.~\eqref{eq:MPS_sym} holds. As the MPS defined by the tensors $A_x$ are normal, they each define a split state $\omega_x$ on $\mathcal{A}$ through Eq.~\eqref{eq:mps_to_state} and thus also a state $\omega_{x,R}$ on $\mathcal{A}_R$ through Eq.~\eqref{eq:mps_to_right_state}.

Let us consider the GNS triple $(\mathcal{H}_x, \pi_x, \Omega_x)$ of each $\omega_{x,R}$. An explicit representation is given as in the previous section; we denote the finite dimensional Hilbert spaces corresponding to the interval $[0,n]$ by $\mathcal{H}_x^{(n)}$ and the corresponding injections by $\phi^{(m, n)}_{x}:\mathcal{H}_x^{(n)}\to \mathcal{H}_x^{(m)}$ for each $x\in X$ and $n,m\in\mathbb Z$ such that $m>n$.
We denote by $\pi_x^{(m)}$
the representation of $\caA_{[0,m]}$ on $\caH_{x}^{(m)}$ constructed as in (\ref{pio}).
Likewise we use notation
$\phi_x^{(m,n)}$, $\phi_x^{(\infty,n)}$ etc.

In the following we give an explicit representation of $u_{x,g}$ defined in Eq.~\eqref{udef}. Through this explicit representation we connect the index defined in Eq.~\eqref{eq:mps_index} to the index defined in Eq.~\eqref{eq:mps_index_2}.

Before giving the explicit representation, notice that for all n,
\begin{equation*}
    \left\langle \psi_n(x) \middle | U_n^*(g) \psi_n(xg) \right \rangle = 1,
\end{equation*}
or graphically,
    
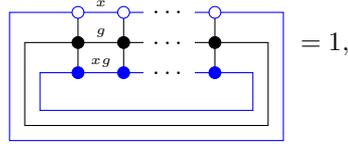
\begin{equation*}
        \begin{tikzpicture}[baseline = 3mm]
          \def\d{0.6};
          \draw[blue] (-0.5,0) rectangle (3*\d+0.5,-0.5);
          \draw (-0.7,0.4) rectangle (3*\d+0.7,-0.7);
          \draw[blue] (-0.9,0.8) rectangle (3*\d+0.9,-0.9);
          \node[irrep] at (\d/2,0.8) {$x$};
          \node[irrep] at (\d/2,0.4) {$g$};
          \node[irrep] at (\d/2,0) {$xg$};
          \foreach \x in {0,1,3}{
            \draw (\x*\d,0) -- (\x*\d,0.8);
            \node[mps] (t) at (\x*\d,0) {};
            \node[mpo] (t) at (\x*\d,0.4) {};
            \node[cmps] (t) at (\x*\d,0.8) {};
          }
          \node[fill=white] at (\d*2,0) {$\dots$};
          \node[fill=white] at (\d*2,0.4) {$\dots$};
          \node[fill=white] at (\d*2,0.8) {$\dots$};
        \end{tikzpicture} \ = 1,
    \end{equation*}
    where the blue dots with $xg$ drawn on the horizontal line correspond to the MPS tensors describing $\ket{\psi_n(xg)}$, the black dots with $g$ written on the horizontal line to the MPO tensors describing the MPU $U_n(g)$ and the white dots with $x$ written on the horizontal line to the MPS tensors describing $\bra{\psi_n(x)}$. Here and in the following the empty dots denote the dagger of the corresponding full dots.
    
    As $1\in\mathbb{C}$ is an injective MPS (with bond dimension 1 and physical dimension 1), we can apply Lemma \ref{lem:main_mps_tool} to conclude that there exists $V_0(x,g)$ and $W_0(x,g)$ such that \eqref{eq:reduction} holds. We can actually construct them explicitly, in two different ways. First, we can consider
    \begin{equation*}
        V_0(x,g) = 
        \begin{tikzpicture}[xscale=-0.5,yscale=-0.4]
          \draw[blue] (0,-0.5)--(1,-0.5) node[midway,irrep] {$\color{black}xg$};
          \draw[blue] (0,-1) -- (-1,-1) node[midway,irrep] {$\color{black}x$};
          \draw (0,0) -- (-1,0) node[midway,irrep] {$g$};
          \draw[blue] (1,-0.5) -- (1,1) -- (-1,1);
          \node[irrep] at (-0.5,1) {$xg$};
          \draw[caction] (0,-1)--(0,0);
          \node[mps,label=left:$\rho^{(xg)}$] at (1,0.25) {};
        \end{tikzpicture} \quad \text{and} \quad 
        W_0(x,g) = 
        \begin{tikzpicture}[xscale=0.5,yscale=-0.4]
          \draw[blue] (0,-0.5)--(1,-0.5) node[midway,irrep] {$\color{black}xg$};
          \draw[blue] (0,-1) -- (-1,-1) node[midway,irrep] {$\color{black}x$};
          \draw (0,0) -- (-1,0) node[midway,irrep] {$g$};
          \draw[blue] (1,-0.5) -- (1,1) -- (-1,1);
          \node[irrep] at (-0.5,1) {$xg$};
          \draw[caction] (0,-1)--(0,0);
        \end{tikzpicture} \ .
    \end{equation*}
    The second way to construct such operators is a bit more involved. First, note that the MPO tensors describing $U_n(g^{-1})$ and $U_n(g)^*$ are related to each other with a gauge transformation as described in Eq.~\eqref{eq:conjugate_gauge}, and thus 
\begin{equation*}
        \hat V_0(x,g) = 
        \begin{tikzpicture}[xscale=-0.5,yscale=0.4]
          \draw[blue] (0,-0.5)--(1,-0.5) node[midway,irrep] {$\color{black}x$};
          \draw[blue] (0,-1) -- (-1,-1) node[midway,irrep] {$\color{black}xg$};
          \draw (0,0) -- (-1,0) node[midway,irrep] {$g$};
          \draw[blue] (1,-0.5) -- (1,1) -- (-1,1);
          \node[irrep] at (-0.5,1) {$x$};
          \draw[halfc action] (0,-1)--(0,0);
          \node[mps,label=left:$\rho^{(x)}$] at (1,0.25) {};
        \end{tikzpicture} \quad \text{and} \quad 
        \hat W_0(x,g) = 
        \begin{tikzpicture}[xscale=0.5,yscale=0.4]
          \draw[blue] (0,-0.5)--(1,-0.5) node[midway,irrep] {$\color{black}x$};
          \draw[blue] (0,-1) -- (-1,-1) node[midway,irrep] {$\color{black}xg$};
          \draw (0,0) -- (-1,0) node[midway,irrep] {$g$};
          \draw[blue] (1,-0.5) -- (1,1) -- (-1,1);
          \node[irrep] at (-0.5,1) {$x$};
          \draw[halfc action] (0,-1)--(0,0);
        \end{tikzpicture} \ 
    \end{equation*}
    also satisfy Eq.~\eqref{eq:reduction}, where the light blue tensors were introduced in Eq.~\eqref{eq:lightblue_action}. Using now Lemma~\ref{lem:fusion_unique}, we conclude that there is a non-zero $c_4(x,g)\in \mathbb C$ and $N\in\mathbb{N}$ such that for all $n>N$
    \begin{equation}\label{eq:sandwich_switch}
        c_4(x,g) \cdot
        \begin{tikzpicture}[baseline = 3mm]
          \def\d{0.6};
          \draw[blue] (-0.5,0) -- (4*\d,0);
          \draw (-0.5,0.4) -- (4*\d,0.4);
          \draw[blue] (-0.5,0.8) -- (4*\d+0.5,0.8)--++(0,-0.6)--++(-0.5,0);
          \foreach \x in {\d/2,3.5*\d}{
              \node[irrep] at (\x,0.8) {$x$};
              \node[irrep] at (\x,0.4) {$g$};
              \node[irrep] at (\x,0) {$xg$};          
          }
          \foreach \x in {0,1,3}{
            \draw (\x*\d,0) -- (\x*\d,0.8);
            \node[mps] (t) at (\x*\d,0) {};
            \node[mpo] (t) at (\x*\d,0.4) {};
            \node[cmps] (t) at (\x*\d,0.8) {};
          }
          \node[fill=white] at (\d*2,0) {$\dots$};
          \node[fill=white] at (\d*2,0.4) {$\dots$};
          \node[fill=white] at (\d*2,0.8) {$\dots$};
          \draw[halfc action] (4*\d,0) -- (4*\d,0.4);
          \draw[decorate, decoration={brace,mirror}] (-0.1,-0.2) -- (3*\d+0.1,-0.2) node[midway,below] {$n$};
        \end{tikzpicture} \ =           
        \begin{tikzpicture}[baseline = 3mm]
          \def\d{0.6};
          \draw[blue] (-0.5,0) -- (4*\d+0.5,0)--++(0,0.6)--++(-0.5,0);
          \draw (-0.5,0.4) -- (4*\d,0.4);
          \draw[blue] (-0.5,0.8) -- (4*\d,0.8);
          \foreach \x in {\d/2,3.5*\d}{
              \node[irrep] at (\x,0.8) {$x$};
              \node[irrep] at (\x,0.4) {$g$};
              \node[irrep] at (\x,0) {$xg$};          
          }
          \foreach \x in {0,1,3}{
            \draw (\x*\d,0) -- (\x*\d,0.8);
            \node[mps] (t) at (\x*\d,0) {};
            \node[mpo] (t) at (\x*\d,0.4) {};
            \node[cmps] (t) at (\x*\d,0.8) {};
          }
          \node[fill=white] at (\d*2,0) {$\dots$};
          \node[fill=white] at (\d*2,0.4) {$\dots$};
          \node[fill=white] at (\d*2,0.8) {$\dots$};
          \draw[caction] (4*\d,0.4) -- (4*\d,0.8);
          \draw[decorate, decoration={brace,mirror}] (-0.1,-0.2) -- (3*\d+0.1,-0.2) node[midway,below] {$n$};
        \end{tikzpicture} \ .
    \end{equation}

Let $m,n\in \mathbb Z$ such that $m-n$ is large enough and let us define the operators $u_{x,g}^{(m, n)}:\mathcal{H}_{xg}^{(n)}\to \mathcal{H}_{x}^{(m)}$ and $w_{x,g}^{(m,n)}: \mathcal{H}_{x}^{(n)}\to \mathcal{H}_{xg}^{(m)}$ as
\begin{align}
    w_{x,g}^{(m, n)}: \ 
    \begin{tikzpicture}[xscale=0.6]
        \draw[blue] (0,0) -- (2.7,0);
        \node[irrep] at (2.5,0) {$x$};
        \foreach \x in {0,2}{
            \draw (\x,0) -- (\x,0.4);
        }
        \foreach \x in {(1,0.25)}{
            \node[fill=white] at \x {$\dots$};
        }
        \draw[fusion] (0,0)--(2,0);
        \draw[decorate, decoration=brace] (-0.1,0.5) -- (2.1,0.5) node[midway,above] {$n$};
    \end{tikzpicture} &\mapsto \frac{1}{\lambda} \cdot 
    \begin{tikzpicture}[xscale=0.6]
        \draw[blue] (0,0) -- (6,0);
        \draw (0,0.4) -- (6,0.4);
        \foreach \x in {0,2}{
            \draw (\x,0) -- (\x,0.8);
            \node[cmpo] at (\x,0.4) {};
        }
        \draw[fusion] (0,0)--(2,0);
        \foreach \x in {3,5}{
            \draw (\x,0) -- (\x,0.8);
            \node[mps] at (\x,0) {};        
            \node[cmpo] at (\x,0.4) {};
        }
        \draw[blue] (6,0.2)--(6.7,0.2);
        \node[irrep] at (6.5,0.2) {$xg$};
        \node[irrep] at (5.5,0.4) {$g$};
        \node[irrep] at (5.5,0) {$x$};
        \draw[action] (6,0)--(6,0.4);
        \foreach \x in {(1,0.4), (4,0.4), (4,0)}{
            \node[fill=white] at \x {$\dots$};
        }
        \draw[decorate, decoration=brace] (-0.1,0.9) -- (2.1,0.9) node[midway,above] {$n$};
        \draw[decorate, decoration=brace] (2.9,0.9) -- (5.1,0.9) node[midway,above] {$m-n$};
    \end{tikzpicture}  \ , \\  \label{eq:u finite def}   
    u_{x,g}^{(m,n)} : \ 
    \begin{tikzpicture}[xscale=0.6]
        \draw[blue] (0,0) -- (2.7,0);
        \node[irrep] at (2.5,0) {$xg$};
        \foreach \x in {0,2}{
            \draw (\x,0) -- (\x,0.4);
        }
        \foreach \x in {(1,0.25)}{
            \node[fill=white] at \x {$\dots$};
        }
        \draw[fusion] (0,0)--(2,0);
        \draw[decorate, decoration=brace] (-0.1,0.5) -- (2.1,0.5) node[midway,above] {$n$};
    \end{tikzpicture} &\mapsto \frac{c_4(x,g)}{\lambda} \cdot 
    \begin{tikzpicture}[xscale=0.6,baseline=1mm]
        \draw[blue] (0,0) -- (6,0);
        \draw (0,0.4) -- (6,0.4);
        \foreach \x in {0,2}{
            \draw (\x,0) -- (\x,0.8);
            \node[mpo] at (\x,0.4) {};
        }
        \draw[fusion] (0,0)--(2,0);
        \foreach \x in {3,5}{
            \draw (\x,0) -- (\x,0.8);
            \node[mps] at (\x,0) {};        
            \node[mpo] at (\x,0.4) {};
        }
        \draw[blue] (6,0.2)--(6.7,0.2);
        \draw[halfc action] (6,0)--(6,0.4);
        \node[irrep] at (6.5,0.2) {$x$};
        \node[irrep] at (5.5,0) {$xg$};
        \node[irrep] at (5.5,0.4) {$g$};
        \foreach \x in {(1,0.4), (4,0.4), (4,0)}{
            \node[fill=white] at \x {$\dots$};
        }
        \draw[decorate, decoration=brace] (-0.1,0.9) -- (2.1,0.9) node[midway,above] {$n$};
        \draw[decorate, decoration=brace] (2.9,0.9) -- (5.1,0.9) node[midway,above] {$m-n$};
    \end{tikzpicture} \ ,
\end{align}
where $\lambda\in \mathbb{R}^+$ is a constant that we fix later and the light blue tensor is the one defined in Eq.~\eqref{eq:lightblue_action}. 

Let us note here that using \eqref{eq:action_zipper} together with \eqref{eq:conjugate_gauge}, we obtain that for large enough $k$,
\begin{equation}\label{eq:trivial_zipper_gray}
  \begin{aligned}
  \begin{tikzpicture}[baseline=1mm]
    \def\d{0.6};
    \draw[blue] (-0.6,0) -- (2*\d+0.6,0);
    \draw (-0.6,0.4) -- (2*\d+0.6,0.4);
    \draw[blue] (2*\d+0.6,0.2) --++ (0.5,0);
    \foreach \x in {0,2}{
      \draw (\x*\d,0) -- (\x*\d,0.8);
      \node[mps] (t) at (\x*\d,0) {};
      \node[mpo] (t) at (\x*\d,0.4) {};
    }
    \draw[halfc action] (2*\d+0.6,0)--++(0,0.4);
    \node[fill=white] at (\d*1,0) {$\dots$};
    \node[fill=white] at (\d*1,0.4) {$\dots$};
    \node[irrep] at (-0.3,0) {$xg$};
    \node[irrep] at (-0.3,0.4) {$g$};
    \node[irrep] at (2*\d+0.3,0) {$xg$};
    \node[irrep] at (2*\d+0.3,0.4) {$g$};
    \node[irrep] at (2*\d+0.85,0.2) {$x$};
    \draw[decorate, decoration=brace] (-0.1,0.9) -- (2*\d+0.1,0.9) node[midway,above] {$k$};
  \end{tikzpicture} \  =  \
  \begin{tikzpicture}[baseline=1mm]
    \def\d{0.6};
    \draw[blue] (-0.6,0) -- (2*\d+0.6,0);
    \draw (-0.6,0.4) -- (2*\d+0.6,0.4);
    \draw[blue] (2*\d+0.6,0.2) --++ (\d+0.5,0);
    \draw (4*\d,0.2)--++(0,0.6);
    \node[mps] at (4*\d,0.2) {};
    \foreach \x in {0,2}{
      \draw (\x*\d,0) -- (\x*\d,0.8);
      \node[mps] (t) at (\x*\d,0) {};
      \node[mpo] (t) at (\x*\d,0.4) {};
    }
    \draw[halfc action] (2*\d+0.6,0)--++(0,0.4);
    \node[fill=white] at (\d*1,0) {$\dots$};
    \node[fill=white] at (\d*1,0.4) {$\dots$};
    \node[irrep] at (-0.3,0) {$xg$};
    \node[irrep] at (-0.3,0.4) {$g$};
    \node[irrep] at (2*\d+0.3,0) {$xg$};
    \node[irrep] at (2*\d+0.3,0.4) {$g$};
    \node[irrep] at (2*\d+0.85,0.2) {$x$};
    \draw[decorate, decoration=brace] (-0.1,0.9) -- (2*\d+0.1,0.9) node[midway,above] {$k-1$};
  \end{tikzpicture} , \\
  \begin{tikzpicture}[baseline=1mm,xscale=-1]
    \def\d{0.6};
    \draw[blue] (-0.6,0) -- (2*\d+0.6,0);
    \draw (-0.6,0.4) -- (2*\d+0.6,0.4);
    \draw[blue] (2*\d+0.6,0.2) --++ (0.5,0);
    \foreach \x in {0,2}{
      \draw (\x*\d,0) -- (\x*\d,0.8);
      \node[mps] (t) at (\x*\d,0) {};
      \node[mpo] (t) at (\x*\d,0.4) {};
    }
    \draw[halfc action] (2*\d+0.6,0)--++(0,0.4);
    \node[fill=white] at (\d*1,0) {$\dots$};
    \node[fill=white] at (\d*1,0.4) {$\dots$};
    \node[irrep] at (-0.3,0) {$xg$};
    \node[irrep] at (-0.3,0.4) {$g$};
    \node[irrep] at (2*\d+0.3,0) {$xg$};
    \node[irrep] at (2*\d+0.3,0.4) {$g$};
    \node[irrep] at (2*\d+0.85,0.2) {$x$};
    \draw[decorate, decoration={brace,mirror}] (-0.1,0.9) -- (2*\d+0.1,0.9) node[midway,above] {$k$};
  \end{tikzpicture} \  =  \
  \begin{tikzpicture}[baseline=1mm,xscale=-1]
    \def\d{0.6};
    \draw[blue] (-0.6,0) -- (2*\d+0.6,0);
    \draw (-0.6,0.4) -- (2*\d+0.6,0.4);
    \draw[blue] (2*\d+0.6,0.2) --++ (\d+0.5,0);
    \draw (4*\d,0.2)--++(0,0.6);
    \node[mps] at (4*\d,0.2) {};
    \foreach \x in {0,2}{
      \draw (\x*\d,0) -- (\x*\d,0.8);
      \node[mps] (t) at (\x*\d,0) {};
      \node[mpo] (t) at (\x*\d,0.4) {};
    }
    \draw[halfc action] (2*\d+0.6,0)--++(0,0.4);
    \node[fill=white] at (\d*1,0) {$\dots$};
    \node[fill=white] at (\d*1,0.4) {$\dots$};
    \node[irrep] at (-0.3,0) {$xg$};
    \node[irrep] at (-0.3,0.4) {$g$};
    \node[irrep] at (2*\d+0.3,0) {$xg$};
    \node[irrep] at (2*\d+0.3,0.4) {$g$};
    \node[irrep] at (2*\d+0.85,0.2) {$x$};
    \draw[decorate, decoration={brace,mirror}] (-0.1,0.9) -- (2*\d+0.1,0.9) node[midway,above] {$k-1$};
  \end{tikzpicture}  \ .
  \end{aligned}
\end{equation}

We will prove that the series $u_{x,g}^{(m,n)}$ extends to a unitary $u_{x,g}:\mathcal{H}_{xg}\to\mathcal{H}_x$ defined in Eq.~\eqref{udef}, and the series $w_{x,g}^{(m,n)}$ extends to the unitary $\left(u_{x,g}\right)^*:\mathcal{H}_{x}\to\mathcal{H}_{xg}$. For that, we need to first understand the properties of $u_{x,g}^{(m,n)}$. First we show that the operators $u_{x,g}^{(m,n)}$ and $w_{x,g}^{(m,n)}$ are compatible with the growing procedure of the Hilbert spaces $\mathcal{H}_{x}^{(n)}$.
\begin{lemma}\label{lem:u_w_zip}
    The operators $u_{x,g}^{(m,n)}$ and $w_{x,g}^{(m,n)}$ satisfy the compatibility equations
    \begin{align*}
        u_{x,g}^{(k,m)} \circ \phi^{(m,n)}_{xg} &= \phi^{(k,m)}_{x} \circ u_{x,g}^{(m,n)}, \\
        w_{x,g}^{(k,m)} \circ \phi^{(m,n)}_{x} &= \phi_{xg}^{(k,m)} \circ w_{x,g}^{(m,n)} ,
    \end{align*}
    for all $k>m>n$ such that $m-n$ is large enough.
\end{lemma}

\begin{proof}
    The map $\phi_{x}^{(m,n)}$ acts on a vector $v\in\mathcal{H}_{x}^{(n)}$ by adding $m-n$ MPS tensors to it. Then $w_{x,g}^{(k,m)}$ adds $k-m$ more MPS tensors and multiplies with $k$ MPO tensors. That is, 
    \begin{equation*}
    w_{x,g}^{(k,m)} \circ \phi^{(m,n)}_{x} : 
    \begin{tikzpicture}[xscale=0.6]
        \draw[blue] (0,0) -- (2.7,0);
        \node[irrep] at (2.5,0) {$x$};
        \foreach \x in {0,2}{
            \draw (\x,0) -- (\x,0.4);
        }
        \foreach \x in {(1,0.25)}{
            \node[fill=white] at \x {$\dots$};
        }
        \draw[fusion] (0,0)--(2,0);
        \draw[decorate, decoration=brace] (-0.1,0.5) -- (2.1,0.5) node[midway,above] {$n$};
    \end{tikzpicture} \mapsto \frac{1}{\lambda} \cdot 
    \begin{tikzpicture}[xscale=0.6,baseline=1mm]
        \draw[blue] (0,0) -- (6,0);
        \draw (0,0.4) -- (6,0.4);
        \foreach \x in {0,2}{
            \draw (\x,0) -- (\x,0.8);
            \node[cmpo] at (\x,0.4) {};
        }
        \draw[fusion] (0,0)--(2,0);
        \foreach \x in {3,5}{
            \draw (\x,0) -- (\x,0.8);
            \node[mps] at (\x,0) {};        
            \node[cmpo] at (\x,0.4) {};
        }
        \draw[blue] (6,0.2)--(6.7,0.2);
        \node[irrep] at (6.5,0.2) {$xg$};
        \node[irrep] at (5.5,0.4) {$g$};
        \node[irrep] at (5.5,0) {$x$};
        \draw[fusion] (6,0)--(6,0.4);
        \foreach \x in {(1,0.4), (4,0.4), (4,0)}{
            \node[fill=white] at \x {$\dots$};
        }
        \draw[decorate, decoration=brace] (-0.1,0.9) -- (5.1,0.9) node[midway,above] {$k$};
    \end{tikzpicture}  \ . 
    \end{equation*}
    Given that $n-m$ is large enough, we can use Eq.~\eqref{eq:action_zipper} $k-m$ times to arrive at the equivalent expression
    \begin{equation*}
    w_{x,g}^{(k,m)} \circ \phi^{(m,n)}_{x} : 
    \begin{tikzpicture}[xscale=0.6]
        \draw[blue] (0,0) -- (2.7,0);
        \node[irrep] at (2.5,0) {$x$};
        \foreach \x in {0,2}{
            \draw (\x,0) -- (\x,0.4);
        }
        \foreach \x in {(1,0.25)}{
            \node[fill=white] at \x {$\dots$};
        }
        \draw[fusion] (0,0)--(2,0);
        \draw[decorate, decoration=brace] (-0.1,0.5) -- (2.1,0.5) node[midway,above] {$n$};
    \end{tikzpicture} \mapsto \frac{1}{\lambda} \cdot 
    \begin{tikzpicture}[xscale=0.6,baseline=1mm]
        \draw[blue] (0,0) -- (6,0);
        \draw (0,0.4) -- (6,0.4);
        \foreach \x in {0,2}{
            \draw (\x,0) -- (\x,0.8);
            \node[cmpo] at (\x,0.4) {};
        }
        \draw[fusion] (0,0)--(2,0);
        \foreach \x in {3,5}{
            \draw (\x,0) -- (\x,0.8);
            \node[mps] at (\x,0) {};        
            \node[cmpo] at (\x,0.4) {};
        }
        \draw[blue] (6,0.2)--(9.2,0.2);
        \foreach \x in {6.7,8.7}{
            \draw (\x,0.2) -- (\x,0.6);
            \node[mps] at (\x,0.2) {};  
        }
        \node[irrep] at (6.4,0.2) {$xg$};
        \node[irrep] at (5.5,0.4) {$g$};
        \node[irrep] at (5.5,0) {$x$};
        \draw[fusion] (6,0)--(6,0.4);
        \foreach \x in {(1,0.4), (4,0.4), (4,0),(7.7,0.2)}{
            \node[fill=white] at \x {$\dots$};
        }
        \draw[decorate, decoration=brace] (-0.1,0.9) -- (2.1,0.9) node[midway,above] {$n$};
        \draw[decorate, decoration=brace] (3-0.1,0.9) -- (5.1,0.9) node[midway,above] {$m-n$};
        \draw[decorate, decoration=brace] (6.6,0.7) -- (8.8,0.7) node[midway,above] {$k-m$};
    \end{tikzpicture}  \ . 
    \end{equation*}
    Finally notice that this action is the same as the action of $\phi^{(k,m)}_{xg} \circ w_{x,g}^{(m,n)}$. The other equation is proven analogously, using Eq.~\eqref{eq:trivial_zipper_gray}.
\end{proof}

We then show that -- after to growing the Hilbert spaces to matching sizes-- $u_{x,g}^{(m,n)}$ and $w_{x,g}^{(m,n)}$ are the adjoint of each other.
\begin{lemma}\label{lem:u_w_dagger}
    Let $\chi\in\mathcal{H}_{xg}^{(n)}$ and $\psi\in\mathcal{H}_x^{(n)}$ be arbitrary. The operators $u_{x,g}^{(m,n)}$ and $w_{x,g}^{(m,n)}$ satisfy 
    \begin{equation*}
        \left\langle \phi_{x}^{(m,n)}(\psi)\middle| u_{x,g}^{(m,n)} (\chi)\right\rangle^{(m)}_{x} = 
        \left\langle w_{x,g}^{(m,n)}(\psi)\middle| \phi_{xg}^{(m,n)} (\chi)\right\rangle^{(m)}_{xg},
    \end{equation*}
    where the scalar product $\langle . |. \rangle^{(m)}_{x}$ on $\mathcal{H}_x^{(m)}$ for the interval $[0,m]$ is defined in Eq.~\eqref{eq:interval_scalar_product}.
\end{lemma}

\begin{proof}
        
   Let us represent the l.h.s.\ using the graphical representation:
   \begin{equation*}
       \left\langle \phi_{x}^{(m,n)}(\psi)\middle| u_{x,g}^{(m,n)} (\chi)\right\rangle_{x}^{(m)} = \frac{c_4(x,g)}{\lambda} \cdot
        \begin{tikzpicture}[xscale=0.6,baseline=3mm]
            \draw[blue] (0,0) -- (6,0);
            \draw[blue] (0,0.8) -- (7,0.8);
            \draw (0,0.4) -- (6,0.4);
            \foreach \x in {0,2}{
                \draw (\x,0) -- (\x,0.8);
                \node[mpo] at (\x,0.4) {};
            }
            \foreach \x in {3,5}{
                \draw (\x,0) -- (\x,0.8);
                \node[mps] at (\x,0) {};        
                \node[cmps] at (\x,0.8) {};        
                \node[mpo] at (\x,0.4) {};
            }
            \draw[fusion] (0,0)--(2,0);
            \draw[cfusion] (0,0.8)--(2,0.8);
            \draw[blue] (6,0.2)--(7,0.2)--(7,0.8);
            \draw[halfc action] (6,0)--(6,0.4);
            \node[irrep] at (6.5,0.2) {$x$};
            \node[irrep] at (5.5,0) {$xg$};
            \node[irrep] at (5.5,0.4) {$g$};
            \node[irrep] at (5.5,0.8) {$x$};
            \foreach \x in {(1,0.4), (4,0.4), (4,0)}{
                \node[fill=white] at \x {$\dots$};
            }
        \draw[decorate, decoration=brace] (-0.1,1) -- (2.1,1) node[midway,above] {$n$};
        \draw[decorate, decoration=brace] (3-0.1,1) -- (5.1,1) node[midway,above] {$m-n$};
        \end{tikzpicture} \ ,
   \end{equation*}
   where the long empty tensor represents $\bra{\psi}$. Using Eq.~\eqref{eq:sandwich_switch}, the r.h.s.\ can be changed to
    \begin{equation*}
       \left\langle \phi_{x}^{(m,n)}(\psi)\middle| u_{x,g}^{(m,n)} (\chi)\right\rangle_{x}^{(m)} = \frac{1}{\lambda} \cdot 
       \begin{tikzpicture}[xscale=0.6,baseline=3mm]
            \draw[blue] (0,0) -- (7,0);
            \draw[blue] (0,0.8) -- (6,0.8);
            \draw (0,0.4) -- (6,0.4);
            \foreach \x in {0,2}{
                \draw (\x,0) -- (\x,0.8);
                \node[mpo] at (\x,0.4) {};
            }
            \foreach \x in {3,5}{
                \draw (\x,0) -- (\x,0.8);
                \node[mps] at (\x,0) {};        
                \node[cmps] at (\x,0.8) {};        
                \node[mpo] at (\x,0.4) {};
            }
            \draw[fusion] (0,0)--(2,0);
            \draw[cfusion] (0,0.8)--(2,0.8);
            \draw[blue] (7,0)--(7,0.6)--(6,0.6);
            \draw[caction] (6,0.8)--(6,0.4);
            \node[irrep] at (6.5,0.6) {$xg$};
            \node[irrep] at (5.5,0) {$xg$};
            \node[irrep] at (5.5,0.8) {$x$};
            \node[irrep] at (5.5,0.4) {$g$};
            \foreach \x in {(1,0.4), (4,0.4), (4,0)}{
                \node[fill=white] at \x {$\dots$};
            }
            \draw[decorate, decoration=brace] (-0.1,1) -- (2.1,1) node[midway,above] {$n$};
            \draw[decorate, decoration=brace] (3-0.1,1) -- (5.1,1) node[midway,above] {$m-n$};
        \end{tikzpicture} \ .  
    \end{equation*}
    This expression is exactly $ \left\langle w_{x}^{(m,n)}(\psi)\middle| \phi_{xg}^{(m,n)} (\chi)\right\rangle_{xg}^{(m)}$ finishing the proof. 
\end{proof}

Let us also show that the two operators are the inverses of each other -- again, up to growing the system size.

\begin{lemma}\label{lem:u_w_inverse}
    There is a choice of $\lambda=:\lambda(x,g)>0$ such that the operators $u_{x,g}^{(m,n)}$ and $w_{x,g}^{(m,n)}$ satisfy 
    \begin{align*}
        u_{x,g}^{(k,m)} \circ w_{x,g}^{(m,n)} = \phi_x^{(k,n)}, \\
        w_{x,g}^{(k,m)} \circ u_{x,g}^{(m,n)} = \phi_{xg}^{(k,n)},
    \end{align*}
    for all $k>m>n$ such that $m-n$ and $k-m$ is large enough.
\end{lemma}
\begin{proof}
Let us first calculate  $u_{x,g}^{(k,m)} \circ w_{x,g}^{(m,n)}(\psi)$ for some $\psi\in \mathcal{H}_x^{(n)}$: 
\begin{equation*}
    u_{x,g}^{(k,m)} \circ w_{x,g}^{(m,n)} (\psi) = 
    \frac{c_4(x,g)}{\lambda^2} \cdot
    \begin{tikzpicture}[xscale=0.6,baseline=3mm]
        \draw[blue] (0,0) -- (6,0);
        \draw (0,0.8) -- (10,0.8);
        \draw[blue] (6,0.2)--(10,0.2);
        \draw[blue] (10,0.5) -- (10.7,0.5);
        \draw (0,0.4) -- (6,0.4);
        \foreach \x in {0,2}{
            \draw (\x,0) -- (\x,1.2);
            \node[cmpo] at (\x,0.4) {};
            \node[mpo] at (\x,0.8) {};
        }
        \foreach \x in {3,5}{
            \draw (\x,0) -- (\x,1.2);
            \node[mps] at (\x,0) {};        
            \node[mpo] at (\x,0.8) {};        
            \node[cmpo] at (\x,0.4) {};
        }
        \foreach \x in {7,9}{
            \draw (\x,0.2) -- (\x,1.2);
            \node[mps] at (\x,0.2) {};        
            \node[mpo] at (\x,0.8) {};
        }
        \draw[fusion] (0,0)--(2,0);
        \draw[action] (6,0)--(6,0.4);
        \draw[halfc action] (10,0.2)--(10,0.8);
        \node[irrep] at (6.5,0.2) {$xg$};
        \node[irrep] at (9.5,0.2) {$xg$};
        \node[irrep] at (5.5,0) {$x$};
        \node[irrep] at (5.5,0.4) {$g$};
        \node[irrep] at (5.5,0.8) {$g$};
        \node[irrep] at (9.5,0.8) {$g$};
        \foreach \x in {(1,0.4), (1,0.8), (4,0.4), (4,0),(4,0.8), (8,0.2), (8,0.8)}{
            \node[fill=white] at \x {$\dots$};
        }
        \draw[decorate, decoration=brace] (-0.1,1.3) -- (2.1,1.3) node[midway,above] {$n$};
        \draw[decorate, decoration=brace] (2.9,1.3) -- (5.1,1.3) node[midway,above] {$m-n$};
        \draw[decorate, decoration=brace] (6.9,1.3) -- (9.1,1.3) node[midway,above] {$k-m$};
    \end{tikzpicture}      
\end{equation*}
Using now Eq.~\eqref{eq:action_zipper} $k-m$ times, we obtain
\begin{equation*}
    u_{x,g}^{(k,m)} \circ w_{x,g}^{(m,n)} (\psi) = 
    \frac{c_4(x,g)}{\lambda^2} \cdot
    \begin{tikzpicture}[xscale=0.6,baseline=3mm]
        \draw[blue] (0,0) -- (6,0);
        \draw (0,0.8) -- (7,0.8);
        \draw[blue] (6,0.2)--(7,0.2);
        \draw[blue] (7,0.5) -- (7.7,0.5);
        \draw (0,0.4) -- (6,0.4);
        \foreach \x in {0,2}{
            \draw (\x,0) -- (\x,1.2);
            \node[cmpo] at (\x,0.4) {};
            \node[mpo] at (\x,0.8) {};
        }
        \foreach \x in {3,5}{
            \draw (\x,0) -- (\x,1.2);
            \node[mps] at (\x,0) {};        
            \node[mpo] at (\x,0.8) {};        
            \node[cmpo] at (\x,0.4) {};
        }
        \draw[fusion] (0,0)--(2,0);
        \draw[action] (6,0)--(6,0.4);
        \draw[halfc action] (7,0.2)--(7,0.8);
        \node[irrep] at (6.5,0.2) {$xg$};
        \node[irrep] at (5.5,0) {$x$};
        \node[irrep] at (5.5,0.4) {$g$};
        \foreach \x in {(1,0.4), (1,0.8), (4,0.4), (4,0),(4,0.8)}{
            \node[fill=white] at \x {$\dots$};
        }
        \draw[decorate, decoration=brace] (-0.1,1.3) -- (2.1,1.3) node[midway,above] {$n$};
        \draw[decorate, decoration=brace] (2.9,1.3) -- (5.1,1.3) node[midway,above] {$k-n$};
    \end{tikzpicture}      \ .
\end{equation*}
Using now the tensors defined in Eq.~\eqref{eq:gray_fusion} together with Eq.~\eqref{eq:mps_index_2}, we obtain
\begin{equation*}
    u_{x,g}^{(k,m)} \circ w_{x,g}^{(m,n)}(\psi) = 
    \frac{c_4(x,g)}{\lambda^2} \cdot \hat\sigma_x(g,g^{-1}) \cdot 
    \begin{tikzpicture}[xscale=0.6,baseline=1mm]
        \draw[blue] (0,0) -- (7,0);
        \draw (0,0.8) -- (6,0.8);
        \draw[densely dotted] (6,0.6)--(7,0.6);
        \draw[blue] (7,0.3) -- (7.7,0.3);
        \draw (0,0.4) -- (6,0.4);
        \foreach \x in {0,2}{
            \draw (\x,0) -- (\x,1.2);
            \node[cmpo] at (\x,0.4) {};
            \node[mpo] at (\x,0.8) {};
        }
        \foreach \x in {3,5}{
            \draw (\x,0) -- (\x,1.2);
            \node[mps] at (\x,0) {};        
            \node[mpo] at (\x,0.8) {};        
            \node[cmpo] at (\x,0.4) {};
        }
        \draw[fusion] (0,0)--(2,0);
        \draw[fusion,gray] (6,0.4)--(6,0.8);
        \draw[action] (7,0)--(7,0.6);
        \node[irrep] at (6.5,0.6) {$e$};
        \node[irrep] at (5.5,0) {$x$};
        \node[irrep] at (7.5,0.3) {$x$};
        \node[irrep] at (5.5,0.4) {$g$};
        \node[irrep] at (5.5,0.8) {$g$};
        \foreach \x in {(1,0.4), (1,0.8), (4,0.4), (4,0),(4,0.8)}{
            \node[fill=white] at \x {$\dots$};
        }
        \draw[decorate, decoration=brace] (-0.1,1.3) -- (2.1,1.3) node[midway,above] {$n$};
        \draw[decorate, decoration=brace] (2.9,1.3) -- (5.1,1.3) node[midway,above] {$k-n$};
    \end{tikzpicture}      \ .
\end{equation*}
Here the black action tensor is the identity. Using now Eq.~\eqref{eq:zipper_g_g_inv}, we obtain that

\begin{equation*}
     u_{x,g}^{(k,m)} \circ w_{x,g}^{(m,n)}(\psi) = \frac{c_4(x,g)}{\lambda^2} \cdot \hat\sigma_x(g,g^{-1}) \cdot \phi_{x}^{k,n} (\psi)  .
\end{equation*}
Let us show that this constant is positive. For that, let us choose $\psi$ such that $\|\psi\|=1$, then 
\begin{align*}
     &\frac{c_4(x,g)}{\lambda^2} \cdot \hat\sigma_x(g,g^{-1}) = \left\langle \phi_{x}^{(k,n)}(\psi) \middle| u_{x,g}^{(k,m)} \circ w_{x,g}^{(m,n)}(\psi) \right\rangle_{x}^{(k)}  =\\
     &\left\langle w_{x,g}^{(k,n)}(\psi) \middle| \phi_{x}^{(k,m)} \circ w_{x,g}^{(m,n)}(\psi) \right\rangle_{x}^{(k)} = \left\langle w_{x,g}^{(m,n)}(\psi) \middle| w_{x,g}^{(m,n)}(\psi) \right\rangle_{x}^{(m)} >0,
\end{align*}
where in the second equality we have used Lemma~\ref{lem:u_w_dagger} and in the last one Lemma~\ref{lem:u_w_zip}. This implies that there is a choice of $\lambda\in \mathbb{R}^+$ such that the expression is $1$.
\end{proof}

\begin{lemma}\label{u_w_beta}
    Let $O\in\mathcal{A}_{[0,n]}$ and $k>m>n$ with $k-m$, $m-n$ large enough. The operators $u_{x,g}^{(m,n)}$ and $w_{x,g}^{(m,n)}$ satisfy
    \begin{equation*}
        u_{x,g}^{(k,m)} \circ \pi_{xg}^{(m)}(O) \circ w_{x,g}^{(m,n)} = \phi_x^{(k,m)}\circ \pi_x^{(m)}\big(\beta_{g}^{(m,n)}(O)\big)\circ\phi_x^{(m,n)} .
    \end{equation*}
\end{lemma}
\begin{proof} 
    Let us apply $u_{x,g}^{(k,m)} \circ \pi_{xg}^{(m)}(O) \circ w_{x,g}^{(m,n)}$ on a vector $\chi\in\mathcal{H}_{x}^{(n)}$. This can be represented using the graphical language as
\begin{align*}
    u_{x,g}^{(k,m)} \circ \pi_{xg}^{(m)}(O) \circ w_{x,g}^{(m,n)} (\chi)= \frac{c_4(x,g)}{\lambda^2} \cdot  
    \begin{tikzpicture}[xscale=0.6,baseline=5mm]
        \draw[blue] (0,0) -- (5.5,0);
        \draw[blue] (5.5,0.2) -- (9,0.2);
        \draw (0,0.4) -- (5.5,0.4);
        \draw (0,1.2) -- (9,1.2);
        \foreach \x in {0,2}{
            \draw (\x,0) -- (\x,1.6);
            \node[cmpo] at (\x,0.4) {};
            \node[mpo] at (\x,1.2) {};
        }
        \draw[fusion,blue] (0,0)--(2,0);
        \draw[fusion] (0,0.8)--(2,0.8);
        \foreach \x in {3,5}{
            \draw (\x,0) -- (\x,1.6);
            \node[mps] at (\x,0) {};        
            \node[cmpo] at (\x,0.4) {};
            \node[mpo] at (\x,1.2) {};
        }
        \draw[action] (5.5,0)--(5.5,0.4);
        \foreach \x in {6,8}{
            \draw (\x,0.2) -- (\x,1.6);
            \node[mps] at (\x,0.2) {};        
            \node[mpo] at (\x,1.2) {};
        }
        \draw[blue] (9,0.7) -- (10,0.7);
        \draw[halfc action] (9,0.2)--(9,1.2);
        \foreach \x in {(1,0.4), (1,1.2), (4,0.4), (4,0),(4,1.2),(7,0.2),(7,1.2)}{
            \node[fill=white] at \x {$\dots$};
        }
        \node[irrep] at (2.5,0) {$x$};
        \node[irrep] at (2.5,0.4) {$g$};
        \node[irrep] at (2.5,1.2) {$g$};
        \node[irrep] at (8.5,0.2) {$xg$};
        \node[irrep] at (8.5,1.2) {$g$};
        \node[irrep] at (9.5,0.7) {$x$};
        \draw[decorate, decoration=brace] (-0.1,1.7) -- (2.1,1.7) node[midway,above] {$n$};
        \draw[decorate, decoration=brace] (2.9,1.7) -- (5.1,1.7) node[midway,above] {$m-n$};
        \draw[decorate, decoration=brace] (5.9,1.7) -- (8.1,1.7) node[midway,above] {$k-m$};
    \end{tikzpicture} \ .
\end{align*}
Using %Eq.~\eqref{eq:trivial_zipper} 
Eq.~\eqref{eq:trivial_zipper_gray}
$k-m$ times, we obtain
\begin{align*}
    u_{x,g}^{(k,m)} \circ \pi_{xg}^{(m)}(O) \circ w_{x,g}^{(m,n)} (\chi)=  \frac{c_4(x,g)}{\lambda^2} \cdot    
    \begin{tikzpicture}[xscale=0.6,baseline=5mm]
        \draw[blue] (0,0) -- (6,0);
        \draw[blue] (6,0.2) -- (7,0.2);
        \draw (0,0.4) -- (6,0.4);
        \draw (0,1.2) -- (7,1.2);
        \foreach \x in {0,2}{
            \draw (\x,0) -- (\x,1.6);
            \node[cmpo] at (\x,0.4) {};
            \node[mpo] at (\x,1.2) {};
        }
        \draw[fusion,blue] (0,0)--(2,0);
        \draw[fusion] (0,0.8)--(2,0.8);
        \foreach \x in {3,5}{
            \draw (\x,0) -- (\x,1.6);
            \node[mps] at (\x,0) {};        
            \node[cmpo] at (\x,0.4) {};
            \node[mpo] at (\x,1.2) {};
        }
        \draw[action] (6,0)--(6,0.4);
        \draw[blue] (7,0.7) -- (8,0.7);
        \draw[halfc action] (7,0.2)--(7,1.2);
        \foreach \x in {(1,0.4), (1,1.2), (4,0.4), (4,0),(4,1.2)}{
            \node[fill=white] at \x {$\dots$};
        }
        \node[irrep] at (2.5,0) {$x$};
        \node[irrep] at (2.5,0.4) {$g$};
        \node[irrep] at (2.5,1.2) {$g$};
        \node[irrep] at (6.5,0.2) {$xg$};
        \node[irrep] at (6.5,1.2) {$g$};
        \node[irrep] at (7.5,0.7) {$x$};
        \draw[decorate, decoration=brace] (-0.1,1.7) -- (2.1,1.7) node[midway,above] {$n$};
        \draw[decorate, decoration=brace] (2.9,1.7) -- (5.1,1.7) node[midway,above] {$k-n$};
    \end{tikzpicture} \ . 
\end{align*}
Using now \eqref{eq:mps_index_2} together with the fact that in Lemma~\ref{lem:u_w_inverse} we have chosen $\lambda$ such that $c_4(x,g)\cdot \hat\sigma_x(g,g^{-1})=\lambda^2$, we obtain
\begin{equation*}
    u_{x,g}^{(k,m)} \circ \pi_{xg}^{(m)}(O) \circ w_{x,g}^{(m,n)} (\chi)=         \begin{tikzpicture}[xscale=0.6,baseline=5mm]
        \draw[blue] (0,0) -- (6,0);
        \draw (0,0.4) -- (5.5,0.4);
        \draw (0,1.2) -- (5.5,1.2);
        \foreach \x in {0,2}{
            \draw (\x,0) -- (\x,1.6);
            \node[cmpo] at (\x,0.4) {};
            \node[mpo] at (\x,1.2) {};
        }
        \draw[fusion,blue] (0,0)--(2,0);
        \draw[fusion] (0,0.8)--(2,0.8);
        \foreach \x in {3,5}{
            \draw (\x,0) -- (\x,1.6);
            \node[mps] at (\x,0) {};        
            \node[cmpo] at (\x,0.4) {};
            \node[mpo] at (\x,1.2) {};
        }
        \draw[fusion,gray] (5.5,0.4)--(5.5,1.2);
        \foreach \x in {(1,0.4), (1,1.2), (4,0.4), (4,0),(4,1.2)}{
            \node[fill=white] at \x {$\dots$};
        }
        \node[irrep] at (2.5,0) {$x$};
        \node[irrep] at (2.5,0.4) {$g$};
        \node[irrep] at (2.5,1.2) {$g$};
        \draw[decorate, decoration=brace] (-0.1,1.7) -- (2.1,1.7) node[midway,above] {$n$};
        \draw[decorate, decoration=brace] (2.9,1.7) -- (5.1,1.7) node[midway,above] {$k-n$};
    \end{tikzpicture}  \ .    
\end{equation*}
Finally using \eqref{eq:gray_g_g_zipper}, we obtain 
\begin{equation*}
    u_{x,g}^{(k,m)} \circ \pi_{xg}^{(m)}(O) \circ w_{x,g}^{(m,n)} (\chi)=             \begin{tikzpicture}[xscale=0.6,baseline=5mm]
        \draw[blue] (0,0) -- (9,0);
        \draw (0,0.4) -- (5.5,0.4);
        \draw (0,1.2) -- (5.5,1.2);
        \foreach \x in {0,2}{
            \draw (\x,0) -- (\x,1.6);
            \node[cmpo] at (\x,0.4) {};
            \node[mpo] at (\x,1.2) {};
        }
        \draw[fusion,blue] (0,0)--(2,0);
        \draw[fusion] (0,0.8)--(2,0.8);
        \foreach \x in {3,5}{
            \draw (\x,0) -- (\x,1.6);
            \node[mps] at (\x,0) {};        
            \node[cmpo] at (\x,0.4) {};
            \node[mpo] at (\x,1.2) {};
        }
        \draw[fusion,gray] (5.5,0.4)--(5.5,1.2);
        \foreach \x in {6,8}{
            \draw (\x,0) -- (\x,0.4);
            \node[mps] at (\x,0) {}; 
        }
        \foreach \x in {(1,0.4), (1,1.2), (4,0.4), (4,0),(4,1.2),(7,0)}{
            \node[fill=white] at \x {$\dots$};
        }
        \node[irrep] at (2.5,0) {$x$};
        \node[irrep] at (2.5,0.4) {$g$};
        \node[irrep] at (2.5,1.2) {$g$};
        \draw[decorate, decoration=brace] (-0.1,1.7) -- (2.1,1.7) node[midway,above] {$n$};
        \draw[decorate, decoration=brace] (2.9,1.7) -- (5.1,1.7) node[midway,above] {$m-n$};
        \draw[decorate, decoration=brace] (5.9,0.5) -- (8.1,0.5) node[midway,above] {$k-m$};
    \end{tikzpicture}  \ .    
\end{equation*}
The r.h.s.\ is $\phi_x^{(k,m)}\circ \pi_x^{(m)}\big(\beta_{g}^{(m,n)}(O)\big)\circ\phi_x^{(m,n)} (\chi)$ finishing the proof.
\end{proof}

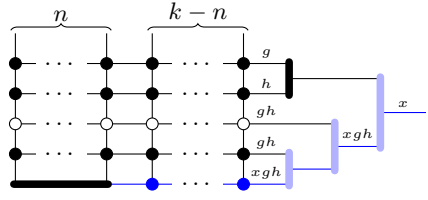
\begin{lemma}\label{lem:u_w_index}Let $m-n$, $k-m$ be large enough.
    The operators $u_{x,g}^{(k,m)}$,$ u_{xg,h}^{(m,n)} $
    $u_{x,gh}^{(k,n)}$satisfy
    \begin{equation}\label{eq:sigmaprime}
      u_{x,g}^{(k,m)} u_{xg,h}^{(m,n)} = \sigma'_x(g,h) \pi_x^{(k)}(v(g,h)) u_{x,gh}^{(k,n)} ,
    \end{equation}
    for some $\sigma'_x(g,h)$ such that  $[\sigma'_x(g,h)] = [\hat{\sigma}_x(g,h)]$, see Eq.\eqref{equivclassigmaMPS}. 
\end{lemma}

\begin{proof}
    Let $\chi\in\mathcal{H}_{xg}^{(n)}$. Then 
    \begin{equation*}
    u_{x,g}^{(k,m)} u_{xg,h}^{(m,n)}(\chi) = \frac{c_4(x,g)}{\lambda(x,g)} \cdot \frac{c_4(xg,h)}{\lambda(xg,h)} \cdot 
    \begin{tikzpicture}[xscale=0.6,baseline=3mm]
        \draw[blue] (0,0) -- (6,0);
        \draw (0,0.8) -- (10,0.8);
        \draw (0,0.4) -- (6,0.4);
        \foreach \x in {0,2}{
            \draw (\x,0) -- (\x,1.2);
            \node[mpo] at (\x,0.4) {};
            \node[mpo] at (\x,0.8) {};
        }
        \foreach \x in {3,5}{
            \draw (\x,0) -- (\x,1.2);
            \node[mps] at (\x,0) {};        
            \node[mpo] at (\x,0.8) {};        
            \node[mpo] at (\x,0.4) {};
        }
        \foreach \x in {7,9}{
            \draw (\x,0.2) -- (\x,1.2);
            \node[mps] at (\x,0.2) {};        
            \node[mpo] at (\x,0.8) {};        
        }
        \draw[fusion] (0,0)--(2,0);
        \draw[blue] (6,0.2)--(10,0.2)--(10,0.8);
        \draw[blue] (10,0.5)--++(1,0);
        \draw[halfc action] (6,0)--(6,0.4);
        \draw[halfc action] (10,0.2)--(10,0.8);
        \node[irrep] at (6.5,0.2) {$xg$};
        \node[irrep] at (5.5,0) {$xgh$};
        \node[irrep] at (10.5,0.5) {$x$};
        \node[irrep] at (5.5,0.4) {$h$};
        \node[irrep] at (5.5,0.8) {$g$};
        \node[irrep] at (9.5,0.2) {$xg$};
        \node[irrep] at (9.5,0.8) {$g$};
        \foreach \x in {(1,0.4), (1,0.8), (4,0.8), (4,0.4), (4,0), (8,0.2),(8,0.8)}{
            \node[fill=white] at \x {$\dots$};
        }
        \draw[decorate, decoration=brace] (-0.1,1.3) -- (2.1,1.3) node[midway,above] {$n$};
        \draw[decorate, decoration=brace] (2.9,1.3) -- (5.1,1.3) node[midway,above] {$m-n$};
        \draw[decorate, decoration=brace] (6.9,1.3) -- (9.1,1.3) node[midway,above] {$k-m$};
    \end{tikzpicture} 
    \end{equation*}
    Using Eq.~\eqref{eq:trivial_zipper_gray}, we obtain
    \begin{equation*}
    u_{x,g}^{(k,m)} u_{xg,h}^{(m,n)}(\chi) = \frac{c_4(x,g)}{\lambda(x,g)} \cdot \frac{c_4(xg,h)}{\lambda(xg,h)} \cdot
    \begin{tikzpicture}[xscale=0.6,baseline=3mm]
        \draw[blue] (0,0) -- (6,0);
        \draw[] (0,0.8) -- (7,0.8);
        \draw (0,0.4) -- (6,0.4);
        \foreach \x in {0,2}{
            \draw (\x,0) -- (\x,1.2);
            \node[mpo] at (\x,0.4) {};
            \node[mpo] at (\x,0.8) {};
        }
        \foreach \x in {3,5}{
            \draw (\x,0) -- (\x,1.2);
            \node[mps] at (\x,0) {};        
            \node[mpo] at (\x,0.8) {};        
            \node[mpo] at (\x,0.4) {};
        }
        \draw[fusion] (0,0)--(2,0);
        \draw[blue] (6,0.2)--(7,0.2)--(7,0.8);
        \draw[blue] (7,0.5)--++(1,0);
        \draw[halfc action] (6,0)--(6,0.4);
        \draw[halfc action] (7,0.2)--(7,0.8);
        \node[irrep] at (6.5,0.2) {$xg$};
        \node[irrep] at (5.5,0) {$xgh$};
        \node[irrep] at (7.5,0.5) {$x$};
        \node[irrep] at (5.5,0.4) {$h$};
        \node[irrep] at (5.5,0.8) {$g$};
        \foreach \x in {(1,0.4), (1,0.8), (4,0.8), (4,0.4), (4,0)}{
            \node[fill=white] at \x {$\dots$};
        }
        \draw[decorate, decoration=brace] (-0.1,1.3) -- (2.1,1.3) node[midway,above] {$n$};
        \draw[decorate, decoration=brace] (2.9,1.3) -- (5.1,1.3) node[midway,above] {$k-n$};
    \end{tikzpicture} 
    \end{equation*}    
    Let us use now Eq.~\eqref{eq:mps_index_2} to obtain
    \begin{equation*}
    u_{x,g}^{(k,m)} u_{xg,h}^{(m,n)}(\chi) = \frac{c_4(x,g)}{\lambda(x,g)} \cdot \frac{c_4(xg,h)}{\lambda(xg,h)} \cdot \hat\sigma_{x}(g, h) \cdot 
    \begin{tikzpicture}[xscale=0.6,baseline=3mm]
        \draw[blue] (0,0) -- (7,0);
        \draw (0,0.8) -- (6,0.8);
        \draw (0,0.4) -- (6,0.4);
        \foreach \x in {0,2}{
            \draw (\x,0) -- (\x,1.2);
            \node[mpo] at (\x,0.4) {};
            \node[mpo] at (\x,0.8) {};
        }
        \foreach \x in {3,5}{
            \draw (\x,0) -- (\x,1.2);
            \node[mps] at (\x,0) {};        
            \node[mpo] at (\x,0.8) {};        
            \node[mpo] at (\x,0.4) {};
        }
        \draw[fusion] (0,0)--(2,0);
        \draw[blue] (6,0.6)--(7,0.6);
        \draw[blue] (7,0.3)--++(1,0);
        \draw[fusion] (6,0.4)--(6,0.8);
        \draw[halfc action] (7,0)--(7,0.6);
        \node[irrep] at (6.5,0.6) {$gh$};
        \node[irrep] at (5.5,0) {$xgh$};
        \node[irrep] at (7.5,0.3) {$x$};
        \node[irrep] at (5.5,0.4) {$h$};
        \node[irrep] at (5.5,0.8) {$g$};
        \foreach \x in {(1,0.4), (1,0.8), (4,0.8), (4,0.4), (4,0)}{
            \node[fill=white] at \x {$\dots$};
        }
        \draw[decorate, decoration=brace] (-0.1,1.3) -- (2.1,1.3) node[midway,above] {$n$};
        \draw[decorate, decoration=brace] (2.9,1.3) -- (5.1,1.3) node[midway,above] {$k-n$};
    \end{tikzpicture} 
    \end{equation*}    
    Let us insert now the identity in the middle:
    \begin{align*}
    u_{x,g}^{(k,m)} u_{xg,h}^{(m,n)}(\chi) = &\frac{c_4(x,g)}{\lambda(x,g)} \cdot \frac{c_4(xg,h)}{\lambda(xg,h)} \cdot \hat\sigma_{x}(g, h) \cdot \frac{1}{\omega(gh,(gh)^{-1},gh)} \cdot \\
    &
    \begin{tikzpicture}[xscale=0.6,yscale=0.4,baseline=7mm]
        \draw[blue] (0,0) -- (7,0);
        \draw (0,4) -- (6,4);
        \draw (0,1) -- (6,1);
        \draw (0,2) -- (6,2);
        \draw (0,3) -- (6,3);
        \foreach \x in {0,2}{
            \draw (\x,0) -- (\x,5);
            \node[mpo] at (\x,1) {};
            \node[cmpo] at (\x,2) {};
            \node[mpo] at (\x,3) {};
            \node[mpo] at (\x,4) {};
        }
        \foreach \x in {3,5}{
            \draw (\x,0) -- (\x,5);
            \node[mps] at (\x,0) {};        
            \node[mpo] at (\x,1) {};        
            \node[cmpo] at (\x,2) {};
            \node[mpo] at (\x,3) {};
            \node[mpo] at (\x,4) {};
        }
        \draw[fusion] (0,0)--(2,0);
        \draw[blue] (7,0.75)--++(1,0);
        \draw[blue] (8,2.125)--++(1,0);
        \draw[densely dotted] (6,1.5)--++(1,0);
        \draw (6,3.5) --++(2,0);
        \draw[fusion] (6,1)--(6,2);
        \draw[fusion] (6,3)--(6,4);
        \draw[halfc action] (7,0)--(7,1.5);
        \draw[halfc action] (8,0.75)--(8,3.5);
        \node[irrep] at (7.5,0.75) {$xgh$};
        \node[irrep] at (5.5,0) {$xgh$};
        \node[irrep] at (8.5,2.125) {$x$};
        \node[irrep] at (5.5,1) {$gh$};
        \node[irrep] at (5.5,2) {$gh$};
        \node[irrep] at (5.5,3) {$h$};
        \node[irrep] at (5.5,4) {$g$};
        \foreach \x in {(1,1), (1,2), (1,3), (1,4),(4,0),(4,1),(4,2),(4,3),(4,4)}{
            \node[fill=white] at \x {$\dots$};
        }
        \draw[decorate, decoration=brace] (-0.1,5.1) -- (2.1,5.1) node[midway,above] {$n$};
        \draw[decorate, decoration=brace] (2.9,5.1) -- (5.1,5.1) node[midway,above] {$k-n$};
    \end{tikzpicture} 
    \end{align*}    
    Using now Eq.~\eqref{eq:mps_index_2}, we obtain
    \begin{align*}
    u_{x,g}^{(k,m)} u_{xg,h}^{(m,n)}(\chi) = &\frac{c_4(x,g)}{\lambda(x,g)} \cdot \frac{c_4(xg,h)}{\lambda(xg,h)} \cdot \frac{\hat\sigma_{x}(g, h)}{\hat\sigma_{xgh}((gh)^{-1},gh)} \cdot \frac{1}{\omega(gh,(gh)^{-1},gh)} \cdot \\
    &
    \begin{tikzpicture}[xscale=0.6,yscale=0.4,baseline=7mm]
        \draw[blue] (0,0) -- (6,0);
        \draw (0,4) -- (6,4);
        \draw (0,1) -- (6,1);
        \draw (0,2) -- (7,2);
        \draw (0,3) -- (6,3);
        \foreach \x in {0,2}{
            \draw (\x,0) -- (\x,5);
            \node[mpo] at (\x,1) {};
            \node[cmpo] at (\x,2) {};
            \node[mpo] at (\x,3) {};
            \node[mpo] at (\x,4) {};
        }
        \foreach \x in {3,5}{
            \draw (\x,0) -- (\x,5);
            \node[mps] at (\x,0) {};        
            \node[mpo] at (\x,1) {};        
            \node[cmpo] at (\x,2) {};
            \node[mpo] at (\x,3) {};
            \node[mpo] at (\x,4) {};
        }
        \draw[fusion] (0,0)--(2,0);
        \draw[blue] (7,1.25)--++(1,0);
        \draw[blue] (8,2.375)--++(1,0);
        \draw[blue] (6,0.5)--++(1,0);
        \draw (6,3.5) --++(2,0);
        \draw[halfc action] (6,0)--(6,1);
        \draw[fusion] (6,3)--(6,4);
        \draw[halfc action] (7,0.5)--(7,2);
        \draw[halfc action] (8,1.25)--(8,3.5);
        \node[irrep] at (7.5,1.25) {$xgh$};
        \node[irrep] at (5.5,0) {$xgh$};
        \node[irrep] at (8.5,2.375) {$x$};
        \node[irrep] at (5.5,1) {$gh$};
        \node[irrep] at (5.5,2) {$gh$};
        \node[irrep] at (5.5,3) {$h$};
        \node[irrep] at (5.5,4) {$g$};
        \foreach \x in {(1,1), (1,2), (1,3), (1,4),(4,0),(4,1),(4,2),(4,3),(4,4)}{
            \node[fill=white] at \x {$\dots$};
        }
        \draw[decorate, decoration=brace] (-0.1,5.1) -- (2.1,5.1) node[midway,above] {$n$};
        \draw[decorate, decoration=brace] (2.9,5.1) -- (5.1,5.1) node[midway,above] {$k-n$};
    \end{tikzpicture} 
    \end{align*}    
    Using again Eq.~\eqref{eq:mps_index_2}, we obtain
    \begin{align*}
    u_{x,g}^{(k,m)} u_{xg,h}^{(m,n)}(\chi) = &\frac{c_4(x,g)}{\lambda(x,g)} \cdot \frac{c_4(xg,h)}{\lambda(xg,h)} \cdot \frac{\hat\sigma_{x}(g, h)}{\hat\sigma_{xgh}((gh)^{-1},gh)} \cdot \frac{\hat\sigma_x(gh,(gh)^{-1})}{\omega(gh,(gh)^{-1},gh)} \cdot \\
    &
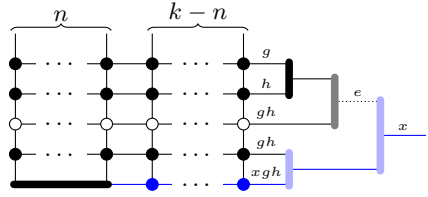
    \begin{tikzpicture}[xscale=0.6,yscale=0.4,baseline=7mm]
        \draw[blue] (0,0) -- (6,0);
        \draw (0,4) -- (6,4);
        \draw (0,1) -- (6,1);
        \draw (0,2) -- (7,2);
        \draw (0,3) -- (6,3);
        \foreach \x in {0,2}{
            \draw (\x,0) -- (\x,5);
            \node[mpo] at (\x,1) {};
            \node[cmpo] at (\x,2) {};
            \node[mpo] at (\x,3) {};
            \node[mpo] at (\x,4) {};
        }
        \foreach \x in {3,5}{
            \draw (\x,0) -- (\x,5);
            \node[mps] at (\x,0) {};        
            \node[mpo] at (\x,1) {};        
            \node[cmpo] at (\x,2) {};
            \node[mpo] at (\x,3) {};
            \node[mpo] at (\x,4) {};
        }
        \draw[fusion] (0,0)--(2,0);
        \draw[densely dotted] (7,2.75)--++(1,0);
        \draw[blue] (8,1.625)--++(1,0);
        \draw[blue] (6,0.5)--++(2,0);
        \draw (6,3.5) --++(1,0);
        \draw[halfc action] (6,0)--(6,1);
        \draw[fusion] (6,3)--(6,4);
        \draw[fusion, gray] (7,2)--(7,3.5);
        \draw[halfc action] (8,0.5)--(8,2.75);
        \node[irrep] at (7.5,2.75) {$e$};
        \node[irrep] at (5.5,0) {$xgh$};
        \node[irrep] at (8.5,1.625) {$x$};
        \node[irrep] at (5.5,1) {$gh$};
        \node[irrep] at (5.5,2) {$gh$};
        \node[irrep] at (5.5,3) {$h$};
        \node[irrep] at (5.5,4) {$g$};
        \foreach \x in {(1,1), (1,2), (1,3), (1,4),(4,0),(4,1),(4,2),(4,3),(4,4)}{
            \node[fill=white] at \x {$\dots$};
        }
        \draw[decorate, decoration=brace] (-0.1,5.1) -- (2.1,5.1) node[midway,above] {$n$};
        \draw[decorate, decoration=brace] (2.9,5.1) -- (5.1,5.1) node[midway,above] {$k-n$};
    \end{tikzpicture} 
    \end{align*}    
    Finally notice using the definition of $v(g,h)$ (see Eq.~\eqref{eq:v graphical}) and the definition of $u_{x,g}^{k,n}$ (see Eq.~\eqref{eq:u finite def})  that this is 
    \begin{align*}
            u_{x,g}^{(k,m)} u_{xg,h}^{(m,n)}(\chi) = &\frac{c_4(x,g)}{\lambda(x,g)} \cdot \frac{c_4(xg,h)}{\lambda(xg,h)} \cdot \frac{\hat\sigma_{x}(g, h)}{\hat\sigma_{xgh}((gh)^{-1},gh)} \cdot \frac{\hat\sigma_x(gh,(gh)^{-1})}{\omega(gh,(gh)^{-1},gh)} \cdot \frac{\lambda(x,gh)}{c_4(x,gh)} \\ & \cdot \pi_x^{(k)}(v(g,h)) \cdot u_{x,gh}^{(k,n)}(\chi).
    \end{align*}
    Note that part of this expression can be simplified using Eq.~\eqref{eq:sigmaconstraint}: we can write 
    $$ \frac{\hat\sigma_x(gh,(gh)^{-1})}{\omega(gh,(gh)^{-1},gh) \hat\sigma_{xgh}((gh)^{-1},gh)} = \frac{\hat\sigma_{x}(gh, e)}{\hat\sigma_{x}(e,gh)}= 1\ , $$
    and denoting $\alpha(x,g)=\frac{\lambda(x,g)}{c_4(x,g)}$ we obtain
    
$$ u_{x,g}^{(k,m)} u_{xg,h}^{(m,n)}(\chi) = \frac{\alpha(x,gh)} {\alpha(x,g)  \cdot\alpha(xg,h)} \hat\sigma_x(g,h) \pi_x^{(k)}(v(g,h)) \cdot u_{x,gh}^{(k,n)}(\chi)$$
We thus obtain that Eq.~\eqref{eq:sigmaprime} holds with
$$ \hat\sigma'_x(g,h)  = \frac{\alpha(x,gh)} {\alpha(x,g)  \cdot\alpha(xg,h)} \hat\sigma_x(g,h), $$
and that $\hat\sigma'_x(g,h)$ is in the same equivalence class (see Eq.\eqref{equivclassigmaMPS}) as $\hat\sigma_x(g,h)$.
\end{proof}

\begin{thm}
    The operators $u_{x,g}^{(m,n)}$ can be extended to unitary operators $u_{x,g}:\mathcal{H}_{xg}\to\mathcal{H}_x$  such that
    \begin{equation*}
        \mathrm{Ad}(u_{x,g})\pi_{xg} = \pi_x\beta_{g,R},
    \end{equation*}
    and such that
    \begin{equation*}
        u_{xg,h} u_{x,gh} = \sigma'_x(g,h) \pi_x(v(g,h)) u_{x,gh},
    \end{equation*}
    where $[\sigma'_x(g,h)] = [{\hat\sigma_x(g,h)}] $. 
\end{thm}
This, combined with subsection \ref{3eq} shows then that the index defined from the MPS formalism coincides with the index defined from the general formalism when applied to MPS and MPO. 

\begin{proof}
    Lemma \ref{lem:u_w_zip} allows one to define $u_{x,g}:\mathcal{H}_{xg}^{loc}\to \mathcal{H}_{x}^{loc}$ and $w_{x,g}:\mathcal{H}_{x}^{loc} \to \mathcal{H}_{xg}^{loc}$ as the inductive limit of $u_{x,g}^{(m,n)}$ and $w_{x,g}^{(m,n)}$, respectively. For example, given $\xi\in\mathcal{H}_{xg}^{loc}$, there is a representative $\xi_n\in \caH_{xg}^{(n)}$ such that $\xi=\phi_x^{(\infty,n)}{(\xi_n)}$ for some $n\in \mathbb N$. Given this representative, $u_{x,g}$ acts as $\xi\mapsto \phi_x^{(\infty,m)}\lmk u_{x,g}^{(m,n)} \xi_n\rmk$. 
    Lemma \ref{lem:u_w_zip} shows
    \begin{align}
        \phi_x^{(\infty, k)} u_{x,g}^{(k,m)}\lmk\phi_{xg}^{(m,n)} (\xi_n)\rmk
        =\phi_x^{(\infty, k)}\phi_x^{(k,m)}\lmk u_{x,g}^{(m,n)}\xi_n\rmk=\phi_x^{(\infty,m)}
        u_{x,g}^{(m,n)}\xi_n,
    \end{align}
     hence the resulting equivalence class is independent of the representative $\xi_n$ of $\xi$.
    
    Due to Lemma \ref{lem:u_w_inverse}, the operators $u_{x,g}$ and $w_{x,g}$ are invertible and $u_{x,g}^{-1} = w_{x,g}$: given a vector $\xi\in \mathcal{H}_{xg}^{loc}$ and a representative $\xi_n\in \mathcal{H}_{x,g}^{[0,n]}$, Lemma \ref{lem:u_w_inverse} shows that $w_{x,g} u_{x,g} \xi = \phi_{\infty\leftarrow [0,k]}\lmk w_{x,g}^{(k,m)} u_{x,g}^{(m,n)} \xi_n\rmk = \phi_{\infty\leftarrow [0,n]}(\xi_n) = \xi$, and similarly, $u_{x,g} w_{x,g} \xi = \xi$ for any $\xi\in\mathcal{H}_x$. Similarly, due to Lemma \ref{lem:u_w_dagger}, $u_{x,g}^{\dagger} = w_{x,g}$.  In particular, their norm is bounded, and thus they can be extended to the whole Hilbert space. The extension keeps all the above properties. 

    Lemma \ref{u_w_beta} shows that the obtained $u_{x,g}$ satisfies Eq.~\ref{udef}, and thus by uniqueness, it coincides (up to a phase) with $u_{x,g}$ obtained in Section \ref{sxgh}. Finally Lemma~\ref{lem:u_w_index} shows that the index obtained from the MPS formalism is in the same equivalence class as the equivalence class of the index obtained in this paper.
\end{proof}

\backmatter

\bmhead{Acknowledgments}

The authors thank the hospitality of the Institute of Mathematical Sciences  (ICMAT) in Madrid where part of this work has been done.
YO was supported by JSPS KAKENHI Grant Number 19K03534 and
22H01127 and also supported by JST CREST Grant Number JPMJCR19T2.
JGR and AM have been partially supported by the European Research Council (ERC) under the European Union’s Horizon 2020 research and innovation programme through the ERC-CoG SEQUAM (Grant Agreement No. 863476).

\begin{appendices}

\end{appendices}

\end{document}